\documentclass[12pt,a4paper]{book} 
\pdfoutput=1

\usepackage{geometry}         
\usepackage{amscd,amsmath,amssymb,amsfonts,xspace,mathrsfs,amsthm,dsfont,bbm}
\usepackage{bbm}
\usepackage{tocloft}
\usepackage[bbgreekl]{mathbbol}
\usepackage{qtree}
\usepackage{color}
\usepackage{pdfpages}
\usepackage[dvipsnames]{xcolor}
\let\normalcolor\relax
\usepackage{bbold}
\usepackage{latexsym}
\usepackage{graphicx}
\usepackage{euscript}
\usepackage{dsfont}
\usepackage{color}
\usepackage{longtable}
\usepackage[colorlinks=true, linkcolor=blue, citecolor=red]{hyperref}

 \newtheorem{theorem}{Theorem}
  \newtheorem{proposition}{Proposition}
  \newtheorem{lemma}{Lemma}
  \newtheorem{conj}{Conjecture}

\theoremstyle{definition}
 \newtheorem{definition}{Definition}[chapter]
 
\numberwithin{equation}{section}

\def\lfig#1#2#3#4#5{
\begin{figure}[t]
 \centerline{\includegraphics[width=#3]{#2}}
 \vspace{#5}
  \caption{#1 \label{#4}}
 \end{figure}
}

\def\bea{\begin{eqnarray}}
\def\eea{\end{eqnarray}}
\def\be{\begin{equation}}
\def\ee{\end{equation}}
\def\ba{\begin{align}}
\def\ea{\end{align}}
\def\bse{\begin{subequations}}
\def\ese{\end{subequations}}

\newcommand{\nn}{\nonumber}

\def\det{\,{\rm det}\, }
\def\diag{{\rm diag}}
\def\sign{{\rm sgn}}

\def\Ch{{\rm Ch}}

\def\Sym{\,{\rm Sym}\, }

\def\Im{\,{\rm Im}\,}

\def\Ad{{\rm Ad}}
\def\ad{{\rm ad}}
\def\Aut{{\rm Aut}}
\def\Re{\,{\rm Re}\,}

\def\({\left(}
\def\){\right)}
\def\[{\left[}
\def\]{\right]}
\def\<{\left\langle}
\def\>{\right\rangle}
\def\hf{{1\over 2}}

\newcommand{\p}{\partial}

\def\SL{SL(2,\IZ)}
\def\cMw{\cM^{!}}

\newcommand{\kahler}{{K\"ahler}\xspace}
\newcommand{\hk}{{hyperk\"ahler}\xspace}
\newcommand{\qk}{{quaternion-K\"ahler}\xspace}
\newcommand{\sch}{{Schr\"oder}\xspace}
\newcommand{\vig}{{Vign\'eras}\xspace}
\newcommand{\pleb}{{Pleba\'nski}\xspace}

\newcommand{\scR}{\mathscr{R}}

\newcommand{\scJ}{\mathscr{J}}
\newcommand{\scK}{\mathscr{K}}
\newcommand{\scL}{\mathscr{L}}

\newcommand{\vth}{\vartheta}

\def\vph{\varphi}

\newcommand{\eps}{\epsilon}

\renewcommand{\d}{\mathrm{d}}
\newcommand{\de}{\mathrm{d}}
\newcommand{\De}{\mathrm{D}}
\newcommand{\I}{\mathrm{i}}

\newcommand{\rmR}{\mathrm{R}}

\newcommand{\esf}{\mathsf{e}}

\def\vph{\varphi}

\newcommand{\Asf}{{\sf A}}

\newcommand{\asf}{{\sf a}}

\def\euF{\EuScript{F}} 
\def\beuF{\bar \euF}

\newcommand{\cL}{\mathcal{L}}
\newcommand{\cD}{\mathcal{D}}

\newcommand{\cV}{\mathcal{V}}
\newcommand{\cC}{\mathcal{C}}
\newcommand{\cS}{\mathcal{S}}
\newcommand{\cG}{\mathcal{G}}
\newcommand{\cK}{\mathcal{K}}
\newcommand{\cM}{\mathcal{M}}

\newcommand{\cN}{\mathcal{N}}
\newcommand{\cE}{\mathcal{E}}
\newcommand{\cX}{\mathcal{X}}

\newcommand{\cR}{\mathcal{R}}
\newcommand{\cT}{\mathcal{T}}
\newcommand{\cJ}{\mathcal{J}}
\newcommand{\cZ}{\mathcal{Z}}
\newcommand{\cI}{\mathcal{I}}

\newcommand{\cQ}{\mathcal{Q}}
\newcommand{\cH}{\mathcal{H}}
\newcommand{\cU}{\mathcal{U}}
\newcommand{\cA}{\mathcal{A}}
\newcommand{\cB}{\mathcal{B}}

\newcommand{\pbbm}{\mathbbm{p}}
\newcommand{\xbbm}{\mathbbm{x}}
\newcommand{\ybbm}{\mathbbm{y}}
\newcommand{\vbbm}{\mathbbm{v}}
\newcommand{\wbbm}{\mathbbm{w}}

\newcommand{\zbbm}{\mathbbm{z}}

\newcommand{\kbbm}{\mathbbm{k}}

\newcommand{\abbm}{\mathbbm{a}}
\newcommand{\bbbm}{\mathbbm{b}}

\def\bbLambda{\mathbb{\Lambda}}

\newcommand{\C}{{\mathbb C}}

\newcommand{\Z}{{\mathbb Z}}

\def\scR{\mathscr{R}}

\def\Ev{\mathscr{E}}
\def\Fv{\mathscr{F}}

\def\Zv{\mathscr{Z}}

\newcommand{\bfg}{{\bf g}}

\newcommand{\bfa}{{\boldsymbol a}}
\newcommand{\bfb}{{\boldsymbol b}}

\newcommand{\bfv}{{\boldsymbol v}}

\newcommand{\bfe}{{\boldsymbol e}}

\newcommand{\bfk}{{\boldsymbol k}}

\newcommand{\bfn}{{\boldsymbol n}}
\newcommand{\bfp}{{\boldsymbol p}}
\newcommand{\bfq}{{\boldsymbol q}}
\newcommand{\bfr}{{\boldsymbol r}}
\newcommand{\bfs}{{\boldsymbol s}}
\newcommand{\bft}{{\boldsymbol t}}
\newcommand{\bfz}{{\boldsymbol z}}
\newcommand{\bfx}{{\boldsymbol x}}
\newcommand{\bfy}{{\boldsymbol y}}

\newcommand{\bftet}{{\boldsymbol \theta}}
\newcommand{\bfmu}{{\boldsymbol \mu}}
\newcommand{\bfnu}{{\boldsymbol \nu}}

\newcommand{\bfLam}{{\boldsymbol \Lambda}}

\newcommand{\bfgam}{{\boldsymbol \gamma}}

\newcommand{\IT}{\mathds{T}}
\newcommand{\IR}{\mathds{R}}
\newcommand{\IC}{\mathds{C}}
\newcommand{\IZ}{\mathds{Z}}
\newcommand{\IQ}{\mathds{Q}}
\newcommand{\IN}{\mathds{N}}
\newcommand{\IH}{\mathds{H}}

\newcommand{\IP}{\mathds{P}}

\newcommand{\Tr}{\mbox{Tr}}

\newcommand{\sgn}{\mbox{sgn}}

\newcommand{\tzeta}{\tilde\zeta}

\newcommand{\talp}{\tilde\alpha}

\newcommand{\tlh}{\tilde h}

\newcommand{\tell}{\tilde\ell}
\newcommand{\tmu}{\tilde\mu}
\newcommand{\tnu}{\tilde \nu}

\newcommand{\txi}{\tilde\xi}

\def\cl0{\tilde c_0}

\def\zl{z^\Lambda}

\def\xil{\xi^\Lambda}
\def\txil{\txi_\Lambda}
\def\tzetal{\tzeta_\Lambda}
\def\zetal{\zeta^\Lambda}

\newcommand{\alpi}{\alpha^{\[i\]}}
\newcommand{\talpi}{\talp^{\[i\]}}
\newcommand{\txii}{\txi^{\[i\]}}
\newcommand{\xii}{\xi_{\[i\]}}
\newcommand{\alpj}{\alpha^{\[j\]}}
\newcommand{\txij}{\txi^{\[j\]}}
\newcommand{\xij}{\xi_{\[j\]}}

\newcommand{\CP}{\IC P^1}

\def\bOm{\bar\Omega}
\def\btau{\bar \tau}
\def\bt{\bar t}
\def\bi{\bar \imath}
\def\bj{\bar \jmath}

\def\ba{\bar a}

\def\by{\bar y}
\def\bz{\bar z}
\def\bw{\bar w}
\def\bY{\bar Y}
\def\bZ{\bar Z}

\def\bW{\bar W}
\def\bX{\bar X}
\def\bF{\bar F}

\def\btau{\bar \tau}
\def\bOm{\bar\Omega}

\def\hq{\hat q}

\def\hzeta{\hat \zeta}
\def\hgam{\hat\gamma}

\def\bfhgam{\hat\bfgam}
\def\bfhmu{\hat\bfmu}

\def\hpartial{\hat \partial}

\def\hq{\hat q}

\def\hcT{\hat\cT}

\def\chg{\check g}

\def\chphi{\check\phi}

\def\gama{\hat\gamma}

\def\Win{W_{\rm inst}}
\def\bWin{\bar W_{\rm inst}}

\newcommand{\glueg}{\mathbbm{g}}

\def\bfg{{\boldsymbol{\gamma}}}
\def\hbfg{\hat\bfg}

\def\CY{\mathfrak{Y}}
\def\CYm{{\widehat{\mathfrak{Y}}}}

\def\ver{v}

\def\ST{\IT^{\rm S}}
\def\LT{\IT^{\ell}}

\def\hgam{\hat\gamma}

\def\hi#1{h^{(#1)}}

\def\hh{h^{(0)}}
\def\thh{\tlh^{(0)}}
\def\than{\tlh^{\rm (an)}}
\def\thh{\tlh^{(0)}}

\def\ai#1{{\alpha}^{[#1]}}

\def\hpert{h^{(0)}}
\def\zpert{z_0}
\def\bzpert{\bz_0}

\def\dIu{\cI}
\def\dId{\cI}

\def\opI{\scJ}
\def\Prz{\mathscr{P}}

\def\scLn#1{\scL^{(#1)}}
\def\bscLn#1{\bar{\scL}^{(#1)}}
\def\scKn#1#2{\cK^{(#1)}_{\bfg,#2}}

\def\cJn#1{\cJ^{(#1)}}
\def\cIn#1{\cI^{(#1)}}
\def\cGn#1#2{\dIu^{(#1)#2}}
\def\tcGn#1#2{\dId^{(#1)}_{#2}}
\def\cLn#1#2{\cL^{(#1)}_{#2}}
\def\bcLn#1#2{\bar{\cL}^{(#1)}_{#2}}

\def\ths#1{\theta^{(#1)}}
\def\vths#1{\vartheta^{(#1)}}
\def\vthls#1{\vartheta^{(#1)||}}
\def\tvths#1{\tilde\vartheta^{(#1)}}
\def\tvthls#1{\tilde\vartheta^{(#1)||}}

\def\vthps#1{\vartheta^{(#1)\perp}}

\def\vthA#1{\varTheta^{(#1)}}

\def\di#1{d_{\Nr_#1}}

\def\ellg#1{\ell_{#1}}

\def\cXsf{\cX^{\rm sf}}

\def\cXr{\mathsf{X}^{\rm ref}}
\def\cXz{\mathsf{X}^{\rm sf}}

\def\hcXr{\mathsf{X}}
\def\hcXrp{\mathsf{X}^{+}}
\def\hcXrm{\mathsf{X}^{-}}
\def\hcXrpm{\mathsf{X}^{\pm}}

\def\gi#1{g^{(#1)}}

\def\girf#1{g^{(#1)\rm{ref}}}

\def\chgirf#1{\chg^{(#1)\rm{ref}}}
\def\vwgi#1{\mathfrak{g}_{#1}}

\def\whvwg{\widehat{\mathfrak{g}}}

\def\cXr{\mathsf{X}^{\rm ref}}

\def\Er{\Ev^{\rm (ref)}}

\def\cJr{\cJ^{\rm (ref)}}
\def\cEf{\cE^{(0)}}

\def\Ef{\Ev^{(0)}}
\def\Efrf{\Ev^{(0){\rm ref}}}

\def\trmRi#1{\tilde \rmR^{(#1)}}

\def\trmRVWi#1{\tilde \frR^{(#1)}}

\def\rmRi#1{\rmR^{(#1)}}

\def\rmRirf#1{\rmR^{(#1)\rm ref}}

\def\scRrf{\scR^{\rm ref}}

\def\whgirf#1{\widehat g^{(#1)\rm{ref}}}

\def\whchgirf#1{\lefteqn{\widehat \chg}\hphantom{\chg}^{(#1)\rm{ref}}}

\def\vwh{h^{\rm VW}}

\def\whh{\widehat h}
\def\twhh{\widehat{\tilde h}}

\def\whg{\widehat g}
\def\whgi#1{\widehat g^{(#1)}}

\def\whG{\widehat G}

\def\bfLami#1{\bfLam^{(#1)}}

\def\bbLami#1{\bbLambda^{(#1)}}

\def\cvths#1{\lefteqn{\smash{\mathop{\vphantom{<}}\limits^{\;\circ}}}\vartheta^{(#1)}}

\def\Gi#1{G^{(#1)}}
\def\cGi#1{\cG^{(#1)}}
\def\Mi#1{M^{(#1)}}

\def\whGi#1{\whG^{(#1)}}

\def\cEPhi{\Phi^{\,\cE}}
\def\EPhi{\Phi^{\,\Ev}}

\def\rPhi{\Phi^{(\bfr)}}

\def\PhiR{\Phi_{\rm R}}
\def\Phid{\Phi_\delta}

\def\Phii#1{\Phi^{(#1)}}
\def\PhiRi#1{\PhiR^{(#1)}}
\def\Phidi#1{\Phid^{(#1)}}
\def\Fvi#1{\Fv^{(#1)}}
\def\whFvi#1{\widehat\Fv^{(#1)}}

\def\vu{\mathfrak{u}}

\def\cEp{\cE^{(+)}}

\def\Ep{\Ev^{(+)}}
\def\Eprf{\Ev^{(+){\rm ref}}}

\def\frm{\mathfrak{m}}

\def\frz{\mathfrak{z}}

\def\frJ{\mathfrak{J}}
\def\frV{\mathfrak{V}}
\def\frD{\mathfrak{D}}
\def\frr{\mathfrak{r}}

\def\frt{\mathfrak{t}}

\def\frS{\mathfrak{S}}
\def\frZ{\mathfrak{Z}}

\def\frR{\mathfrak{R}}
\def\frL{\mathfrak{L}}

\def\frLn#1#2{\frL^{(#1)}_{\bfg,#2}}
\def\bfrLn#1#2{\bar{\frL}^{(#1)}_{\bfg,#2}}

\def\bfrZ{\bar\frZ}

\def\cbfr{c_\bfr}

\def\cAr{\cA^{(\bfr)}}

\def\cCf{\cC}

\def\tleta{\tilde\eta}

\def\vnc#1#2{#1^{[#2]}}

\newcommand{\q}{\mbox{q}}

\def\Nr{r}

\def\rdcr{\hat r}                  

\def\sl{a}                    
\def\Ms{s}

\def\Kr{K^{\rm (ref)}}
\def\frKr{\mathfrak{K}}
\def\Wr{W^{\rm (ref)}}

\def\cHr{\cH^{\rm ref}}
\def\cJr{\cJ^{\rm (ref)}}

\def\IS{\lefteqn{\textstyle\sum}\int}
\def\ker{\mathcal{K}}
\def\tIS#1{\resizebox{0.013\vsize}{!}{$\displaystyle{\IS_{#1}}$}}

\def\opI{\scJ}
\def\Prz{\mathscr{P}}

\def\prodi#1#2{\lefteqn{\prod_{#1}^{#2}}\mathop{\phantom{\prod}}\nolimits}

\def\Gact{\Gamma_{\rm act}}

\title{PhD Thesis in String Theory}
\author{Your Name}
\date{\today}

\begin{document}
\begin{titlepage}
\includepdf{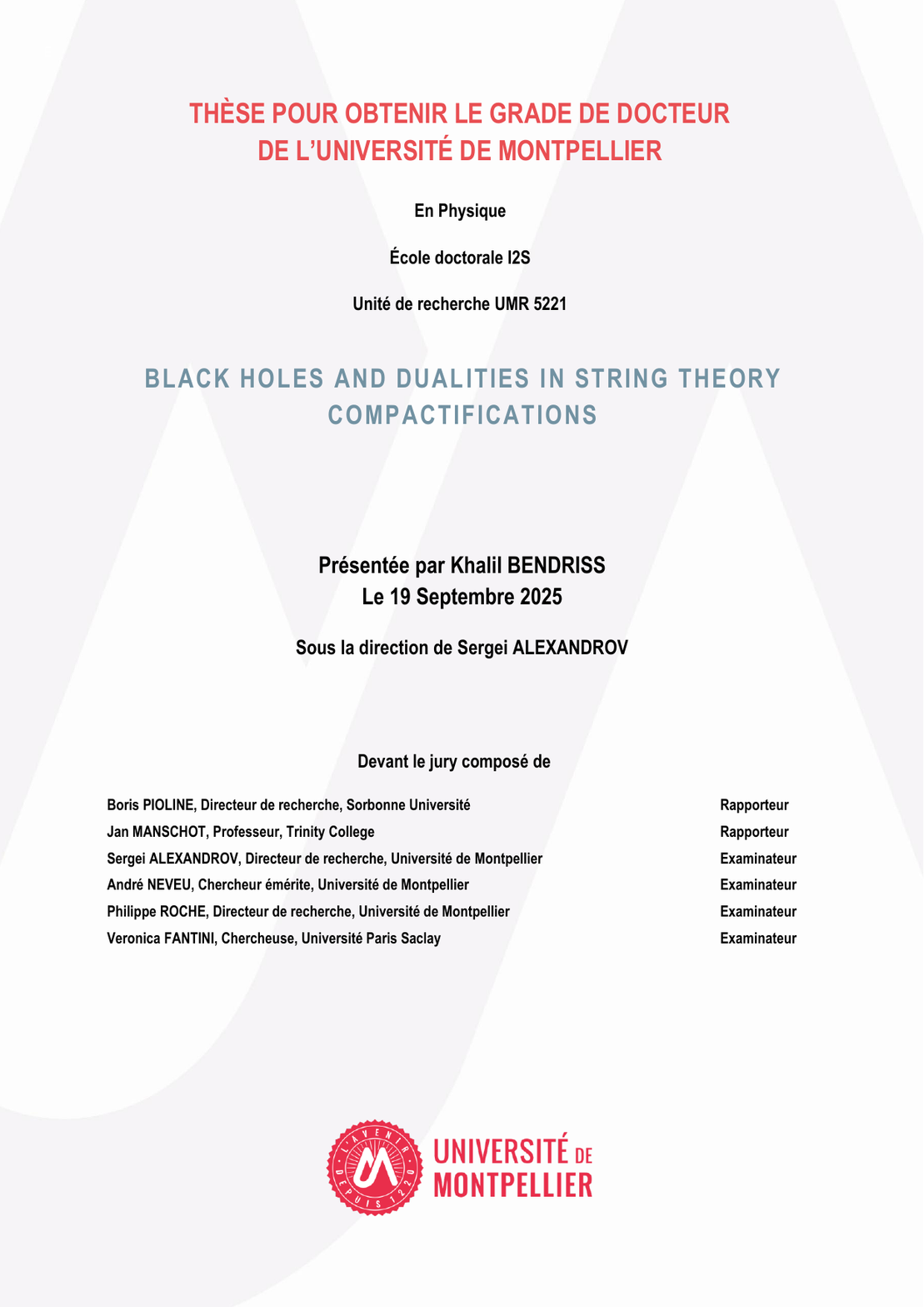}
%
\end{titlepage}

\chapter*{Abstract}
In this thesis we consider three problems appearing in, or related to, type II string theory compactified on a Calabi-Yau manifold.
In the first one we study the hypermultiplet moduli space (HM), by working on its twistor space. 
We use previously found data on the latter, obtained by applying mirror symmetry and S-duality, in order to compute the NS5-instanton corrections to the HM metric, in the one-instanton approximation. 

These corrections are weighted by D4-D2-D0 BPS indices, which coincide with rank 0 Donaldson-Thomas invariants and count the (signed) number of BPS black hole microstates. 
These invariants exhibit wall-crossing behavior and induce a Riemann-Hilbert problem. This problem can describe the D-instanton corrected twistor space of the hypermultiplet moduli space in type II string theory. But it is more general as it may describe other setups and is of independent interest for mathematicians. We consider a quantum deformation of this problem, induced by the refined BPS indices. Using a formulation of the problem in terms of a non-commutative Moyal star product, we provide a perturbative solution to it. By considering this solution in the adjoint form we identify an object that we interpret as a generating function for coordinates on the, still mysterious, quantum analog of the twistor space. 
In particular, its unrefined limit gives an all-order expression for the generating function of the coordinates on the classical twistor space. 

Finally, we study the modular properties of the D4-D2-D0 BPS indices, more precisely of their generating functions. It was previously argued, using S-duality, that the generating functions are higher depth mock modular forms. Moreover, they satisfy a modular completion equation, which fixes their shadow in terms of other (lower rank) generating functions. We start by bringing about a significant simplification to these equations and recovering subtle contributions that were overlooked. We do this using various types of trees and identities between generalized error functions. Then, we provide (a recipe for) solutions to these modular completion equations, up to all the holomorphic modular ambiguities that need to be fixed independently. In order to do this we perform a few extensions to the problem, including a refinement parameter and an extension of the lattice. Then, we use indefinite generalized theta series and Jacobi-like forms to write the solutions. For some cases we provide explicit solutions to the original problem; for others we provide solutions to the extended problem and conjecture that they do reduce to solutions to the original one.

\chapter*{Acknowledgments}
Diving deep into the laws of the universe and grappling with the beautiful mathematics that emerge has always been a dream, and I am very grateful for the opportunity to fulfill it.  

I reserve special thanks to my supervisor Sergei Alexandrov who made this journey extremely interesting. Our black board discussions have always been very rich and stimulating and often led to significant breakthroughs in our understanding of the different problems. Your intuition and creativity are impressive and it was a pleasure to witness them firsthand. Not to forget your willingness to tackle the, sometimes monstrous, computations. Furthermore, I highly value the personal bond we created during this time, namely through the multiple dinners we shared with speakers or other colleagues, the football match we played, and won, together and finally, the trip to Venice after the conference in Trieste. 

I hold very heartfelt appreciation to the member of my examination committee. I am extremely thankful to Boris Pioline and Jan Manschot for accepting to review my dissertation. I am also very humbled to have Andr\'e Neveu as an Examiner, as well as Philippe Roche and Veronica Fantini. Your participation is an essential pillar of the research community. 

I wish to thank my MSc thesis supervisor Guillaume Bossard, for preparing me to this beautiful PhD subject and for subsequent discussions whenever we met at a conference. I also thank Boris Pioline, Abhiram Kidambi, Severin L\"ust and Ruben Minasian for the discussions we had at the different conferences where we met and for shaping my scientific culture. 

I am grateful to "Ecole Doctorale: Information, Structures, Systèmes" for the financial support during this thesis. Particularly to Mme Lucyna Firlej for her responsiveness and help at various points.

On a more personal note, I want to thank the different people that welcomed me to this beautiful city of Montpellier. I start with my first officemate, Lo\"ic Fernandez, whom I only saw during the "happy" times when he was writing his dissertation but it was enough to start our friendship. Then to the one who changed cities mid-way, Romain Simon, know that Montpellier misses you! Special thanks go to the very good friends I made during this time and who made me feel at home: Giacomo Ferrante, Th\'eo Simon, Hugo Plombat, Julien Perradin, Adam Bernard, Simon Schreyer, Benjamin Benhamou Bui, Juliette Plo, Thomas Alfonsi, Michele Frigerio, Daniel Johnson and Elsa Teixeira. 

Finally, I want to thank my family for their unwavering support and for the sacrifices they made to stand with me in this journey. For Henda, Khaled, Meriam, Adam and Youssef, thank you a lot, I love you.

\setcounter{page}{1}
\setcounter{equation}{0}
\setcounter{footnote}{0}
\

\newpage
\setcounter{tocdepth}{2}
\tableofcontents

\chapter{Introduction}
The study of non-perturbative phenomena is a cornerstone of high energy physics, allowing to make profound observations beyond the reach of conventional perturbation theory. In String Theory this is particularly essential. Objects like instantons are often needed to get the whole picture, for example by resolving some singularities or restoring some symmetries, while black holes, are considered to be the main testing ground for the expected properties of a theory of quantum gravity. The underlying structure of string theory is rich in non-perturbative dualities and symmetries, allowing to probe these effects. These investigations often reveal a powerful interplay between elegant mathematical structures and deep physical results. Accordingly, in this chapter we will quickly introduce the physical and mathematical notions that will be treated in the rest of this dissertation.

In the following thesis we focus on type II string theories compactified on a Calabi-Yau (CY) threefold: a three (complex) dimensional K\"ahler manifold with a vanishing first Chern class. This yields $N=2$ supersymmetry in the non-compact dimensions, ensuring a nice blend of tractable computations and rich mathematical structure. In the absence of fluxes, the geometric and topological data of the CY manifold fully determine the properties of the resulting 4d theory. This is encoded in metric of the moduli space which fixes the low energy action at the two derivative level. The moduli space is the space of deformations of the CY metric and of massless excitations of the strings, and in the presence of $N=2$ supersymmetry it factorizes 
\be
\cM_{4d}=\cV \times \cQ,
\ee
where the first term denotes the vector multiplet (VM) and the second term the hypermultiplet (HM) moduli space. The metric on $\cV$ is tree-level exact and is well-known. On the contrary, the HM metric receives all possible $g_s$ corrections, including non-perturbative D-instantons and NS5-instantons corrections. In this dissertation we will focus on properties of $\cQ$ and related objects.

Our first project was to compute the NS5-instanton corrections to the metric on $\cQ$, in the one instanton approximation. The main difficulty in doing this comes from the fact that the HM moduli space is quaternion-K\"ahler (QK), making it very difficult to study. Unlike K\"ahler spaces its metric can't be expressed in terms of a single function. However, its properties are nicely encoded in its \textit{twistor space} $\cZ_\cQ$ \cite{MR664330,Alexandrov:2008nk} that is much more convenient to work with. 
First, the twistor space, which is a $\CP$-bundle over $\cQ$ with a canonical complex structure, is fully described by a set of holomorphic functions. These are called \textit{holomorphic transition functions} and they set up a Riemann-Hilbert problem whose solutions are the coordinates on $\cZ_\cQ$. 
In the presence of D-instantons, the induced deformations to $\cZ_\cQ$ were studied and the transition functions computed in \cite{Alexandrov:2008gh,Alexandrov:2009zh}, while for NS5-instantons, in the one instanton approximation, the same was achieved in \cite{Alexandrov:2010ca}.
On the other hand, once we obtain the twistor geometry we can apply a procedure, explained in \cite{Alexandrov:2008nk} and applied to D-instantons in \cite{Alexandrov:2014sya}, in order to get the metric on $\cQ$. We apply this procedure in \cite{Alexandrov:2023hiv} to the proposal from \cite{Alexandrov:2010ca} to retrieve the linear NS5-instanton corrections to the hypermultiplet metric. We conclude this work by performing two independent cross-checks on our resulting metric.

The non-perturbative contributions to the metric considered previously are weighted by BPS indices which count BPS black hole micro-states and coincide with the so-called generalized Donaldson-Thomas (DT) invariants of the CY. These indices are piecewise constant and exhibit \textit{wall-crossing} behavior described by the Kontsevich-Soivbelman formula \cite{ks}. Such jumps set up a Riemann-Hilbert problem \cite{Bridgeland:2016nqw} similar to the one encountered for D-instantons above. In fact these indices have refined counterparts that are obtained by adding a fugacity parameter $y$. These refined indices offer much more information about the BPS spectrum and they verify a refined version of the wall-crossing formula. In addition, the parameter $y$ effectively quantizes the moduli space and it is natural to look for a quantum Riemann-Hilbert problem on its twistor space. We propose a framework for this problem where the non-commutativity is ensured by a Moyal product involving the parameter $y$ and we find an asymptotic expansion for the piecewise holomorphic functions solving it \cite{Alexandrov:2025abc}.

Another project we worked on is solving an equation involving the generating functions of BPS indices. Despite their importance, little is known about how to compute these indices systematically. This is partly due to wall-crossing. A rather surprising result in that regard is that, in the large volume attractor point, generating functions of rank 0 DT invariants are \textit{higher depth mock} modular forms \cite{Alexandrov:2018lgp}! 
Higher depth mock modular forms are generalizations of the usual notion of mock modular forms \cite{Zwegers-thesis}, and we will define them later. 
It was by enforcing S-duality on the hypermultiplet moduli space on the type IIB side that Alexandrov and Pioline \cite{Alexandrov:2018lgp} derived the transformation rules of these functions under $SL(2,\Z)$, in the form of modular anomaly equations. The anomaly of each generating function is written in terms of a sum, indexed by trees, of indefinite theta series multiplied by lower charge generating functions. Whereas the kernels of the indefinite theta series are made of sums over different sets of trees and involving generalized error functions and derivatives thereof. We considerably simplify the kernels using several identities between generalized error functions \cite{Alexandrov:2024jnu}. It is worth noting that the constraint coming from each equation is only strong enough to fix the solution up to a holomorphic modular ambiguity. In \cite{Alexandrov:2024wla}, we restrict to one modulus CY and give a recipe to solve these refined modular anomaly equations up to all the (relevant) holomorphic modular ambiguities.

\section{Outline}

The first three chapters provide the necessary background needed to present our results. 

Chapter \ref{chap-string} begins by establishing the necessary background in string theory, motivating the core themes of this dissertation. Section \ref{sec-string-typeII} introduces the spectrum and symmetries of type II string theories. Section \ref{sec-string-CY} then discusses compactifications on Calabi-Yau threefolds, presenting their associated deformation moduli spaces. Building on this, section \ref{sec-string-Moduli} describes the moduli spaces for both type IIA and IIB theories. These spaces fully determine the low-energy effective action, detailed in section \ref{sec-string-EFT}. Finally, section \ref{sec-string-BHBPS} introduces BPS states and indices counting them. We also explain how they are related to instanton effects and to BPS black holes.

Chapter \ref{chap-twist} introduces the geometric framework of twistor spaces, which is a powerful tool used in this work in chapters \ref{chap-NS5HM} and \ref{chap-qRH}. The chapter start in section \ref{sec-twist-qk} with a presentation of \qk (QK) spaces. We provide a valuable description of four-dimensional QK spaces in terms of a single potential, due to Przanowski \cite{Przanowski:1984qq} in subsection \ref{subsec-twist-Prz}. Following this, we proceed in section \ref{sec-twistor} to the explicit construction of the twistor space associated with a QK manifold and, in section \ref{sec-twist-transition}, describe the properties of th twistor space and the essential data required to define it.

Chapter \ref{chap-mod} shifts focus to the mathematical tools of modular forms and their generalizations, which are fundamental notions for chapter \ref{chap-DTall}. We begin in section \ref{sec-mod-VV} with a review of classical modular forms and vector-valued modular forms. In section \ref{sec-mod-mock}, we extend this discussion to (higher-depth) mock modular forms. The chapter also covers Jacobi forms and slight variations thereof in section \ref{sec-mod-Jacobi}. Finally, in section \ref{sec-mod-theta}, we define the notion of theta series and special functions, namely generalized error functions, that will appear in the study of generating functions of BPS indices.

The following chapters shift focus from this foundational background to present the original contributions of this dissertation.

Chapter \ref{chap-NS5HM} presents the first result of this thesis, detailing the computation of NS5-brane instanton corrections to the hypermultiplet moduli space metric, based on our work in \cite{Alexandrov:2023hiv}. We begin in section \ref{sec-NS5-pert} by describing the perturbative metric, introducing the twistor space framework from the outset to prepare the reader for the more complex instanton analysis. In section \ref{sec-NS5-Instanton}, we present the main calculation of the NS5-instanton corrected metric. 
To validate our findings, the final sections of the chapter explore two crucial limits. In section \ref{sec-NS5-UHM}, we compare our result with the Przanowski description, relevant for rigid CY spaces, introduced earlier, while in section \ref{sec-NS5-smallgs}, we examine the small string coupling limit, matching the metric's form to expectations from string amplitude computations. This analysis culminates in a prediction for a specific string amplitude in the limit of small Ramond-Ramond fields.

Chapter \ref{chap-qRH} addresses our work on the quantum Riemann-Hilbert problem that is induced by refined BPS indices, as detailed in \cite{Alexandrov:2025abc}. We begin in section \ref{sec-qTBA-classical} by reviewing the classical Riemann-Hilbert problem and demonstrating how its solution can be derived from a TBA-like integral equation. In section \ref{sec-qTBA-qTBA}, we introduce a quantum deformation via the Moyal star product and then formulate the quantum Riemann-Hilbert problem. A central result is the explicit construction of a formal solution to this quantum problem, which correctly reduces to the classical solution in the unrefined limit. Furthermore, in section \ref{sec-qTBA-genfun}, we identify a function that acts as a refined generating function for Darboux coordinates, providing multiple all order perturbative expansions for it and for its unrefined limit.
Finally, \ref{sec-qTBA-uncoupled} provides an explicit computation of the formal solution, performed in the case of so-called \textit{uncoupled} BPS structure, which reproduces a result from \cite{Barbieri:2019yya}.

Chapter \ref{chap-DTall} delves into the modular properties of generating functions for BPS indices, presenting the findings from our papers \cite{Alexandrov:2024jnu, Alexandrov:2024wla}. The first part of the chapter, based on \cite{Alexandrov:2024jnu}, defines these generating functions and the equations they satisfy (\S\ref{sec-DTall-def}), introduces a tree-based formalism to define the coefficients\footnote{These are functions playing the role of coefficients in the completion equation.} in these equations (\S\ref{sec-DTall-Kernel}), and presents significant simplifications of these coefficients (\S \ref{sec-DTall-simplify}). 
The second part of the chapter, corresponding to the work in \cite{Alexandrov:2024wla}, specializes to Calabi-Yau spaces with $b_2=1$. 	Therein, our goal is to solve the completion equations mentioned above and thus find (at least the anomalous part of) the generating functions of BPS indices.
First, we introduce \textit{anomalous coefficients} that parametrize the ambiguity in the generating functions and show that they satisfy their own modular completion equations (\S \ref{sec-DTall-gi}). Our goal in this second half of the chapter is to solve for these anomalous coefficients. We first present two families of solutions (\S \ref{sec-DTall-HeckeVW}) based on Vafa-Witten theory and Hecke-like operators. Motivated by the limitations of these methods, we develop a more general strategy (\S\ref{sec-DTall-Extensions}) for solving the completion equations, using the simplest non-trivial case to motivate tractable extensions to the problem. We demonstrate how this framework solves this case (\S\ref{sec-DTall-sol2}) and conclude by presenting a general recipe to solve the completion equations for any anomalous coefficient (\S\ref{sec-DTall-soln}).

Finally, in chapter \ref{chap-concl} we provide both a summary of our results and a comprehensive outlook. This final chapter will first consolidate the principal achievements concerning NS5-instanton corrections, the quantum Riemann-Hilbert problem, and the generating functions of BPS indices. It will then map out future directions of research, identifying a set of specific and important problems.

\chapter{Background on Type II Superstring Theory}
\label{chap-string}

Having established in the previous chapter the main theme of this thesis, namely the study of the hypermultiplet moduli space arising in the context of Type II string theory compactifications, we now turn to laying the necessary theoretical groundwork. A thorough understanding of the HM, its geometric properties, and particularly the quantum corrections it receives, requires a clear exposition of its origins within string theory. 
This chapter provides a focused review of the essential aspects of 10d Type II string theories and their compactifications on Calabi-Yau threefolds, which give rise to the four-dimensional N=2 supergravity theories where the hypermultiplet moduli space appears.

We start in 10d, outlining the fundamental constituents and massless spectrum of type IIA and type IIB string theories.
We will present the fields belonging to the Neveu-Schwarz-Neveu-Schwarz (NS-NS) sector, common to both theories, and the Ramond-Ramond (R-R) sector, which distinguishes them. The introduction of D-branes as non-perturbative, dynamical objects carrying R-R charge will also be crucial as they play a vital role as sources for non-perturbative effects, such as the instanton corrections central to this work.

Subsequently, we will discuss compactifications of the theory, focusing specifically on Calabi-Yau threefolds. We will give the defining properties of these manifolds and explain their role in preserving precisely N=2 supersymmetry upon going from ten to four dimensions. 
This dimensional reduction yields a rich structure in the four-dimensional effective field theory, characterized by N=2 supergravity coupled to a specific spectrum of matter multiplets.
The central outcome of this process, for the purposes of this thesis, is the emergence of distinct moduli spaces associated with the massless fields of the four-dimensional theory. 
These scalars organize into vector multiplets and hypermultiplets. We will explain how each of these fields originates from the 10d spectrum and the moduli of the compactification manifold and draw a disctinction between the geometry of the VM moduli space and that of the HM.

This exposition provides an essential context for appreciating the challenges involved in the study of the hypermultiplet moduli space. It also motivates the techniques, such as twistor theory and the analysis of BPS states, explored in the subsequent chapters. We begin our exposition with foundational elements of Type II superstrings in 10d.

\section{Type II Superstring Theory}
\label{sec-string-typeII}

String theory naturally generalizes quantum field theory by considering particles as vibrating strings rather than point-like objects. These strings sweep out a two-dimensional surface called \textit{worldsheet}, and spacetime emerges as the target space of a nonlinear sigma model defined on this surface

We add supersymmetry to this theory and we obtain five consistent string theories in 10d with no instabilities.
Among them, there are two so-called Type II theories obtained using oriented closed strings. As is the case for any string theory, a graviton-like particle emerges from the excitations of the closed strings showing that gravity is necessarily present. In fact, the low energy limit of type IIA/B string theory gives type IIA/B supergravity in 10d. Therefore, in order to obtain the low energy effective description of compactified string theories, it is convenient to look at the corresponding 10d supergravity theories. We focus on the bosonic sector, as supersymmetry allows to recover the fermionic part. 

\subsection{The massless bosonic spectrum}
\label{subsec-BosSpectrum}

The bosonic spectrum of type II supergravity is subdivided into two parts, the Neveu-Schwarz Neveu-Schwarz (NS-NS) sector and the Ramond-Ramond (RR) sector. 

The NS-NS sector, identical for both type IIA and type IIB, contains three fields. First, the ten dimensional metric $\hat g_{XY}$ (where $X,Y=0\dots9$ are spacetime indices), which describes gravity. Next, there is the Kalb-Ramond field, a two-form $\hat B_2$ sourced by the fundamental string. We denote its field strength $H_3=\de B_2$. Finally, the ten dimensional dilaton $\hat \Phi$ which is a scalar field whose VEV determines the string coupling constant $g_s$. String interactions are governed by $g_s$ and in this sense they are perturbative, whereas non-perturbative effects scale exponentially in $-1/g_s^\alpha$ for some $\alpha>0$. Note that we use hats on ten dimensional fields to differentiate them from their four dimensional counterparts to be introduced later.  

In contrast, the RR sector consists of gauge potentials and is different between type IIA and type IIB. It comprises $p$-form potentials $\hat A_p$ with $p=1,3$ in the former and $p=0,2,4$ in the latter. The scalar $\hat A_0$ is often combined with the dilaton $\hat \Phi$ into the \textit{complexified} string coupling constant $\tau=\hat A_0 + \I/g_s$. In general, to these potentials correspond field strengths $F_{p+1}=\de \hat A_p$ that enter the 10d action of supergravity and therefore appear in the equations of motion. In type IIB an additional self-duality constraint is imposed on $F_5$. It is evident that these $p$-forms can be viewed as electromagnetic potentials and it turns out that they are sourced by D-branes.

\subsection{D-branes and NS5-branes}
A D$p$-brane is a solitonic solution of supergravity extending over $p$ spatial dimensions. It is a non-perturbative, dynamical object in string theory. D-branes were first encountered as surfaces on which endpoints of strings can end. It was later understood however that they are fundamental in string theory and supergravity. Indeed, they are essential to a number of dualities in the theory, as we will explain shortly. Moreover, they play the role of sources for the RR potentials \cite{Polchinski:1995mt}, akin to electric charges for the electromagnetic field in Maxwell's theory. More precisely, a D$p$-brane is charged under the potential $\hat A_{p+1}$ and sources the field strength $F_{p+2}$. The allowed branes differ between the theories, namely we have D$p$-branes with $p=0,2,4,6,8$ in type IIA and $p=-1,1,3,5,7,9$ in type IIB. Their non-perturbative nature can be seen in their tension being proportional to $1/g_s$. 

Complementary to this picture are Neveu-Schwarz branes (NS5-branes). These are ($5+1$)-dimensional objects, magnetic dual to the fundamental strings. It is known on general grounds that their tension is proportional to $1/g_s^2$ \cite{Becker:1995kb}, making them suppressed even compared to D-branes when $g_s\ll 1$. 

Upon compactification, branes can wrap cycles in the internal manifold resulting in many interesting effects. When Euclidean D$p$-branes (or NS5-branes) wrap a non-trivial $p+1$-cycle they are seen as point-like objects from the 4d theory perspective. Hence, this configuration gives instanton corrections which, among other things, contribute to the metric on the hypermultiplet moduli space. This is a crucial part of this thesis.

\subsection{Symmetries}
At last, string theory has a beautiful set of symmetries allowing to navigate between the different theories elegantly. Some of these dualities are extremely powerful as they relate the perturbative regime of the theory to the non-perturbative one. Focusing on type IIA and type IIB theories we can give a quick picture of the dualities involving them. 

First, we have T-duality which relates type IIA string theory compactified on a circle of radius $R$ to type IIB string theory compactified on another circle with radius $1/R$. A particular feat of this duality is that it offsets the dimension of a D$p$-brane by $1$ transforming it to a D$(p+1)$-brane if the duality is performed along a direction transverse to its surface or to a D$(p-1)$-brane if it performed along a direction parallel to it. 

Then, we have S-duality which acts within type IIB and exchanges the strong coupling regime with the weak coupling regime by sending $g_s\to 1/g_s$. More precisely, it acts on the complexified string coupling $\tau$ with an $SL(2,\IZ)$-group action
\be
\qquad 
\tau \to \frac{a\tau + b}{c\tau +d}, 
\qquad 
\begin{pmatrix}
	a& b\\
	c& d
\end{pmatrix}\in SL(2,\IZ).
\ee 
Actually, at the classical level the symmetry group is full $SL(2,\IR)$. However, it is completely broken by $\alpha'$ corrections and it's only upon including D(-1)-D1 instantons that we recover a symmetric action of $SL(2,\IZ)$ \cite{Alexandrov:2009qq,RoblesLlana:2006is}, which is expected to hold after all quantum corrections are incorporated. For instance D3-instantons are expected to transform into themselves while D5-instantons are mapped to NS5-instantons.

\section{Calabi-Yau compactifications}
\label{sec-string-CY}

To bridge the gap between the 10d string theory and the four-dimensional spacetime we live in, the extra six spatial dimensions must be compactified on a 6 (real) dimensional manifold $\CY$. The choice of this internal manifold $\CY$ profoundly influences the resulting four-dimensional effective field theory, determining its symmetries, field content and interactions. Since this thesis focuses on theories with $N=2$ supersymmetry in 4d, the manifold $\CY$ is required to be a Calabi-Yau threefold. 

\subsection{Definition and properties}
In general, a Calabi-Yau manifold $\CY$ is a $n$ (complex) dimensional compact complex manifold with a Ricci-flat K\"ahler metric. We focus specifically on the case $n=3$. It can be shown that $\CY$ has a vanishing first Chern class $c_1\(\CY\)$ and that its holonomy group is in $SU(3)$. This second property is crucial for preserving $N=2$ supersymmetry upon compactification because it allows the existence of a unique covariantly constant spinor. 

The topology of these manifolds is very important for our purposes. A particularly interesting and fundamental property is the Hodge diamond which provides the dimensions of the Dolbeault cohomology groups $H^{p,q}(\CY)$
\be
\begin{array}{ccccccc}
	&&& 1 &&& \\
	&& 0 && 0 && \\
	& 0 && h^{1,1} && 0 & \\
	1 && h^{2,1} && h^{2,1} && 1 \\
	& 0 && h^{1,1} && 0 & \\
	&& 0 && 0 && \\
	&&& 1 &&& \\
\end{array}
\ee
It implies a few immediate consequences: 
\begin{itemize}
\item There are no non-trivial one- or five-dimensional cycles in a $\CY$.
\item The Euler characterestic can be obtained as 
\be
\chi_\CY = \sum_{p,q}(-1)^{p+q} h^{p+q}(\CY) = 2\(h^{1,1}(\CY)-h^{2,1}(\CY)\).
\ee
\item There is a unique, up to a holomorphic factor, nowhere vanishing holomorphic $(3,0)$-form, commonly denoted by $\Omega$.
\end{itemize}

It is clear that $\Omega$ fixes completely the complex structure of $\CY$. In addition $\CY$ admits a $(1,1)$ K\"ahler form $J$ which encodes its K\"ahler structure.

\subsection{Moduli spaces of K\"ahler and complex structures}
\label{subsec-CplxKahlerModuli}
A significant part of the low energy effective theory is determined by the space of deformations of the compactification manifold. For a Calabi-Yau this is the space of deformations of its metric, which is completely determined by the complex and K\"ahler structures. In fact, from the string theory point of view it is worthwhile to consider the complexification of the K\"ahler structure by considering the combination $\hat B_2+\I J$ where $\hat B_2$ is the Kalb-Ramond field. 
Since the complex and the (complexified) K\"ahler structures are encoded in the $(3,0)$-form $\Omega$ and the complexified K\"ahler form, which are both closed, then their infinitesimal deformations are given by elements of the groups $H^{2,1}(\CY,\IC)$ and $H^{1,1}(\CY,\IC)$. Denoting the resulting moduli spaces respectively by $\cM_C(\CY)$ and $\cM_K(\CY)$, the full CY moduli space is given by 
\be
\label{CY-moduli}
\cM_\CY= \cM_C(\CY) \times \cM_K(\CY).
\ee

The geometry of both factors is such that they are special K\"ahler spaces. This means that they have a natural K\"ahler metric that can be expressed in terms of a single holomorphic function $F(X)$ called the prepotential. For a $n$-dimensional space, this function should be homoegeneous of degree two in the $n+1$ variables $X^\Lambda,\, \Lambda=0,\dots,n$. It defines the K\"ahler potential as 
\be
\label{prepot-to-pot}
K = -\log \[2\Im \(\bX^\Lambda F_{\Lambda }(X)\)\],
\ee
which in turn defines the metric as 
\be
\label{Kahler-metric}
g_{a\bar b}=\partial_{z^a}\partial_{\bz^{\bar b}} K(z,\bz).
\ee
The coordinates $X^\Lambda$ are related to the actual coordinates of the special K\"ahler space as $z^a=\frac{X^a}{X^0}$ with $a=1,\dots,n$. A metric on each of the factors \eqref{CY-moduli} can be determined by $\Omega, J$ respectively. 

Let's start with $\cM_C$: the holomorphic $(3,0)$-form $\Omega$ gives rise to the following K\"ahler potential,
\be
\label{Kahlerpot-cplx}
K = -\log\(\I \int_{\CY} \Omega \wedge \bOm\).
\ee
In order to construct the corresponding prepotential, we start by introducing holomorphic coordinates on $\cM_C$. For this, we consider the full de Rham cohomology group 
\be
\label{decomp-H3}
H^{3}(\CY)=H^{3,0} \oplus H^{2,1} \oplus H^{1,2} \oplus H^{0,3} 
\ee
and its dual $H_3(\CY)$. Then we choose a basis of $H_3$ given by three-cycles $\cA^\Lambda$ and $\cB_\Lambda$ such that the only non-vanishing intersection numbers are $\cA^\Lambda \ast \cB_\Sigma = \delta^\Lambda_\Sigma$. They can be thought of as $A$ and $B$ cycles respectively. With this we define, 
\be
X^\Lambda = \int_{\cA^\Lambda}\Omega, \qquad F_\Lambda = \int_{\cB_\Lambda}\Omega.
\ee
Of course, these $2(h^{2,1}+1)$ coordinates are not all independent since the dimension of $\cM_C$ is given by $h^{2,1}$, the dimension of the group $H^{2,1}(\CY)$. This is due to the fact that infinitesimal deformation of $\Omega\in H^{3,0}$ are given by elements of $H^{2,1}$. Therefore, one can always choose the A and B cycles such that the $X^\Lambda$ are all independent and the $F_\Lambda$ are functions thereof. In addition, the Riemann bilinear identity 
\be
\label{Riemann-bilinear}
\int_{\CY} \chi \wedge \psi =\sum_\Lambda \(\int_{\cA^\Lambda} \chi \int_{\cB_\Lambda} \psi - \int_{\cB_\Lambda}\chi \int_{\cA^\Lambda} \psi\) ,
\ee 
valid for closed 3-forms $\chi$ and $\psi$, applied to $\Omega$ and $\partial_\Lambda \Omega$ respectively, ensures that $F_\Lambda$ are derivatives of a homogeneous function which can be obtained as 
\be
\label{prepot-cplx}
F=\hf X^\Lambda F_\Lambda.
\ee 
Then, if we apply \eqref{prepot-to-pot} to it we get the K\"ahler potential \eqref{Kahlerpot-cplx}. In conclusion, $F$ is the prepotential on the space of complex structure moduli.

Next, we deal with $\cM_K$. For similar reasons as the previous paragraph, this space has dimension $h^{1,1}$. The 2-form $J$ will be used to construct the metric on this space. In this case, we are interested in the even cohomology group 
\be
\label{decomp-Heven}
H^{\rm even} = H^{0} \oplus H^2 \oplus H^4 \oplus H^6.
\ee
Both $H^0$ and $H^6$ are one dimensional and generated respectively by $1$ and the volume form $\omega_\CY$. We supplement the former by a basis of $H^2$ to get $\omega_I=(1,\omega_i)$ where $I=0,\dots,h^{1,1}$ and $i=1,\dots,h^{1,1}$. We supplement the latter by a basis of $H^4$ to get $\omega^I=(\omega_\CY,\omega^i)$ such that 
\be
\label{Heven-basis-wedge}
\omega_i\wedge \omega^j = \delta_i^j \omega_\CY, 
\qquad 
\omega_i\wedge \omega_j= \kappa_{ijk}\omega^{k},
\ee
where the first equality ensures a duality correspondence and the second defines the triple intersection numbers of the CY. Choosing the dual basis $\gamma^i$ of 2-cycles and $\gamma_i$ of 4-cycles, they can also be defined as $\kappa_{ijk}=\int_{\CY}\omega_i\wedge \omega_j\wedge \omega_j = \langle\gamma_i,\gamma_j,\gamma_k \rangle$ Then the natural metric on $\cM_K(\CY)$ is defined by 
\be
g_{i\bj}=\frac{1}{4\cV}\int_{\CY}\omega_i \wedge \star \omega_j = \partial_i \partial_{\bj} (-\log 8 \cV),
\ee 
where $\cV = \frac16 \int_{\CY} J\wedge J \wedge J$ is the volume of the CY.

The holomorphic coordinates on $\cM_K$ arise as 
\be
v^i = b^i + \I t^i = \int_{\gamma^i}(\hat B_2+\I J).
\ee
Then we define the homogeneous coordinates as $X^I = (1,v^i)$ and the prepotential on the moduli space of K\"ahler structure can be written as 
\be
\label{prepot-Kahler}
F(X)= -\frac16 \frac{\kappa_{ijk}X^{i} X^{j} X^{k} }{X^0}.
\ee
As we will see this cubic prepotential is only an approximation, valid in the limit of large volume, of the potential appearing in string compactificaitons.

\subsection{Mirror symmetry}
Before going to the moduli space of the compactified theory, let's look at an interesting duality. Mirror symmetry relates seemingly different theories compactified on different Calabi-Yau threefolds.

Given a CY space $\CY$ with at least one complex structure modulus, then there exists a mirror CY space $\hat \CY$ satisfying 
\be
h^{3-p,q}(\CY)=h^{p,q}(\hat\CY),
\ee
with the the only non-trivial relation being
\be
h^{2,1}(\CY)=h^{1,1}(\hat \CY).
\ee
We see that under this exchange the dimensions of the spaces $\cM_C(\CY)$ and $\cM_K(\hat \CY)$ coincide. This is indicative of a deeper connection, namely that the complex structure moduli space of $\CY$ is exchanged with the K\"ahler structure moduli space of $\hat \CY$ and vice versa.

In string theory mirror symmetry manifests through a duality between the type IIA theory compactified on $\CY$ with the type IIB theory compactified on $\hat \CY$. This duality is extremely interesting because it maps perturbative corrections on one side to non-perturbative on the other side. Furthermore, it exchanges D2-branes in type IIA with D(-1)-D1-D3-D5 branes in type IIB.

\section{Moduli spaces}
\label{sec-string-Moduli}
Now we turn to the full moduli space of the compactified string theory. It is obtained by taking the field content of the corresponding 10d theory and expanding it on a basis of harmonic forms on $\CY$. This means that we only take the lowest Kaluza-Klein modes of the compactification. Then, we dualize the eventual 2-forms and 3-forms. For even $p$-forms one can take the basis introduced below \eqref{decomp-Heven} whereas for odd ones the spaces $H^1,H^5$ are trivial and thus we consider only the basis of $H^3$. We choose for it a basis dual to the A and B cycles introduced in \eqref{decomp-H3} and we label it $\alpha_\Lambda,\beta^\Lambda$. We keep the notations of \ref{subsec-CplxKahlerModuli} where indices $a$ and $i$ run respectively over $h^{2,1}$ and $h^{1,1}$ values starting from $1$ and $\Lambda, I$ are their corresponding extensions starting from $0$. 

We consider only the bosonic spectrum since the fermionic contributions can in principle be recovered using supersymmetry. The treatment of type IIA and type IIB will be done in separate subsections and will follow closely \cite{Alexandrov:2011va}.

\subsection{Type IIA}
\label{4d-fields-A}
We start with the NS-NS sector. The 10d metric $\hat g_{XY}$, upon compactification, yields a 4d metric $g_{\mu\nu}$ and, as discussed in section \ref{sec-string-CY}, $h^{2,1}$ fields corresponding to the holomorphic coordinates $z^a$ on $\cM_C$ and $h^{1,1}$ fields corresponding to the real coordinates $t^i$ on $\cM_K$. The Kalb-Ramond field $\hat B_2$ gives a 4d two-form $B_2$ and $h^{1,1}$ scalars $b^i$ which combine with $t^i$ from the metric to give the complexified K\"ahler moduli $v^i = b^i+\I t^i$, as was also discussed in \ref{sec-string-CY}. We dualize the 2-form $B_2$ and trade it for the axion $\sigma$. Finally, the dilaton $\hat \Phi$ reduces to a single scalar $\phi$. 

Next, in the RR sector we start with the 1-form $\hat A_1$ which simply reduces to a 4d 1-form $A_1^0$. Then, the 3-form $\hat A_3$ gives a 4d form $A_3$, $h^{1,1}$ 1-forms $A_1^i$ and $2(h^{2,1}+1)$ real scalars $\zeta^\Lambda, \tzeta_\Lambda$. The 3-form can be dualized to a constant and thus can be safely ignored\footnote{It was shown in \cite{Louis:2002ny} that it can an important role by inducing a gauge charge for the NS-axion $\sigma$.}.

These fields organize into $N=2$ multiplets: 
\begin{align*}
	&\text{gravitational multiplet} &&(g_{\mu \nu}, A_1^0) \\
	&\text{(universal) hypermultiplet} &&(\sigma, \phi, \zeta^0, \tilde{\zeta}_0) \\
	&h^{2,1} \text{ hypermultiplets}&&(z^a, \zeta^a, \tilde{\zeta}_a) \\
	&h^{1,1} \text{ vector multiplets}&&(A_1^i, v^i \equiv b^i + it^i).
\end{align*}
The second multiplet is present regardless of the Hodge numbers of the manifold and is thus called the \textit{universal} hypermultiplet. When $h^{2,1}=0$ the latter represents the entirety of the HM moduli space and in that case its description can be achieved much more simply (as we will see in chapter \ref{chap-NS5HM}). 

\subsection{Type IIB}
\label{4d-fields-B}
In this case, as was seen in subsection \ref{subsec-BosSpectrum}, the 10d NS-NS sector is the same. The only difference in the 4d reduction is that we denote the scalar dual to the 4d 2-form $B_2$, $\psi$ instead of $\sigma$. 

For the RR sector we have the 0-form $\hat A_0$ which gives a scalar $c^0$ and the two form $\hat A_2$ which gives a 4d 2-form $A_2$ that can be dualized into a scalar $c_0$ and $h^{1,1}$ scalars $c^i$, whereas the 4-form requires a little bit more attention. Due to the self-duality condition on its field strength, it only gives rise to $h^{1,1}$ 2-forms $D_2^i$ and $1+h^{2,1}$ vectors $A_1^\Lambda$. After dualizing the 2-forms into scalars $c_i$, we rearrange the fields into $N=2$ multiplets 
\begin{align*}
	&\text{gravitational multiplet} &&(g_{\mu \nu}, A_1^0) \\
	&\text{(universal) hypermultiplet} &&(\psi, c_0, \phi, c^0) \\
	&h^{1,1} \text{ hypermultiplets}&&( b^i+\I t^i, c^i, c_i) \\
	&h^{2,1} \text{ vector multiplets}&&(A_1^a, z^a).
\end{align*}

This concludes our exposition of the different multiplets appearing in type IIA and type IIB. We can see that the spectrum is always made of one gravitational multiplet coupled to two types of matter multiplets: vector multiplets whose bosonic sector contains a gauge field and a complex scalar, and hypermultiplets each having 4 real scalars.

We summarize the dependence of the number of HM and VM on the Hodge numbers: there are $h^{1,1}$ vector multiplets in type IIA and $h^{2,1}$ in type IIB. The number of hypermultiplets in type IIA is given by $1+h^{2,1}$ while in type IIB it is given by $1+h^{1,1}$. This is in perfect agreement with the statement of mirror symmetry which exchanges the Hodge numbers of mirror Calabi-Yau manifolds. Since each hypermultiplet has four real degrees of freedom and each vector multiplet has two, we get the dimension of each moduli space by multiplying the corresponding number of multiplets by four and two respectively.

\section{Low energy effective action}
\label{sec-string-EFT}

The action of $N=2$ supergravity coupled to $n_V$ vector multiplets and $n_H$ hypermultiplets is, to a large extent, fixed by supersymmetry. When we restrict to the lowest order in $\alpha'$, in the Einstein frame, the bosonic part of the action reads 
\be
S_{\rm eff}=\hf \int R\star1 + S_{\rm VM}+S_{\rm HM}.
\ee
We see that at this level, the VM and HM moduli spaces are decoupled and we can look at the action associated to each of them separately. 

For the vector multiplet, which is a special \kahler space, the effective action is completely fixed by its holomorphic prepotential $F$ and the gauge fields \cite{Cecotti:1989qn,deWit:1984pk}.
Whereas, for the hypermultiplet, the effective action is given by the simple non-linear $\sigma$-model 
\be
\label{Seff-HM}
S_{\rm HM} = \int \de^4 x\, g_{\alpha \beta}(q) \partial^\mu q^\alpha \partial_\mu q^\beta,
\ee
where $q^\alpha$ denote the $4n_H$ scalars parametrizing the HM moduli space. The metric $g_{\alpha \beta}$ in this case is required by supersymmetry to be that of a quaternion-K\"ahler manifold \cite{Bagger:1983tt}, which is notoriously difficult to describe. This is in part why the hypermultiplet is the main interest of this thesis. 

In conclusion, the low energy effective theory is completely determined by the metric on the factorized moduli space 
\be
\label{fact-VM-HM}
\cM_{\rm 4d} = \cM_{\rm VM} \times \cM_{\rm HM},
\ee
where the first factor is required to be special K\"ahler and the second quaternion-K\"ahler. This follows entirely from $N=2$ supersymmetry and is thus expected to hold at the quantum level. 

The prepotential governing the geometry of the vector multiplet moduli space is generally tree level exact and is well known. For type IIB theory, it is also free of $\alpha'$-corrections and is therefore given by \eqref{prepot-cplx}. Whereas for type IIA the relevant function agrees with \eqref{prepot-Kahler} only in the large volume approximation. However, the physically relevant prepotential can be obtained from the one in type IIB using mirror symmetry \cite{Candelas:1990rm,Hosono:1993qy}, which maps the classical result of type IIB into perturbative and non-perturbative $\alpha'$ corrections in type IIA. Indeed, this symmetry implies that the moduli spaces of type IIA and type IIB compactified on mirror CY manifolds coincide and due to the factorization \eqref{fact-VM-HM} this implies 
\be
\qquad \cM_{\rm VM}^{\rm A/B}(\CY) = \cM_{\rm VM}^{\rm B/A}(\hat \CY),  \qquad \cM_{\rm HM}^{\rm A/B}(\CY) = \cM_{\rm HM}^{\rm B/A}(\hat \CY).
\ee
In conclusion, the situation for the vector multiplet moduli space is well under control.

In contrast, the hypermultiplet sector is much more complicated. There are two main sources for this difficulty. First, the geometry of the moduli space is quaternion-K\"ahler and there is, in general, no analog to the holomorphic prepotential that we used in the VM case. Second, the metric on this space gets all types of $g_s$ corrections both perturbative and non-perturbative. The latter are especially complicated since the rules of the string instanton calculus are not established for NS5-instantons so that the straightforward microscopic calculation cannot be performed. For the reasons above we will use twistor space techniques to get a better description of the HM moduli space.

\section{Black holes and BPS states}
\label{sec-string-BHBPS}
It has long been known in $N=2$ theories that BPS states provide important contributions to the non-perturbative sector. From understanding the microscopic origin of the Bekenstein-Hawking entropy of certain supersymmetric black holes to computing instanton contributions to the low energy effective theory of type II string theory compactified on a CY, the study of the BPS spectrum is key. 

In the context of $N=2$ theories, a BPS state is defined as a "short" representation of the supersymmetry algebra with an associated charge $\gamma$ under the gauge group of the theory. In fact it is annihilated by half of the super-generators and thus preserves half the supersymmetry. This property endows these states with stability against decays into non-BPS representations. Furthermore, they are generally invariant under continuous deformations of parameters. Therefore, one can study such states at weak coupling and deduce information about them at strong coupling. However, for some moduli this is not the case and we know that the BPS spectrum changes discontinuously across one-dimensional \textit{walls of marginal stability}, which will be defined later. 

In type II string theory D-branes wrapping non-contractible, calibrated cycles on the CY give BPS states charged under the 4d gauge fields discussed in subsections \ref{4d-fields-A} and \ref{4d-fields-B}. 
The charge vector $\gamma=(p^\Lambda,q_\Lambda)$ is an element of a lattice with Dirac-Schwinger-Zwanziger product $\gamma_{12}=\langle \gamma_1,\gamma_2\rangle$ that will be useful later. Additionally, its components $p^\Lambda,q_\Lambda$ give the decomposition of the cycle being wrapped by the brane in a certain basis of $H_3(\CY,\IZ)$ or $H_{\rm even}(\CY,\IZ)$. The BPS condition can then be written using the \textit{central charge} which defines a central extension of the supersymmetry algebra. This function $Z(\gamma;u)$ is linear in charges and holomorphic in the moduli $u$. In general, all states satisfy the inequality 
\be
\label{BPS-mass-cond}
M(\gamma;u)\geq |Z(\gamma;u)|,
\ee
and BPS states are those that saturate the bound. If a BPS state exists for a given moduli, we expect it to continue to exist at least in a neighborhood around it. 

Let us denote $\cH_{\gamma,u}^{\rm BPS}$ the Hilbert space of BPS states with charge $\gamma$ and at moduli $u$. The dimension of this space is finite and we can define on it an index that counts the (signed) number of BPS multiplets \cite[eqn (3.21)]{Alexandrov:2020qpb}
\be
\label{BPS-Trace}
\Omega(\gamma;u)=\,\Tr_{\cH'_{\gamma,u}}\[(-1)^{2J_3}\] \in\IZ,
\ee
where $J_3$ is the generator of rotations along a given axis and $\cH'_{\gamma,u}$ is the Hilbert space of BPS states with charge $\gamma$ and we factorized the center of mass degrees of freedom.
It conjecturally coincides with the so-called generalized Donaldson-Thomas (DT) invariants. These are topological invariants defined for a given Calabi-Yau and they foster a lot of interest by mathematicians and are quite hard to compute. The integral nature of these DT invariants is conjectural. It is supported however by the fact that all the examples that have been computed are integer.

We will start by looking at some examples of BPS states arising in type II string theories compactified on a Calabi-Yau $\CY$. There, we will also look at the role played by mirror symmetry in relating these different examples. Next, we talk about BPS bound states and wall crossing. And finally we will talk about refined BPS indices.

\subsection{BPS instantons and black holes}
\label{subsec-string-BPSBH}
The first example we will look at gives instanton contributions to the hypermultiplet moduli space metric, in which BPS indices appear as weights. This is relevant for chapters \ref{chap-NS5HM} and \ref{chap-qRH}.

In type IIA these states arise as even Euclidean D$p$-branes wrapping $p+1$ cycles. Obviously, since the compactification manifold is a Calabi-Yau, the only admissible case is $p=2$. So we get BPS states by wrapping D2-branes on special Lagrangian 3-cycles on $\CY$. By decomposing this cycle $\Gamma\in H_3(\CY,\IZ)$ in the basis of \eqref{decomp-H3} as $q_\Lambda \cA^\Lambda - p^\Lambda \cB_\Lambda$ we get the electromagnetic charge vector $\gamma=(p^\Lambda,q_\Lambda)$, with $\Lambda=0,\dots,h^{2,1}$, associated to the BPS state. These states, from the 4d low energy effective theory, are seen as points, hence they are instantons. 

In type IIB they come from odd Euclidean D$p$-branes wrapping even cycles. There, all the values $p=-1,1,3,5$ are admissible and the BPS state is given by any bound state of branes wrapping holomorphic even cycles on the $\CY$, $\Gamma\in H_{\rm even}$. The dimension of this group is $2(1+h^{1,1})$ and we can use the basis of even cycles $(\gamma_I,\gamma^I)$ defined below \eqref{Heven-basis-wedge} with $I=0,\dots,h^{1,1}$. We can see then that the charge associated to the branes should be $\gamma=(p^I, q_I)$.

The Hodge numbers giving the number of components of the charge vectors in both type IIA and type IIB point to the fact that these instanton contributions are mapped to each other under mirror symmetry. Indeed $h^{2,1}$ and $h^{1,1}$ are exchanged under this symmetry and more importantly the hypermultiplet moduli space of type IIA on $\CY$ and that of type IIB on $\hat \CY$ are equivalent and they should therefore receive the same instanton corrections.

There is another possibility in both type IIA and type IIB of getting instanton contributions from BPS states. It is given by NS5-branes wrapping the whole Calabi-Yau. This can form bound states with the other D-branes mentioned above and we just add an integer $k$ that denotes the NS5-charge, to $\gamma$. Computing these contributions is the main goal of chapter \ref{chap-NS5HM}.

The second example we will look at gives particle states that are extended in time. When the string coupling constant $g_s\gg1$ they can be interpreted as BPS black holes. This is relevant for chapter \ref{chap-DTall} and we use notations adapted to the corresponding paper.

In type IIA they arise as bound states of D6-D4-D2-D0 branes wrapping even cycles on the $\CY$. In a similar fashion to the even cycles encountered in type IIB above, we have a charge vector associated to the BPS state $\gamma=(p^\Lambda,q_\Lambda)$ with $2(h^{1,1}+1)$ components. The indices counting the number of these states are $\Omega(\gamma;z^a\equiv b^a+\I t^a)$ and depend on the K\"ahler moduli. When these states are interpreted as black hole solutions to the supergravity low energy theory, $z^a$ give the value of the moduli at spacial infinity.

\subsection{Wall crossing}
The BPS indices exhibit a very interesting behavior called \textit{wall-crossing} \cite{Denef:2000nb}. In fact, despite being protected against decay into non-BPS states, a BPS state can still decay into two or more BPS states. But there are stringent conditions on such a decay. 

First of all, for an irreducible charge vector there is no such wall-crossing behavior since it cannot decay into lighter constituents. Let's look at the next-simplest case, namely a charge $\gamma=\gamma_1+\gamma_2$ that can be decomposed into two primitive charge vectors. In some region of the moduli space the bound state is energetically favored and supersymmetric, therefore, it can be formed. The radius of this configuration is 
\be
\label{Radius-BPS12}
R_{12}=\hf \langle \gamma_1,\gamma_2\rangle \, \frac{|Z(\gamma_1;z^a)+Z(\gamma_2;z^a)|}{\Im\(Z(\gamma_1;z^a)\overline{Z(\gamma_2;z^a)}\)},
\ee
and can be obtained through supersymmetry constraints \cite{Denef:2000nb} as well as through solving the supergravity equations \cite{Bates:2003vx}. One can see that a necessary condition for the existence of such a configuration is that $R_{12}$ has to be positive, hence we get the Denef stability condition 
\be
\langle \gamma_1,\gamma_2\rangle\Im\(Z(\gamma_1;z^a)\overline{Z(\gamma_2;z^a)}\) >0.
\ee

In general, such a bound state does not decay, because the binding energy is negative. However, for special values of the moduli, the central charges of the two constituents align
\be
Z(\gamma_1,z^a)\,\overline{Z(\gamma_2,z^a)}\in \IR_+,
\ee 
and the binding energy becomes null. This condition defines the above-mentioned walls of marginal stability which separate the region where the configuration is stable and the one where it is not. This process can be seen by looking at how the radius \eqref{Radius-BPS12} varies when we approach such a wall. Indeed, the imaginary part in the denominator approaches zero and the radius therefore goes to $\infty$ signaling that the two particles are being separated.

In fact we can find precisely how many BPS states of charge $\gamma$ will disappear from the spectrum by using the primitive wall-crossing formula \cite{Denef:2007vg}
\be
\label{primitive-wc}
\Delta \Omega(\gamma,z^a)=(-1)^{\gamma_{12}}\gamma_{12} \, \Omega(\gamma_1,z^a)\Omega(\gamma_2,z^a),
\ee
where $\gamma_{12}=\langle \gamma_1,\gamma_2\rangle$ and the difference $\Delta \Omega$ is between the index in the stable region and the index in the unstable. 
More generally, increasingly complicated decays can happen at the wall of marginal stability when the total charge $\gamma$ can be decomposed in multiple ways. And it is remarkable that a formula allowing to compute BPS indices on one side of any wall in terms of the BPS indices on the other exists \cite{ks}. This is the Kontsevich-Soibelman wall crossing formula (KSWCF) and it is written elegantly in terms of an equality of a (possibly infinite) product of operators on each side of the wall.

\subsection{Refined BPS indices}
In this subsection we introduce the refined BPS indices. 

First, note that mathematically the unrefined BPS indices $\Omega(\gamma;z^a)$, introduced physically in \eqref{BPS-Trace}, are (roughly) defined as the Euler numbers of the moduli space $\cM_{\gamma,z^a}$ of semi-stable coherent sheaves 
\be
\label{BPS-Euler-Sheaves}
\Omega(\gamma,z^a)=\sum_{p=0}^{2d_\C(\cM_{\gamma,z^a})}(-1)^{p-d_\IC(\cM_{\gamma,z^a})}b_{p}(\cM_{\gamma,z^a}),
\ee
where $d_\C(\cM)$ is the complex dimension of $\cM$.
Their refined counterparts, are defined instead as Laurent polynomials giving access to the Betti numbers of $\cM_{\gamma,z^a}$
\be
\label{refined-BPS}
\Omega(\gamma,z^a,y)=P(\cM_{\gamma,z^a},-y),
\ee
where 
\be
\label{refBPS-Laurent-pol}
P(\cM,y)=\sum_{p=0}^{2d_\C(\cM)}y^{p-d_\IC(\cM)}b_{p}(\cM).
\ee

Physically, the refinement corresponds to turning on an $\Omega$-background \cite{Moore:1998et,Nekrasov:2002qd} and is realized through
adding a new parameter $y$ which may be thought of as a fugacity conjugate to the angular momentum $J^3$ carried by the BPS state in 4d. In this way, refined BPS indices contain more information about the spectrum by being sensitive to the angular momentum $J^3$. 
They can be obtained by replacing the formula \eqref{BPS-Trace} with
\be
\label{BPS-refined-Trace}
\Omega(\gamma,z^a;y)=\,\Tr_{\cH'_{\gamma,u}}\[(-y)^{2J_3}\] \in\IZ[y,y^{-1}].
\ee

Notice that in both definitions, \eqref{BPS-refined-Trace} and \eqref{refBPS-Laurent-pol}, we recover the unrefined invariants when we take $y\to1$. 

In fact, the equivalence of these two equations is conjectural\footnote{The status is the same for unrefined BPS indices.}.
The refined indices verify refined wall-crossing relations \cite{ks}. However, they have a serious problem. Namely, they turn out to be not protected by supersymmetry. This means that they can change under variation of the complex structure moduli of the hypermultiplet moduli space. The situation is resolved for some cases, namely for non-compact Calabi-Yau spaces where a $\IC^\times$ action was used to modify \eqref{BPS-refined-Trace} and get a definition that is protected by supersymmetry \cite{Gaiotto:2010be}. 
Unfortunately, for the cases which we are interested in, such as compact CYs, we do not have a deformation-invariant definition yet.

\chapter{Twistorial description}
\label{chap-twist}

This chapter presents the main technical details about twistor spaces of quaternion-K\"ahler manifolds. This is directly relevant for chapters \ref{chap-NS5HM} and \ref{chap-qRH}. 

Understanding the hypermultiplet moduli space of type II string theory compactified on a Calabi-Yau is the motivation for most of our work during this PhD. The geometry of this moduli space is quaternion-K\"ahler, which makes it very difficult to study. Although it is very rich and receives $g_s$ corrections stemming from one-loop, D-instantons and NS5-instantons, it is highly challenging to compute its metric directly. However, quaternion-K\"ahler spaces are in one-to-one correspondence with twistor spaces and one can retrieve all information on the former using the latter. In this chapter we show how this correspondence can be achieved and we present the main twistorial objects that will be useful later. Our discussion closely follows \cite{Alexandrov:2011va}.

\section{Quaternion-K\"ahler spaces}
\label{sec-twist-qk}
A quaternion-K\"ahler manifold is a $4n$ (real) dimensional manifold with a Riemannian metric $g_\cQ$ and holonomy group contained in $Sp(n)\times SU(2)$ \cite{quatman}. 

Unlike the name suggests, in general, such a manifold is not K\"ahler and does not have an integrable almost complex structure. 
Instead, it has a triplet of locally defined almost complex structures $J^i,\, i=1,2,3$ satisfying the quaternionic algebra 
\be
\label{qk-algebra}
J^i \,J^j =\varepsilon^{ijk} J^k-\delta^{ij},
\ee
with $\varepsilon^{ijk}$ the Levi-Civita symbol. These define the vector of quaternionic 2-forms 
\be
\vec\omega_\cQ(X,Y)=g_\cQ(\vec{J}X,Y),
\ee
which are locally defined, non-closed analogs of K\"ahler 2-forms. If we denote $\vec{p}$ the $SU(2)$ part of the Levi-Civita connection associated to $g_\cQ$, then $\vec \omega_\cQ$ is covariantly closed with respect to it 
\be
\label{closed2forms}
\de \vec \omega_\cQ + \vec p \times \vec \omega_\cQ = 0,
\ee
where we used the notation $(\vec a \times \vec b)^i = \varepsilon^{ijk} a^j \wedge b^k $. Furthermore, it is proportional to the $SU(2)$ part of the curvature
\be
\label{omega-curv}
\de \vec p + \hf\, \vec p \times \vec p = \frac{\nu}{2}\, \vec \omega_\cQ,
\ee  
where the proportionality coefficient is related to the Ricci scalar curvature through $R=4n(n+2) \nu$. In particular, when $\nu\to0$, $\cQ$ becomes hyperk\"ahler, but we are mostly interested in the case of negative curvature, $\nu<0$. 

The equation \eqref{omega-curv} allows to compute $\vec \omega_\cQ$ and thus $g_\cQ$ from $\vec p$. However, direct computations of any of these quantities is very difficult. The alternative is to go to the twistor space of $\cQ$ which can be described by a set of holomorphic functions.

\subsection{Four dimensional case}
\label{subsec-twist-Prz}

Before presenting the actual twistor space technique for \qk spaces, let's present another characterization relevant for the four dimensional case. This case can be dealt with quite nicely using the Przanowski description. In chapter \ref{chap-NS5HM} we perform a cross-check where we compare the metric computed from the latter description with the one computed using the twistor space. Therefore, it is relevant to explain the Przanowski description already in the present section. 

In \cite{Przanowski:1984qq} it was proven that locally, on a four-dimensional self-dual Einstein space
(which is a characterization of QK geometry in four dimensions) with a negative curvature,
one can always find complex coordinates $z^\alpha$ ($\alpha=1,2$) such that the metric takes the following form
\be
\de s^2_z[h] = -\frac{6}{\Lambda}\(h_{\alpha \bar \beta}\de z^\alpha \de \bz^\beta + 2 e^h |\de z^2|^2\),
\label{metPrz}
\ee
where $h_{\alpha}= {\partial h}/{\partial z^\alpha}$, etc. This metric is completely determined by
a single real function $h(z^\alpha, \bz^\alpha)$ that must satisfy the following non-linear partial differential equation
\bea
\Prz_z[h]\equiv h_{1\bar1}h_{2\bar2} - h_{1\bar2}h_{\bar1 2}+\(2h_{1\bar1}-h_{1}h_{\bar1}\)e^{h}=0.
\label{Prz-master-equation}
\eea

Of course, the equation is too complicated to be solved in general. However, the problem significantly simplifies
if one already knows a solution $\hpert$ describing some QK space and one is interested in linear deformations of this space.
The point is that such deformations are governed by the {\it linearization} of \eqref{Prz-master-equation}
around $\hpert$ \cite{Alexandrov:2006hx,Alexandrov:2009vj}
\be
\bigl(\Delta+1\bigr)\( r^2\delta h\)=0,
\label{linPrz}
\ee
where $\Delta$ is the Laplace-Beltrami differential operator defined by the perturbative metric and we expanded $h=\hpert+\delta h$ keeping only terms linear in $\delta h$. And this linear equation is much easier to solve. In the case we are interested in, the perturbative metric is given in \eqref{1lmetric} and defines the differential operator
\be
\Delta=\!
\frac{r^2}{r+2c}\[
(r+c)\p_r^2+\frac{16(r+2c)^2}{r+c}\, \p_{\sigma}^2
+\frac{2}{\tau_2}\left|(\tzeta+\tau\zeta)\p_\sigma+\tau\p_{\tzeta}-\p_\zeta\right|^2
\!\!-\frac{r+2c}{r}\, \p_r
\].
\label{Laplacereal}
\ee

\section{The twistor space}
\label{sec-twistor}
There is a very natural way to define the twistor space. One starts with the triplet $J^i$ on $\cQ$ which generate a two-sphere of almost complex structures (on $\cQ$) obtained as a normalized linear combination of the components of $\vec J$
\be
J(t,\bar t) = \frac{1-t\bar t}{1+ t \bar t} \, J^3+ \frac{t+\bar t}{1+t\bar t} \, J^2 + \frac{\I (\bar t -t)}{1+t\bar t} \, J^1, \qquad t\in \CP,
\ee
parametrized by $t\in \CP$. The twistor space $\cZ_\cQ$ is simply the canonical $\CP$-bundle over $\cQ$ with $t$ as the fiber coordinate. A point on $\cZ_\cQ$ is essentially a point on $\cQ$ with a choice of almost complex structure. 

The resulting bundle is a $2n+1$ complex dimensional K\"ahler space. In the following, we will work in a local patch on the base and we will omit the corresponding coordinates in the arguments of functions.

The main characteristic of the twistor space is that it admits a \textit{holomorphic contact} structure. This means that it has a holomorphic 1-form $\cX$ such that $\cX \wedge \(\de \cX\)^n$ is a nowhere-vanishing holomorphic top-form. This contact form will be of great importance to us, so let us construct it. We start with the globally defined (1,0)-form 
\be
\label{Dt-1form}
\De t=\de t + p^{+} - \I t p^{3} + t^2 p^{-},
\ee
where $p^{\pm}=-\hf (p^1\mp p^2)$ are certain combinations of the components of the $SU(2)$ connection. Then we define the holomorphic contact 1-form as
\be
\label{contact-1form-canonical}
\cX^{\[i\]} = \frac{4}{\I t}e^{\Phi^{\[i\]}} \De t,
\ee
where $\Phi^{\[i\]}(t)$ is determined by the requirement that $\cX^{\[i\]}$ is holomorphic and is locally defined hence the patch index $\[i\]$. We call this function $\Phi^{\[i\]}$ the \textit{contact potential} and it plays a central role in the description of the twistor space. 
A crucial feature of the contact structure is that, locally and under a proper choice of coordinates, it can be written in the canonical form 
\be
\label{cX-Darboux}
\cX^{\[i\]} = \de \alpi + \xii^\Lambda \,\de \txii_\Lambda,
\ee 
where $\Lambda=0\dots n-1$. The functions on the right hand side are a set of holomorphic coordinates on $\cZ_\cQ$ and they are called \textit{Darboux coordinates}. They are extremely important in our construction as they allow to determine the geometry of the twistor space and thus of the quaternionic base. They are only defined locally and develop singularities, namely branch cuts and poles, when extended beyond the patches where they are defined.

Another important feature of the twistor space is that it has a real structure. This makes it manifest that the resulting metric is real. This structure is defined by the action of the antipodal map $\tau (t)=-1/\bar t$, and it must be compatible with the other structures of $\cZ_\cQ$. Namely, it should preserve the holomorphic contact form in a sense we will see below. 

\section{Patches and transition functions}
\label{sec-twist-transition}
In order to facilitate working with the real structure we will choose a covering $\hat \cU_i$ of $\cZ_\cQ$ that is adapted to it.
Since we work in a specific patch on the base space, we will only consider the projection of $\hat \cU_i$ on $\CP$, that we denote $\cU_i$. We require that the antipodal map preserve the covering through the relation
\be
\label{rlty-patch}
\tau(\cU_i)=\cU_{\bi}.
\ee
Then, the compatibility of the real structure with the holomorphic contact structure is expressed
on the contact 1-form as 
\be
\label{rel-cond}
\overline{\tau\(\cX^{\[i\]}\)}=\cX^{\[\bi\]},
\ee
and similarly on Darboux coordinates.

In order to go from one patch to another we apply holomorphic contact transformations to the coordinates on $\cZ_\cQ$. These are transformations which preserve the contact structure and they are generated by holomorphic transition functions $H^{\[ij\]}\(\xii,\txij,\alpj\)$. These transition functions depend on initial "position" and final "momentum" $\(\xii, \txij, \alpj\)$ coordinates, which will make the contact transformations complicated. For Darboux coordinates the gluing conditions are written as \cite{Alexandrov:2009zh}
\be
\label{glu-Darboux}
\begin{split}
	\xij^\Lambda =& \xii^\Lambda -\partial_{\txij_\Lambda}H^{\[ij\]}+\xij^\Lambda \partial_{\alpj}H^{\[ij\]},
	\\
	\txij_\Lambda =& \txii_\Lambda + \partial_{\xii^\Lambda}H^{\[ij\]},
	\\
	\alpj =& \alpi + H^{\[ij\]} - \xii^\Lambda \partial_{\xii^\Lambda}H^{\[ij\]},
\end{split}
\ee
and they imply the transformation
\be
\label{CHP3-cX-O2}
\cX^{\[i\]}=f_{ij}^2\,\cX^{\[j\]},
\ee
on the contact 1-form, where $f_{ij}=1-\partial_{\alpj}H^{\[ij\]}$. 

Up to now, we only imposed that the transition functions be holomorphic but there are other constraints they need to satisfy. First, one expects the transition function $H^{\[ji\]}$ to generate the inverse contactomorphism to $H^{\[ij\]}$. Second, on the overlap of three patches $\cU_{i}\cup \cU_j\cup \cU_k$ one should get the same result by using either the transition function $H^{\[ik\]}$ or the combination $H^{\[ij\]}$ followed by $H^{\[jk\]}$. And finally, in order to satisfy \eqref{rel-cond} for Darboux coordinates, we need to have 
\be
\label{rel-cond-Hij}
\tau\(H^{\[ij\]}\) = H^{\[\bi\bj\]}.
\ee

The transition functions play an analogous role to the potential for K\"ahler spaces. Their knowledge is enough to reconstruct the geometry of the twistor space as well as that of the base space. In fact, given a covering of $\cZ_\cQ$, the corresponding transition functions and their behaviour around $t=0$ and $t=\infty$, it is in principle possible to solve the set of equations \eqref{glu-Darboux} and thus reconstruct the Darboux coordinates. In practice, it is much better to work with the integral form of the gluing conditions
\be
\label{glu-integral-Darboux}
\begin{split}
\xii^\Lambda (t) =& A^\Lambda + t^{-1}Y^\Lambda -t \bar Y^{\Lambda} -\hf \sum_{j}\oint_{C_j}\frac{\de t'}{2\pi \I t'}\frac{t'+t}{t'-t}\(\partial_{\txij_\Lambda}H^{\[ij\]}+\xij^\Lambda \partial_{\alpj}H^{\[ij\]}\),
\\
\txii_\Lambda(t) =& B_\Lambda + \hf \sum_{j}\oint_{C_j} \frac{\de t'}{2\pi \I t'}\frac{t'+t}{t'-t} \partial_{\xii^\Lambda}H^{\[ij\]},
\\
\alpi            =& B_\alpha +\hf \sum_{j}\oint_{C_j} \frac{\de t'}{2\pi \I t'}\frac{t'+t}{t'-t} (H^{\[ij\]}-\xii^\Lambda\partial_{\xii^\Lambda}H^{\[ij\]}) - 2\I c_\alpha \log t,
\end{split}
\ee
where $C_j$ is a contour going around $\cU_j$ in the counter-clockwise direction and the sum goes over all patches $\cU_j$. The real parameters $A^\Lambda,B_\Lambda,B_\alpha$ and complex $Y^\Lambda$ parametrize the $4n+1$ dimensional space of solutions.
They can be viewed as coordinates on the $4n$-dimensional base $\cQ$ after absorbing the extra degree of freedom by rotating $t$ such that $Y^0$ becomes real. The numerical parameter $c_\alpha$ is part of the so-called \textit{anomalous dimensions}\footnote{Anomalous dimensions originate from asymptotic conditions on the Darboux coordinates. For a complete discussion see \cite{Alexandrov:2008nk}.}
that modify the Darboux coordinates. These account for perturbative corrections to the metric $g_\cQ$ and in the setup of chapter \ref{chap-NS5HM} the $c_\alpha$ contribution will be enough to describe all of them \cite{Alexandrov:2008nk,Robles-Llana:2006ez}. Although these integral equations are very difficult to solve in general, one can follow an iterative approach by starting from a given expression and then computing corrections around it order by order \cite{Alexandrov:2009zh}. 

Importantly, we also have an integral representation for the contact potential
\be
\label{int-contact-pot}
\Phi^{\[i\]}(t) = \phi - \hf \sum_{j}\oint_{C_j} \frac{\de t'}{2\pi \I t'}\frac{t'+t}{t'-t} \log\(1-\partial_\alpj H^{\[ij\]}\),
\ee
wher $\phi$ is the $t$-independent part and it reads
\be
\label{contact-tindep}
\phi= \frac{\frac{1}{8\pi}\sum_j \oint_{C_j}\frac{\de t }{t}\(t^{-1}Y^\Lambda -t\bar Y^\Lambda\)\partial_{\xii^\Lambda}H^{\[ij\]} +c_\alpha}{2\cos\[\frac{1}{4\pi}\sum_j \oint_{C_j}\frac{\de t}{t} \log\(1-\partial_\alpj H^{\[ij\]}\)\]}.
\ee

In conclusion, the twistor data is given by open patches $\cU_i$ and a minimal set of transition functions between them $H^{\[ij\]}$, that verify some constraints. However, there is a potential problem when Darboux coordinates develop branch cut singularities that extend beyond multiple patches. In that case, the naive contours going around a single patch $\cU_i$ are not closed anymore and we need to devise complicated alternatives. Luckily, in \cite{Alexandrov:2008ds} it was shown that the formulas \eqref{glu-integral-Darboux} remain valid using the (now open) contours $C_i$ going around $\cU_i$.

\section{Procedure to get the metric}
\label{sec-twist-metric}
The general procedure to derive the metric on a QK manifold $\cM$ from the knowledge of Darboux coordinates \eqref{glu-integral-Darboux} and the contact potential \eqref{int-contact-pot} on its twistor space $\cZ_\cM$ was described in detail in \cite{Alexandrov:2008nk,MR1001707}.
Here we present it in the form adapted to the twistor description of $\cM_H$ given in the previous subsection.
It consists of several steps:
\begin{enumerate}
	\item
	At the first step, we redefine the coordinates $\txi_\Lambda$ and $\talp$ into\footnote{The transformation \eqref{redefDC}
		can be seen as a contact transformation defining Darboux coordinates in the patch around $t=0$.}
	\be
	\begin{split}
		\txi^{\[+\]}_\Lambda=&\, \txi_\Lambda-F_\Lambda(\xi),
		\\
		\ai{+}=&\,-\hf\( \talp +\xi^\Lambda\txi_\Lambda\)+F(\xi).
	\end{split}
	\label{redefDC}
	\ee
	The advantage of the new coordinates is that their expansion around the north pole $t=0$ of $\CP$ does not
	contain singular $t^{-1}$ terms.
	As a result, we can write the following Laurent expansion for the set of holomorphic coordinates on
	the twistor space and the contact potential \eqref{eqchip}
	\be
	\begin{split}
		\xi^\Lambda=&\,\xi^\Lambda_{-1}t^{-1}+\xi^\Lambda_0+\xi^\Lambda_1 t+O(t^2),
		\\
		\txi^{\[+\]}_\Lambda=&\,\txi^{\[+\]}_{\Lambda,0}+\txi^{\[+\]}_{\Lambda,1}\,t+O(t^2),
		\\
		\ai{+}=&\,4\I c\log t+\ai{+}_{0}+\ai{+}_1 t+O(t^2),
		\\
		\Phi=&\, \Phi_0+\Phi_1 t +O(t^2).
	\end{split}
	\label{expDc}
	\ee

	\item
	In terms of the new coordinates, the contact one-form \eqref{cX-Darboux} is given by
	\be
	\cX=\de\ai{+}+\xi^\Lambda\de\txi^{\[+\]}_\Lambda.
	\label{cXpole}
	\ee
	Substituting the expansions \eqref{expDc} into this expression
	and comparing it with the canonical form $Dt$ \eqref{Dt-1form} using \eqref{contact-1form-canonical}, one finds the components of
	the SU(2) connection
	\be
	\begin{split}
		p^+ &=\frac{\I}{4}\, e^{-\Phi_0}
		\, \xi^\Lambda_{-1}  \de\txi^{[+]}_{\Lambda,0}\, ,
		\\
		p^3 &= -\frac{1}{4}\, e^{-\Phi_0} \left( \de\alpha^{[+]}_0 +
		\xi^\Lambda_0   \de\txi^{[+]}_{\Lambda,0} +
		\xi^\Lambda_{-1}   \de\txi^{[+]}_{\Lambda,1}  \right) -\I \Phi_1 p^+\, .
	\end{split}
	\label{connection}
	\ee

	\item
	Then one computes the triplet of quaternionic 2-forms \eqref{omega-curv} with $\nu=-8$.
	In particular, for $\omega^3$ the formula reads
	\be
	\omega^3 = -4\({\rm d} p^3-2\I  p^+ \wedge p^-\).
	\label{Kform}
	\ee
	
	\item
	Next, one specifies the almost complex structure $J^3$ by providing a basis of (1,0) forms on $\cM$.
	Such a basis was found in \cite{Alexandrov:2008nk} and, after some simplifications, it takes the following form
	\be
	\label{defPi}
	\pi^a =\de \(\xi^a_{-1} /\xi^0_{-1} \) ,
	\qquad
	\tilde\pi_\Lambda= \de\txi^{\[+\]}_{\Lambda,0}  ,
	\qquad
	\tilde\pi_\alpha = \frac{1}{2\I}\,\de\ai{+}_0 +2c \,\de\log\xi^0_{-1} .
	\ee
	
	\item
	Finally, the metric is recovered as $g(X,Y) = \omega^3(X,J^3 Y)$.
	To do this in practice, one should rewrite $\omega^3$, computed by \eqref{Kform} in terms of differentials of
	coordinates on $\cM$, in the form which makes explicit that it is of (1,1) Dolbeault type.
	Using for this purpose a basis $\pi^X$, which can be taken to be $(\pi^a,\tilde\pi_\Lambda,\tilde\pi_\alpha)$, the final result should look like
	\be
	\omega^3=2\I g_{X\bY} \pi^X\wedge \bar\pi^{Y},
	\label{metom}
	\ee
	from which the metric readily follows as $\de s^2 =2 g_{X\bY} \pi^X \otimes \bar\pi^{Y}$.
	
\end{enumerate}

\chapter{Modular Forms}
\label{chap-mod}

Modular forms are holomorphic functions on the upper half-plane $\IH$ which satisfy a certain transformation property under the modular group $\SL$. In mathematics, their importance can hardly be overstated and it is nicely illustrated by the following quote attributed to Martin Eichler (1912-1992) : "There are five fundamental operations in mathematics: addition, subtraction, multiplication, division, and modular forms." They are also very important in string theory as they emerge from weak-strong dualities like S-duality in type IIB.

Every time modularity appears in physics should be seen as a great opportunity to get valuable information. Indeed, the theory of modular forms is very rigid and thus allows to determine exact expressions for functions with minimal initial data. In general, to fix a modular form we only need to know a finite number of its Fourier coefficients. Furthermore, modular forms have stringent conditions on the asymptotic growth of their coefficients which is often very useful in physics, for example in considerations about entropy. 

More broadly, there are numerous generalizations to modular forms and we will consider some of them in our work. These related notions offer the double advantage of maintaining the rigid structure of modularity while dealing with a bigger set of objects. In this chapter we will first introduce and define the standard notion of modularity. Then, we will introduce \textit{mock} modular forms that first appeared in Ramanujan's last letter to Hardy. Schematically, these are functions with an anomalous modular transformation that can be cured, through a specific procedure, at the price of introducing a holomorphic anomaly. Afterwards, we introduce \textit{Jacobi}, \textit{Jacobi-like} and \textit{mock Jacobi} forms. Finally, we provide a standard way to construct these objects using theta series and generalized error functions.

\section{Modular forms and Vector-valued modular forms}
\label{sec-mod-VV}
\subsection{Definition and examples}
There are two key objects essential for understanding and defining modular forms. First, the upper half-plane
\be
\label{def-upper-halfplane}
\IH=\{\tau\in\IC| \Im \tau>0\},
\ee
which contains all complex numbers with strictly positive imaginary part. A standard practice is to denote $\tau_2\equiv \Im \tau$ for $\tau\in\IH$. On this space we can define an action of the modular group
\be
\SL = \left\{g=\begin{pmatrix}
	a& b\\ c& d
\end{pmatrix} \bigg| a,b,c,d\in \IZ,  ad-bc=1 \right\}.
\ee
An important fact about this group is that it is generated by two transformations, namely the translation and inversion matrices given by
\be
\label{mod-TS}
T=\begin{pmatrix}
	1& 1\\ 0& 1
\end{pmatrix}
,\qquad 
S=\begin{pmatrix}
	0 & -1\\1& 0 
\end{pmatrix}.
\ee
The action of $\SL$ on the upper half-plane can be defined through
\be
\tau \to g(\tau)=\frac{a\tau+b}{c\tau+d}. 
\ee
We can then do a quick computation to show that $g(\tau)$ is indeed in $\IH$. We start with 
\be
	\Im \frac{a\tau+b}{c\tau+d}=\Im \frac{(a\tau+b)(c\btau+d)}{|c\tau+d|^2}
	= \Im \frac{ac|\tau|^2+bd+ad\tau+bc\btau}{|c\tau+d|^2},
\ee
where we multiplied by the conjugate of the denominator and expanded the numerator. Then, we drop the real terms and remain with
\be
\Im \frac{a\tau+b}{c\tau+d}
= \frac{(ad-bc)\Im \tau}{|c\tau+d|^2}
=\frac{\Im\tau}{|c\tau+d|^2}>0,
\ee
where we used the fact that $\det g= ad-bc=1$.

Now we get to the main definition. A modular form of weight $k\in \IZ$ is a holomorphic function on $\IH$, bounded when $\Im\tau\to\infty$, that transforms in the following way 
\be
\label{mod-transformation}
f\(\frac{a\tau+b}{c\tau+d}\)=(c\tau+d)^k f(\tau),
\ee
for any element $g\in\SL$. It follows immediately that a modular form is necessarily $1$-periodic and thus can be written as a Fourier series
\be
\label{modular-Fourier}
f(\tau) = \sum_{n=0}^{\infty}a_n \q^n,
\ee
where we used the standard notation $\q=e^{2\pi \I \tau}$ verifying $|q| <1$. 

We can see that constant functions are modular forms of weight $0$ and the zero function is a modular form for any weight.
A natural question to ask, however, is whether for a given weight $k$ a non-trivial modular form exists. Let us denote $\cM_{k}$ the set of modular forms of weight $k$. 
First, it is easy to see, by applying the transformation $\begin{pmatrix}
	-1&0\\0&-1
\end{pmatrix}$, that the set is trivial for odd weight $\cM_{2k+1}=\{0\}$. On the other hand, when $k<0$ there are also no modular forms in $\cM_k$ since their transformation property contradicts the condition that they are bounded when $\tau_2\to\infty$. Finally, when $k\geq0$ and even, a less trivial but well-known result is that the dimension of the vector-space $\cM_k$ is always finite and given by 
\be
\label{dim-modForms}
\dim \cM_{k}=
\begin{cases}
	\lfloor k/12\rfloor\, \qquad\qquad &{\rm if }\, k\equiv 2 \mod 12, \\
	\lfloor k/12\rfloor+1\, \qquad  &{\rm if }\, k\not\equiv 2 \mod 12.
\end{cases}
\ee
This already gives a peek into the the power of modular forms. We only need to compute a few Fourier coefficients \eqref{modular-Fourier} to determine completely a function that transforms with a given weight.

A family of important examples of modular forms, for any admissible weight, are the Eisenstein series $E_k$\footnote{There are numerous different normalization for the Eisentein series. We choose the one where the constant coefficient is normalized to $1$.}. 
For even $k\geq4$, they are defined as
\be
\label{mod-Eisenstein}
E_k(\tau)=1-\frac{2k}{B_k}\sum_{n=1}^{\infty}\sigma_{k-1}(n)\q^n,
\ee
where the $B_k$ appearing in the denominator are the Bernoulli numbers given by
\be
\label{mod-Bernoulli-nb}
\frac{t}{\exp(t)-1}=\sum_{k=0}^{\infty} \frac{B_k}{k!}t^k.
\ee
and $\sigma_k(n)=\sum_{\substack{d|n\\d>0}}d^k$ is the sum of divisors of $n$ to the power $k$. The function $E_k$ is a modular form of weight $k$.

The family of functions $E_k$ can actually be extended with the interesting case $k=2$ of the second Eisenstein series 
\be
\label{E2-def}
E_2(\tau)=1-24\sum_{n=1}^{\infty}\sigma_1(n)\q^n.
\ee
This function is not a modular form, but it is a \textit{quasi-modular} form. It has an anomalous transformation 
\be
\label{E2-anom}
E_2\(\frac{a\tau+b}{c\tau+d}\)=(c\tau+d)^2\(E_2(\tau)+\frac{6}{\pi\I}\frac{c}{c\tau+d}\),
\ee
that can be cured at the price of adding a non-holomorphic term. Namely,
\be
\widehat E_2(\tau,\btau)=E_2(\tau)-\frac{3}{\pi\Im(\tau)}
\ee
transforms as a modular form of weight 2. 

Two of the Eisenstein series 
\be
\begin{split}
E_4(\tau)=&1+240 q + 2160 q^2 + \dots,
\\
E_6(\tau)=& 1 - 504q -16632 q^2 + \dots,
\end{split}
\ee
play a special role in the theory of modular forms.
In fact these two elements are extremely interesting because they freely generate the full ring of modular forms $\oplus_{k=4}^{\infty}\cM_k$. Any function $f\in \cM_k$ can be written in a unique way as a sum of terms $E_4^a E_6^b$ where $4a+6b=k$. We can see that for a given $k$ the number of solutions $(a,b)$ is exactly equal to the dimension of $\cM_k$ and then it remains only to show that the different monomials $E_4^a E_6^b$ are not linearly dependent. 

The first weight for which we have more than one solution is $k=12$ which is solved by $(3,0)$ and $(0,2)$. This is directly related to a very important modular form, namely the \textit{discriminant} form $\Delta(\tau)$ defined as
\be
\label{discriminant-Delta}
\Delta(\tau)
=\frac{E_4^3-E_6^2}{1728}
=q\prod_{n=1}^{\infty}(1-q^n)^{24}.
\ee 
This is the appropriate point to say that modular forms can be defined with different growth conditions at infinity. In fact, \eqref{discriminant-Delta} is the first example of a modular form that vanishes when $\tau\to\I\infty$. It is thus called a \textit{cusp} form, and we denote the space of all cusp forms $S_k$. We can equivalently say that $S_k$ is the space of modular forms with vanishing  constant coefficient $c_0$ in \eqref{modular-Fourier}. In that case, the remaining Fourier coefficients have even more stringent growth conditions, namely the $a_n$ grow as $O(n^{k-1})$ for a standard modular form and as $O(n^{k/2})$ for cusp forms. 

Nonetheless, we are mostly interested in a bigger space of functions called \textit{weakly holomorphic} modular forms and denoted $\cMw_k$. A function $f$ is called weakly holomorphic if it transforms as a modular form and behaves as $f(\tau)=O(q^{-N})$ when $\tau\to\I\infty$, for some $N$. The Fourier expansion of such functions typically starts at $-N$
\be
f(\tau)=\sum_{n=-N}^{\infty}a_n q^n, \qquad f\in \cMw_k
\ee 
and the $a_n$ grow as $O(e^{c\sqrt{n}})$ asymptotically. Notice that we can have weakly holomorphic modular forms of negative, but still even, weight $k\in 2\IZ$.

\subsection{Vector Valued modular forms}
The first generalization we consider is when modular forms transform nicely only under a subgroup $\Gamma$ of the full $\SL$. 
We keep the transformation law \eqref{mod-transformation}, require $\Gamma$ to be of finite index\footnote{The index of a subgroup $\Gamma \subset \SL$ is the number of left (or right) cosets of $\Gamma$ in $\SL$.} 
in $\SL$ and impose further growth conditions at the cusps\footnote{The cusps can be defined as the orbit classes of the boundary $\partial \IH=\IQ\cup \{\infty\}$ under $\Gamma$. For $\Gamma=\SL$ there is only one (and we often take $\infty$ as its representative).}.
Unlike what we have seen above, this new class of modular forms for a subgroup $\Gamma$ can have odd or half-integral weight $k\in \hf\IZ$. 

However, we can trade this notion for that of \textit{Vector Valued} (VV) modular forms, which is the one we use in our work. 

Let $V$ be a vector space on $\IC$ of dimension $d$ and $f$ a vector valued function on it. We denote its components $f_\mu(\tau)$ where $\mu=1,\dots,d$. Let $M_{\mu\nu}(g)$ be a representation of the full modular group $\SL$ on $V$, we say that $f$ is a vector valued modular form of weight $k\in\hf\IZ$ and multiplier system $M_{\mu\nu}$ if we have the transformation rule
\be
f_\mu(g(\tau))=(c\tau+d)^k \sum_{\nu}M_{\mu \nu}(g) f_{\nu}(\tau),
\ee
under the full modular group. Since $\SL$ is generated by the two matrices $T,S$ \eqref{mod-TS}, it suffices to provide $M_{\mu\nu}(T),M_{\mu\nu}(S)$. The growth conditions now need to be applied only when $\tau\to\I\infty$, but on all of the components.

An interesting example of a modular form with a non-trivial multiplier system and half integer weight is the Dedekind eta function 
\be
\label{Dedekind-eta}
\eta(\tau)=q^{1/24}\prod_{n=1}^{\infty}(1-q^n),
\ee
which is a modular form of weight $1/2$ and multiplier system given by
\be
M^{(\eta)}(T)=e^{\frac{\pi \I}{12}},
\qquad
M^{(\eta)}(S)=e^{-\frac{\pi\I}{4}}.
\ee
It is related to the discriminant function \eqref{discriminant-Delta} through $\eta(\tau)=\Delta(\tau)^{1/24}$.

In what follows we will invariably call modular form any object that transforms appropriately under $\SL$, possibly with a non-trivial multiplier system. In particular, we will always allow the weight to be half integer	.

\section{Mock modular forms}
\label{sec-mod-mock}
Mock modular forms give a generalization of the notions we studied in the previous chapter while retaining a similarly rigid structure. They appear naturally in physics \cite{Alexandrov:2025sig}, including in this present work, and in many areas in mathematics. Therefore, they are a very nice tool to have.

A mock modular form is a holomorphic function that transforms in an almost modular way, it has an \textit{anomaly}. This failure to transform as in \eqref{mod-transformation} is completely determined by an ordinary holomorphic modular form called, \textit{shadow}. We will describe this more precisely in this section, drawing a distinction between \textit{pure} and \textit{mixed} mock modular forms. And then we will generalize this notion further to include \textit{higher depth} mock modular forms.

Mock modular forms can be found with all of the different growth conditions at the cusps seen above as well as with a non-trivial multipler system. This being said, we will define them in the simplest case of a trivial multiplier system and without choosing specific growth conditions. 

Let $g$ be a holomorphic modular form of weight $2-k$. An important object for what follows is the function $g^*$ defined as the former's so-called \textit{non-holomorphic Eichler integral}
\be
\label{mock-Integral-Eichler}
g^*(\tau,\btau) = (-2\pi\I)^{1-k}\int_{\btau}^{-\I \infty}\frac{\overline{g(\bz)}}{(\tau-z)^k}\de z.
\ee

A mock modular form $h$ of weight $k$ and shadow $g$ is a holomorphic function on the upper half-plane such that the combination 
\be
\label{mock-compl}
\hat{h}(\tau,\btau)=h(\tau)+g^*(\tau,\btau),
\ee
transforms according to \eqref{mod-transformation}. The above combination gives a \textit{completion} $\hat{h}$ of the mock modular form $h$. Starting from this completion one can apply the shadow operator 
\be
\label{mock-shadow-op}
(4\pi\tau_2)^k\partial_{\btau} \hat{h} = -2\pi\I \,\bar g,
\ee
where $\tau_2=\Im \tau$,
and obtain the (complex conjugate of the) shadow function $g$. 

More precisely, we just defined \textit{pure} mock modular forms. It turns out however that most known examples of pure mock modular forms have negative powers of $q$. For this reason and because many examples appearing in physics do not fit this definition, we will introduce \textit{mixed} mock modular forms\footnote{In fact in \cite{Dabholkar:2012nd} it is argued that mixed mock modular forms are the most natural objects to consider. See the discussion at the end of subsection 7.3 in that paper.}.

A mixed mock modular form $h$ is a holomorphic function such that there exist numbers $r_j$ and corresponding holomorphic modular forms $f_j, g_j$ with weights $k+r_j$ and $2+r_j$, such that the completion 
\be
\label{mock-compl-mixed}
\hat h = h+\sum_{j}f_j g_j^*,
\ee
transforms as a modular form of weight $k$ and $g_j^*$ are the non-holomorphic Eichler integrals corresponding to $g_j$. Acting with the shadow operator \eqref{mock-shadow-op} now gives
\be
\label{mock-drv-hdepth}
(4\pi\tau_2)^k\partial_{\btau} \hat{h}= -2\pi\I\sum_{j}(4\pi\tau_2)^{k+r_j}f_j \bar g_j.
\ee
From now on we take the term mock modular forms to mean this more general class of functions. 

Finally, we define a mock modular form of \textit{depth} n to be a holomorphic function whose modular completion can be constructed such that in \eqref{mock-drv-hdepth} the functions $f_j$ are replaced by completions $\hat f_j$ of mock modular forms of depth $n-1$ \cite{bringmann2017higher,Alexandrov:2025sig}.

\section{Jacobi forms}
\label{sec-mod-Jacobi}
Jacobi forms were introduced and studied by Eichler and Zagier \cite{MR781735}. In the first subsection, we start by defining the standard notion of Jacobi forms. Then, we provide some of their important properties and define mock Jacobi forms. We will not prove most of the claims, but we refer the interested reader to \cite{Dabholkar:2012nd}. Then, in subsection \ref{subsec-Jacobi-like} we will define Jacobi-like forms and prove a few propositions about them.

\subsection{Jacobi and mock Jacobi forms}
\label{subsec-Jacobi-mockJ}
Let $\vph(\tau,z)$ be a holomorphic function in $\tau,z\in\IH\times\IC$. We say that $\vph$ is a \textit{Jacobi form} of weight $k$ and index $m$ if it obeys the transformation properties,
\be
\label{Jacobi-modular-transf}
\vph\(\frac{a\tau+b}{c\tau+d},\, \frac{z}{c\tau+d}\)=(c\tau+d)^k e^{\frac{2\pi \I m c z^2}{c\tau+d}}\vph(\tau,z)
\ee
for any element of $\SL$, and 
\be
\label{Jacobi-elliptic-transf}
\vph(\tau, z+\lambda\tau+\mu)=e^{-2\pi \I m(\lambda^2\tau+2\lambda z)}\vph(\tau,z)
\ee
for any integers $\lambda,\mu\in\IZ$. Setting the elliptic variables $z = 0$, the above reduces to the definition of a modular form.

We see that $\vph$ is doubly periodic in $\tau,z$ and can be written using $q=e^{2\pi\I\tau},\, y=e^{2\pi \I z}$ as a Fourier series
\be
\label{Jacobi-Fourier}
\vph(\tau,z)=\sum_{n,r}c(n,r)\,q^n\,y^r.
\ee
This form allows to define the different types of Jacobi forms with different growth at the cusps. The function $\vph(\tau,z)$ is called a \textit{holomorphic Jacobi form} if 
\be
c(n,r)=0 \,\qquad \text{unless}\,\qquad 4mn\geq r^2.
\ee
It is called a \textit{Jacobi cusp form} if it satisfies the stronger condition 
\be
c(n,r)=0\,\qquad \text{unless}\,\qquad 4mn>r^2. 
\ee
We give it the name \textit{weak} Jacobi form if it satisfies instead the condition 
\be
c(n,r)=0 \,\qquad \text{unless}\,\qquad n\geq0,
\ee
and finally, we call it \textit{weakly holomorphic Jacobi form} if it satisfies
\be
c(n,r)=0 \,\qquad \text{unless}\,\qquad  n \geq n_0,
\ee
for some integer $n_0\in\IZ$. These conditions dictate the asymptotic growth of the Jacobi form as can be seen in \cite{Dabholkar:2012nd}.

The above definitions form a chain of implications, from the weakest to the strongest. Accordingly, they induce a series of spaces that can be ordered by inclusion. We will see that the space corresponding to weak Jacobi forms, which obviously has a ring structure, can be generated by a small number of functions that we will define later. Furthermore, it is clear that a weakly holomorphic Jacobi form, for a given $n_0$, can be lifted to a weak Jacobi form after a multiplication by an appropriate power of the discriminant form \eqref{discriminant-Delta}. More generally, like for modular forms, the asymptotic growth of Fourier coefficients of a Jacobi form is dictated by the growth at the cusps of the form itself and the precise relations can be found in \cite{Dabholkar:2012nd}.

We will now define mock Jacobi forms. In order to do that, we start by presenting a characterization of Jacobi forms that will prove instrumental, namely the theta expansions. 

The elliptic transformation property \eqref{Jacobi-elliptic-transf} is equivalent to the periodicity of the Fourier coefficients
\be
c(n,r)=C(\Delta,r),
\ee
where $\Delta=4mn-r^2$ is called the discriminant of the monomial $q^ny^r$ and $C(\Delta,r)$ depends only on $r\mod 2m$. Using this property with the Fourier expansion of $\vph$ \eqref{Jacobi-Fourier} allows us to find its theta expansion 
\be
\label{Jacobi-theta-expansion}
\vph(\tau,z)=\sum_{\ell\in\IZ/2m\IZ} h_\ell(\tau) \theta^{(m)}_{\ell}(\tau,z),
\ee
where 
\be
\label{Thetaexp-theta}
\theta^{(m)}_{\ell}(\tau,z) = \sum_{r\in2m\IZ+\ell} q^{\frac{r^2}{4m}}\,y^r
\ee
is an element of a very special type of Jacobi forms, namely the unary theta series that will be defined later. Meanwhile, the coefficients $h_\ell(\tau)$ are modular forms of weight $k-\hf$ and are weakly holomorphic, holomorphic or cuspidal if $\vph$ is a weak Jacobi form, a holomorphic Jacobi form or a Jacobi cusp form, respectively. More precisely, the vector of all $2m$ functions $h_\ell$ transforms like a modular form of weight $k-\hf$ under $\SL$. Conversely, any expression \eqref{Jacobi-theta-expansion} where the $h_\ell$ transform as a VV modular form of weight $k-\hf$, with an adapted multiplier system, under $\SL$ gives a Jacobi form of weight $k$ and index $m$. 

With this in mind, let's proceed to define (depth $n$) mock Jacobi forms. Let $\vph(\tau,z)$ be a holomorphic function in $\tau,z$, satisfying only the transformation property \eqref{Jacobi-elliptic-transf}. Hence, we relax the equation \eqref{Jacobi-modular-transf} and instead require that the coefficients $h_\ell(\tau)$ in the expansion \eqref{Jacobi-theta-expansion} be (depth $n$) mock modular forms of weight $k-\hf$. Then, $\vph$ is a (depth $n$) mock Jacobi form and its completion
\be
\label{Jacobi-mock-completion}
\hat\vph(\tau,\btau,z)=\sum_{\ell\in\IZ/2m\IZ} \hat h_\ell(\tau) \vth_{m,\ell}(\tau,z),
\ee
given in terms of the completions $\hat h_\ell$, transforms as in \eqref{Jacobi-modular-transf} and \eqref{Jacobi-elliptic-transf}.

Lastly, we provide a few well-known generalizations relevant to our work: we allow for multiple elliptic parameters, non-holomorphicity and vector valuedness. 

Let $\vph_\mu(\tau, \bfz)$ be a finite set of (in general, non-holomorphic\footnote{In the case of non-holomorphic Jacobi forms the theta expansion \eqref{Jacobi-theta-expansion} may not exist.}) functions, labelled by $\mu$, on $\IH \times \IC^n$, and $x \cdot y = \sum_{i,j=1}^n Q_{ij}x_i y_j$ denotes a bilinear form on $\IC^n$. Then $\varphi_\mu(\tau, \bfz)$ is a vector valued (multi-variable) Jacobi form of weight $(w, \bar{w})$ and (matrix valued) index $mQ_{ij}$ if it satisfies the following transformation properties:
\be
\label{Jacobi-VV-transf}
\begin{split}
\vph_\mu(\tau, \bfz + \bfa \tau + \bfb) =& e^{-2\pi im(\bfa^2 \tau + 2\bfa \cdot \bfz)} \vph_\mu(\tau, \bfz), \quad \bfa, \bfb \in \IZ^n,
\\
\vph_\mu \left( \frac{a\tau + b}{c\tau + d}, \frac{\bfz}{c\tau + d} \right) =& (c\tau + d)^w (c\bar{\tau} + d)^{\bw} e^{2\pi im\frac{c \bfz^2}{c\tau + d}} \sum_{\nu} M_{\mu\nu}(g) \vph_\nu(\tau, \bfz),
\end{split}
\ee
where we use boldface symbols for $n$-dimensional vectors, $M_{\mu\nu}(\rho)$ is a multiplier system and we omit the $\btau$ dependence. Notice that we allow in the second transformation property, the modular transformation, for an anti-holomorphic weight $\bw$.

The matrices $M_{\mu\nu}(\rho)$ must furnish a representation of the group $\SL$ generated by two transformations $T,S$ in \eqref{mod-TS} and thus in order to define the multiplier system, it is enough to specify it for these two generators. In summary, to characterize the modular behavior of a Jacobi form, it is sufficient to provide its modular weight $(w, \bw)$, index $m$ and two matrices $M_{\mu\nu}(T)$ and $M_{\mu\nu}(S)$.

\subsection{Jacobi-like forms}
\label{subsec-Jacobi-like}
In this subsection we will study a class of functions, that play a crucial role in \cite{Alexandrov:2024wla}, closely related to Jacobi forms but that have no elliptic transformation property. This class of functions is called Jacobi-like forms. They were first introduced in \cite{Zagier:1994,Cohen1997} and further studied in the mathematics literature in \cite{Lee2001,Lee2006,2010arXiv1007.4823C}. 

A Jacobi-like function can be defined as a formal power series in $X$, with the coefficients being functions on the upper half-plane 
\be
\label{Jacobi-like-Laurent}
\Phi(\tau,X)=\sum_{n=n_0}^{\infty}\phi_n(\tau)(2\pi\I X)^n,
\ee
with $n_0\in\IZ$, and they transform under $\SL$ as 
\be
\label{Jacobi-like-transformation}
\Phi\(\frac{a\tau+b}{c\tau+d},\frac{X}{(c\tau+d)^2}\)=(c\tau+d)^k\,e^{\frac{cX}{c\tau+d}}\,\Phi(\tau,X).
\ee

We can relate them more closely to Jacobi forms by the identification $X\to 2\pi\I m z^2$ which gives us a function $\phi(\tau,z)$ that obeys the second transformation in \eqref{Jacobi-modular-transf} with weight $k$ and index $m$. Building on this relation, we will procede in a way that is most convenient for our purposes, namely we will call Jacobi-like form any function $\phi(\tau,z)$ that satisfies the second transformation property of \eqref{Jacobi-VV-transf}, with potentially multiple elliptic parameters and a non-trivial multiplier system. Additionally, we allow for expansions in \textit{odd} powers of $z$\footnote{The original definition of Jacobi-like forms implies that they have an expansion in even powers of z. However, once one allows for a non-trivial multiplier system, there is no much sense keeping this condition. In practice, the functions appearing in the main text are functions of $(\tau,z)$ and $\bfz = (z_1, . . . , z_n)$ which behave as Jacobi-like forms with respect to $z$ with an expansion in even powers, up to an overall shift in the power, and as usual Jacobi forms with respect to $\bfz$.}.

Next, we are interested in modular properties of the expansion coefficients of a Jacobi-like form around a point where one of the variables, say $z_1$, vanishes. For simplicity, we restrict ourselves to the case n = 1 of only one elliptic variable and set $Q_{11}=1$, but the propositions below are trivially generalized to $n > 1$ provided the quadratic form is factorized, i.e. $Q_{1i} = 0$ for $i > 1$.

We start with a particularly simple and interesting example of Jacobi-like forms that can be constructed from the quasi-modular form $E_2(\tau)$ as 
\be
\label{E2-Jacobilike}
e^{-\frac{m}{3}\pi^2E_2(\tau)z^2}.
\ee
It has weight 0 and index $m$. The automorphy factor in the modular transformation \eqref{Jacobi-modular-transf} is given precisely by the modular anomaly of $E_2$! This combination is very helpful in constructing a modular differential operator and getting information about the modular properties of the coefficients of the Laurent expansion \eqref{Jacobi-like-Laurent} of a Jacobi-like form.

\begin{proposition}\label{prop-JacobiE2}
	Let $\vph(\tau,z)$ be a Jacobi-like form of modular weight $w$ and index $m$. Then
	\be
	\tilde \vph(\tau,z) = e^{\frac{m}{3}\pi^2 E_2(\tau) z^2} \vph(\tau,z)
	\label{def-tphi}
	\ee	
	is a Jacobi-like form of the same weight and vanishing index, and the coefficients of its Laurent expansion
	$\tilde \vph(\tau,z)=\sum_{n=n_0}^\infty h_n(\tau) z^n$
	are modular forms of weight $w+n$.
\end{proposition}
\begin{proof}
	The fact that the index of $\tilde \vph$ vanishes comes from the fact that we multiplied by precisely a function with the opposite index to that of $\vph$. Then it is clear, by applying the modular transformation to $\tilde \vph$ and writing its $z$-expansion on both sides, that each coefficient $h_n$ transforms in a modular way.
\end{proof}

This simple observation can be used to prove

\begin{proposition}
	\label{prop-Jacobi-n}
	Let $\vph(\tau,z)$ be a Jacobi-like form of modular weight $w$ and index $m$, and having a smooth limit at $z\to 0$.
	We define the following differential operator
	\be 
	\cD_m^{(n)}=\sum_{k=0}^{\lfloor n/2\rfloor}c_{n,k}E_2^k(\tau) \,\p_z^{n-2k},
	\label{defcDmn}
	\qquad
	c_{n,k}=\frac{n!\(\frac{2m}{3}\pi^2\)^k }{(2k)!!(n-2k)!}\, .
	\ee
	Then 
	\be
	\phi^{(n)}(\tau)\equiv \cD_m^{(n)}\vph(\tau,z)|_{z=0},	
	\label{coeffJac}
	\ee
	is a vector valued modular form of weight $w+n$.
\end{proposition}
\begin{proof}
	If $\phi_\mu$ is smooth at small $z$, the same is true for the function $\tilde \vph_\mu$ \eqref{def-tphi}
	and hence its expansion coefficients are given by the derivatives with respect to $z$ evaluated at $z=0$.
	According to Proposition \ref{prop-JacobiE2}, such derivatives $\partial_z^n\tilde \vph_\mu(\tau,0)$
	transform as modular forms of weight $w+n$.
	On the other hand, we have
	\be
	\partial_z^n\tilde \vph_\mu(\tau,0) = \sum_{k=0}^{\lfloor n/2 \rfloor} \frac{n!}{(2k)!(n-2k)!} \, 
	\left.\(e^{-\frac{x^2}{2}}\frac{\de^{2k}}{\de^{2k} x}\, e^{\frac{x^2}{2}}\)\right|_{x=0}\, 
	\(\frac{2m}{3} \pi^2 E_2(\tau)\)^k \partial_z^{n-2k} \vph_\mu(\tau,0).
	\ee
	Taking into account that $\left.\(e^{-\frac{x^2}{2}}\frac{\de^{2k}}{\de^{2k} x}
	\, e^{\frac{x^2}{2}}\)\right|_{x=0}=(2k-1)!!$, where $n!!=n(n-2)\dots1$,
	we conclude that
	\be
	\partial_z^n\tilde \vph_\mu(\tau,0) =\cD_m^{(n)}\vph_\mu(\tau,z)|_{z=0},
	\ee
	which proves the statement of the proposition.
\end{proof}

These two propositions, albeit simple, will be of great benefit in the construction of \cite{Alexandrov:2024wla}, explained in chapter \ref{chap-DTall}.

\section{Theta series}
\label{sec-mod-theta}
In the previous section we defined Jacobi forms and saw briefly how they can be related to a VV modular form through their theta expansion. In fact the theta series \eqref{Thetaexp-theta} providing that link is part of a bigger class of functions called generalized theta series. We will start in subsection \ref{subsec-Generalized-theta} by constructing these series from lattices and then providing a theorem conditioning their modular properties. Then, first we restrict to 1-dimensional lattices and find the standard family of unary theta series in subsection \ref{subsec-unary-theta}. Second, in \ref{subsec-basis-weakJacobi} we give a family of functions that generate the ring of weak Jacobi forms. Whereas in subsection \ref{subsec-theta-convergence} we discuss the issue of convergence when the lattice is not of negative definite signature. 
Finally, in the last part we define the generalized error functions, introduced in \cite{Alexandrov:2016enp}, that are very important for constructing modular completions of the theta series from the previous subsection.

\subsection{Generalized theta series}
\label{subsec-Generalized-theta}

Let us define
\be
\vth_{\bbmu}(\tau, \zbbm;\bbLambda, \Phi, \pbbm) = \sum_{\kbbm\in \bbLambda + \bbmu + \hf \pbbm} 
(-1)^{\pbbm\ast \kbbm} \Phi\(\sqrt{2\tau_2}\(\kbbm+\bbbeta \) \) \q^{-\hf \kbbm^2} 
e^{2\pi\I\zbbm\ast \kbbm},
\label{gentheta}
\ee
where $\q=e^{2\pi\I\tau}$, 
$\bbLambda$ is a $d$-dimensional lattice equipped with a bilinear form $\xbbm\ast \ybbm$ 
such that the associated quadratic form has signature $(n,d-n)$ and is integer valued, 
$\pbbm$ is a characteristic vector satisfying 
$\kbbm\ast(\kbbm+\pbbm) =0 \mod 2$ for $\forall \kbbm\in \bbLambda$, 
$\bbmu \in \bbLambda^{\star}/\bbLambda$, and $\zbbm=\bbalpha-\tau\bbbeta\in\IC^d$
with $\bbalpha,\bbbeta\in\IR^d$.
(We follow the convention to denote $d$-dimensional quantities by blackboard letters.)
The Vign\'eras theorem \cite{Vigneras:1977} asserts that if the kernel $\Phi(\xbbm)$ satisfies 
suitable decay properties as well as
the following differential equation 
\be
\[\p_\xbbm^2+2\pi (\xbbm\ast\p_\xbbm-\lambda)\]\Phi(\xbbm)=0,
\label{Vigdif}
\ee
where $\lambda$ is an integer parameter,
then the theta series is a vector valued (multi-variable) Jacobi form\footnote{More precisely, 
	the elliptic transformation \eqref{Jacobi-elliptic-transf} can generate an additional sign factor 
	$(-1)^{\pbbm\ast(\abbm+\bbbm)}$.} 
with the following modular properties: 
\be
\begin{split}
	w(\vth) &= \(\hf \,(d+\lambda), -\hf\, \lambda\),
	\qquad\qquad
	m(\vth) = - \hf\, \ast,
	\\
	\Mi{\vth}_{\bbmu \bbnu}(T) &= e^{-\pi \I\(\bbmu + \hf \pbbm\)^2 } \delta_{\bbmu\bbnu},
	\qquad
	\Mi{\vth}_{\bbmu \bbnu}(S) = \frac{e^{(2n-d)\frac{\pi\I}{4}}}{\sqrt{|\bbLambda^\ast/\bbLambda|}} \,
	e^{\frac{\pi \I}{2}\pbbm^2} e^{2\pi \I \bbmu \ast \bbnu},
\end{split}
\label{mult-genth}
\ee
where by $\ast$ in the formula for the index we mean the matrix representing the bilinear form.
The multiplier system here forms the Weil representation of the modular group, defined by the lattice $\bbLambda$.
A particularly interesting case is when the multi-variable Jacobi form is reduced to the usual
Jacobi form by choosing $\zbbm=\bbtheta z$ where $\bbtheta\in \bbLambda$.
Then the index is a scalar and is given by
\be 
m(\vth)=-\hf\, \bbtheta^2.
\ee

Note that if we restrict to modular forms (without elliptic parameter) the expression \eqref{gentheta} becomes 
\be
\vth_\bbmu(\tau)=\sum_{\kbbm\in\bbLambda+\bbmu} \Phi(\kbbm;\tau_2)\, e^{-\pi\I\tau \kbbm^2}.
\label{gentheta-noz}
\ee
This can be obtained simply by setting $\zbbm=0$ in the expression giving a Jacobi form, but we prefer to write here in order to illustrate its simplicity.

Another interesting case of Jacobi forms is when the lattice has a negative definite\footnote{Usually, this case corresponds to the positive definite quadratic forms. In our conventions it is negative due to the minus sign in the power of $\q$ in \eqref{gentheta}.}
quadratic form. 
In that case, the convergence is ensured by the powers of $\q$ and the kernel can be taken to be constant $\Phi(\xbbm)=1$. This choice of $\Phi$ trivially solves \eqref{Vigdif} with $\lambda=0$ and the resulting $\vth_\bbmu$ is then a (multi-variable) Jacobi form holomorphic in $\tau$.

The case of lattices with indefinite signature is more nuanced and we will treat it below. We start however by looking at the case of 1-dimensional negative definite lattices.

\subsection{Unary theta series}
\label{subsec-unary-theta}

Let us specialize \eqref{gentheta} to the case where $d=1$, $n=0$ and $\bbLambda= m\IZ$ so that 
the bilinear form is $\xbbm\ast\ybbm=-\xbbm\ybbm/m$.
We also take $\pbbm=-mp$ where $p$ is odd for odd $m$ and arbitrary integer otherwise, 
$\zbbm=-m z$ and $\Phi=1$ (hence $\lambda=0$).
Then the theta series reduces to 
\be
\label{Vignerasth}
\vths{m,p}_{\mu}(\tau,z)=
\!\!\!\!
\sum_{{k}\in \IZ+\frac{\mu}{m}+\frac{p}{2}}\!\!
(-1)^{mpk}\, \q^{m k^2/2}\, y^{mk},
\ee
where again $y=e^{2\pi\I z}$.
Its modular properties follow from \eqref{mult-genth} and are given by
\be
\begin{split}
	&w(\vths{m,p}) = 1/2,
	\qquad\qquad
	m(\vths{m,p}) = m/2,
	\\
	\Mi{m,p}_{\mu\nu}(T)&= e^{\frac{\pi\I}{m} \(\mu+\tfrac{mp}{2} \)^2}\,\delta_{\mu\nu}\, ,
	\qquad
	\Mi{m,p}_{\mu\nu}(S)=
	\frac{e^{-\frac{\pi\I}{2}\, m p^2}}{\sqrt{\I m}}\,
	e^{-2\pi\I\,\frac{\mu\nu}{m}}.
\end{split}
\label{eq:thetatransforms}
\ee

For even $m=2\kappa$, we can choose $p=0$. Then \eqref{Vignerasth} gives the function that appeared in \eqref{Thetaexp-theta},
\be
\ths{\kappa}_\mu(\tau,z)\equiv \vths{2\kappa,0}_{\mu}(\tau,z)
=\sum_{k\in 2\kappa\IZ+\mu}\q^{\frac{k^2}{4\kappa} }\,y^{k}.
\label{deftheta}
\ee
If $z=0$, we will simply drop the last argument and write $\ths{\kappa}_\mu(\tau)$.
The multiplier system \eqref{eq:thetatransforms} reduces to
\be
\Mi{\ths{\kappa}}_{\mu \nu}(T) =
e^{\frac{\pi \I}{2\kappa} \,\mu^2}\delta_{\mu \nu},
\qquad
\Mi{\ths{\kappa}}_{\mu \nu}(S)=
\frac{1}{\sqrt{2\I\kappa}} \, e^{-\frac{\pi \I}{\kappa}\,\mu \nu} .
\label{mult-theta}
\ee

On the other hand, specifying $m=p=1$ in \eqref{Vignerasth},
we reproduce the standard Jacobi theta function 
\be
\label{free-theta-1}
\theta_1(\tau,z) =\vths{1,1}_{0}(\tau,z)= \sum_{k\in Z+\hf} \q^{k^2/2} (-y)^k
\ee
whose modular properties are
\be
\begin{split}
	w(\theta_1) &= 1/2,
	\qquad\quad\ \
	m(\theta_1) = 1/2,
	\\
	\Mi{\theta_1}(T) &= e^{\pi \I/4} ,
	\qquad
	\Mi{\theta_1}(S) = e^{-3\pi \I/4}.
\end{split}
\label{multi-theta-N}
\ee
This function was used in \cite{Alexandrov:2024wla} to construct a suitable, more tractable, extension of the original problem. In particular, we used its behavior around $z=0$ where it vanishes and has the expansion,
\be
\theta_1(\tau,z) = -2\pi \eta(\tau)^3 z -4\pi^2 \I \eta'(\tau) \eta(\tau)^2 z^3 + O(z^5).
\label{theta1-z3}
\ee
$\eta(\tau)$ is the Dedekind eta function introduced in \eqref{Dedekind-eta}.

\subsection{Generating family of (scalar) weak Jacobi forms}
\label{subsec-basis-weakJacobi}
As we promised, we will define here a generating family of scalar weak Jacobi forms with a trivial multiplier system. There are mainly two great benefits to having such a set of functions. First, starting from the knowledge of a few Fourier coefficients of a weak Jacobi form, one can find a closed form expression for it. Second, given a function $\phi(\tau,z)$ holomorphic in both variables, one can check if it's a weak Jacobi form without computing its, often complicated, transformation behavior. 

We start by defining the four standard Jacobi theta functions,
\be
\label{Jacobi-theta}
\begin{split}
\vth_1(\tau,z) =& -\I \sum_{n=-\infty}^{\infty} (-1)^n q^{\hf (n-\hf)^2}\,y^{n-\hf},\\
\vth_2(\tau,z) =& \sum_{n=-\infty}^{\infty}q^{\hf(n-\hf)^2}\, y^{n-\hf},\\
\vth_3(\tau,z) =& \sum_{n=-\infty}^{\infty} q^{\hf n^2} \, y^n,\\
\vth_4(\tau,z) =& \sum_{n=-\infty}^{\infty} (-1)^n q^{\hf n^2}\, y^n.
\end{split}
\ee
These functions are very useful for the theory of Jacobi forms. The first one $\vth_1$ is the unary theta series with $p=m=1$ defined in \eqref{free-theta-1}. Additionally, after a rescaling of the arguments, $\vth_2$ and $\vth_3$ are the two components of the unary theta series \eqref{deftheta} defined with $\kappa=1$.

To these functions one often associates the set of \textit{Thetanullwerte}. These are obtained simply by taking $z=0$ in the above expressions and they are modular forms of weight $\hf$ for certain subgroups of $\SL$.  
For $\vth_1$ the resulting Thetanullwerte vanishes but for the other three we get
\be
\label{Thetanullwerte}
\begin{split}
	\vth_2(\tau,0) \equiv \vth_2(\tau) =& 2\sum_{n=0}^{\infty}q^{\hf(n+\hf)^2},\\
	\vth_3(\tau,0) \equiv \vth_3(\tau) =& 1+2\sum_{n=1}^{\infty} q^{\hf n^2},\\
	\vth_4(\tau,0) \equiv \vth_4(\tau) =& 1+2\sum_{n=1}^{\infty} (-1)^n q^{\hf n^2}.
\end{split}
\ee
Then we can define the weak Jacobi forms \cite{Dabholkar:2012nd} needed to generate the full ring,
\be
\label{Basis-phis}
\begin{split}
\varphi_{-1,2}(\tau,z) =& \,\frac{\vth_1(\tau,2z)}{\eta^3(\tau)},\\
\varphi_{-2,1}(\tau,z) =& \,\frac{\vth_1^2(\tau,z)}{\eta^6(\tau)},\\
\varphi_{0,1}(\tau,z)  =& \, 4\(\frac{\vth_2(\tau,z)}{\vth_2(\tau)}+\frac{\vth_3(\tau,z)}{\vth_3(\tau)}+\frac{\vth_4(\tau,z)}{\vth_4(\tau)}\).
\end{split}
\ee

Finally, the ring of weak Jacobi forms with integer weight and index is generated by complex polynomials
\be
\IC[E_4,E_6,\varphi_{-1,2},\varphi_{-2,1},\varphi_{0,1}],
\ee
where $E_4,E_6$ are the Eisenstein series defined previously. Note that this is not a basis as the three forms defined in \eqref{Basis-phis} are not independent instead they are related through \cite{Dabholkar:2012nd}
\be
432 \vph_{-1,2}^2=\vph_{-2,1}\(\vph_{0,1}^3-3E_4\,\vph_{-2,1}^2\vph_{0,1}+2E_6\vph_{-2,1}^3\).
\ee

\subsection{Convergence of indefinite theta series}
\label{subsec-theta-convergence}

Let us now consider theta series with a quadratic form of indefinite signature, namely $(n,d-n)$ with $0<n<d$.
In this case the kernel $\Phi(\xbbm)$ cannot be trivial since otherwise the theta series would be divergent.
On the other hand, a non-trivial kernel would spoil holomorphicity in $\tau$ unless $\Phi(\xbbm)$ is a piece-wise constant 
function\footnote{It is possible also to multiply it by a homogeneous polynomial in $\xbbm$ since 
	the non-holomorphic dependence can then be canceled by multiplying by a power of $\tau_2$.}.
Thus, the only way to get a convergent and holomorphic theta series is to take $\Phi(\xbbm)$ 
to be a combination of sign functions.
The following theorem from \cite{Alexandrov:2020bwg} (generalizing results of 
\cite{Nazaroglu:2016lmr,Alexandrov:2017qhn,funke2017theta}) 
provides the simplest choice of such kernel

\begin{theorem}
	\label{th-conv}
	Let the signature of the quadratic form be $(n,d-n)$ and
	\be
	\Phi(\xbbm)=\prod_{i=1}^n\Bigl(\sgn(\vbbm_{1,i}\ast \xbbm)-\sgn(\vbbm_{2,i}\ast\xbbm)\Bigr).
	\label{kerconverge}
	\ee
	Then the theta series \eqref{gentheta} is convergent provided:
	\begin{enumerate}
		\item
		for all $i\in \Zv_{n}=\{1,\dots,n\}$,
		$\vbbm_{1,i}^2,\vbbm_{2,i}^2\geq 0$;
		\item
		for any subset $\cI\subseteq \Zv_{n}$ and any set of $s_i\in \{1,2\}$, $i\in\cI$,
		\be
		\Delta_{\cI}(\{s_i\})\equiv \mathop{\det}\limits_{i,j\in \cI}(\vbbm_{s_i,i}\ast \vbbm_{s_j,j})\geq 0;
		\label{condDel}
		\ee
		\item
		for all $\ell\in\Zv_n$ and any set of $s_i\in \{1,2\}$, $i\in\Zv_n\setminus\{\ell\}$,
		\be
		\vbbm_{1,\ell\perp\{s_i\}}\ast \vbbm_{2,\ell\perp\{s_i\}}>0,
		\label{condscpr}
		\ee
		where $_{\perp\{s_i\}}$ denotes the projection on the subspace orthogonal to the span of
		$\{\vbbm_{s_i,i}\}_{i\in \Zv_n\setminus\{\ell\}}$;
		\item
		if $\vbbm_{s,i}^2=0$, then $\exists \alpha_{s,i}\in \IR$ such that $\alpha_{s,i}\vbbm_{s,i}\in \bbLambda$.
	\end{enumerate}
	\label{th-converg}
\end{theorem}

Note that the last condition requiring that the (rescaled) null vectors, i.e. satisfying $\vbbm_{s,i}^2=0$, 
that appear in the definition of the kernel belong to the lattice
is important. If such a null vector is present, it is also important to keep the elliptic variable $\zbbm$ generic
because the theta series has poles at the points where $\exists {\kbbm}\in \bbLambda+\bbmu+\hf\,\pbbm$ such that
$\vbbm_{s,i}\ast (\kbbm+\bbbeta)=0$. In particular, theta series involving null vectors are typically 
divergent in the limit $\zbbm\to 0$.

\subsection{Generalized error functions}
\label{subsec-gen-erf}

In the previous subsection, we provided a class of functions $\Phi(\xbbm)$
that define convergent and holomorphic indefinite theta series.
However, in contrast to the usual theta series with negative definite quadratic form, 
they are not modular. This can be seen, for example, from the fact that 
the discontinuities of the signs spoil the Vign\'eras equation \eqref{Vigdif}.
Nevertheless, there is a simple recipe to construct their modular completions 
\cite{Zwegers-thesis,Alexandrov:2016enp,Nazaroglu:2016lmr}.

This is achieved with help of the generalized error functions
introduced in \cite{Alexandrov:2016enp,Nazaroglu:2016lmr} (see also \cite{kudla2016theta}).
They are defined by
\bea
E_n(\cM;\vu)&=& \int_{\IR^n} \de \vu' \, e^{-\pi\sum_{i=1}^n(u_i-u'_i)^2} \prod_{i=1}^n \sign(\cM^{\rm tr} \vu')_i\, ,
\label{generr-E}
\eea
where $\vu=(u_1,\dots,u_n)$ is $n$-dimensional vector and $\cM$ is $n\times n$ matrix of parameters.
However, to get kernels of indefinite theta series we need functions depending on a $d$-dimensional vector
rather than $n$-dimensional one. To define such functions, let $\cV$ be $d\times n$ matrix 
which can be viewed as a collection of $n$ vectors, $\cV=(\vbbm_1,\dots,\vbbm_n)$, 
and it is assumed that these vectors span a positive definite subspace in $\IR^d$ endowed with the quadratic form $\ast$,
i.e. $\cV^{\rm tr}\ast\cV$ is positive definite. We also introduce a $n\times d$ matrix
$\cB$ whose rows define an orthonormal basis for this subspace.
Then we set
\be
\Phi_n^E(\cV;\xbbm)=E_n(\cB\ast \cV;\cB\ast\, \xbbm).
\label{generrPhiME}
\ee
The detailed properties of these functions can be found in \cite{Nazaroglu:2016lmr}.
Most importantly, they do not depend on $\cB$,
solve the Vign\'eras equation \eqref{Vigdif} with $\lambda=0$ and at generic, large $\xbbm$ reduce to
$\prod_{i=1}^n \sgn (\vbbm_i\ast\,\xbbm)$.
Thus, to construct a completion of the theta series whose kernel is a combination of sign functions,
it is sufficient to replace each product of $n$ sign functions by $\Phi_n^E$ with matrix of parameters $\cV$
given by the corresponding vectors $\vbbm_i$.

Finally, if one of the vectors is null, it reduces the rank of the generalized error function.
Namely, for $\vbbm_\ell^2=0$, one has
\be
\Phi_n^E(\{\vbbm_i\};\xbbm)=\sgn (\vbbm_\ell\,\ast\xbbm)\,\Phi_{n-1}^E(\{\vbbm_i\}_{i\in \Zv_{n}\setminus\{\ell\}};\xbbm).
\label{Phinull}
\ee
In other words, for such vectors the completion is not required.

The generalized error functions satisfy an important identity, which generalizes the one given in 
\cite[Eq.(D.13)]{Alexandrov:2018lgp}:
\begin{proposition}
	\label{prop-identPhiE}
	\be
	\Phi_{n}^E\(\{\bfv_i+\delta_{i,n}\bfv_{n-1}\};\bfx\)+\Phi_n^E\(\{\bfv_i+\delta_{i,n-1}\bfv_{n}\};\bfx\) 
	-\Phi_n^E\(\{\bfv_i\};\bfx\) = \Phi_{n-2}^E\(\{\bfv_i\};\bfx\),
	\label{identPhiE}
	\ee
	where for $n=2$ on the r.h.s. we set by definition $\Phi_0^E=1$.
\end{proposition}

\subsection{A peek into mock modularity}
Let's look at a simple example that will appear in our work. Let's take a theta series $\vth_{\bbmu}(\tau, \zbbm;\bbLambda, \Phi, \pbbm)$ following \eqref{gentheta} with kernel 
\be
\label{Phi-ex}
\Phi(\xbbm)=E_1(\xbbm\ast \vbbm) - \sgn(\xbbm\ast\vbbm').
\ee
Assuming that the condition for convergence are satisfied, the theta series is modular if and only if $\vbbm'$ is a null vector belonging to the lattice. When $\vbbm'$ is not null, then the series can be considered as a term that completes a mock modular Jacobi form! Indeed, provided the lattice contains a null vector $\wbbm$, a mock Jacobi form is given by $\vth_{\bbmu}(\tau, \zbbm;\bbLambda, \Phi', \pbbm)$ with 
\be
\label{Phigi-ex}
\begin{split}
	\Phi'(\xbbm)
	=&\,\sgn(\xbbm\ast\vbbm') - E_1(\xbbm \ast \wbbm ),
	\\
	=& \,\sgn(\xbbm\ast\vbbm') - \sgn(\xbbm \ast \wbbm ),
\end{split}
\ee
where we used \eqref{Phinull}. Then, the fact that $\xbbm=\sqrt{2\tau_2}(\kbbm+\bbbeta)$ has non-holomorphic dependence on $\tau$ becomes irrelevant since inside the sign functions we can simply scale away the $\sqrt{\tau_2}$ dependence and thus ensures holomorphicity. Furthermore, the sum of the two kernels solves \eqref{Vigdif} and thus the sum of the two theta series is modular. Note that $\wbbm$ needs to belong to the lattice, otherwise, the indefinite theta series would diverge.

\chapter{The hypermultiplet metric}
\label{chap-NS5HM}

The geometry of the hypermultiplet (HM) moduli space in type II string theories compactified on a Calabi-Yau threefold serves as a main thread connecting the different parts of this thesis. As established in Chapter \ref{chap-string}, a complete understanding of its metric is very important, as it directly determines the two-derivative effective action of the underlying four-dimensional $N=2$ supergravity.
However, the structure of the HM moduli space, is notoriously complex. In contrast to the vector multiplet (VM) moduli space, whose metric receives no $g_s$ corrections, the HM metric is subject to one-loop as well as all types of non-perturbative corrections in the string coupling. While significant progress has been made in computing the perturbative one-loop correction and non-perturbative corrections arising from D-branes, the contributions from NS5-brane instantons have remained mysterious. The primary goal of this chapter is to elucidate these corrections, presenting an explicit computation of the one-instanton NS5 contribution to the HM metric and detailing the methods used to obtain it.

The main tool at our disposal for constraining these corrections is S-duality of type IIB string theory. A direct application of this symmetry on the \qk metric of the HM is, however, complicated. QK metrics are governed by a set of intricate differential constraints, making a direct construction very difficult. To circumvent this difficulty and efficiently impose the constraints of S-duality, we employ the twistor space formalism, which was presented in Chapter \ref{chap-twist}. This powerful technique allows one to encode the entire QK metric in a set of holomorphic transition functions defined on the corresponding twistor space $\cZ_\cM$. The procedure to reconstruct the metric from this twistorial data is systematic, and was reviewed in detail. It is crucial to note, however, that this powerful construction, as currently understood, presents inconsistencies at the multi-instanton level \cite{Alexandrov:2014mfa, Alexandrov:2014rca}. Our analysis is therefore necessarily restricted to the one-instanton approximation, where the method is well-defined and trustworthy.

To obtain the NS5-instanton corrections, we build upon the twistorial description of these instantons developed in \cite{Alexandrov:2010ca} and apply the systematic procedure of \cite{Alexandrov:2008nk} to derive the explicit deformation of the HM metric. The validity of this central result is then backed up by two non-trivial consistency checks. Specifically, in the case of a rigid Calabi-Yau, the HM moduli space simplifies to a four-dimensional manifold whose geometry can be described, in a framework developed by Przanowski \cite{Przanowski:1984qq}, by a single function, as explained in subsection \ref{subsec-twist-Prz}. Instanton effects in this framework are captured by solutions to the linearization of the governing non-linear differential equation. As shown in \cite{Alexandrov:2006hx}, this approach yields three families of solutions. We find that our general, twistorial-based result, when specialized to this rigid limit, coincides with the one given by the Przanowski description.

The second limit we take is that of small string coupling $g_s$.
We get inspiration from \cite{Alexandrov:2021shf} and show that at the leading order in $g_s$ the one-instanton contribution must have the following structure 
\be
\de s^2_{\rm NS5}
\simeq\sum_\bfgam C_\bfgam\, e^{-S_\bfgam}\( \cA_\bfgam^2+\cB_\bfgam \de S_\bfgam\),
\label{sqaurestr}
\ee
where $\bfgam$ is a charge vector labeling bound states of NS5 and D-branes, $S_\bfgam$ is an instanton action,
$C_\bfgam$ is a function of the moduli scaling as a power of $g_s$,
and $\cA_\bfgam$, $\cB_\bfgam$ are one-forms on $\cM_H$ such that the coefficients of $\cA_\bfgam$ are given by
certain string amplitudes in an NS5-background, while $\cB_\bfgam$ cannot be fixed by this analysis.
We show that the metric which we derived does exhibit the structure \eqref{sqaurestr}
with specific $C_\bfgam$ and one-forms $\cA_\bfgam$ and $\cB_\bfgam$. This provides a prediction for three-point sphere
correlation functions where one of the vertex operators corresponds to a hypermultiplet scalar
and two others represent fermionic zero modes of the background NS5-brane.
While for generic background fields the prediction is somewhat involved, it drastically simplifies
in the limit of small RR fields.
We hope that this prediction will help to understand how these amplitudes can be computed directly
using worldsheet techniques.

This chapter closely follows the work presented in \cite{Alexandrov:2023hiv} and is organized as follows. In section \ref{sec-NS5-pert}, we begin by reviewing the perturbative metric of the hypermultiplet moduli space. We choose to present the perturbative metric by passing through the corresponding twistor space in order to prepare the reader for the more involved D- and NS5-instanton corrected twistor space, which is treated in section \ref{sec-NS5-Instanton}. At the end of the treatment, we present the schematic form of the resulting one-instanton corrected metric. The subsequent sections are dedicated to verifying and analyzing this result in various physical limits. 
In Section \ref{sec-NS5-UHM}, we specialize our computation to the case of rigid Calabi-Yau threefolds and show that our result perfectly matches the metric deformation derived from the Przanowski formalism, confirming its validity. Finally, Section \ref{sec-NS5-smallgs} is devoted to a detailed analysis of the small string coupling $(g_s\ll1)$ limit. We first compute the expected form of the metric correction from a string amplitude perspective in subsection \ref{subsec-NS5-squareAmplitudes}. We then demonstrate in subsection \ref{subsec-NS5-gslimit} that our general metric indeed reduces to this expected form. A further simplification is explored in subsection \ref{subsec-NS5-smallRR}, where we consider the limit of small RR-axion fields, yielding a result that is simpler and more accessible to independent verification techniques.

\section{Perturbative hypermultiplet moduli space}
\label{sec-NS5-pert}
We start by describing the perturbative hypermultiplet moduli space in type IIA. Although the expression of its metric was found in \cite{Alexandrov:2007ec} (based on earlier works \cite{Antoniadis:1997eg,Gunther:1998sc,Antoniadis:2003sw,Robles-Llana:2006ez}) using the c-map procedure, we will present it here using the twistor space approach. This allows to then pass to the instanton corrected metric easily. 

We start by giving the data on the twistor space corresponding to the perturbative HM moduli space. Then, we describe the procedure that allows to extract the expression of the metric from the Darboux coordinates. Finally, we give explicitly the 1-loop corrected metric and discuss it briefly.

\subsection{Twistor space}
Here we will quickly give a reminder of the physical fields, already introduced in subsection \ref{4d-fields-A}, that constitute the hypermultiplet moduli space in type IIA setting. They also parametrize the base of the twistor space, which we describe afterwards, including one-loop $g_s$ corrections. 

For type IIA string theory, compactified on a Calabi-Yau threefold $\CY$, the hypermultiplet moduli space has $4(h^{2,1}(\CY)+1)$ (real) dimensions. One can see that even for rigid CYs the moduli space still exists and has $4$ dimensions, in this case it is called the {\it universal} hypermultiplet. The real fields parametrizing this particular HM are 
\begin{itemize}		
	\item
	RR fields $\zeta^0,\tzeta_0$ 
	arising as period integrals of the RR 3-form of type IIA string theory over a symplectic basis of cycles in $H_3(\CY,\IZ)$;
	
	\item
	four-dimensional dilaton\footnote{Throughout the thesis we use the name `dilaton' for its exponential given by the field $r$.}
	$r\equiv e^{-2\phi_{4}} \sim g_s^{-2}$;
	
	\item
	NS axion $\sigma$ which is dual to the $B$-field in four dimensions.
\end{itemize}
For the other hypermutliplets, we have the $4h^{2,1}$ fields given by the (real) RR fields $\zeta^a,\tzeta_a$ and the (complex) complex structure moduli $z^a$ with $a=1,\dots,h^{2,1}$. When we want to consider the whole vector of RR fields, our convention is to use, instead of the index $a$, the index $\Lambda=0,\dots,h^{2,1}$. 

The complex fields $z^a$ parametrize the complex structure moduli space $\cM_C\subset\cM_H$ of the Calabi-Yau. It is a special \kahler manifold with metric completely fixed by the holomorphic prepotential $F$ defined in \eqref{prepot-cplx} with the projective coordinates given by $X^\Lambda=(1,z^a)$. In fact the prepotential fixes completely the classical metric on the whole hypermultiplet moduli space! On the twistor space, this translates to the fact that $F(X)$ determines all transition functions. 

The precise twistor space description yielding the perturbative metric on the HM was found in \cite{Alexandrov:2008nk}. Working in a local neighborhood of a fixed point in the basis space, the covering of $\CP$ is given by three patches $\cU_+,\cU_-$ and $\cU_0$. The first two cover a neighborhood of $t=0$ and $t=\infty$ respectively, whereas $\cU_0$ covers the rest of the sphere. This covering is such that it respects the reality conditions with $\cU_+$ and $\cU_-$ exchanged under \eqref{rlty-patch} while $\cU_0$ is invariant.

There are two independent transition functions between these patches. They are defined as 
\be
H^{[+0]}=F(\xi), \qquad H^{[-0]}(\xi)=\bF(\xi),
\ee
where $F$ is the prepotential of the special \kahler space $\cM_C$. This gives the classical twistor space. The one-loop corrections \cite{Robles-Llana:2006ez}, however, are implemented by taking the anomalous parameter $c_\alpha = -2c$ where
\be
\label{oneloop-c}
c = -\frac{\chi_\CY}{192\pi}
\ee
was found in \cite{Robles-Llana:2006ez} and $\chi_\CY$ is the Euler characteristic of the CY. 
One can then use equations \eqref{glu-integral-Darboux} to find all the Darboux coordinates in the central patch $\cU_0$
\be
\label{pert-Darboux0}
\begin{split}
	\xi_{\rm pert}^\Lambda =& \zeta^\Lambda + \cR\,(t^{-1}z^\Lambda-t\bz^\Lambda),
	\\
	\txi_{\rm pert,\,\Lambda}^{[0]} =& \tzeta_\Lambda + \cR\(t^{-1}F_\Lambda(z)-t\bF_\Lambda(\bz)\),
	\\
	\talp_{\rm pert}^{[0]}=& \sigma + \cR\(t^{-1}W(z)-t\bW(\bz)\)-8\I c \log t,
\end{split}
\ee
where 
\be
W(z)=F_\Lambda \zeta^\Lambda - z^\Lambda\tzeta_\Lambda,
\ee
and we used the physical fields as coordinates to write the solution. Their relation to the coordinates used in \eqref{glu-integral-Darboux} is the following 
\be
A^\Lambda = \zetal, \qquad 
B_\Lambda = \tzetal - \zeta^\Sigma\Re F_{\Lambda \Sigma}(z), \qquad
B_\alpha = -\hf(\sigma + \zetal B_\Lambda), \qquad
Y^\Lambda=\cR \zl.
\ee
And the dilaton coincides with the $t$-independent part \eqref{contact-tindep} of the contact potential, and can be obtained from the variable $\cR$ 
\be
\label{NS5-Phi-pert}
r=e^{\Phi_{\rm pert}}=\frac14 \cR^2 K - c.
\ee

Furthermore, we used the shifted variable 
\be
\label{talp}
\talpi=-2\alpi - \xii^\Lambda\txii_\Lambda.
\ee
This redefinition ensures that $\talpi$ is invariant under the action of the symplectic group. On the other hand, $(\xil_{\[0\]},\txil^{\[0\]})$ form a symplectic vector. 
Moreover, the contact 1-form \eqref{cX-Darboux}, after this redefinition, reads
	\be
	\cX = -\hf \(\de \talp +\txi_\Lambda \de \xi^\Lambda -\xi^\Lambda \de \txi_\Lambda\).
	\ee
These properties show that the whole construction is symplectic invariant which agrees with expectations from type IIA string theory. One should note however that this symplectic invariance is not manifest in other patches.

\subsection{The metric}
Upon applying the procedure presented in chapter \ref{chap-twist}, we find that the one-loop corrected metric is given simply in terms of the 1-forms 
\be
\begin{split}
	\frZ_\Lambda=&\,\de\tzeta_\Lambda-F_{\Lambda\Sigma}\de\zeta^\Sigma,
	\qquad
	\frJ=z^\Lambda \frZ_\Lambda=z^\Lambda\de\tzeta_\Lambda-F_\Lambda\de\zeta^\Sigma,
	\\
	&\,\quad
	\frS=\frac{1}{4}\(\de \sigma + \tzeta_\Lambda \de \zeta^\Lambda - \zeta^\Lambda \de \tzeta_\Lambda+8c\cA_K \),
\end{split}
\label{def-pertforms}
\ee
as \cite{Alexandrov:2007ec}
\be
\begin{split}
	\de s^2_{\rm pert} =&\, \frac{r+2c}{r^2(r+c)}\, \de r^2
	-\frac{1}{r} \(N^{\Lambda\Sigma} - \frac{2(r+c)}{rK}\, z^\Lambda \bz^\Sigma\) \frZ_\Lambda\bfrZ_\Sigma
	\\
	&\,
	+\frac{r+c}{r^2(r+2c)} |\frS|^2
	+ \frac{4(r+c)}{r}\, \cK_{a\bar{b}} \de z^a \de \bz^b.
\end{split}
\label{1lmetric}
\ee
Here we denoted the matrix $N_{\Lambda\Sigma} = -2\Im F_{\Lambda\Sigma}$, with inverse $N^{\Lambda\Sigma}$. While 
\be
\cA_K =  \frac{\I}{2}\(\cK_a\de z^a-\cK_{\bar a}\de \bz^a\)=\Im\p\log K
\label{kalcon}
\ee
is the K\"ahler connection on $\cM_C$
with $\cK = -\log K $.

Let's say a few words about the metric \eqref{1lmetric}. It receives no further perturbative corrections and will be deformed only under the effect of instantons. Besides, \eqref{1lmetric} has three singularities at $r=0,\,-2c,\,-c$. The last two arise due to the one-loop correction and occur only when the Euler characteristic is positive. In fact, the first two turn out to be coordinate singularities while in the latter the quadratic curvature invariant diverges\footnote{At the level of the twistor space, the basis of holomorphic (1,0)-forms \eqref{defPi} becomes degenerate at the point $r=-2c$.} \cite{Alexandrov:2009qq}. It is believed that the inclusion of D-brane instantons and NS5-instantons, at all orders, should smooth it out. 

The hypermultiplet metric \eqref{1lmetric} carries an action of the symplectic group $Sp(2h^{2,1}+2,\IZ)$.
It leaves $r$ and $\sigma$ invariant and transforms $(X^\Lambda,F_\Lambda)$ and $(\zeta^\Lambda,\tzeta_\Lambda)$
as vectors. However, since a generic symplectic transformation affects the prepotential $F$, it is not a true isometry of $\cM_H$.
Only a subgroup of $Sp(2h^{2,1}+2,\IZ)$ which is realized as monodromies around singularities of the complex structure moduli space
is a true isometry.
The symplectic invariance can be seen as a characteristic feature of the type IIA formulation
and is expected to hold at the non-perturbative level.

It is also invariant under Peccei-Quinn symmetries acting by shifts on the RR fields and the NS axion
\be
T_{\eta^\Lambda,\tleta_\Lambda,\kappa}\ : \
(\zeta^\Lambda,\tzeta_\Lambda,\sigma)\mapsto
(\zeta^\Lambda+\eta^\Lambda,\tzeta_\Lambda+\tleta_\Lambda,\sigma+2\kappa-\tleta_\Lambda\zeta^\Lambda+\eta_\Lambda\tzeta_\Lambda).
\label{Heis}
\ee
At the perturbative level, the parameters $(\eta^\Lambda,\tleta_\Lambda,\kappa)$ can take any real value, whereas
instanton corrections break these isometries to a discrete subgroup with $(\eta^\Lambda,\tleta_\Lambda,\kappa)\in\IZ^{2h^{2,1}+3}$.
In particular, D-instantons break continuous shifts of the RR fields, but leave the invariance along $\sigma$, while
NS5-instantons break them all. The fact that the transformations \eqref{Heis} form the non-commutative
Heisenberg algebra plays an important role for description of NS5-instantons
(see, e.g., \cite{Witten:1996hc,Dijkgraaf:2002ac,Pioline:2009qt,Bao:2009fg}).

Finally, as we saw in section \ref{sec-string-EFT} mirror symmetry implies that $\cM_H$ in type IIA compactified on $\CY$ is identical to
the same moduli space in type IIB compactified on a mirror CY $\CYm$. Furthermore,
in this mirror type IIB formulation, $\cM_H$ must carry an isometric action of the S-duality group $SL(2,\IZ)$.
Its action at the level of the twistor space can be written efficiently \cite{Alexandrov:2014rca}.
In particular, S-duality was crucial to get NS5-instanton corrections from the D5-instanton ones \cite[Fig. I.1]{Alexandrov:2011va}.

\section{Instanton corrected hypermultiplet moduli space}
\label{sec-NS5-Instanton}
In this section we will describe the instanton corrected \qk manifold $\cM_H$. We follow the same logic as the previous section and we start by describing its twistor space and then giving the expression of the one-instanton corrected metric.

\subsection{Twistor space}

In this subsection we will start by giving a slight generalization of equations \eqref{glu-integral-Darboux}, which allows for open contours. Then, we will present the data determining the D-instanton and NS5-instanton corrected twistor space in succession.

\paragraph{New integral form}\

\noindent
The generalization we consider here consists in allowing open contours, with transition functions associated to them, in the integral equations determining the Darboux coordinates. The idea behind this generalization is that if $H^{\[ij\]}$ in \eqref{glu-integral-Darboux} have branch cuts then their integrals over closed contours $C_i$, can be replaced by integrals of their discontinuities $H_i$ along the open contours $\ell_i$. This is explained in more details, and was obtained in, \cite{Alexandrov:2009zh} and we will apply it here.

First, we define a notation for the integral combination 
\be
\opI[H]=\int_{\ell} \frac{\de t'}{t'} \,
\frac{t' + t}{t' - t}\, H\(\xi_{\rm pert}(t'), \txi^{\rm pert}(t'), \talp^{\rm pert}(t')\)
\label{def-frI}
\ee
which looks like the combinations used in \eqref{glu-integral-Darboux} except that now the function is evaluated on the perturbative Darboux coordinates \eqref{pert-Darboux0} and that the integration path is a line $\ell$ associated to the function $H$. The expression for the Darboux coordinates then becomes
\be
\begin{split}
	&
	\xi^\Lambda  = \xi_{\rm pert}^\Lambda
	+\frac{1}{4\pi \I} \sum_i \opI_i\[\hpartial_{\txi_\Lambda }H_i\] ,
	\qquad\qquad
	\txi_\Lambda =   \txi^{\rm pert}_\Lambda
	-\frac{1}{4\pi \I}   \sum_i \opI_i\[\hpartial_{\xi^\Lambda } H_i\],
	\\
	&
	\talp =  \talp^{\rm pert}
	+\cR\(t^{-1}\Win-t\bWin\)
	+\frac{1}{4\pi \I}   \sum_i \opI_i\[\( 2- \xi^\Lambda \hpartial_{\xi^\Lambda}-\txi_\Lambda \hpartial_{\txi_\Lambda}\)H_i\],
\end{split}
\label{genDC}
\ee
with
\be
\hpartial_{\xi^\Lambda }=\partial_{\xi^\Lambda }-\txi_\Lambda\p_{\talp}\, ,
\qquad
\hpartial_{\txi_\Lambda }=\partial_{\txi_\Lambda }+\xi^\Lambda \p_{\talp }\, ,
\label{def-hatder}
\ee
\be
\Win=\frac{1}{4\pi \I}\sum_i \int_{C_i} \frac{\de t'}{t'}
\( z^\Lambda\,\hpartial_{\xi^\Lambda }+F_\Lambda \,\hpartial_{\txi_\Lambda} \)H_i^{\rm pert}\, .
\ee
To get a real metric from \eqref{genDC}, one should impose an additional condition that
the set $\{ \ell_i\}$ is invariant under the antipodal map $\varsigma[t]=-1/\bt$, while the set $\{H_i\}$
is invariant under the combination of $\varsigma$ with complex conjugation, i.e. for each $i$ there is $\bi$
such that
\be
\varsigma[\ell_i]=\ell_{\bi},
\qquad
\overline{\varsigma[H_i]}=H_{\bi}.
\label{reality}
\ee

To complete the twistorial description of linear deformations, we give also a formula
for a modification of the contact potential $\Phi$ which is given by
\be
\Phi(t)=\phi+\frac{1}{2\pi \I} \sum_i \opI_i\[ \p_{\talp}H_i\],
\label{eqchip}
\ee
where the $t$-independent part reads as
\be
e^\phi=\frac{1}{4}\,\cR^2 K-c -\frac{\cR}{16\pi} \sum_i\int_{\ell_i}\frac{\de t}{t}
\[\(t^{-1} z^{\Lambda}-t \bz^{\Lambda} \)\hpartial_{\xi^\Lambda }
+\(t^{-1} F_{\Lambda}-t \bF_{\Lambda} \)\hpartial_{\txi_\Lambda }\]H_i^{\rm pert}\, .
\label{contpotconst}
\ee

\paragraph{D-instantons data}
\noindent\\
D-instantons have been incorporated into the twistorial description of $\cM_H$ at linear order in \cite{Alexandrov:2008gh}
and to all orders in the instanton expansion in \cite{Alexandrov:2009zh}.
Here we present the first simplified version which is sufficient at one-instanton level and fits the framework above.

As we described in subsection \ref{subsec-string-BPSBH}, in type IIA, each D-instanton is characterized by a charge vector $\gamma=(p^\Lambda, q_\Lambda)$.
It is integer valued and characterizes the 3-cycle wrapped by the D2-brane generating the instanton,
in the same basis of $H_3(\CY,\IZ)$ that is used to define RR fields $(\zeta^\Lambda,\tzeta_\Lambda)$.
Given the charge $\gamma$, we give the expression of the central charge function
\be
\label{defZ}
Z_\gamma(z) = q_\Lambda z^\Lambda- p^\Lambda F_\Lambda(z),
\ee
which was used to define the BPS bound \eqref{BPS-mass-cond}. We also have the generalized Donaldson-Thomas (DT) invariant $\Omega_\gamma$ defined in \eqref{BPS-Trace} and we remind that from the physical viewpoint it counts
the (signed) number of BPS instantons of a given charge. In the following we will mainly use its rational version
\be
\label{BPS-rational}
\bOm_\gamma = \sum_{d|\gamma}  \frac{1}{d^2}\,  \Omega_{\gamma/d} ,
\ee
which takes into account multi-covering effects and
allows to simplify many equations being more suitable for implementing S-duality \cite{Manschot:2010xp,Alexandrov:2012au}.
An important property of DT invariants is that $\Omega_{-\gamma}=\Omega_\gamma$.

Finally, we define the so-called BPS ray
\be
\ellg{\gamma}= \{ t :\,  Z_\gamma(z)/t \in \I\IR^{-} \},
\label{rays}
\ee
which joins the north and south poles of $\CP$ along the direction determined by the phase of the central charge,
and the following transition function assigned to $\ell_\gamma$
\be
H_\gamma
=\frac{\sigma_\gamma\,\bOm_\gamma}{4\pi^2}\, e^{-2\pi \I (q_\Lambda \xi^\Lambda-p^\Lambda\txi_\Lambda)} ,
\label{prepH}
\ee
where $\sigma_\gamma$ is the so-called quadratic refinement. This is a sign factor that must satisfy
$
\sigma_{\gamma+\gamma'} =
(-1)^{\langle \gamma, \gamma' \rangle}\, \sigma_\gamma\, \sigma_{\gamma'}$,
where $\langle\gamma,\gamma'\rangle=q_\Lambda p'^\Lambda-q'_\Lambda p^\Lambda$ is the skew-symmetric product of charges.
In the following it is chosen\footnote{How this choice is reconciled
	with symplectic invariance is explained in \cite[\$2.3]{Alexandrov:2014rca}.} to be $\sigma_\gamma=(-1)^{q_\Lambda p^\Lambda}$.
In \eqref{prepH} we don't indicate the patches where $\xi$ and $\txi$ are evaluated because at the one-instanton level this is irrelevant. The set of all $(\ell_\gamma,H_\gamma)$ for which DT invariants are non-vanishing comprise the twistor data of D-instantons.

Our one-instanton approximation corresponds to keeping only terms linear in $\bOm_\gamma$,
while we allow for D-instantons of different (in particular, proportional) charges.
Thus, it is not about extracting the dominant instanton contribution, but rather the linear response of the metric
to the change of the contact structure by $(\ell_\gamma,H_\gamma)$.
This approach allows us to get results independent of particular values of DT invariants and
to keep track of the charge dependence in the resulting instanton corrections.

\paragraph{NS5-instantons data}\

\noindent
The twistor data incorporating NS5-instantons in the one-instanton approximation (as it was defined in the previous paragraph)
have been found in \cite{Alexandrov:2010ca} by translating the above construction of D-instantons
to the mirror type IIB formulation and applying S-duality to D-instantons with a non-vanishing D5-brane charge.
The duality was applied at the level of the twistor space where it acts by a holomorphic transformation preserving
the contact structure. In particular, its action on the fiber coordinate $t$ and the Darboux coordinates
$(\xi^\Lambda,\txi_\Lambda,\talp)$ is well known \cite{Alexandrov:2008nk,Alexandrov:2012bu}
and therefore allows to get the contours and transition functions incorporating NS5-instantons as images of $(\ell_\gamma, H_\gamma)$
under this action.

Then, we take the twistor data encoding NS5-instantons obtained in \cite{Alexandrov:2010ca},
translate them back into type IIA language, and apply them as deformations of the perturbative twistor space.
The price to pay for using type IIB twistor data in type IIA is the absence of manifest symplectic invariance.
We also believe that the complicated structure of the resulting metric is partially a consequence of this hybrid approach
and there should exist a genuine type IIA formulation of NS5-instantons.
However, in the absence of such formulation, we have to proceed as just described, but we hope that
our results can shed light on this and other issues related to the geometry of $\cM_H$.

After these preliminary comments, let us describe the twistor data for NS5-instantons
after they have been translated (partially) to type IIA variables.
To this end, let us introduce an integer valued charge vector
$\bfg=(k,p,\hgam)$ with $k\ne 0$ and $\hgam=(p^a,q_a,q_0)$.
Here $k$ denotes NS5-brane charge, while the other components are related to bound D-branes.
In particular, the standard D-brane charge vector is obtained as
$\gamma=(p^0,\hgam)$ where $p^0=\gcd(k,p)$.
On the type IIB side, $p^0$ is D5-brane charge, while $\hgam$ encodes D3-D1-D(-1)-charges.

Given the charge $\bfg$, we define the contour $\ell_{\bfg}$ as a half-circle\footnote{More precisely,
	$\ell_{\bfg}$ is the image of the BPS ray $\ell_\gamma$ under
	some $SL(2,\IZ)$ transformation.
	In particular, the ordering of the original BPS rays is preserved.
}
stretching between the two zeros of $\xi^0_{\rm pert}(t)-p/k$
and the associated transition function
\be
\label{5pqZ2}
H_{\bfg}=\frac{\bOm_\gamma}{4\pi^2}\,e^{-\pi \I k \(\talp+(\xi^\Lambda-2n^\Lambda) \txi_\Lambda\)}
\Psi_{\hbfg}(\xi-n) ,
\ee
where
$\hbfg$ is a reduced charge and our results hold for arbitrary function $\Psi_{\hbfg}$ (provided it ensures
convergence of integrals along $\ell_{\bfg}$). Its concrete form \cite[Eqn. (2.26)]{Alexandrov:2023hiv} will be important only in deriving
the small string coupling approximation in section \ref{sec-NS5-smallgs}.

The (hidden) symplectic invariance ensures that it should be possible to rewrite the construction of NS5-instantons
in other ``frames" where a different combination appears in the exponent and in the argument of $\Psi$.
However, a map between different frames is expected to be non-trivial and
to involve an integral transform (see, e.g., \cite{Alexandrov:2015xir}),
similarly to a change between coordinate and momentum representations in quantum mechanics.

\subsection{The metric}

In the paper \cite{Alexandrov:2023hiv} we produced an expression for the instanton corrected metric. Its expression is quite cumbersome and we will present it here mainly for illustration purposes. Namely, we will not define all the functions that enter it. 

To this end, we use the 1-form notations \eqref{def-pertforms}
which arise naturally already in the perturbative metric \eqref{1lmetric}.
Besides, we also define a one-form labeled by D-instanton charge
\be
\cCf_\gamma= N^{\Lambda\Sigma}\(q_\Lambda-\Re F_{\Lambda\Xi}p^\Xi\)\(\de\tzeta_\Sigma-\Re F_{\Sigma\Theta}\de\zeta^\Theta\)
+\frac14\, N_{\Lambda\Sigma}\,p^\Lambda\,\de\zeta^\Sigma.
\label{connC}
\ee
Finally, we will use various functions labeled by charges $\gamma$ and $\bfg$ which are defined in \cite[Appendix C.]{Alexandrov:2023hiv}
as expansion coefficients of the integral transform \eqref{def-frI} of the transition functions and its derivatives. 
For D-instantons 
the integral transform reads 
\be
\cJ_\gamma(t)= \opI_\gamma\[e^{-2\pi \I\(q_\Lambda \xi^\Lambda-p^\Lambda \txi_\Lambda\)}\],
\label{defcJg}
\ee
and its value at $t=0$ is given by
\be
\begin{split}
	\cJn{0}_\gamma \equiv &\, \cJ_\gamma(0)= \int_{\ellg{\gamma}}\frac{\de t}{t}\, e^{-2\pi \I \Xi_\gamma(t)},
\end{split}
\label{cJn0}
\ee
where 
\be
\Xi_\gamma(t)\equiv q_\Lambda \xi_{\rm pert}^\Lambda-p^\Lambda \txi^{\rm pert}_\Lambda.
\ee                                                                                     
Other expansion coefficients $\cJn{\pm n}_\gamma$ correspond to the coefficients of $t^{\pm n}$ in $\cJ_\gamma(t)$ and are written in a similar form to \eqref{cJn0}. We will not write them as the goal here is simply to give an idea of how the solution looks.

For NS5-instantons the situation is more involved as we have multiple integral transforms and multiple expansion coefficients for each. 
The simplest transform reads
\be
\begin{split}
	\cI_\bfg(t)=&\, \opI_\bfg\[e^{2\pi\I k \alpha_n }\Psi_{\hbfg}\],
\end{split}
\label{NS5-cIt}
\ee
and its expansion coefficients are denoted $\cIn{\pm n}_\bfg$ with $n=1,2$. We also have many other functions obtained by adding to the integrand in \eqref{NS5-cIt}, derivatives with respect to the Darboux coordinates or combinations of the Darboux coordinates themselves. These functions, as well as their expansion coefficients will be written using different scripts of $\cL$ and $\cK$. One can recognize integrals related to NS5-corrections by the $\bfg$ index.

While for D-instantons, using integration by parts,
all relevant quantities can be expressed only through three such functions, $\cJn{n}_\gamma$ with $n=0,\pm1$,
for NS5-instantons this does not seem to be possible and we have to deal with many different functions.
Using all these definitions, the one-instanton corrected metric on $\cM_H$ is found to be
\be
\de s^2  =\de s^2_{\rm pert}
+\frac{\cR}{8 \pi^2 r}\sum_{\gamma}\sigma_\gamma \bOm_\gamma \, \frD_\gamma
+\frac{\cR}{8\pi^2r }\sum_{\bfg} \bOm_\gamma \,k\, \frV_{\bfg},
\label{ds2final2}
\ee
where $\frD_\gamma$ and $\frV_{\bfg}$ encode the D-brane and NS5-brane instanton contributions, respectively,
They are given by
\bea
\frD_\bfg &=&
\frac{\pi}{\cR}\cJn{0}_\gamma \[\cR^2 |Z_\gamma|^2 \(\(\frac{\de r}{r+c}+\de\log\frac{|Z_\gamma|^2}{K}\)^2
+\(\frac{\frS}{r+2c}+2\Im\p\log\frac{Z_\gamma}{K}\)^2\)+\de\Theta_\gamma^2-4\cC_\gamma^2
\]
\nn\\
&+&
\( Z_\gamma\cJn{1}_\gamma+\bZ_\gamma\cJn{-1}_\gamma\)\Biggl[
\frac{\I\frS^2}{4(r+2c)^2}-\frac{\I(r+c)\,\cA_K\frS}{r(r+2c)}
-\frac{\I\de r\,\de\log K}{4(r+c)}-\frac{\I |\frJ|^2}{2r K}
-\I\cK_{ab}\de z^a \de \bz^b
\nn\\
&&
+2\pi\cC_\gamma\(\frac{\frS}{r+2c}-2\,\cA_K\)
\Biggr]
+\(\frac{r+c}{r(r+2c)}\,\frS-4\pi \I\cC_\gamma \) \Bigl(\cJn{1}_{\gamma}\de Z_\gamma -\cJn{-1}_{\gamma}\de\bZ_\gamma\Bigr)
\nn\\
&+&
\frac{\I\de r}{2(r+c)}\, \de \Bigl(Z_\gamma \cJn{1}_{\gamma}+\bZ_\gamma \cJn{-1}_{\gamma}\Bigr)
- \frac{\I\cR}{r}\,\Im(\bZ_\gamma\,\frJ)\,\de \cJn{0}_{\gamma}.
\label{frDtot2}
\eea
and
{\footnotesize
\bea
\frV_\bfg &=&
\(\scLn{1}_\bfg+\bscLn{-1}_\bfg\) \Biggl[
\frac{\I}{8}\(\frac{(\de r)^2}{(r+c)^2}-\frac{\frS^2}{(r+2c)^2}\)
+\frac{\I}{2r}\,\cA_K\frS
+\frac{\I|\frJ|^2}{2r K}
+\frac{\I}{2}\,\frac{|\p K|^2}{K^2}+\I\cK_{ab}\de z^a \de \bz^b
\Biggr]
\nn\\
&-&
\frac{\I}{2}\(\cLn{1}{\bfg,\Lambda} \de z^\Lambda+\bcLn{-1}{\bfg,\Lambda}\de\bz^\Lambda\) \de \log K
-\(\cLn{1}{\bfg,\Lambda} \de z^\Lambda-\bcLn{-1}{\bfg,\Lambda}\de\bz^\Lambda\)\(\frac{\frS}{2r}+\cA_K\)
\nn\\
&-&
\frac{2(r+c)}{r(r+2c)}\,\frS\,\bigg[\frac{\I}{\cR}\, \cIn{0}_\bfg\(\frac{c\de r}{r+c}+r\de\log K\)
+\Re( \cIn{1}_\bfg\,\frJ)
\bigg]
\nn\\
&-&
2\pi k \Biggl\{
\frac{2}{\cR} \,\cIn{0}_\bfg\biggl(\frac{(r+2c)^2}{(r+c)^2}\,(\de r)^2
-\frS^2+\frac{2r(r+2c)}{r+c}\, \de r\, \de\log K
-4r\cA_K\frS-\frac{4(r+c)}{K}\,|\frJ|^2
\nn\\
&&\quad
+\frac{r^2(\de K)^2}{K^2}-4r^2\cA_K^2
\biggr)
+\cR\Im\(\(\scLn{2}_\bfg-\bscLn{0}_\bfg\)\frJ\) \(\frac{\frS}{r+2c}-2\cA_K\)
\nn\\
&&
+ 2\(\scLn{1}_\bfg+\bscLn{-1}_\bfg\)
\(\frac{\de r}{r+c}\(\frS-2c\cA_K\) -\de\log K\( \frac{c\frS}{r+2c}+2r\cA_K\)\)
\nn\\
&&
+4\(\frac{r+2c}{r+c}\, \de r+r\de\log K\)\biggl[\Im(\cIn{1}_\bfg\frJ)
-\cR^{-1}\Re\(N^{\Lambda\Sigma}\bcLn{0}{\bfg,\Lambda}\,\frZ_\Sigma\)
+\Im\(\cLn{1}{\bfg,\Lambda}\de z^\Lambda\)\biggr]
\nn\\
&&
+4\Re\frJ
\[\Im\( N^{\Lambda\Sigma}\bcLn{-1}{\bfg,\Lambda}\, \frZ_\Sigma\)
-\cR\Re\(\cLn{0}{\bfg,\Lambda}\de z^\Lambda\) \]
\\
&&
-4\Im\frJ
\[\Re\( N^{\Lambda\Sigma}\bcLn{-1}{\bfg,\Lambda}\frZ_\Sigma\)
+\cR\Im\(\cLn{0}{\bfg,\Lambda}\de z^\Lambda\) \]
\nn\\
&&
+\frac{\cR}{2}\,\scK_\bfg\Biggl[\frac{\frS^2}{4(r+2c)^2}-\frac{(\de r)^2}{4(r+c)^2}
+\frac{\de r\, \de\log K}{2(r+c)}-\frac{\cA_K\frS}{r+2c}
-\frac{d K^2}{4K^2}+\cA_K^2\Biggr]
\nn\\
&&
-\biggl(\frac{\frS}{r+2c}-2\cA_K\biggr)
\biggl[N^{\Lambda\Sigma}
\Re\biggl(\Bigl(z^\Xi\scKn{1}{\Xi\bar\Sigma}+\bz^\Xi\scKn{-1}{\bar\Xi\bar\Sigma}\Bigr)\frZ_\Lambda\biggr)
-\cR\Im\Bigl(\Bigl(z^\Lambda\scKn{2}{\Lambda\Sigma}-\bz^\Lambda\scKn{0}{\bar\Lambda\Sigma}\Bigr)\de z^\Sigma\Bigr)\biggr]
\nn\\
&&
+\Bigl(\scKn{0}{\Lambda\bar\Sigma}+\scKn{0}{\bar\Sigma\Lambda}\Bigr)
\biggl(\frac{1}{\cR}\,N^{\Lambda\Xi} N^{\Sigma\Theta} \frZ_\Theta \bar\frZ_\Xi-\cR\de z^\Lambda \de\bz^\Sigma\biggr)
+2\I N^{\Lambda\Sigma}\Bigl(\scKn{1}{\Sigma\Xi}\,\bar\frZ_\Lambda\de z^\Xi-\scKn{-1}{\bar\Sigma\bar\Xi}\,\frZ_\Lambda \de\bz^\Xi\Bigr)
\Biggr\}
\nn\\
&-&
\frac{\I}{2}\,\de\log K\(z^\Lambda\frLn{1}{\Lambda}+\bz^\Lambda\bfrLn{-1}{\Lambda}\)
-\frac{2\I}{\cR}\(\frac{c\,\frS}{r+2c}+2r \cA_K\)\de\cIn{0}_\bfg
\nn\\
&+&
\frac{\I}{\cR}\(\de\cGn{0}{\Lambda}_\bfg\de\tzeta_\Lambda-\de\tcGn{0}{\bfg,\Lambda} \de \zeta^\Lambda\)
- \frac{\cR}{r}\, \Re\[z^\Lambda\frLn{0}{\Lambda}\,\bar\frJ\].\nn
\label{frVtot1}
\eea}
Note that the variable $\cR$ appears only in the instanton terms and therefore in our approximation it can be expressed
through $r$ using the perturbative relation \eqref{NS5-Phi-pert}.

The result for the D-instanton corrections given by the second term in \eqref{ds2final2}
can be compared with the linearization (in DT invariants $\bOm_\gamma$)
of the metric found in \cite{Alexandrov:2014sya}. We do not provide any details of this simple exercise which
shows a perfect match between the two metrics.

The result for NS5-instanton corrections given by the last term in \eqref{ds2final2} is new.
Note that the specific form of the function $\Psi_{\hbfg}(\xi)$ appearing in \eqref{5pqZ2} has not been used and
it is valid for any such function ensuring convergence of the corresponding integrals.
The complicated nature of the result is expected since NS5-instantons break all continuous isometries
of the moduli space. It is also does not exhibit any fibration or other nice geometric structure.
Nevertheless, as will be shown in section \ref{sec-NS5-smallgs}, in the small string coupling limit, if one neglects terms
proportional to the differential of the instanton action, $\frV_\bfg$ reduces to the square of a one-form,
and this is precisely the structure expected from the analysis of string amplitudes.

\section{Universal Hypermultiplet}
\label{sec-NS5-UHM}

The case of a rigid CY, i.e. without complex structure moduli ($h^{2,1}(\CY)=0$), is special.
Therein, the spectrum contains only the {\it universal hypermultiplet} \cite{Strominger:1997eb}
comprising the dilaton $r$, NS axion $\sigma$  and a pair of RR fields\footnote{We drop
	indices on quantities labeled by $\Lambda,\Sigma,\dots$ since in this case that take a single value.}
$\zeta,\tzeta$, so that the moduli space is four-dimensional. Such QK manifolds are known to
have an alternative description due to Przanowski \cite{Przanowski:1984qq}
which we presented in subsection \ref{subsec-twist-Prz}.
This fact allows to test our one-instanton corrected metric which must be compatible with this description.
In particular, it must produce a solution of the linearized differential equation.
We showed that the metric is indeed consistent with the Przanowski description and
obtained the corresponding potential solving the differential equation.
In this section we will briefly explain the description due to Przanowski and how we approached the consistency check.

\subsection{Perturbative case}

For $\cM_H$, the perturbative metric has a well known Przanowski description \cite{Alexandrov:2006hx}
in terms of coordinates
\be
z^1=-\bigl(r+c\log(r+c)\bigr)-\frac{\I}{4}\, (\sigma+\zeta\tzeta+\tau\zeta^2),
\qquad
z^2=\frac{\I}{2}\, (\tzeta+\tau\zeta)
\label{pert-zz}
\ee
and the Przanowski potential
\be
\hpert=-\log\frac{2\tau_2 r^2}{r+c}\, ,
\label{pert-h0}
\ee
where $\tau\equiv \tau_1+\I\tau_2$
is a fixed complex parameter with $\tau_2>0$, related to the holomorphic prepotential in the rigid case \cite{Bao:2009fg}
\be
F(X) = -\frac{\tau}{2}\, X^2.
\label{prepUHM}
\ee
It is easy to check that with these definitions the equation \eqref{Prz-master-equation} is satisfied and
the metric \eqref{metPrz} reproduces \eqref{1lmetric} provided one sets the cosmological constant to be $\Lambda=-3/2$.

\subsection{Instanton corrected case}

Since we want to write the metric using the physical coordinates $\vph^{m}\equiv(r,\sigma,\zeta,\tzeta)$, we need to take into account not only the linear deformations of the real potential $h(z)$ but also deformations of the complex coordinates $z^i$. 

In such a situation, the variation of the metric \eqref{metPrz} is not the naive 
\be
\delta_h \de s^2=4\delta h_{\alpha\bar\beta} \de z^\alpha\de\bz^\beta+8 e^{\hpert} \delta h |\de z^2|^2,
\label{linPrzmet}
\ee
but gets an additional contribution. To find it, let $(z^\alpha(\vph),h(z))$
and $(\zpert^\alpha(\vph),\hpert(\zpert))$ denote the deformed and non-deformed complex coordinates
and Przanowski potential. Note that $\zpert^\alpha$ are given by \eqref{pert-zz}, where we only take the perturbative part of $r$.
By construction, we have $\Prz_z[h]=\Prz_{\zpert}[\hpert]=0$.
Using these functions, we further define
\be
\delta_\vph h=h(z(\vph))-\hpert(\zpert(\vph)).
\ee
Note that this variation of the Przanowski potential is different from the one used above
which reads $\delta h(z)=h(z)-\hpert(z)$ and satisfies the linearized Przanowski equation.
The relation between these two functions is obtained as follows
\be
\delta h=\delta_\vph h-\(\hpert(z(\vph))-\hpert(\zpert(\vph))\)
=\delta_\vph h -\hpert_\alpha \delta z^\alpha- \hpert_{\bar\alpha} \delta \bz^\alpha,
\label{res-deltah}
\ee
where we introduced $\delta z^\alpha(\vph)=z^\alpha(\vph)-\zpert^\alpha(\vph)$.

Now we can write the full deformation of the metric which keeps the coordinates $\vph^m$ fixed as
\be
\delta \de s^2=\de s^2_{z(\vph)}[h]-\de s^2_{\zpert(\vph)}[\hpert]
=\delta_h \de s^2+\delta_z\de s^2,
\label{deltamet}
\ee
where the first term is defined in \eqref{linPrzmet} and
\be
\delta_z \de s^2=\de s^2_{z(\vph)}[\hpert]- \de s^2_{\zpert(\vph)}[\hpert]
\label{var-metric}
\ee
is the deformation of the Przanowski metric defined by the {\it non-deformed} potential $\hpert$
under the variation of the complex coordinates as functions  of the fields $\vph^m$.
Using \eqref{metPrz} with $\Lambda=-3/2$ and taking into account that our non-perturbed potential \eqref{pert-h0} satisfies
$\hpert_\alpha=\hpert_{\bar \alpha}$,
$\hpert_{\alpha \bar \beta}=\hpert_{\alpha\bar\beta}$, etc., we find an explicit formula
\be
\begin{split}
	\delta_z \de s^2
	=&\, 8\hpert_{\alpha \beta\gamma}\Re(\delta z^\gamma)\de \zpert^\alpha \de \bzpert^\beta
	+8\hpert_{\alpha \beta}\Re(\de\delta z^\alpha \de\bzpert^\beta)
	\\
	&\, + 16 e^{\hpert}\( \hpert_\alpha \Re(\delta z^\alpha)|\de \zpert^2|^2
	+\Re(\de\delta z^2 \de \bzpert^2)\).
\end{split}
\label{coord-Przmt}
\ee

After this preliminary analysis, we give the data for the Przanowski description, which were shown in \cite{Alexandrov:2009vj} to follow directly from the twistorial construction for a generic four-dimensional QK space.
Although this description is not unique because there is a large ambiguity in the choice of coordinates $z^\alpha$
(which also affects the Przanowski potential $h$),
it was found that a particularly convenient choice is given by
\be
z^1 = \frac{\I}{2}\,\ai{+}_0 -2c \log \xi_{-1},
\qquad
z^2 = \frac{\I}{2}\,\txi^{\[+\]}_0,
\label{genz}
\ee
where in our case $\xi_{-1}=\cR$, $\ai{+}$ and $\txii_{+}$ are defined in \eqref{redefDC}.
The Przanowski potential should then be equal to
\be
h=-2\phi+2\log(\xi_{-1}/2)= -2\log \frac{2r}{\cR}\, .
\label{genh}
\ee
Thus, it is sufficient to plug in these identifications into \eqref{coord-Przmt} and \eqref{linPrzmet} and to verify that
the resulting metric reproduces \eqref{ds2final2}. Furthermore, one can check that \eqref{Prz-master-equation} is satisfied.

\section{Small string coupling limit}
\label{sec-NS5-smallgs}

In this section we extract the small string coupling limit to the hypermultiplet metric
that we calculated in section \ref{sec-NS5-Instanton}.
This is the limit where we expect to establish a connection with string amplitudes.
For D-instantons this has already been done in \cite{Alexandrov:2021shf},
therefore here we concentrate on NS5-instantons.
First, we obtain a general structure of the instanton corrected metric following from analysis
of the effective action and string amplitudes. Then we show that exactly the same structure emerges
in the small string coupling limit of the metric \eqref{ds2final2}, thereby providing predictions
for a certain class of string amplitudes in NS5-brane background.
This regime, however,  is still difficult to obtain directly in string theory, so in the last subsection we consider an additional limit of small RR fields which allows us to recover the standard action
and crucially simplifies our predictions for the amplitudes.

\subsection{Instantons from string amplitudes}
\label{subsec-NS5-squareAmplitudes}

The analysis of this subsection is very similar to the one in \cite[\$5 and \$6.1]{Alexandrov:2021shf}
and \cite{Alexandrov:2023hiv}.

Our goal is to relate the metric on the hypermultiplet moduli space to scattering amplitudes of physical fields.
Since the relevant fields are massless scalars, the first non-trivial amplitudes are 4-point functions.
Therefore, we need to generate a 4-point interaction vertex from a metric dependent term in the effective action.
The simplest possibility is to consider the kinetic term for hypermultiplet scalars $\vph^m$ parametrizing $\cM_H$
\be
-{1\over 2} \, \int \de^4 x\( g^{\rm pert}_{mn} +  \sum_\bfg \, e^{-S_\bfg}\( h^{(\bfg)}_{mn}(\vph)+\cdots \)\)
\p_\mu \vph^m \p^\mu \vph^n .
\label{escalar}
\ee
Here we substituted the expected form of the metric in the small string coupling limit where
it is equal to the perturbative metric plus instanton corrections proportional to the exponential of the instanton action.
We kept only NS5-instanton contributions, denoted NS5-instanton action by $S_\bfg$ and
the leading term in the expansion of the tensor multiplying the exponential by $h^{(\bfg)}_{mn}$.
In the limit $g_s\to 0$, we expect that $S_\bfg\sim g_s^{-2}$ and
assume that the fields $\vph^m$ are normalized so that they stay constant.

Let us now expand the fields around their expectation values $\phi^m$. If $\lambda^m=\vph^m-\phi^m$ denotes the fluctuations,
then the expansion of \eqref{escalar} generates infinitely many interaction vertices for these fluctuations. In particular,
the leading instanton contribution to the $\lambda^4$-term is obtained by bringing down
two factors of $\lambda^m$ from the expansion of the instanton action and is given by
\be
-{1\over 4} \sum_\bfg \int \de^4 x\,  e^{-S_\bfg(\phi)}\, \p_p S_\bfg(\phi) \, \p_q S_\bfg(\phi) \,
h^{(\bfg)}_{mn}(\phi)\,  \lambda^p \,
\lambda^q \, \p_\mu\lambda^m \p^\mu\lambda^n .
\label{eleadscalar}
\ee
This term induces an instanton contribution to 4-point functions of fields $\lambda^{m_i}$ which reads as
\be
(2\pi)^4 \delta^{(4)}\left(\sum_i p_{i}\right) e^{-S_\bfg}
\bigg[\p_{m_1} S_\bfg\,  \p_{m_2} S_\bfg \,  h^{(\bfg)}_{m_3 m_4}\,  p_{34}
+\ {\mbox{inequivalent perm.} \atop \mbox{of 1,2,3,4}}
\bigg]\, ,
\label{eampin}
\ee
where $ p_{i}^\mu$ is the momentum carried by $\lambda^{m_i}$ and $p_{ij}=\eta_{\mu\nu} p_{i}^\mu p_{j}^{\nu}$.

The amplitude \eqref{eampin} induced by the term \eqref{eleadscalar} in the effective action is to be compared with
the explicit computation of the instanton amplitude in string theory.
First, we note that NS5-instanton perturbation theory is similar to the one for D-instantons
\cite{Sen:2020cef}, but with open string diagrams ending on D-branes replaced by closed string diagrams
in the presence of an NS5-brane. In particular, the instanton action should be given by the sphere diagram
in the NS5-background which agrees with its scaling as $g_s^{-2}$.
The overall normalization factor should be given by the exponential of the torus diagram,
and each insertion of a closed string vertex operator,
at leading order in $g_s$, gives rise to a factor given by the sphere one-point function of this operator.

This is not the end of the story, however, since the instanton breaks half of the $N=2$ supersymmetry.
The four broken supercharges imply the existence of four Goldstino zero modes.
To get a non-vanishing result from integration over these modes, their vertex operators should
be inserted in the sphere diagrams composing our amplitude.
As a result, schematically, the NS5-instanton contribution to the 4-point function we are interested in is given by
\be
(2\pi)^4 \delta^{(4)}\left(\sum_i p_{i}\right)\bOm_\gamma \cN_\bfg\, e^{-S_\bfg}
\int \[\prod_{\alpha,\dot\alpha=1,2}\de\chi^\alpha\de\chi^{\dot\alpha}\] A_{\bfg,\alpha\dot\alpha\beta\dot\beta}^{m_1m_2m_3m_4}
\chi^\alpha\chi^{\dot\alpha}\chi^\beta\chi^{\dot\beta},
\label{gen4point}
\ee
where $\cN_\bfg$ is the normalization factor computed by torus with removed zero modes,
$\chi^\alpha$ are the fermionic zero modes, and  $A_{\bfg,\alpha\dot\alpha\beta\dot\beta}^{m_1m_2m_3m_4}$ is a sum of products
of four sphere diagrams, each with one closed string vertex operator corresponding to one of $\lambda^{m_i}$,
and four fermion zero modes distributed among the four spheres.
We also included the factor of $\bOm_\gamma$ which for primitive $\gamma$ counts the number of BPS instantons in a given homology class,
and for non-primitive charges takes also into account multi-covering effects. The fact that these effects combine to give the rational
DT invariant \eqref{BPS-rational} can be argued in the same way as for D-instantons \cite{Alexandrov:2021shf}.

This expression can be further simplified, because each sphere diagram must carry even number of fermion zero modes.
Hence only two situations are possible: either all four zero modes are inserted on one sphere, or
two spheres carry two zero modes each and two spheres are without them.
Moreover, one can argue that the former configuration does not contribute.
Indeed, the sphere one-point function of the vertex operator corresponding to $\lambda^m$,
without additional insertions of the fermion zero modes, is simply given by
the derivative of the instanton action $-\p_m S_\bfg$.
In particular, it does not depend on the momentum carried by the vertex operator.
Therefore, all momentum dependence in the case where all four fermion zero modes are
inserted on a single sphere comes from this sphere diagram.
However, the Lorentz invariance implies that it should be a function of $p^2$
where $p$ is the momentum carried by the vertex operator on this sphere. But since
$p^2=0$, this contribution does not depend on $p_{i}$'s at all and would give rise
to a potential term in the effective action.
Since instanton corrections should not generate any potential, this amplitude is expected to vanish.

Thus, the only surviving contribution is the one where we have two of the zero modes on
one sphere, two on another sphere, and two spheres without zero modes which, as we already noted, produce
the factors $-\p_m S_\bfg$. Let us estimate the sphere diagram with the zero mode insertions.
Note that to have a non-vanishing coupling with the momentum vector, one of the zero modes must carry dotted index
and the other one should carry undotted index. Then the full diagram can be represented as
\be
\I \cA^{(\bfg)}_m(\phi) \, p_\mu\, \gamma^\mu_{\alpha\dot\alpha}  ,
\label{eterm}
\ee
where $\cA^{(\bfg)}_m(\phi)$ is a function of background fields independent of the momentum.
Collecting all contributions, we find that
\be
A_{\bfg,\alpha\dot\alpha\beta\dot\beta}^{m_1m_2m_3m_4}=
-\p_{m_1} S_\bfg\,  \p_{m_2} S_\bfg \,
\Bigl(\cA^{(\bfg)}_{m_3}\, p_{3,\mu}\, \gamma^\mu_{\alpha\dot\alpha}\Bigr)
\Bigl( \cA^{(\bfg)}_{m_4} \, p_{4,\nu}\, \gamma^\nu_{\beta\dot\beta} \Bigr)
+\ {\mbox{inequivalent perm.} \atop \mbox{of 1,2,3,4.}}
\ee
Substituting this result into \eqref{gen4point}, integrating over the zero modes, and using the fact that
$\eps^{\alpha\beta}\eps^{\dot\alpha\dot\beta}\,\gamma^\mu_{\alpha\dot\alpha} \,\gamma^\nu_{\beta\dot\beta}
=-\Tr(\gamma^\mu\gamma^\nu)=-2\, \eta^{\mu\nu}$,
one obtains that the NS5-instanton contribution to the 4-point function has
the following form\footnote{We are sloppy here about numerical factors.
	Moreover, as shown in \cite{Alexandrov:2021shf}, there is also an additional factor that must be taken into account coming from
	a difference between the four-dimensional metric in the string frame used to calculate string amplitudes
	and in the frame used to write the effective action \eqref{escalar} with vector and hypermultiplets decoupled.
	We assume that all such factors have been absorbed into $\cN_\bfg$.}
\be
(2\pi)^4 \delta^{(4)}\left(\sum_i p_{i}\right)\bOm_\gamma \cN_\bfg\, e^{-S_\bfg}
\bigg[\p_{m_1} S_\bfg\,  \p_{m_2} S_\bfg \,  \cA^{(\bfg)}_{m_3}\cA^{(\bfg)}_{m_4}\,  p_{34}
+\ {\mbox{inequivalent perm.} \atop \mbox{of 1,2,3,4}}
\bigg].
\label{gen4point-final}
\ee
Comparing \eqref{gen4point-final} with \eqref{eampin}, one finds that they have exactly the same structure.
This allows to extract the metric $h^{(\bfg)}_{mn}(\vph)$:
\be
h^{(\bfg)}_{mn}=\bOm_\gamma \cN_\bfg \cA^{(\bfg)}_m\cA^{(\bfg)}_n.
\label{e56}
\ee

However, the above argument is not quite exact because it is insensitive to the terms in the action \eqref{escalar}
proportional to $\p_m S_\bfg$. Indeed, such terms can be generated either by integration by parts or by
a change of variables involving non-perturbative terms \cite{Alexandrov:2021shf}.
In either case the scattering amplitudes should not be affected and hence \eqref{e56} is valid only up to addition
of the gradient of the instanton action.

To recapitulate, it is convenient to use the language of differential forms.
Let us define $\cA_\bfg=\cA^{(\bfg)}_m(\vph)\de\vph^m$. Then the above analysis of string amplitudes shows that
in the small $g_s$ limit the NS5-instanton contribution to the hypermultiplet metric should be of the form
\be
\de s^2_{\rm NS5}\simeq
\sum_\bfg \bOm_\gamma \cN_\bfg\, e^{-S_\bfg}\( \cA_\bfg^2+\cB_\bfg \de S_\bfg\)
\label{predictNS5}
\ee
with some one-form $\cB_\bfg$ which this analysis cannot fix.
Below we verify that the metric \eqref{ds2final2} does fit this form and find all functions and one-forms
appearing in \eqref{predictNS5} explicitly. On one hand, this provides another non-trivial check on our results,
and on the other hand, gives a prediction for the amplitudes $\cA^{(\bfg)}_m$.

\subsection{Limit from twistor space result}
\label{subsec-NS5-gslimit}

\subsubsection{Definition of the limit}
Before extracting the small string coupling limit, we should definite it.
Namely, we should specify how various fields behave in this limit. Naively, it is enough to send
the variable $r$, related to the dilaton, to infinity and to keep all other variables fixed.
However, this naive limit suffers from a problem.
It is easy to see already for the classical metric obtained from \eqref{1lmetric} by setting $c=0$
that different terms have different scaling in $r$.
This makes it difficult even to formulate what is meant by the leading order metric in the
large $r$ limit.

On the other hand, in \cite{Alexandrov:2021shf} it was noticed that one does get a homogeneous scaling in $g_s$
for both the classical metric and the small string coupling limit of D-instanton corrections provided we take this limit
as
\be
r,\sigma \sim g_s^{-2},
\qquad
\zeta^\Lambda,\tzeta_\Lambda \sim g_s^{-1},
\qquad
z^a \sim g_s^0,
\qquad
g_s\to 0,
\label{scaling}
\ee
which also implies $\cR\sim g_s^{-1}$.
Besides, in this modified limit the D-instanton corrections have been shown to acquire essentially
the same quadratic structure as in \eqref{predictNS5}
and matched exactly against computations of string amplitudes.
This strongly suggests that \eqref{scaling} is the correct limit to consider for NS5-instantons as well.

In fact, the origin of the scaling \eqref{scaling} can be easily understood from the supergravity action in
ten dimensions\footnote{We thank Ashoke Sen for clarification of this issue.}
where the kinetic terms in the NS sector are multiplied by the factor $e^{-2\phi_{(10)}}\sim g_s^{-2}$. Since such a factor is absent in
the RR kinetic terms, the RR fields should scale as $g_s^{-1}$ so that the whole action scales uniformly.
Finally, the scaling of $\sigma$ follows from the dualization of the $B$-field.
Moreover, this rescaling of the RR fields is necessary to match them with their worldsheet counterparts \cite{Polchinski:1998rr}.
Therefore, it is also necessary to establish a correspondence between the small $g_s$ expansion of the effective action
and the genus expansion of string theory. Hence, if we want to derive predictions for any string amplitudes,
we must study the limit \eqref{scaling} rather than the naive one where only $r$ scales with $g_s$.
Below we do it for the NS5-instanton contribution to the metric \eqref{ds2final2}.

\subsubsection{Saddle point evaluation}

As an important preliminary step, let us evaluate in the small $g_s$ limit, as it is defined in \eqref{scaling},
the following integral
\be
\int_{\ell_\bfg}\frac{\de t}{t}\, f(t) \, e^{-2\pi\I k\cS_\bfg(t)},
\label{intft}
\ee
where $f(t)$ is a polynomial in $t$ and $t^{-1}$, and the precise action, including $\Psi_{\hbfg}$ contribution to \eqref{5pqZ2}, is given by
\be
\cS_\bfg(t)=\frac{1}{2} \(\talp+(\xi^\Lambda-2n^\Lambda) \txi_\Lambda\)
- F(\xi-n)-\frac{m_a(\xi^a-n^a) +Q}{k(\xi^0-n^0)}-\frac{m_0}{k}\, ,
\label{effS}
\ee
and all Darboux coordinates in \eqref{effS} are set to their perturbative expressions \eqref{pert-Darboux0}. The terms $m_\Lambda$ are  combinations of the charges and their precise definition is not important for the structure.
This type of integrals multiplies all terms in \eqref{frVtot1} with positive $k$
and thus encodes NS5-instanton corrections to the hypermultiplet metric.

In the limit, the ``effective action" $\cS_\bfg (t)$ can be expanded as $\cS_\bfg(t)=-4\I c  \log t +\sum_{\ell\ge 0} \cS_{\bfg,\ell}$
where $\cS_{\bfg,\ell}$ scales as $g_s^{\ell-2}$ and we extracted the only term having a logarithmic dependence on $t$.
Note that the expansion starts from the term scaling as $g_s^{-2}$,
as is expected for NS5-instantons.
For our purpose, it is sufficient to keep in the exponential only terms with non-positive scaling power, i.e. with $\ell=0,1,2$.
Then the resulting integral can be evaluated by saddle point.
It is easy to see that at the leading order the result is given by
\be
\frac{f(t_0)\, e^{- S_\bfg}}{t_0^{1+8\pi k c}\sqrt{\I k\cS''_{\bfg,0}}}\, ,
\qquad
S_\bfg=2\pi\I k\(\cS_{\bfg,0}+ \cS_{\bfg,1}+\cS_{\bfg,2}-\hf\, \frac{(\cS'_{\bfg,1})^2}{\cS''_{\bfg,0}}\),
\label{saddleresult}
\ee
where all $\cS_{\bfg,\ell}$ and their derivatives (denoted by primes) are evaluated at $t_0$ which is a solution of
the leading order saddle point equation $\cS'_{\bfg,0}=0$.

From \eqref{effS}, we find that
\begin{subequations}
	\bea
	\cS_{0}(t)&=& \hf\(\sigma+\zeta^\Lambda\tzeta_\Lambda\)-\cR^2\Re(\bz^\Lambda F_\Lambda(z))
	+ \cR\zeta^\Lambda \(t^{-1} F_\Lambda(z) -t \bF_\Lambda(\bz)\)
	\nn\\
	&&\qquad
	+\cR^2\(t^{-2} F(z)+t^2 \bF(\bz)\)-F(\xi(t)),
	\label{cS0}
	\\
	\cS_{\bfg,1}(t)&=& - n^\Lambda \[ \tzeta_\Lambda + \cR\(t^{-1}F_\Lambda(z) - t \bF_\Lambda(\bz)\) - F_{\Lambda}(\xi(t))\],
	\label{Sg1}
	\\
	\cS_{\bfg,2}(t)&=&-\frac{1}{2}\,n^\Lambda n^\Sigma F_{\Lambda \Sigma}(\xi(t))-\frac{m_a}{k}\, \frac{\xi^a(t)}{\xi^0(t)}-\frac{m_0}{k}\, .
	\eea
\end{subequations}
Note that we dropped the index $\bfg$ on $\cS_0$ because this part of the effective action does not depend on any charges.
Taking the first derivative of \eqref{cS0}, one finds that the equation on $t_0$ can be written as
\be
t_0^{-1} z^\Lambda \euF_\Lambda+t_0 \bz^\Lambda\beuF_\Lambda=-\I\cR K,
\label{eqt0}
\ee
where we introduced
\be
\begin{split}
	\euF_\Lambda &= F_\Lambda(\xi(t_0)) - \xi^\Sigma(t_0) F_{\Lambda\Sigma}(z),
	\\
	\beuF_\Lambda &= F_{\Lambda}(\xi(t_0)) - \xi^\Sigma(t_0) \bF_{\Lambda\Sigma}(\bz).
\end{split}
\label{reverseSaddle}
\ee
Note that for generic prepotential this equation is highly non-linear and cannot be solved explicitly, while in the rigid case
($h^{2,1}(\CY)=0$) where $F(X)$ is quadratic and given by \eqref{prepUHM} one finds
$\euF_0=0$, $\beuF_0=-2\I\tau_2 \xi(t_0)$ and $t_0=\zeta/\cR$.
The second derivative appearing in \eqref{saddleresult} is found to be
\be
\begin{split}
	t_0^2 \cS''_{0}(t_0)
	=&\,
	\cR\Bigl[ \zeta^\Lambda\( t_0^{-1} F_{\Lambda}(z)-t_0 \bF_\Lambda(\bz)\)  +4\cR \(t_0^{-2} F(z) +t_0^2 \bF(\bz)\)
	\\
	&\, -\(\zeta^\Lambda\(t_0^{-1}z^\Sigma-t_0\bz^\Sigma\) +2\cR\(t_0^{-2} z^\Lambda z^\Sigma+t_0^2 \bz^\Lambda\bz^\Sigma\)\)
	F_{\Lambda\Sigma}(\xi(t_0))\Bigr].
\end{split}
\ee
Finally, the instanton action defined in \eqref{saddleresult} is given by
\bea
S_\bfg &=&
2\pi\I k\Biggl[
\hf\(\sigma+ \zeta^\Lambda \tzeta_\Lambda\) - \cR^2 \Re (\bz^\Lambda F_\Lambda)
+\cR \zeta^\Lambda (t_0^{-1} F_\Lambda-t_0 \bF_\Lambda)  + \cR^2 \(t_0^{-2} F +  t_0^2 \bF\)  - F(\xi(t_0))
\nn\\
&&-n^\Lambda\( \tzeta_\Lambda + \cR\(t_0^{-1}F_\Lambda - t_0 \bF_\Lambda \)- F_{\Lambda}(\xi(t_0))\)
-\frac{1}{2}\,n^\Lambda n^\Sigma F_{\Lambda \Sigma}(\xi(t_0)) - \frac{m_a}{k}\,\frac{\xi^a(t_0)}{\xi^0(t_0)} -\frac{m_0}{k}
\nn\\
&&- \frac{\cR^2}{2t_0^2\, \cS''_{\bfg,0}(t_0)}\, \Bigl(n^\Lambda \Bigl( t_0^{-1}F_\Lambda +t_0\bF_\Lambda
- \(t_0^{-1}z^\Sigma +t_0 \bz^\Sigma\) F_{\Lambda \Sigma}(\xi(t_0))\Bigr)\Bigr)^2\Biggr] .
\label{exp-saddle-I0}
\eea

The result \eqref{exp-saddle-I0} appears to be quite complicated
and its physical significance is not clear to us. However, one can note that all complications come from keeping the RR fields large
so that $t_0$ remains finite and $F(\xi(t_0))$ does not reduce to $F(z)$.
Probably it is not too surprising that large RR fields lead to a weird instanton action since they couple to the self-dual 3-form
living on the world-volume of the NS5-brane, which makes the problem inherently quantum.
Below, in section \ref{subsec-NS5-smallRR} we show that making the background RR fields small,
one reduces \eqref{exp-saddle-I0} to the expected instanton action.
Nevertheless, even without taking this additional limit, we are able to show that NS5-corrections to the hypermultiplet metric
match the quadratic structure \eqref{predictNS5} predicted by the analysis of string amplitudes.

\subsubsection{The metric and its square structure}

Using the results of the previous subsection, we conclude that at the leading order in the limit \eqref{scaling} one has
\be
\cIn{0}_\bfg\approx
- \frac{k}{p^0}\,\frac{\xi^0(t_0)}{t_0^{1+8\pi k c}}\,\frac{ e^{-S_\bfg}}{\sqrt{\I k \cS_{0}''(t_0)}}\, ,
\label{saddleI0-neg}
\ee
while all other integral functions, in this limit, are proportional to it.

Let us now see how the one-instanton corrected metric \eqref{ds2final2} simplifies in our limit.
We will consider only NS5-corrections given by $\frV_\bfg$ \eqref{frVtot1} with positive $k$
and extract its leading order contribution.
Then we can use the leading order result \eqref{saddleI0-neg} and express all other integrals in terms of $\cIn{0}_\bfg$.
In addition, there are the following simplifications:
\begin{itemize}
	\item
	One can drop all terms proportional to the one-loop parameter $c$ since they are always of subleading order.
	\item
	The terms in the first three lines of \eqref{frVtot1} are subleading compared to the rest of the expression and thus can also be dropped.
	\item
	The variables $r$ and $\cR$ can be exchanged (even in the perturbative part of the metric)
	using the classical relation $r=\cR^2 K/4$.
\end{itemize}
As a result, the NS5 one-instanton contribution reduces to
\bea
\frV_\bfg&\approx&
-2\pi k\cIn{0}_\bfg \Biggl[
\frac{4}{\cR} \,\biggl(
(\de r)^2 -\frS^2 + r^2\,\frac{(\de K)^2}{K^2}-4r^2\cA_K^2
-\frac{2r}{K}\,|\frJ|^2
\biggr)
\nn\\
&&
+\frac{\cR}{2r^2} (z^\Lambda \euF_\Lambda) (\bz^\Sigma\beuF_\Sigma)
\biggl( (\de r)^2-\frS^2
-2r\de r\, \de\log K+4r\cA_K\frS
+r^2\,\frac{(\de K)^2}{K^2}-4r^2\cA_K^2\biggr)
\nn\\
&&
+ \frac{2}{r}\,\euF\Bigl(\de r\,\frS -2r^2\cA_K\de\log K\Bigr)
-\frac{\I\cR}{2r}\, \euF\(t_0^{-1}\frJ + t_0 \bar\frJ\) \(\frS-2r\cA_K\)
\nn\\
&&
-2\I\( \de r+r\de\log K\)\Bigl(t_0^{-1}\frJ  + t_0 \bar\frJ
-\I\cR^{-1} N^{\Lambda\Sigma} \(\beuF_\Lambda \frZ_\Sigma + \euF_\Lambda \bar\frZ_\Sigma\)
+t_0^{-1}\euF_\Lambda\de z^\Lambda + t_0 \beuF_\Lambda \de \bz^\Lambda \Bigr)
\nn\\
&&
+2\I N^{\Lambda \Sigma}\Bigl(t_0\beuF_\Lambda\frZ_\Sigma\bar\frJ +t_0^{-1}\euF_\Lambda\bar\frZ_\Sigma\frJ \Bigr)
-2\cR\Bigl(\euF_\Lambda \de z^\Lambda\bar\frJ+ \beuF_\Lambda \de \bz^\Lambda\frJ\Bigr)
\label{smallfrV}\\
&&
-\frac{1}{2r}\, \euF\(\frS-2r\cA_K\)
\Bigl(N^{\Lambda \Sigma}\(\beuF_\Sigma\frZ_\Lambda + \euF_\Sigma\bar\frZ_\Lambda \)
+\I \cR \(t_0^{-1}\euF_\Lambda \de z^\Lambda  + t_0 \beuF_\Lambda \de \bz^\Lambda \) \Bigr)
\nn\\
&&
+2\euF_\Lambda \beuF_\Sigma \Bigl(\cR^{-1}N^{\Lambda\Xi} N^{\Sigma\Theta} \frZ_\Theta \bar\frZ_\Xi-\cR \de z^\Lambda \de\bz^\Sigma\Bigr)
+2\I N^{\Lambda\Sigma} \Bigl(
t_0^{-1}\euF_\Sigma \euF_\Xi\bar\frZ_\Lambda\de z^\Xi
+t_0\beuF_\Sigma \beuF_\Xi\frZ_\Lambda \de\bz^\Xi
\Bigr))
\Biggr]
\nn\\
&-&
\[ \frac{\I}{2}\,\euF\,\de\log K
+\frac{4\I r}{\cR}\,\cA_K
-\frac{1}{\cR}\, N^{\Lambda \Sigma}\(\euF_\Lambda \bar\frZ_\Sigma - \beuF_\Lambda \frZ_\Sigma\)
+ \frac{\cR}{2r}\(z^\Lambda \euF_\Lambda \bar\frJ - \bz^\Lambda \beuF_\Lambda \frJ\) \]\de \cIn{0}_\bfg,
\nn
\eea
where we introduced another convenient notation
\be
\euF= t_0^{-1} z^\Lambda \euF_\Lambda- t_0 \bz^\Lambda \beuF_\Lambda.
\label{def-euF}
\ee

It is straightforward to verify that the expression \eqref{smallfrV} can be rewritten as
\be
\frV_\bfg \approx -\frac{\pi k}{\cR}\, \cIn{0}_\bfg \Bigl(\cA^2+\cB \de \cS_0 \Bigr),
\label{quadfrV}
\ee
where
\be
\begin{split}
	\cA =&\,
	2\de r + \frac{\cR\euF}{2r} \, \frS - \I\cR\(t_0^{-1}\frJ+ t_0 \bar\frJ \)
	- N^{\Lambda \Sigma}\(\euF_\Lambda \bar\frZ_\Sigma+ \beuF_\Lambda \frZ_\Sigma \)
	\\
	&\,
	+ 2r\de \log K - \cR\euF \cA_K
	-\I\cR \(t_0^{-1}\euF_\Lambda \de z^\Lambda +t_0 \bar\euF_\Lambda \de \bz^\Lambda \),
\end{split}
\label{cA1form-square}
\ee
\be
\begin{split}
	\cB =&\,
	\frac{\cR\euF}{2r}\,\de r-2 \frS
	+ \I N^{\Lambda \Sigma}\(\euF_\Lambda \bar\frZ_\Sigma- \beuF_\Lambda \frZ_\Sigma \)
	- \frac{4\I}{K}\(z^\Lambda \euF_\Lambda \,\bar\frJ- \bz^\Lambda \beuF_\Lambda\, \frJ\)
	\\
	&\,
	- \cR\(t_0^{-1}\frJ- t_0 \bar\frJ \)
	+ \frac{1}{2}\,\cR\euF\, \de\log K + 4r\cA_K
	- \cR \(t_0^{-1}\euF_\Lambda  \de z^\Lambda - t_0  \beuF_\Lambda \de \bz^\Lambda \)
\end{split}
\ee
and
\be
\begin{split}
	\de\cS_{0}=&\,
	-\[\frac{\cR\euF}{2r}\,\de r -2\frS
	+\I N^{\Lambda \Sigma}\(\euF_\Lambda \bar\frZ_\Sigma- \beuF_\Lambda \frZ_\Sigma \)\]
	\\
	&\,
	-\cR\(t_0^{-1} \frJ - t_0 \bar\frJ\)
	+\frac{1}{2}\,\cR\euF\,\de\log K+4r\cA_K
	-\cR \(t_0^{-1}\euF_\Lambda  \de z^\Lambda - t_0  \beuF_\Lambda \de \bz^\Lambda \) .
\end{split}
\label{fulldS0}
\ee
Note that the only dependence on the charge vector is in the overall coefficient, while $\cA$ and $\cB$ are charge independent.
One can also check that for negative $k$ the result is obtained by complex conjugation, namely, $\frV_{-\bfg}=-\bar\frV_\bfg$.
Therefore, combining \eqref{ds2final2}, \eqref{saddleI0-neg} and \eqref{quadfrV}, one finds that the full NS5-instanton correction
to the hypermultiplet metric in the small string coupling limit is given by
\be
\de s^2_{\rm NS5}\simeq
\frac{1}{4\pi r }\sum_{\bfg\; :\; k>0} \bOm_\gamma \, \frac{k^3}{p^0}\,
\Re \[\frac{\xi^0(t_0)}{t_0^{1+8\pi k c}}\,\frac{ e^{-S_\bfg}}{\sqrt{\I k \cS_{0}''(t_0)}}\,
\Bigl( \cA^2+\cB \de \cS_0\Bigr)\].
\label{resultNS5}
\ee
This is precisely the form \eqref{predictNS5} of the instanton contribution that we found from the analysis of string amplitudes.
Furthermore, comparing  \eqref{predictNS5} and \eqref{resultNS5}, we can identify (for positive $k$)
\be
\begin{split}
	& \cA_\bfg=f_\bfg \,\cA+g_\bfg\, \de \cS_0,
	\qquad
	\cB_\bfg=\frac{1}{2\pi\I k} \(f_\bfg^2\cB-2f_\bfg g_\bfg \cA-g_\bfg^2\de \cS_0\),
	\\
	&\hspace{3cm}
	\cN_\bfg=\frac{k^3}{8\pi r p^0}\,\frac{\xi^0(t_0)}{t_0^{1+8\pi k c}}\, \frac{f_\bfg^{-2}\, }{\sqrt{\I k \cS_{0}''(t_0)}}\, ,
\end{split}
\label{predict}
\ee
where $f_\bfg$ and $g_\bfg$ are {\it a priori} unknown functions of the moduli.
It is tempting to speculate that $g_\bfg=0$ and $f_\bfg$ is a constant.
But even keeping these functions arbitrary, the identifications \eqref{predict}
provide a large set of predictions for the amplitudes $\cA^{(\bfg)}_m(\phi)$ of one closed string vertex operator
and two fermion zero modes in the NS5-brane background.

\subsection{The limit of small RR fields}
\label{subsec-NS5-smallRR}

Since, at present, calculation of string amplitudes in a non-trivial RR-background appears to be an outstanding problem,
we find it natural to consider our results in the additional limit of small $\zeta,\tzeta$.

We start by analyzing the saddle point equation \eqref{eqt0}.
It is easy to realize that its solution is proportional to $\zeta^\Lambda/\cR$, which is an exact result for a quadratic prepotential.
Therefore, in the limit of small RR fields, the first and third terms in the expression for $\xi_{\rm pert}^\Lambda $ \eqref{pert-Darboux0}
are suppressed by two orders comparing to the second term. Expanding around it, we find
\be
\begin{split}
	\euF_\Lambda =&\,
	-\frac{t_0}{2\cR}\, F_{\Lambda\Sigma\Theta}(z) \(\zeta^\Sigma-\cR t_0 \bz^\Sigma\)\(\zeta^\Theta-\cR t_0 \bz^\Theta\)
	+O(g_s^4\zeta^5),
	\\
	\beuF_\Lambda =&\,  -\frac{\I\cR}{t_0}\, N_{\Lambda\Sigma} z^\Sigma -\I N_{\Lambda\Sigma}\(\zeta^\Sigma-\cR t_0 \bz^\Sigma\)
	+O(g_s^2\zeta^3).
\end{split}
\ee
Substituting these expansions into \eqref{eqt0}, it is easy to solve the resulting equation on $t_0$. This gives
\be
t_0=\frac{N_{\Lambda\Sigma}\bz^\Lambda \zeta^\Sigma}{\cR N_{XY}\bz^X \bz^Y}+O((g_s\zeta)^3),
\label{t0-small}
\ee
consistently with the expectation that $t_0\sim \zeta/\cR$.

We can now perform the same expansion in the instanton action \eqref{exp-saddle-I0}.
One can observe that ignoring the next order term in \eqref{t0-small} corresponds to ignoring the terms of order $O(g_s^2\zeta^{2+n})$
with $n=0,1,2$ in contributions that scale as $g_s^{-n}$ in the limit \eqref{scaling}.
Dropping such terms and taking into account that
\be
\cS''_{0}(t_0)\approx \I \cR^2 N_{\Lambda\Sigma}\bz^\Lambda\bz^\Sigma,
\ee
it is straightforward to check that $S_\bfg$ reduces to the NS5-brane instanton action
\be
S_\bfg^{(0)}=4\pi k r +\pi\I k \left(  \sigma + \zeta^\Lambda \tzeta_\Lambda
-2 n^\Lambda \tzeta_\Lambda
- \bar\cN_{\Lambda\Sigma} (\zeta^\Lambda-n^\Lambda) (\zeta^\Sigma-n^\Sigma)\right)
-2\pi \I m_\Lambda z^\Lambda,
\label{NS5instact}
\ee
(plus a trivial constant term $4\pi k c$). 
In particular, this reproduces the result from \cite{Alexandrov:2010ca} where the naive limit $r\to\infty$ without scaling other fields was used
to extract the same expression.
On one hand, this expression reproduces the result of a classical analysis of instanton solutions in $N=2$ supergravity \cite{deVroome:2006xu}
and on the other hand, it makes contact with the Gaussian NS5-partition function obtained by holomorphic factorization.
This establishes a link with the known results about these instantons and shows that they should emerge
from string amplitudes only in the approximation of small RR fields.

This double limit procedure might seem equivalent to the naive single limit where all fields are
fixed and only $r$ scales, but this is not the case.
Indeed, there are terms in the metric, that survive the naive limit but are dropped
in the limit \eqref{scaling} even before taking the RR fields to be small.
Had we scaled only $r$ from the start, these terms would have remained relevant and would change our results.
For example, we would have to change the saddle point \eqref{t0-small} by replacing $\zeta^\Lambda$
by $\zeta^\Lambda-n^\Lambda$.
The reason why it is the double limit rather than the naive one that should be considered is
our interest in predictions for string amplitudes.
The point is that the first limit \eqref{scaling} evaluated in the previous subsection ensures a relation between
various expansion terms and string amplitudes, while the second limit of small RR fields is supposed to be taken
already in each such term separately. In this way it simply gives the corresponding string amplitudes
in a particular region of the moduli space.
Instead, the naive limit mixes contributions from different string diagrams.
For example, in \eqref{NS5instact} the terms linear in $n^\Lambda$ originate
from $\cS_{\bfg,1}$ \eqref{Sg1} scaling as $g_s^{-1}$ and therefore are expected to capture
disk amplitudes with boundary on D-branes bound to the NS5-brane, while in the naive limit
they have a trivial scaling and are mixed with sphere contributions from $\cS_0$.

Finally, we evaluate the limit of small RR fields for the one-forms \eqref{cA1form-square}-\eqref{fulldS0},
which according to our reasoning should provide predictions for the same limit of the sphere three-point functions.
To this end, it is also useful to note that for the function defined in \eqref{def-euF} one obtains
a very simple result
\be
\euF=\I\cR K+O(g_s^3\zeta^4).
\ee
Then keeping only terms that are at most quadratic in the RR fields and using notation
\be
\hzeta^\Lambda=\zeta^\Lambda-\frac{N_{\Sigma\Theta}\bz^\Sigma \zeta^\Theta}{N_{XY}\bz^X \bz^Y}\, \bz^\Lambda,
\ee
one gets\footnote{See \cite[Footnote 14]{Alexandrov:2023hiv} for some subtleties in the expression of $\cB$.}
\be
\begin{split}
	\hspace{-0.4cm}
	\cA \approx &\,
	2\de r + 2\I \frS+\I\zeta^\Lambda\frZ_\Lambda -\I\cR t_0 \bz^\Lambda(\frZ_\Lambda+\bar\frZ_\Lambda)
	-\frac{\I}{2}\,F_{\Lambda\Sigma\Theta}\hzeta^\Lambda\hzeta^\Sigma \de z^\Theta
	- \cR t_0 N_{\Lambda\Sigma}\hzeta^\Lambda \de \bz^\Sigma,
	\\
	\hspace{-0.4cm}
	\cB \approx&\,
	2\I \de r -2 \frS+\frac{2\cR}{t_0}\, \frJ -\zeta^\Lambda\frZ_\Lambda +\cR t_0 \bz^\Lambda(\frZ_\Lambda+\bar\frZ_\Lambda)
	-\frac{1}{2}\,F_{\Lambda\Sigma\Theta}\hzeta^\Lambda\hzeta^\Sigma \de z^\Theta
	- \I\cR t_0 N_{\Lambda\Sigma}\hzeta^\Lambda \de \bz^\Sigma,
	\\
	\hspace{-0.4cm}
	\de\cS_{0}\approx &\,
	-2\I \de r +2 \frS +\zeta^\Lambda\frZ_\Lambda -\cR t_0 \bz^\Lambda(\frZ_\Lambda-\bar\frZ_\Lambda)
	-\frac{1}{2}\,F_{\Lambda\Sigma\Theta}\hzeta^\Lambda\hzeta^\Sigma \de z^\Theta
	- \I\cR t_0 N_{\Lambda\Sigma}\hzeta^\Lambda \de \bz^\Sigma.
\end{split}
\label{limit-onrforms}
\ee
It is interesting that both $\cA$ and $\cB$ are very similar to $\de \cS_0$.
In particular, one has a very simple relation
\be
\cA\approx\I\de\cS_0-2\cR t_0 \(\I\bar\frJ+ N_{\Lambda\Sigma}\hzeta^\Lambda \de \bz^\Sigma\).
\ee

Combining the one-forms \eqref{limit-onrforms} with the identifications \eqref{predict}, one obtains predictions
for the sphere three-point functions in the NS5-background in the limit of small RR fields.
This can be viewed as one of the main results of our work.

\chapter{Quantum Riemann-Hilbert problem}
\label{chap-qRH}

The D-instanton corrected Darboux coordinates, or their Fourier modes, can be obtained as solutions to a Riemann-Hilbert (RH) problem induced by the BPS indices. Such a RH problem was introduced in \cite{Gaiotto:2008cd}, in the context of four-dimensional $N=2$ gauge theories. 
Its solutions describe the geometry of the moduli space of circle compactifications of 4d $N=2$ super-Yang-Mills theory. The same problem also describes the hypermultiplet moduli space in Calabi-Yau string compactifications \cite{Alexandrov:2008gh, Alexandrov:2011va}.
Furthermore, a variation of the problem, obtained by taking the conformal limit, was studied in \cite{Gaiotto:2014bza,Bridgeland:2016nqw}. Its solutions give rise to interesting geometric structures \cite{Bridgeland:2019fbi, Alexandrov:2021wxu}, encoded in the Joyce and \pleb potentials which satisfy isomonodromy and heavenly equations, respectively. 
These functions are related to instanton generating functions and to the contact potential \cite{Alexandrov:2018lgp, Alexandrov:2021wxu}. Crucially, this classical problem can be recast into a more tractable form as a Thermodynamic Bethe Ansatz (TBA)-like integral equation \eqref{TBAeq-sh}.

This chapter investigates a quantized version of this problem. A similar deformation of the original problem was already considered in \cite{Alexandrov:2019rth, Cecotti:2014wea, Barbieri:2019yya, Chuang:2022uey}. Our work here complements and extends these previous efforts in formulating, analyzing, and solving this problem. 
Non-commutativity is introduced by turning on a refinement parameter, which is known in some cases to effectively quantize the moduli space \cite{Gaiotto:2010be,Cecotti:2014wea,Alexandrov:2019rth}. At the same time, the BPS invariants are replaced by their refined counterparts \eqref{refBPS-Laurent-pol}.
While these refined invariants are not, in general, deformation invariant, they crucially satisfy the refined Kontsevich-Soibelman (KS) wall-crossing formula. This property is sufficient to set up a well-defined quantum Riemann-Hilbert problem, where the refinement parameter introduces a fundamental non-commutativity. One can regard the solutions to this problem as refined analogues of Darboux coordinates on the quantized moduli space.

Previous work in \cite{Alexandrov:2019rth} proposed an integral equation based on the non-commutative Moyal star product as a candidate for describing this quantum system. However, the solutions to that equation did not directly solve the quantum RH problem and, critically, lacked a well-defined classical (unrefined) limit. Despite these shortcomings, they proved useful for defining quantum analogues of the Joyce and \pleb potentials, that solve deformed versions of the isomonodromy and heavenly equations \cite{Alexandrov:2021wxu}. A central result of this chapter is to resolve these issues. 

We introduce a new set of variables, constructed from the solutions of the integral equation in \cite{Alexandrov:2019rth}, that do solve the quantum RH problem. We prove that these new solutions possess a regular unrefined limit and demonstrate that they correctly reduce to solutions of the classical RH problem in that limit. We further provide multiple perturbative expansions for these functions, which in turn give a new representation for the classical solutions themselves.

Interestingly, these new solutions can be derived from the action of a single, charge-independent potential. We propose that this object should be interpreted as a generating function for the refined Darboux coordinates. While this potential itself does not have a smooth unrefined limit, we show that its logarithm does. We prove that in the unrefined limit, this logarithm reduces to the classical generating function for Darboux coordinates, whose expression was previously known only up to the second order in DT invariants \cite{Alexandrov:2017qhn}. Our work not only verifies this agreement at second order but also provides an explicit expression for this unrefined generating function valid to all orders in the instanton expansion.

Finally, we restrict to the case of the so-called uncoupled BPS indices that was also studied in \cite{Barbieri:2019yya}. In this case, the expression of the solutions to the quantum RH problem gets significant simplifications that allow us to compute it explicitly and provide the result in terms of the modified Gamma function. This explicit expression matches the one found in \cite{Barbieri:2019yya}. 

This chapter follows closely \cite{Alexandrov:2025abc} while focusing on one setup among the three explored in that paper.

\section{Classical Riemann-Hilbert problem}
\label{sec-qTBA-classical}

We will first define the notion of (unrefined) BPS structure, which axiomatizes the main ingredients needed to pose the Riemann-Hilbert problem. Then, we introduce the setup and formulate the RH problem. Finally, we give a perturbative solution of this problem.

\subsection{Unrefined BPS structure}
\begin{definition}
\label{definition-BPS}
A BPS structure is given by 
\begin{enumerate}
	\item a finite-rank free abelian group $\Gamma\simeq \IZ^{\oplus 2n}$ with a skew-symmetric bilinear form:
	\be
	\langle -,-  \rangle : \Gamma \times \Gamma \xrightarrow{} \IZ,
	\ee
	
	\item a homomorphism of abelian groups 
	\be
	Z: \Gamma \to \IC,
	\ee
	
	\item a map $\Omega:\Gamma \to \IQ$,
\end{enumerate}
which have to satisfy the following conditions 
\begin{itemize}
	\item symmetry: $\Omega(\gamma)=\Omega(-\gamma)$ for all $\gamma\in\Gamma$, 
	\item support property: given a norm $||\cdot||$ on the vector space $\Gamma\otimes_\IZ \IR$, there is a constant $C>0$ such that 
	\be
	\Omega(\gamma)\neq 0 \implies |Z(\gamma)|>C\cdot ||\gamma||.
	\ee
\end{itemize}
\end{definition} 

The BPS indices considered in this thesis, satisfy this definition.
In the language of the previous chapters, which we will continue using throughout this chapter, the lattice $\Gamma$ corresponds to the charge lattice, the function $Z$ to the central charge and the skew-symmetric product $\langle -,- \rangle$ to the Dirac product. The map $\Omega$ corresponds to the BPS indices \eqref{BPS-Trace}. Despite them being valued in $\IQ$, we will assume that this map is always integer\footnote{For mathematicians, these indices are defined in a different way and are only conjectured to be integer.}, which is the case for \eqref{BPS-Trace}. We also have rational invariants defined from $\Omega(\gamma)$ using \eqref{BPS-rational}.
Finally, a charge $\gamma\in\Gamma$ will be called \textit{active} if $\Omega(\gamma)\neq0$.

\subsection{Riemann-Hilbert problem and TBA equations}
\label{subsec-qTBA-TBA}

The problem that we will define is the RH problem introduced in \cite{Bridgeland:2016nqw}.
In this case, we start with a $4n$-dimensional complex \hk manifold $\cM$, given by a torus fibration over the space of stability conditions of the Calabi-Yau \cite{Bridgeland:2019fbi}.
We take complex coordinates on this space $\theta^a$ and $z^a$ indexed by $a=1,\dots,2n$.
Then, given the lattice $\Gamma$ we have two symplectic vectors $\theta, z$, parametrized by these coordinates, such that $z^a=Z_{\gamma^a}$ and $\theta^a=\theta_{\gamma^a}$ and for any charge $\gamma\in\Gamma$
\be
\label{thet-Z}
\theta_\gamma=\langle \gamma,\theta\rangle,
\qquad
Z_\gamma=\langle \gamma,z\rangle.
\ee


Using these vectors we define the functions
\be
\cXsf_\gamma(t) =  e^{2\pi\I(\theta_{\gamma}-Z_\gamma/t)}.
\label{cXsf}
\ee 	
Then, we take the BPS rays $\ell_\gamma$ defined in \eqref{rays} and
for each one, we denote
\be 
\Gamma_\ell=\{\gamma\in\Gamma_\star\ :\ \I Z_{\gamma}\in \ell, \ \Omega(\gamma)\ne 0\},
\label{setGl}
\ee  
where $\Gamma_{\!\star}=\Gamma\backslash\{0\}$ and we see that for an active ray $\Gamma_\ell$ is non-empty. Finally, we define the Riemann-Hilbert problem,

{\bf RH problem:} {\it  Find piece-wise holomorphic functions $\cX_\gamma(t)$ such that
	\begin{enumerate}
		\item 
		$\cX_\gamma\cX_{\gamma'}=\cX_{\gamma+\gamma'}$;
		\item 
		at $t\to 0$ the functions $\cX_\gamma$ reduce to $\cXsf_\gamma$ and at $t\to \infty$ they behave polynomially in $t$;
		\item 
		they jump across active BPS rays in such a way that, if  
		$\cX_\gamma^\pm$ are values on the clockwise and anticlockwise sides of $\ell$, respectively,
		then they are related by the Kontsevich-Soibelman (KS) symplectomorphism 
		\be 
		\cX_\gamma^-=\cX_\gamma^+ \prod_{\gamma'\in\Gamma_\ell}
		\(1-\sigma_{\gamma'} \cX_{\gamma'}^+\)^{\Omega(\gamma')\langle\gamma',\gamma\rangle},
		\label{KSjump}
		\ee 
		where $\sigma_\gamma$ is a sign factor known as quadratic refinement and defined by the relation
		$\sigma_\gamma\, \sigma_{\gamma'} = (-1)^{\langle\gamma,\gamma'\rangle} \sigma_{\gamma+\gamma'}$.
	\end{enumerate} 
}

Geometrically, the functions $\cX_\gamma(t)$ play the role of Darboux coordinates
on the twistor space over the moduli space $\cM$, where the variable $t$ parametrizes 
the $\IP^1$ fiber of the twistor fibration. Hence, solving the RH problem gives a complete description of $\cM$.
The first condition means that there are only $2n$ independent functions $\cX_\gamma(t)$.
The second condition constrains the asymptotic behavior of the solutions. 
Finally, the third condition is the KS wall-crossing formula and it illustrates the relation to BPS indices. Notice that if all of the latter vanish, the jumps become trivial and the functions \eqref{cXsf} would solve the RH problem. 
In that case the $\cXsf_\gamma(t)$ determine the geometry of $\cM$, which is given by a torus fibration $T^{2n}\to \cM\to\cB$, hence
the index `sf' meaning `semi-flat'.

The above RH problem can be reformulated as a system of TBA-like equations.
These equations are simplified if one trades the integer BPS indices for the rational ones \eqref{BPS-rational}, where they can be written in the following concise form
\be
\cX_0=\cX_0^{\rm sf}
\exp\[\IS_1 K_{01} \cX_1 \].
\quad
\label{TBAeq-sh}
\ee
Here we used notations introduced in \cite{Alexandrov:2021wxu}:
$\cX_i=\cX_{\gamma_i}(t_i)$, and 
\be
\IS_i=\sum_{\gamma_i\in\Gamma_{\!\star} }
\frac{\sigma_{\gamma_i}\bOm(\gamma_i)}{2\pi\I}\int_{\ell_{\gamma_i}}\frac{\de t_i}{t_i^2}\, ,
\qquad
K_{ij}=\gamma_{ij}\,\frac{t_i t_j}{t_j-t_i}\,,
\label{short3}
\ee
with $\gamma_{ij}=\langle\gamma_i,\gamma_j\rangle$.
The equation \eqref{TBAeq-sh} is equivalent to the the RH problem in all three setups considered in \cite{Alexandrov:2025abc}. In each of them, the form of the TBA equation remains unchanged while the weighted sum-integral \tIS{i} and the kernel $K_{ij}(t_i,t_j)$ are modified. 
Notice that the latter has a simple pole at $t_i=t_j$.
The idea is that crossing a BPS ray, one picks up the residue at the pole and that produces the KS jump \eqref{KSjump}.

\subsection{Perturbative Solution}
The TBA equation \eqref{TBAeq-sh} can be solved by iterations. This results in a formal expansion in powers of
BPS indices given by \cite{Filippini:2014sza}
\be
\label{Hexpand}
\cX_0 = \cXsf_0 \, \sum_{n=0}^{\infty} \left(
\prod_{i=1}^n  \IS_i\,  \cXsf_i \right) \sum_{\cT\in \IT_{n+1}^{\rm r}}\,
\frac{\prod_{e\in E_T}K_{s(e)t(e)}}{|\Aut(\cT)|}\, ,
\ee
where $\IT_n^{\rm r}$ is the set of rooted trees with $n$ vertices, each vertex
of the tree is labeled with a charge $\gamma_i\in\Gamma_{\!\star} $, with  $\gamma_1$ associated
to the root vertex, $E_T$ is the set of edges of $T$, 
and $s(e)$, $t(e)$ denote the source and target\footnote{We always orient the edges away from the root.} vertices of edge $e$, respectively. 
One of the results of our paper \cite{Alexandrov:2025abc} is the simpler expression \eqref{compact-hcXr-lim} for the solution $\cX_0$, that will be given later and which is valid in both setup 1 and setup 3 defined above.

In \cite{Bridgeland:2019fbi}, the Joyce and \pleb potentials were introduced as functions that encode the geometry of the complex \hk space $\cM$.
Although they are not directly needed in our work, they constitute an additional motivation for studying the Riemann-Hilbert problem, thus we will give some properties of the \pleb potential for illustration. The solutions of the classical RH problem can be used \cite{Alexandrov:2021wxu} to define a function on $\cM$ 
\be 
W= \IS_1 \cX_1\( 1-\hf\, \IS_2 K_{12} \cX_2\),
\label{Wfull-short}
\ee 
which satisfies the heavenly equation given by 
\be
\label{Heavenly-eqn}
\frac{\partial^2 W}{\partial z^a \partial \theta^b}-\frac{\partial^2W}{\partial z^b\partial\theta^a} =
\frac{1}{(2\pi)^2}\sum_{c,d=1}^{m} \omega^{cd} \frac{\partial^2 W}{\partial \theta^a\partial \theta^c}\frac{\partial^2 W}{\partial \theta^b \partial \theta^d},
\ee
where $\omega^{ab}=\langle\gamma^a,\gamma^b\rangle$ and the $\gamma^a$ form a basis of $\Gamma$. 
It coincides with the \pleb potential introduced\footnote{In the paper \cite{Bridgeland:2020zjh}, the function $W$ is called \textit{Joyce potential}.} 
in \cite{Bridgeland:2020zjh}.
This potential encodes the \hk geometry on $\cM$ \cite{Bridgeland:2019fbi}.
In the context of integrable systems, and in relation to the TBA-like equation \eqref{TBAeq-sh}, it is the conformal limit of the critical value of the Yang-Yang functional \cite{Alexandrov:2010pp}. 
In another setup describing the large volume limit of D3-instanton correction to the hypermultiplet moduli space in type IIB string theory on a CY, it coincides with the instanton generating potential introduced in \cite{Alexandrov:2018lgp} but does not solve the corresponding \pleb heavenly equation.

\section{Quantum Riemann-Hilbert problem}
\label{sec-qTBA-qTBA}

We start this section by defining the Moyal product and formulating the quantum RH problem. Then, we give the new variables that solve it and provide a perturbative expression for them. While these variables are defined in terms of the $\cXr_\gamma$ that already appeared in \cite{Alexandrov:2019rth}, we propose an integral equation that can be used to define them directly. Finally, we give an important lemma and illustrate how it was used in our work to prove many equations and identities. 

Throughout this section we work in Setup 1 from section \ref{subsec-qTBA-TBA}. We have analogous results for Setup 3 and they can be found in \cite{Alexandrov:2025abc}. 

\subsection{Formulating the problem}
First, we give the definition of a refined BPS structure. In fact, this is almost the same as definition \ref{definition-BPS} and the only change is the third statement which becomes
\begin{itemize}
	\item[3.]$\Omega$ is a polynomial
\be 
\label{LaurentOm}
\Omega:\Gamma \to \IQ\[y^{\pm1}\],\qquad \Omega(\gamma)= \sum_{n\in\IZ}\Omega_n(\gamma) \,y^{n}.
\ee
\end{itemize}
This makes contact with our discussion in chapter \ref{chap-string}, where refined BPS indices $\Omega(\gamma,y)$ were given as Poincar\'e polynomials in the refinement parameter $y$, that reduce to the original BPS invariants when $y=1$. 
Switching on $y$ induces a non-commutative structure \cite{Gaiotto:2010be,Cecotti:2014wea,Alexandrov:2019rth} 
which can be formalized as a quantum RH problem. Similarly to the unrefined case, we have rational refined BPS indices $\bOm(\gamma,y)$ defined in terms of $\Omega(\gamma,y)$ as in equation \eqref{BPS-rational}.

The non-commutative Riemann-Hilbert problem can be introduced by considering automorphisms on the quantum torus algebra \cite{Barbieri:2019yya}, but our approach follows \cite{Alexandrov:2019rth} where the Moyal star product ensures the correct commutation relations between a set of functions that play the role of algebra generators. The full dictionary between this approach and the one using automorphisms can be found in \cite{Alexandrov:2025abc}, while here we will directly formulate the problem using the Moyal product.

Representing the refinement parameter as  $y=e^{2\pi\I\alpha}$,
for any two functions $f,g$ on $\cZ=\IP^1\times \cM$, we define the Moyal star product
\be
f \star g =  f \exp\[ \frac{\alpha}{2\pi\I}\sum_{a,b}\omega_{ab}\,
\overleftarrow{\p}_{\!\theta^a}\overrightarrow{\p}_{\!\theta^b} \] g.
\label{starproduct}
\ee
We will use the script $\mathsf{X}$ for refined variables and keep $\cX$ for unrefined ones. Then, we take\footnote{While here this relation is a mere change of script, for other setups in \cite{Alexandrov:2025abc} it appears with a non-trivial rescaling.} $\cXz_\gamma(t)=\cXsf_\gamma(t)$ which satisfy 
the quantum torus algebra \cite{Barbieri:2019yya} with respect to the star product
\be
\cXz_{\gamma}(t)\star\cXz_{\gamma'}(t) =
y^{\langle\gamma,\gamma'\rangle}\cXz_{\gamma+\gamma'}(t).
\label{qtorus-alg-x}
\ee
Note that the previous relation can be generalized for functions at different parameter $t$,
\be
\cXz_{\gamma}(t)\star\cXz_{\gamma'}(t') =
y^{\langle\gamma,\gamma'\rangle}\cXz_{\gamma}(t)\cXz_{\gamma'}(t').
\label{starXX}
\ee 

Now we can write down the quantum RH problem:\\
{\bf qRH problem:} {\it  Find functions $\hcXr_\gamma(t)$ such that
	\begin{enumerate}
		\item 
		they satisfy
		\be 
		\hcXr_\gamma(t)\star \hcXr_{\gamma'}(t)=y^{\langle\gamma,\gamma'\rangle}\,\hcXr_{\gamma+\gamma'}(t);
		\label{starXXh}
		\ee 
		\item 
		they reduce to $\cXz_\gamma(t)$ at small $t$ and grow at most polynomially in $t$ when $t\to\infty$;
		\item 
		they jump across active BPS rays in such a way that, if  
		$\hcXrpm_\gamma$ are values on the clockwise and anticlockwise sides of $\ell$, respectively,
		then they are related by
		\be 
		\hcXrm_\gamma=\prodi{\gamma'\in\Gamma_\ell}{}_{\,\star}
		\exp_\star\Bigl(\sigma_{\gamma'}\bOm(\gamma',y)\,\kappa(\langle\gamma,\gamma'\rangle)\, 
		y^{\langle\gamma,\gamma'\rangle}\hcXr_{\gamma'} \Bigr)\star  \hcXrp_\gamma,
		\label{qKSjump-Xstar}
		\ee 
		where $\prod_\star$ denotes the star product and $\exp_\star(x)=\sum_{n=0}^\infty x\star \cdots \star x/n! $.
	\end{enumerate} 
}
\noindent
Note that, whereas the commutation relation \eqref{starXX} holds for arbitrary $t$ and $t'$,
the relation \eqref{starXXh} is imposed only at equal $t$'s.

\subsection{Perturbative solution}
\label{subsec-qTBA-sol}

As explained in the introduction, we will rely on the variables introduced in \cite{Alexandrov:2019rth}. In order to write the integral equation giving them, we will use the symbol \tIS{i} as in subsection \ref{subsec-qTBA-TBA} just with $\bOm(\gamma_i)$
replaced by $\bOm(\gamma_i,y)$, whereas the integration kernel taken to be  
\be
\Kr_{ij}=\frac{\ker_{ij}}{y-1/y}\, , 
\qquad
\ker_{ij}=\frac{t_i t_j}{t_j-t_i}\, .
\label{defkerref}
\ee
We then define the functions 
\be
\cXr_0 =\cXz_0\star\(1+\IS_1 \Kr_{01}\cXr_1\).
\label{inteqH-star}
\ee
We note first that \eqref{inteqH-star} does not have an unrefined limit since the integration kernel diverges at $y\to1$. Nonetheless, for generic $y$ one can solve iteratively it and find\footnote{There is a subtle but important ambiguity that appears in expressions with multiple nested integrals. It does not appear in the unrefined case of section \ref{sec-qTBA-classical} and is resolved by taking a certain symmetric prescription explained in \cite{Alexandrov:2025abc}.} 
\be  
\cXr_0=\cXz_0\sum_{n=0}^\infty \(\prod_{k=1}^n\IS_k \Kr_{k-1,k}\,\cXz_k\) y^{\sum_{j>i=0}^n\gamma_{ij}}.
\label{expXref}
\ee 
Naturally, these functions do not have an unrefined limit and thus cannot be the solutions we are looking for. Despite this, if we take the combination 
\be 
\Wr = \IS_1\, \cXr_1,
\label{Wref}
\ee
it turns out to be regular at $y=1$, where it reduces to \eqref{Wfull-short}, and to solve a $\star$-deformed version of the \pleb heavenly equation \cite{Strachan:1992em,Takasaki:1992jf,Bridgeland:2020zjh}! This suggests that \eqref{expXref} should be closely related to the actual solutions of the qRH. 

Inspired by this, we define
\be 
\label{hcXranz}
\begin{split} 
	\hcXr_0=&\, \(1+\cJr(t_0)\)_\star^{-1}\star \cXr_0
	\\
	=&\, \(1+\cJr(t_0)\)_\star^{-1}\star \cXz_0 \star \(1+\cJr(t_0)\),
\end{split}
\ee 
where the star index means that $(1+x)^{-1}_\star=\sum_{n=0}^\infty(-1)^n x\star \cdots \star x $ and the charge-independent 
\be 
\cJr(t_0) =\IS_1\,\Kr_{01} \cXr_1.
\label{Jref}
\ee
We claim that functions \eqref{hcXranz} solve the quantum RH problem defined above.

To show that this indeed the case, we have to check the three conditions imposed on $\hcXr_\gamma$.
The first is the product relation \eqref{starXXh}.
It is trivial to check that it does hold:
\be  
\label{hcXrcomm}
\begin{split} 
	\hcXr_{\gamma}(t) \star \hcXr_{\gamma'}(t) =&\, 
	\(1+\cJr\)_\star^{-1}\star \cXz_{\gamma}(t)\star \cXz_{\gamma'}(t) \star \(1+\cJr\)
	\\
	=&\,
	y^{\langle\gamma,\gamma'\rangle}\(1+\cJr\)_\star^{-1}\star \cXz_{\gamma+\gamma'}(t) \star \(1+\cJr\)
	\\
	=&\, y^{\langle\gamma,\gamma'\rangle}\hcXr_{\gamma+\gamma'}(t).
\end{split} 
\ee 
The second condition is also obvious since $\Kr_{01}$ \eqref{defkerref} and hence $\cJr$ 
vanish at $t_0=0$ and become $t_0$-independent at large $t_0$.
The proof of the wall-crossing relation requires a careful but straightforward computation and can be found in \cite{Alexandrov:2025abc}.

Furthermore, we prove\footnote{The proof is sketched in subsection \ref{subsec-qTBA-proofs}.} that, due to the following key relation satisfied by the kernel \eqref{defkerref}
\be 
\label{cyc3}
\ker_{ij}\ker_{ik}=\ker_{ij}\ker_{jk}+\ker_{ik}\ker_{kj},
\ee 
one has the perturbative expansion
\be
\label{linhHr-comm}
\hcXr_0=\cXz_0+
\sum_{n=1}^\infty \[\prod_{k=1}^n \IS_k\ker_{k-1,k}\] 
\left\{\left\{\left\{\cXz_0,\cXz_1\right\}_\star,\cXz_2\right\}_\star,\cdots ,\cXz_n\right\}_\star,
\ee
where 
\be
\{f,g\}_\star=\frac{f \star g-g\star f}{y-1/y}.
\label{Mbracket}
\ee
Since the star product of any functions $\cXz_\gamma$ at different arguments $t$ is known \eqref{starXX}, 
one can explicitly evaluate all commutators, which gives
\be  
\hcXr_0= \cXz_0\,\sum_{n=0}^\infty \prod_{j=1}^n \[\IS_j \kappa\(\sum_{i=0}^{j-1}\gamma_{ij}\) 
\ker_{j-1,j}\cXz_j\],
\label{compact-hcXr}
\ee 
where $\kappa(x)=\frac{y^{x}-y^{-x}}{y-y^{-1}}$.
This representation immediately ensures that $\hcXr_\gamma$ has a well-defined unrefined limit.
Indeed, since $\lim_{y\to 1}\kappa(x)=x$, one gets
\be  
\mathop{\lim}\limits_{y\to 1} \hcXr_0= \cXz_0\,\sum_{n=0}^{\infty}\prod_{j=1}^n \[\IS_j \(\sum_{i=0}^{j-1}\gamma_{ij}\) 
\ker_{j-1,j}\cXz_j\] .
\label{compact-hcXr-lim}
\ee   
A remarkable fact, which we prove in \cite{Alexandrov:2025abc}, is that the expansion 
\eqref{compact-hcXr-lim} is actually identical to the one in \eqref{Hexpand}.
Therefore, the unrefined limit of $\hcXr_\gamma$ coincides with the solution of the classical RH problem $\cX_\gamma$!
Of course, this is exactly what is expected from a solution of the quantum RH problem.

\subsection{An integral equation}
\label{subsec-qTBA-Integral}

Similarly to \eqref{TBAeq-sh} one can wonder if it's possible to get the $\hcXr_\gamma$ through a single (quantum) TBA-like equation, as opposed to the current two-step definition. In particular, we would like such an equation to have an unrefined limit where it would reduce to the unrefined TBA-like equation.

A first approach is to invert the relation \eqref{hcXranz} perturbatively in $\hcXr_\gamma$ and it produces the following expansion
\be
\label{cXr-hcXr}
\cXr_0=\sum_{n=0}^\infty \[\prod_{k=1}^n\IS_k \Kr_{k-1,k}\hcXr_k \,\star \] \hcXr_0,
\ee 
where $\hcXr_k$'s are ordered in the descending order according to their index. This implies
\be 
\cJr(t_0) =\sum_{n=1}^\infty \[\prod_{k=1}^n\IS_k \Kr_{k-1,k}\hcXr_k \,\star \],
\label{Jref-hX}
\ee 
which can be combined with the second line in \eqref{hcXranz} and produce the integral equation
\be 
\sum_{n=0}^\infty \[\prod_{k=1}^n\IS_k \Kr_{k-1,k}\hcXr_k\, \star \] \hcXr_0=
\cXz_0\star \sum_{n=0}^\infty \[\prod_{k=1}^n\IS_k \Kr_{k-1,k}\hcXr_k\, \star \].
\label{expeq}
\ee 
Unfortunately, this equation involves infinitely many terms of increasing order and, despite our efforts,
we have not been able to find a simple function that could produce it after a perturbative expansion.

We got inspiration for a second approach from the relation \eqref{linhHr-comm}, which produces the integral equation
\be
\label{intgeqn-comm}
\hcXr_0 = \cXz_0 + \sum_{n=1}^{\infty}\[\prod_{k=1}^n\IS_k \ker_{k-1,k}\] \left\{\cdots\left\{\left\{\cXz_0,\hcXr_n\right\}_\star,\hcXr_{n-1}\right\}_\star\cdots,\hcXr_1 \right\}_\star,
\ee
which also has an infinite number of terms.
Its proof is much more non-trivial and can be found in \cite{Alexandrov:2025abc}.
Its advantage is that it makes manifest that $\hcXr_\gamma$ has a well defined unrefined limit
because the star bracket \eqref{Mbracket} then reduces to the ordinary Poisson bracket.
Note a difference with \eqref{linhHr-comm}: whereas there one starts commuting $\cXz_0$ from the left, 
here one starts from the right.

\subsection{Trees and proofs}
\label{subsec-qTBA-proofs}
\lfig{A representation of the identity \eqref{cyc3} in terms of rooted trees.
	Here $T_k$'s are any rooted subtrees.}
{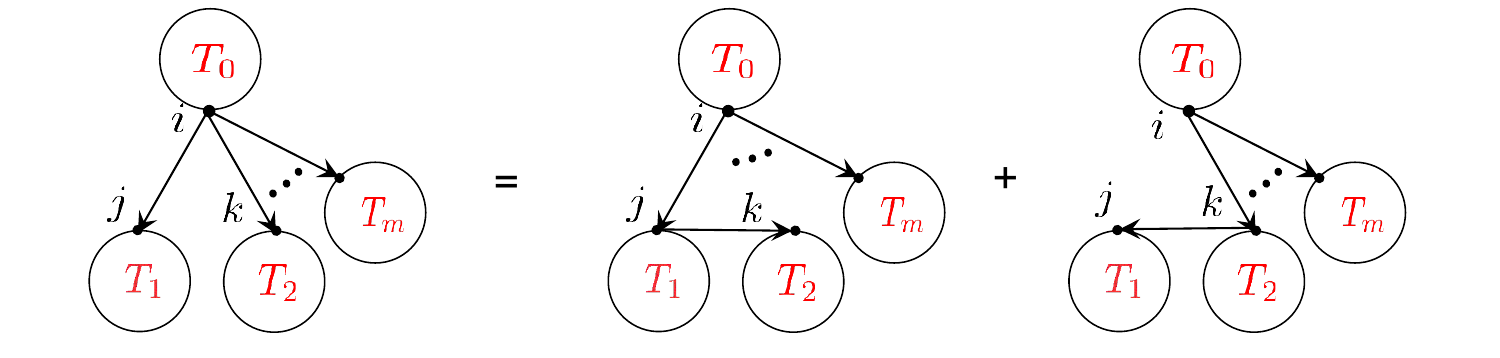}{15cm}{fig-identity}{-0.3cm}

Some of the equations and representations we obtained previously, and some of those that we will encounter in the next section rely on the identity \eqref{cyc3}. Here we will give a consequence of this identity that can be used on any product of kernels as well as a geometric interpretation of it. 
First, in all equations appearing in this work,
the kernels $\ker_{ij}$ can be thought of as factors assigned to edges of a labeled rooted tree\footnote{As a convention, we always orient edges away from the root.} (see, e.g., \eqref{Hexpand}).
Hence, we introduce 
\be 
\frKr_T=\prod_{e\in E_T}\ker_{s(e)t(e)}.
\label{def-frK}
\ee  
Then the identity \eqref{cyc3} is equivalent to the identity between the factors $\frKr_T$ associated 
to the labeled rooted trees shown in Fig. \ref{fig-identity}. 
It can be applied to any vertex that has at least two children and produces two trees 
where the number of children of that vertex decreased by one and the depth of some vertices increased by one.
It is clear that, applying this identity recursively, one can express $\frKr_T$ for any tree as 
a sum over linear trees for which every vertex has only one child. 
In fact, one can make an even stronger statement which was extensively used in \cite{Alexandrov:2025abc}. 
To this end, let us recall that a rooted tree induces a natural partial ordering on its vertices: 
$v<v'$ if the path from $v'$ to the root passes through $v$.
Besides, we introduce the labeling map $\ell_T \ :\ V_T\to \IN$ from vertices of a tree $T$ to natural numbers.
Then the following statement holds

\begin{lemma}\label{lemma-tree}
	For a tree $T$ with $n$ vertices, one has 
	\be 
	\frKr_T=\sum_{T'\in\IT_{n}^{\rm lin}(T)} \frKr_{T'},
	\label{eq-lemma}
	\ee 
	where $\IT_{n}^{\rm lin}(T)$ is the set of linear labeled rooted trees with $n$ vertices such that 
	the labeling preserves the partial ordering of the original tree $T$ in the sense that
	$v<v'\Rightarrow \ell_{T'}^{-1}(\ell_T(v))< \ell_{T'}^{-1}(\ell_T(v'))$.
\end{lemma}

Using this lemma we prove the representation \eqref{linhHr-comm} for $\hcXr_\gamma$, the equivalence of \eqref{Hexpand} and \eqref{compact-hcXr-lim}, the integral equation \eqref{intgeqn-comm} as well as the representation \eqref{qTBA-genfunc} of the (logarithm of the) refined generating function. All these proofs can be found in \cite{Alexandrov:2025abc}.

For illustration, we schematize the proof of the representation \eqref{linhHr-comm} for the perturbative solution of the quantum RH problem 
encoded in $\hcXr_\gamma(t)$. We proceed by induction in the perturbation order. 

In the linear approximation, $\cJr(t_0)$ in \eqref{hcXranz} can be replaced by 
$\mbox{\resizebox{0.018\vsize}{!}{$\displaystyle{\IS_1}$}}\Kr_{01} \cXz_1$
and it is trivial to see that keeping only the first order in the expansion 
gives 
\be
\hcXr_0=\cXz_0+
\IS_1\ker_{01}\left\{\cXz_0,\cXz_1\right\}_\star+O(\bOm(\gamma,y)^2),
\label{linhHr-comm-1}
\ee
consistently with \eqref{linhHr-comm}.

So let us assume that the representation \eqref{linhHr-comm} holds up to order $n-1$.
To prove that it continues to hold at the next order, we note that the expansion of \eqref{hcXranz} 
in powers of $\cXr_\gamma$ leads to
\be
\hcXr_0 = \cXz_0 + \sum_{m=1}^{\infty} (-1)^{m-1}\IS_1\ker_{01}\[\prod_{k=2}^m \IS_k\Kr_{0k}\cXr_k \star\] 
\left\{\cXz_0,\cXr_1\right\}_\star.
\label{expinref}
\ee
The order of factors in the product over $k$ is not important, but we choose it to be descending for definiteness.
Note that the product of kernels in the term of $n$-th order 
can be represented by a rooted tree with $m$ vertices labeled from 1 to $m$ 
all connected to the root labeled by 0. Furthermore, if one assigns to the $k$-th vertex (except the root)
the factor $\mbox{\resizebox{0.018\vsize}{!}{$\displaystyle{\IS_k}$}}\cXr_k/(y-y^{-1})$
and accepts an additional rule that they are multiplied using the star product from left to right, except the rightmost vertex which appears in a commutators with $\cXz_0$,
one arrives at the unique labeled rooted tree shown in Fig. \ref{fig-trees-pr1}(a).
Substituting further the expansion \eqref{expXref} of $\cXr_\gamma$ in terms of $\cXz_\gamma$, one arrives at
\be
\hcXr_0 = \cXz_0 + \sum_{n=1}^{\infty}\sum_{m=1}^{n} \frac{(-1)^{m-1}}{(y-y^{-1})^{n-1}}\[ \prod_{i=1}^n \IS_i\]
\sum_{\sum_{k=1}^m n_k=n}\frKr_{T_{\bfn}}
\[\prod_{k=2}^{m}\prod_{i=j_{k-1}+1}^{j_k}\cXz_i\star\] 
\left\{\cXz_0,\prod_{i=1}^{n_1}\cXz_i\right\}_\star,
\label{expinsf}
\ee
where $n_k\geq 1$, $\bfn=(n_1,\dots n_m)$, $j_k=\sum_{l=1}^{k} n_l$, $T_{\bfn}$ is the tree shown in 
Fig. \ref{fig-trees-pr1}(b),
and $\cXz_i$'s are ordered according to vertices of the tree: from left to right and from top to bottom. 

\lfig{The rooted trees corresponding to the expansions of $\hcXr_0$ in $\cXr_i$ and $\cXz_i$, respectively.}
{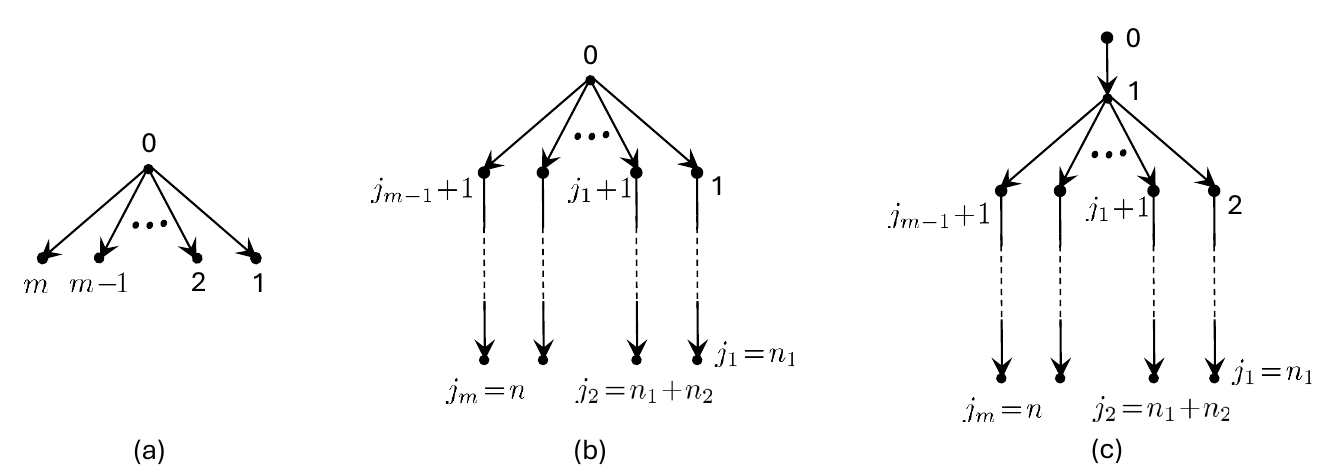}{16cm}{fig-trees-pr1}{-0.4cm}

The statement we want to prove would follow if we can show that the $n$-th term in \eqref{expinsf} 
is equal to 
\be  
\sum_{m=1}^{n-1} \frac{(-1)^{m-1}}{(y-y^{-1})^{n-2}} \[ \prod_{i=1}^n \IS_i\]
\sum_{\sum_{k=1}^m n_k=n}\frKr_{T'_{\bfn}}
\[\prod_{k=2}^{m}\prod_{i=j_{k-1}+1}^{j_k}\cXz_{i}\star\] 
\left\{\left\{\cXz_0,\cXz_1\right\}_\star,\prod_{i=2}^{n_1}\cXz_{i}\right\}_\star,
\label{expinsf-n}
\ee 
where $n_k\geq 1$ for $k\geq 2$ whereas $n_1\geq 2$, and $T'_{\bfn}$ is the tree shown in Fig. \ref{fig-trees-pr1}(c).
Indeed, redefining $n_1\to n_1+1$, it is easy to see that 
this expression is the same as the term of order $n-1$ in \eqref{expinsf} with $\cXz_0$ replaced by 
$\mbox{\resizebox{0.018\vsize}{!}{$\displaystyle{\IS_1}$}}\ker_{01}\{\cXz_0,\cXz_1\}$
and therefore it is subject to the induction hypothesis, which immediately gives the $n$-th order term in \eqref{linhHr-comm}.

To show the equality of \eqref{expinsf-n} and the $n$-th term in \eqref{expinsf}, we compare their contributions
after application to them of Lemma \ref{lemma-tree}, upon which they are all given by linear trees that differ only by distribution of labels.
A quick analysis of the $n$-th term in \eqref{expinsf} shows that all contributions of linear trees cancel except the linear trees generated by all $T_{\bfn}$
\begin{itemize}
	\item 
	 where the vertex attached to the root has label 1;
	\item
	 where $n_2=1$ and the vertex attached to the root has label $k_1+1$.
\end{itemize}
The idea behind these cancellations is that different trees can have the same ordering of $\cXz$ in their product. Namely, if we start from a given tree, take a linear subtree (with more than $1$ element, excluding the root) and move its top vertex to a newly created branch that is adjacent to it to the left, forming a new tree that has one more linear subtree in total, then this new tree has the same product of $\cXz$ as the initial one. By applying Lemma \ref{lemma-tree} to all trees and carefully grouping the linear contributions, we find many cancellations. In fact, the only reason all contributions do not cancel is that the product of $\cXz$ associated to the rightmost linear subtree appears inside a commutator and thus spoils the cancellations. 

Then, we further subdivide the first type of remaining contributions into two types according to whether $n_1=1$ or $n_1>1$ and upon an appropriate choice of triplets of contributions, one from each type, and a convenient relabeling, we find that they recombine exactly into the sum of linear contributions coming from \eqref{expinsf-n} to which we applied Lemma \ref{lemma-tree}. Of course this is not a proof but rather a sketch of the more rigorous proof found in \cite{Alexandrov:2025abc}.

\section{Adjoint representation and generating functions}
\label{sec-qTBA-genfun}

As was noticed in \cite{Barbieri:2019yya}, the solution of the quantum RH problem can be represented in the adjoint form.
Namely, let us consider the map 
\be
\Psi=\Ad_{\psi}\ :\ \cXz_\gamma  \mapsto \hcXr_\gamma,
\label{mapPsi}
\ee
where multiplication and inversion are evaluated using the star product.
Remarkably, the solution $\hcXr_\gamma$ \eqref{hcXranz} found in the previous section 
has the adjoint form $\hcXr_\gamma=\psi\star \cXz_\gamma \star\psi^{-1}$
with
\be
\psi^{-1}(t)=1+\cJr(t).
\label{respis}
\ee
Furthermore, as was emphasized above, $\cJr(t)$ and hence $\psi(t)$ are independent of charges.
Thus, any of them can be considered as a {\it generating function} of the set of functions $\hcXr_\gamma(t)$ for different charges.
Their perturbative expansion can be obtained from \eqref{Jref-hX} by evaluating the star products.

This result has a nice application: by taking the unrefined limit, 
it can be used to derive a generating function of Darboux coordinates $\cX_\gamma$.
Before, such a generating function was known only up second order in the perturbative expansion \cite{Alexandrov:2017qhn}
and a generalization to higher orders did not seem to be obvious.
Now we will find it to all orders.

To this end, let us apply the Campbell identity to the map \eqref{mapPsi}.
It allows to rewrite it as 
\be 
\hcXr_\gamma=e^{\ad_{\log_\star\!\psi}} \cXz_\gamma=e^{-\{ \cHr,\, \cdot\, \}_\star}\cXz_\gamma,
\label{genHref}
\ee 
where $\log_*(1+x)=\sum_{n=1}^\infty \frac{(-1)^{n-1}}{n} x\star \cdots \star x$, 
the star bracket was defined in \eqref{Mbracket} and we introduced 
\be 
\cHr=(y-y^{-1})\log_\star \(1+\cJr\).
\ee 
In the unrefined limit $y\to 1$, $\hcXr_\gamma$ and $\cXz_\gamma$ reduce to $\cX_\gamma$ and $\cXsf_\gamma$, respectively,
while the star bracket becomes the ordinary Poisson bracket. As a result, the relation
\eqref{genHref} takes the form
\be 
\cX_\gamma=e^{-\{\cH,\, \cdot\, \}} \cXsf_\gamma,
\label{cXPb}
\ee 
where 
\be 
\cH=\lim_{y\to 1} \cHr.
\ee 
Since $\cJr$ does {\it not} have a well-defined unrefined limit, one could worry that $\cHr$ 
does not have it either. However, miraculously, it does and using the identity \eqref{cyc3} as well as a result from \cite{Solomon:1968}, we prove that it can be written in a form where the unrefined limit is manifest,
\be
\label{qTBA-genfunc}
\cHr =
\sum_{n=1}^\infty \[\prod_{k=1}^n \IS_k\ker_{k-1,k}\] 
\sum_{\sigma \in S_n} \frac{(-1)^{s_\sigma+1}}{\binom{n}{s_\sigma}ns_\sigma}\,\left\{\left\{\cXz_{\sigma(1)},\cXz_{\sigma(2)}\right\}_\star,\cdots ,\cXz_{\sigma(n)}\right\}_\star,
\ee
where $s_\sigma$ is the number of ascending runs of $\sigma$. Alternatively $s_\sigma$ can be viewed as $1+$ the number of $1\leq i\leq n-1$ such that $\sigma(i)>\sigma(i+1)$.

The previous expression gives the unrefined limit in terms $\cXsf_\gamma$, but one could want to have its unrefined expression in terms of the $\cX_\gamma$. For this, we provide the conjecture, that was tested up to order $n=6$,
\begin{conj} 
	\be
	\cH = \IS_1 \ker_{01}\cX_1
	-\sum_{n=2}^{\infty}\frac{1}{n(n-1)} \sum_{\sigma \in S_n^{(+)}}\prod_{k=1}^{n}
	\[\IS_{k}\gamma_{\sigma(k-1)\sigma(k)}\,\ker_{k-1,k} \, \cX_{k}\],
	\ee
	where $S_n^{(+)}$ is the set of permutations of $n$ elements such that 
	$\sigma(\lfloor\frac{n}{2}\rfloor)<\sigma(\lceil\frac{n}{2}+1\rceil)$.
\end{conj}
These two formulae provide a generalization of the second order result from \cite{Alexandrov:2017qhn}
to all orders.

\section{Uncoupled case}
\label{sec-qTBA-uncoupled}

Finally, we look at some BPS structures that verify the following properties
\begin{itemize}
	\item  {\it finite}, i.e. $\Omega(\gamma,y)=0$ for all but finitely many $\gamma\in\Gamma$ 
	(we will call the set of such charges $\Gact$);
	\item {\it uncoupled}, i.e. $\langle \gamma,\gamma'\rangle=0$ for any $\gamma,\gamma'\in\Gact$;
	\item {\it palindromic}, i.e. $\Omega_n(\gamma) = \Omega_{-n}(\gamma)$ for all $n \in \IZ$ and $\gamma \in \Gamma$;
	\item {\it integral}, i.e. $\Omega_n(\gamma) \in \IZ$ for all $n \in \IZ$ and $\gamma \in \Gact$.
\end{itemize}
These turn out to be simple enough that we can obtain a much more explicit expression for our solution \eqref{compact-hcXr}. In particular we will show that it reproduces the one found in \cite{Barbieri:2019yya}. 

Since the BPS structure is uncoupled we can assume, up to rotating the charge lattice, that all active charges are electric, i.e they are given by $(0,q_\Lambda)$. When computing $\hcXr_0$ we have to consider two cases. The first, trivial, one is when $\gamma$ is electric and we get 
\be
\hcXr_{(0,q_{\Lambda})}  = \cXz_{(0,q_{\Lambda})}.
\ee 
It is simple to see that all terms in \eqref{compact-hcXr} with $n\geq1$ vanish, either because $\Omega(\gamma_i)=0$ or, for active charges, because their Dirac product with $\gamma$ vanishes. 

For a charge $\gamma_0$ with non-trivial magnetic part, the result is more rich as active charges do contribute to the sums \tIS{j} in
\be  
\begin{split} 
	\hcXr_0 =&\, \cXz_0\,\sum_{n=0}^\infty \prod_{j=1}^n \[\IS_j \kappa\(\sum_{i=0}^{j-1}\gamma_{ij}\) 
	\ker_{j-1,j}\cXz_j\],
	\\
	=&\, \sum_{n=0}^\infty \prod_{k=1}^n \[\IS_k \kappa(\gamma_{0k}) 
	\ker_{k-1,k}\cXsf_k\] \cXsf_0.
\end{split}
\label{hcXr-mlocal-1}
\ee
The sum inside the Dirac products reduced to a single term because all charges $\gamma_1,\dots,\gamma_n$ as well as $\gamma_j$ are active and thus their product vanishes. In fact, for a given $n$ the contribution is symmetric under any permutation of the labels in $\tIS{k}$ . 
Due to this observation, we can use Lemma \eqref{lemma-tree} to go from a kernel structure given by a symmetric sum over linear trees, to one given by star trees\footnote{Star trees are rooted trees where all vertices are children of the root.} with an appropriate factor
\be  
\begin{split} 
	\hcXr_0 =& 
	\,  \sum_{n=0}^\infty \frac{1}{n!}\prod_{k=1}^n \[\IS_k \kappa(\gamma_{0k}) 
	\ker_{0k}\cXsf_k\] \cXsf_0
	\\
	= &\,\cXsf_0\exp\[\IS_1 \kappa(\gamma_{01}) \ker_{01}\cXsf_1\].
\end{split}
\label{hcXr-mlocal-2}
\ee
Then, we use \eqref{cXsf} to explicitly compute the exponent and we find
\be   
\hcXr_0=\cXsf_0 \prod_{\substack{\gamma_1\in\Gact	\\ \Re(Z_{\gamma_1})>0}}
\prod_{k=1-|\gamma_{01}|}^{|\gamma_{01}|-1}\prod_{n\in\IZ}
\Lambda\(\frac{Z_{\gamma}}{t_0},1-\[\vth_{\gamma_1}+(k+n)\alpha \]\)^{-s_{01}\Omega_n(\gamma_1) },
\label{hcXr-res}
\ee
where 
\be
\label{defLambda}
\Lambda(z,\eta) =\frac{e^z\,\Gamma(z+\eta)}{\sqrt{2\pi} z^{z+\eta-1/2}}
\ee
is the modified gamma function \cite{Barbieri:2019yya}
and
\be
[x] = 
\begin{cases}
	x - \lfloor \Re(x)\rfloor, &\text{if } \Im x \geq 0,
	\\
	x + \lfloor \Re(-x)\rfloor+1,\quad &\text{if } \Im x <0.
\end{cases}
\ee
This result coincides with \cite[Th.5.1]{Barbieri:2019yya} up to a different sign convention for the Dirac product and 
the presence of the brackets $[\; \cdot\;]$ in the argument of the $\Lambda$-function.
The last discrepancy has already appeared in the classical case \cite{Alexandrov:2021wxu}.
It simply means that the solution constructed in \cite{Barbieri:2019yya} should be seen as 
an analytic continuation of the branch obtained for $\Re\vth_\gamma\in(0,1)$.

\chapter{Generating functions of D4-D2-D0 BPS indices}
\label{chap-DTall}

As established in section \ref{sec-string-BHBPS}, BPS indices are of fundamental importance as they count black hole microstates and also determine the weights of instanton contributions in type II string theory compactified on Calabi-Yau (CY) threefolds. This dual role means that computing them explicitly is very valuable. Furthermore, their significance is amplified by the fact that they are equivalent to generalized Donaldson-Thomas (DT) invariants, which are important objects in enumerative geometry.

This chapter focuses on BPS indices that count black holes in type IIA string theory, specifically those with vanishing D6-brane charge. These indices coincide with rank 0 DT invariants. Their generating functions, which we denote as $h_p$ (where $p$ represents the D4-brane charge vector with components $p^a, a=1,\dots,b_2(\CY)$), possess remarkable modular properties \cite{Maldacena:1997de,Alexandrov:2012au,Alexandrov:2016tnf,Alexandrov:2018lgp}. Depending on the characteristics of the D4-brane charge $p$, these generating functions $h_p$ are modular or mock modular forms of a specific depth (cf. chapter \ref{chap-mod} for definitions). More precisely, these generating functions satisfy a \textit{modular anomaly equation}\footnote{We will invariably refer to it as \textit{modular anomaly equation}, \textit{completion equation} or \textit{anomaly equation} throughout this work.} \cite{Alexandrov:2018lgp}. This equation defines a non-holomorphic completion of $h_p$ that transforms like a genuine modular form, thereby fixing its shadow.

The modular anomaly equation is expressed as a sum involving indefinite theta series contracted with a vector composed of generating functions corresponding to D4-brane charges lower than $p$. Each of these indefinite theta series is given by a sum over an electric charge lattice with a specific kernel function. The intricate structure of this completion equation presents several challenges. The first paper discussed here \cite{Alexandrov:2024jnu} primarily addresses the complexity of the kernels. Originally, these kernels were expressed through very complicated combinations of generalized error functions \cite{Alexandrov:2016enp,Nazaroglu:2016lmr} and their derivatives, with the structure of each term determined by sums over three distinct types of trees. The work in this first paper focused on the expression of these kernels. First, a subtle but crucial correction to the known expression was identified. It is determined by rational coefficients, denoted $e_\cT$, associated with unrooted labeled trees $\cT$. Following this correction, it was shown that the contributions from one of the three types of trees could be entirely eliminated. This was achieved by exploiting identities among the generalized error functions and the tree structures, effectively demonstrating that certain complex tree contributions could be canceled out, resulting in a more manageable expression for the kernels involving sums over only two types of trees.

The second major challenge, which is the central focus of the paper \cite{Alexandrov:2024wla}, is to solve the system of modular anomaly equations. For this goal, we restrict our attention to Calabi-Yau threefolds with a second Betti number $b_2=1$, corresponding to a single K\"ahler modulus and, for reasons that will be clear later, in this case we denote $r$ the D4-brane charge. A critical aspect of the anomaly equations is that they only fix the mock modular part of $h_r$, leaving a holomorphic modular ambiguity. This ambiguity can only be fixed by independently computing a finite number of the DT invariants appearing as Fourier coefficients of $h_r$. Thus, the complete determination of a generating function $h_r$ necessitates a two-step procedure: first, finding a function that satisfies the given modular anomaly, and second, fixing the remaining holomorphic modular ambiguity using explicit computations of BPS indices. Our ambition is to systematically address the first step for generating functions $h_r$ associated with any D4-brane charge $r$. However, the recursive nature of the anomaly equation constitutes a direct obstruction: the equation for $h_r$ depends on the full generating functions $h_{r_i}$ for all lower charges $r_i < r$. This implies an inherent dependence of $h_\Nr$ on all the modular ambiguities of lower charge. Consequently, one cannot, in principle, solve the anomaly for $h_r$ (step 1) without having already fixed all ambiguities for every $h_{r_i}$ with $r_i < r$ (step 2).

To circumvent this problem, we disentangle the dependence on these unknown modular ambiguities. Instead of solving for a unique $h_r$, the objective becomes to find its structural form in terms of these yet-undetermined holomorphic modular ambiguities, that we denote $\hh_{r_i}$. This is accomplished by introducing a new family of functions, called \textit{anomalous coefficients} and denoted $\gi{\bfr}$. Here, $\bfr = (\Nr_1,\dots, \Nr_n)$ represents a partition of the total charge $\Nr$. An ansatz is then proposed where each $h_r$ is expressed as a homogeneous polynomial in these modular ambiguities $\hh_{r_i}$ (where for each monomial $\sum\Nr_i=\Nr$), and the coefficients of this polynomial are precisely these anomalous coefficients $\gi{\bfr}$. A central theorem of this work then establishes a crucial equivalence: the condition that the generating functions $h_\Nr$ satisfy their respective modular anomaly equations is equivalent to the anomalous coefficients $\gi{\bfr}$ themselves being mock modular forms of depth $(n-1)$ (where $n$ is the number of components in $\bfr$) and satisfying their own set of specific modular anomaly equations. The primary goal then shifts to solving these anomaly equations for the anomalous coefficients $\gi{\bfr}$. As was the case for $h_r$, the functions $\gi{\bfr}$ are only fixed up to a holomorphic modular ambiguity. However, for the latter we can take any solution as long as we use it for higher rank anomaly equations.

We first provide two infinite families of solutions using two different methods. The first family covers cases where $\bfr$ has $n=2$ components, i.e. $\bfr=(\Nr_1,\Nr_2)$. In this case, functions with the same anomaly were studied in \cite{Dabholkar:2012nd} where they found the optimal-growth solutions. They achieve this by applying specific Hecke-like operators to a set of "seed" functions. Although we don't necessarily need the optimal growth property, we can use their method for $\gi{\Nr_1,\Nr_2}$ as it gives \textit{a} solution to our problem. The second family comes from scenarios where all charges in $\bfr$ are trivial\footnote{In this case we also require a trivial triple intersection number for the Calabi-Yau, $\kappa=1$}, i.e., $\bfr = (1,\dots, 1)$. It was observed that the anomaly equations for these $\gi{1,\dots,1}$ functions are identical to those satisfied by the normalized generating functions of $SU(n)$ Vafa-Witten (VW) invariants on the projective plane $\IP$. As solutions for these VW invariants are known \cite{Manschot:2014cca,Alexandrov:2020dyy} for any $SU(n)$, this directly provides explicit solutions for this subset of anomalous coefficients. However, neither of these two approaches could be readily generalized to determine the anomalous coefficients $\gi{\bfr}$ for an arbitrary collection of charges $\bfr$.

Consequently, a more general and alternative strategy, relying on indefinite theta series, was developed to determine generic $\gi{\bfr}$. This method is more involved and requires several extensions to the original problem. These extensions must be done in a way that ensures that a solution to the modified problem can be systematically reduced to a solution to the original one. The first modification is the introduction of a refinement parameter, typically denoted $y$ (where $y = e^{2\pi\I z}$), which elevates the anomalous coefficients from functions of $\tau$ alone to mock Jacobi-like forms that also depend on an elliptic variable $z$. This refinement procedure is known to simplify the structure of the modular anomaly equation \cite{Alexandrov:2019rth}. The second modification is an extension of the charge lattice. The electric D2-brane charge lattice, which is naturally associated with the problem, is artificially enlarged. This extension is designed to ensure the existence of null vectors within the lattice, a crucial feature for constructing solutions using indefinite theta series. Following these two steps, namely refinement and lattice extension, the anomalous coefficients, now depending on $z$ and the extended lattice structure, satisfy a new, modified modular anomaly equation.

A key theorem then provides a family of solutions to this new anomaly equation for these refined, extended anomalous coefficients. These solutions are explicitly constructed as indefinite theta series. Once this solution is obtained, the task is to reduce to a solution $\gi{\bfr}$ of the original modular anomaly equation. This involves two main steps. First, the reduction to the original lattice, which is achieved by applying specific modular derivatives with respect to auxiliary elliptic variables $z_i$ that were introduced alongside the lattice extension. These extra elliptic variables were introduced to ensure that the derivatives act in a factorized way. The second step is taking the unrefined limit, which means evaluating the behavior as the main elliptic parameter $z\to0$ (or equivalently, $y\to1$). For this limit to exist, the refined solutions $\girf{\bfr}$ should have a zero of order $(n-1)$ at $z=0$. This order of vanishing is necessary to cancel a prefactor $(y-y^{-1})^{1-n}$ that arises when relating refined and unrefined quantities. Ensuring this behavior near $z=0$ requires a precise choice for the holomorphic modular ambiguity. We propose an explicit form for this ambiguity and conjecture that it guarantees the needed cancellations of coefficients. 

Using this general method, in this dissertation we show how the unrefined anomalous coefficients $\gi{\bfr}$ can be obtained. In \cite{Alexandrov:2024wla} we give explicit results for $n=2$ and $n=3$ charges. For $n>3$, while the solution for the refined, extended coefficients is fully specified by the theorems, the analytical evaluation of the unrefined limit proves to be quite involved. We also claim that our choice of solution in that case has the correct behavior near $z=0$ and one only needs to compute its unrefined limit. We performed a few cross-checks for the validity of the construction. For $n=2$, the solutions derived from the indefinite theta series method were found to be consistent (up to expected purely modular terms) with those obtained using Hecke-like operators. For $n=3$, specifically for $\bfr=(1,1,1)$ and intersection number $\kappa=1$, the result was also confirmed to align with the solution derived from $SU(3)$ Vafa-Witten invariants \cite{Manschot:2010nc}, again up to modular ambiguities. These successful comparisons provide additional support for the general construction and the underlying conjectures regarding the unrefined limit. 

This chapter follows closely \cite{Alexandrov:2024jnu,Alexandrov:2024wla}. We start in section \ref{sec-DTall-def} by defining the generating functions and their completion equations. Then, in section \ref{sec-DTall-Kernel} we give the explicit definition of the kernels appearing in said equations, in terms of a sum over three types of trees that we also define. In section \ref{sec-DTall-simplify} we take a closer look at the anomaly equations where we present the subtle correction noticed in \cite{Alexandrov:2024jnu} as well as the simplified version of the kernels. Starting from section \ref{sec-DTall-gi} we restrict to Calabi-Yau spaces with $b_2=1$ and define the anomalous coefficients $\gi{\bfr}$, which our goal for the rest of the chapter will be to find. In section \ref{sec-DTall-HeckeVW} we present the two infinite families of solutions relying on Hecke-like operators and on generating functions of VW invariants. In section \ref{sec-DTall-Extensions} we present our general approach for finding all solutions $\gi{\bfr}$ and perform the two extensions to the problem, namely refinement and lattice extension. Finally, in section \ref{sec-DTall-sol2} we show how the solution is constructed for $n=2$ and in \ref{sec-DTall-soln} we give the solutions for any number of charges.


\section{The generating function and its completion}
\label{sec-DTall-def}

We are interested in D4-D2-D0 BPS states characterized by the charge vector $\gamma=(p^0=0,p^a,q_a,q_0)$
where $a=1,\dots, b_2(\CY)$. All the definitions will be given for generic Betti number $b_2(\CY)$ of the Calabi-Yau, but starting from section \ref{sec-DTall-gi} we will restrict to $b_2=1$. The D6-brane charge $p^0$ is vanishing, as indicated, and will be omitted from now on. The D4-brane charge $p^a$ is an element of the lattice $\Lambda=H_4(\CY,\IZ)$,
while the D2-brane charge $q_a$ belongs to the shifted dual lattice $\Lambda^\star+\hf\, \kappa_{ab}p^b$. The symmetric quadratic form $\kappa_{ab}$ used to lower the index of $p^a$ is given by the triple intersection number on the CY and the magnetic charge $p$ through $\kappa_{ab}=\kappa_{abc}p^c$ and is known to have signature $(1,b_2-1)$ when $p$ is an ample divisor. 
The Dirac product between charges depends only on their reduced form $\hgam_i=(p_1^a,q_{1,a})$ and is given by, 
\be
\label{Dirac-product}
\gamma_{ij}\equiv \langle \hgam_i, \hgam_j\rangle=p_j^aq_{i,a}-p_i^aq_{j,a}.
\ee
Finally, the D0 brane charge $q_0$ is a shifted integer $q_0\in\IZ -\frac{1}{24}c_{2,a}p^a$ \cite{Alexandrov:2010ca}, where $c_{2,a}$ are the components of the second Chern class of $\CY$. 

In what follows, as we did in chapter \ref{chap-NS5HM}, we will actually be working with the rational BPS indices \eqref{BPS-rational}
because their generating functions turn out to be the ones with interesting modular properties. The two types of invariants coincide if $\gamma$ is primitive, but in general they are different and $\bOm\in\IQ$ is not necessarily an integer.

The $\bOm(\gamma,z^a)$ depend on the \kahler moduli $z^a$ and, in this form, their generating functions do not possess any special modular properties. However, if we follow \cite{deBoer:2008fk} and evaluate them at the large volume attractor chamber defined by 
\be
z^a_\infty(\gamma) = \lim_{\lambda\to+\infty} (-q^a+\I \lambda p^a),
\ee
we get the Maldacena-Strominger-Witten (MSW) invariants $\bOm^{\rm MSW}(\gamma)$ \cite{Maldacena:1997de} which do not depend on the moduli and it is them we will use to define the generating functions. However, we still need to make some preliminary definitions before introducing these generating functions.

The $\bOm^{\rm MSW}$ are invariant under spectral flow symmetry acting on the D2 and D0 brane charge via
\be
\label{spec-flow}
q_a\to q_a-\kappa_{ab} \epsilon^b, \qquad q_0\to q_0-\epsilon^a q_a + \hf \epsilon^2,
\ee
with $\epsilon\in\Lambda$ and $\epsilon^2=\kappa_{ab}\epsilon^a\epsilon^b$. This action leaves the combination 
\be
\label{q0-comb}
\hq_0\equiv q_0-\hf \kappa^{ab}q_a q_b
\ee
invariant and can be used to reduce $q_a$. In fact $q_a\in\Lambda^*+\hf \kappa_{ab}p^b$ and we can decompose it as
\be
q_a=\kappa_{ab}\eps^b+\mu_a +\hf\, \kappa_{ab} p^b,
\label{defmua}
\ee 
and then act with the spectral flow to eliminate the first part. 
Therefore, the $\bOm^{\rm MSW}(\gamma)$ depend only on the D4-charge $p^a$, the modified D0 charge $\hq_0$ and the residue class $\mu_a$. Accordingly, they will be denoted $\bOm_{p,\mu}(\hq_0)$ (notice that we drop the superscript MSW).

Furthermore, due to the Bogomolov-Gieseker bound \cite{bayer2011bridgeland} the BPS indices (and thus the MSW invariants) vanish unless 
\be
\hq_0 \leq \hq_0^{\rm max}=\frac{1}{24}(\kappa_{abc}p^ap^bp^c+c_{2,a}p^a).
\ee
This allows us to define the generating function of D4-D2-D0 BPS indices for a given D4 charge $p^a$
\be
\label{def-genh-pa}
h_{p,\mu}= \sum_{\hq_0\leq\hq_0^{\rm max}}\bOm_{p,\mu}(\hq_0) \q^{-\hq_0}.
\ee
Notice that we do not use different scripts to write the vector $p$ as opposed to its components $p^a$. According to the results of \cite{Maldacena:1997de,Alexandrov:2012au,Alexandrov:2016tnf,Alexandrov:2018lgp},
under the standard $SL(2,\IZ)$ transformations $\tau\mapsto \frac{a\tau+b}{c\tau+d}$,
the $h_p$ behave as higher depth vector valued mock modular forms
of weight $-\hf b_2-1$ and multiplier system closely related to 
the Weil representation attached to the lattice $\Lambda$ \cite{Alexandrov:2019rth},
\be
\begin{split}
\label{Multsys-h-b2}
\Mi{h_{p}}_{\mu\nu}(T)=& e^{\frac{\pi\I}{12}c_{2,a}p^a+\pi\I\(\mu+\hf p\)^2},
\\
\Mi{h_p}_{\mu\nu}(S) =& \frac{(-\I)^{-1-b_2/2}}{\sqrt{|\Lambda^*/\Lambda|}}e^{-2\pi\I\(\frac{1}{4}p^2+\frac18 c_{2,a}p^a+\mu\nu\)},
\end{split}
\ee
where $\(\mu+\hf p\)^2=\kappa^{ab}\mu_a\mu_b+\mu_a p^a + \frac14 \kappa_{ab}p^ap^b$.

When the charge vector $p^a$ can be decomposed into a sum of $r$ positive irreducible charge vectors $p_{i}^a$, we say that the degree of reducibility of $p^a$ is $r$ and in that case, the generating function associated to such a vector is a mock modular form of depth $r-1$.
More precisely, the generating function $h_p$ has a completion\footnote{Here we use notations consistent with \cite{Alexandrov:2024wla} where we split the sum over reduced charges $\hgam_i$ into its magnetic part in \eqref{exp-whh} and electric part that is absorbed inside $\rmRi{\{p_i\}}_{\mu,\bfmu}$. We use an alternative notation in the other paper relevant to this chapter \cite{Alexandrov:2024jnu}. The difference between these notations is discussed in \cite[\S 4]{Alexandrov:2024jnu}} $\whh_p$,
\be
\whh_{p,\mu}(\tau,\btau)= h_{p,\mu}(\tau)+
\sum_{n=2}^{r}\sum_{\sum_{i=1}^n p_i=p}
\sum_{\bfmu}
\rmRi{\{p_i\}}_{\mu,\bfmu}(\tau, \btau)
\prod_{i=1}^n h_{p_i,\mu_i}(\tau),
\label{exp-whh}
\ee
where $\bfp=\{p_1,\dots,p_n\}$ are ample divisors, $\bfmu=\{\mu_1,\dots,\mu_n\}$ and the sum over them goes over all values $\mu_i=0,\dots,|\det \kappa_{i,ab}|-1$ where $\kappa_{i,ab}=\kappa_{abc}\,p_i^c$. The function $\whh_{p,\mu}$
should transform as a modular form of the same weight and with the same multiplier system as $h_p$. The outer sum goes up to $\Nr$ which is the maximum number of elements in a decomposition of $p$. The functions appearing as coefficients are given by  
\be
\rmRi{\{p_i\}}_{\mu,\bfmu}(\tau, \btau)=
\sum_{\sum_{i=1}^n q_i=\mu+\hf p \atop q_i\in\Lambda + \mu_{i}+\hf p_i } 
\Sym\Bigl\{(-1)^{\sum_{i<j} \gamma_{ij} }\scR_n(\{\gama_i\};\tau_2)\, e^{\pi\I \tau Q_n(\{\gama_i\})}\Bigr\}
\label{defrmRi}
\ee
where $\Sym$ denotes symmetrization (with weight $1/n!$). The sum goes over decompositions of $q=\mu+\hf p$ into $n$ charges $q_i$ with fixed residue classes $\mu_i$. This yields a $(n-1)b_2$-dimensional lattice obtained as $n$ shifted copies of the single electric charge lattice $\prod_{i=1}^{n}\(\Lambda+\mu_i+\hf p_i\)$ on which we impose the constraint $\sum_{i=1}^{n}q_i=q$. We denote the quadratic form on this lattice 
\be
\label{def-Qn}
Q_n\(\{\hgam_i\}\)=\kappa^{ab}q_aq_b-\sum_{i=1}^{n}\kappa_i^{ab}q_{i,a}q_{i,b}
\ee 
where $\kappa_i^{ab}$ is the inverse of $\kappa_{i,ab}$. Since $\kappa_{i,ab}$ has signature $(1,b_2-1)$, then the signature of $Q_n$ is $((n-1)(b_2-1),n-1)$ and thus \eqref{defrmRi} defines an indefinite theta series, with kernel given by $\scR_n(\{\gama_i\};\tau_2)$. This kernel is defined using a complicated expression involving two types of trees and derivatives of generalized error functions. We will define it in the next section.

\section{Trees and kernels}
\label{sec-DTall-Kernel}
In this section we will define the kernels $\scR_n(\{\gama_i\};\tau_2)$ of the indefinite theta series \eqref{defrmRi}, before showing their simplification in section \ref{sec-DTall-simplify}. In order to do this, we proceed in two steps: first we introduce three types of trees, namely \sch trees $\ST_n$, unrooted labeled trees $\LT_{n}$ and marked unrooted labeled trees $\LT_{n,m}$. Next, we show how to associate to each tree a specific factor, constructed out of generalized error functions \eqref{generrPhiME} and derivatives thereof, in a way that gives the intended kernels. The following subsections address these two issues in turn.

\subsection{Trees}
\label{subsec-DTall-trees}

We start with \sch trees, which appear first in later definitions.
Let $\IT_n^{\rm S}$ be the set of 
Schr\"oder trees with $n$ leaves. We denote $V_T$ the set of vertices excluding the leaves, and we define $\IT_n^{\rm S}$ as the set of rooted labeled trees such that any intermediate vertex $v\in V_T$ has $k_v\geq 2$ children. 
We denote the number of elements in $V_T$ by $n_T$, 
and the root vertex by $v_0$.
Furthermore, vertices of $T$ are labeled by charges so that the leaves carry charges $\hgam_i$, 
whereas the charges assigned to other vertices
are given recursively by
the sum of charges of their children, $\hgam_v=\sum_{v'\in\Ch(v)}\hgam_{v'}$ (see Fig. \ref{fig-3Schroder}). The number of these trees for each $n$ is given by the super-Catalan numbers\footnote{They can be found as sequence A001003 in the OEIS.}.

\lfig{Three examples of \sch trees, using the notation $\gamma_{i+j}=\gamma_i+\gamma_j$. The tree a) belongs to $\ST_2$ and the trees b) and c) belong to $\ST_5$. The trees a) and c) have $n_T=1$ and the tree b) has $n_T=3$.
}
{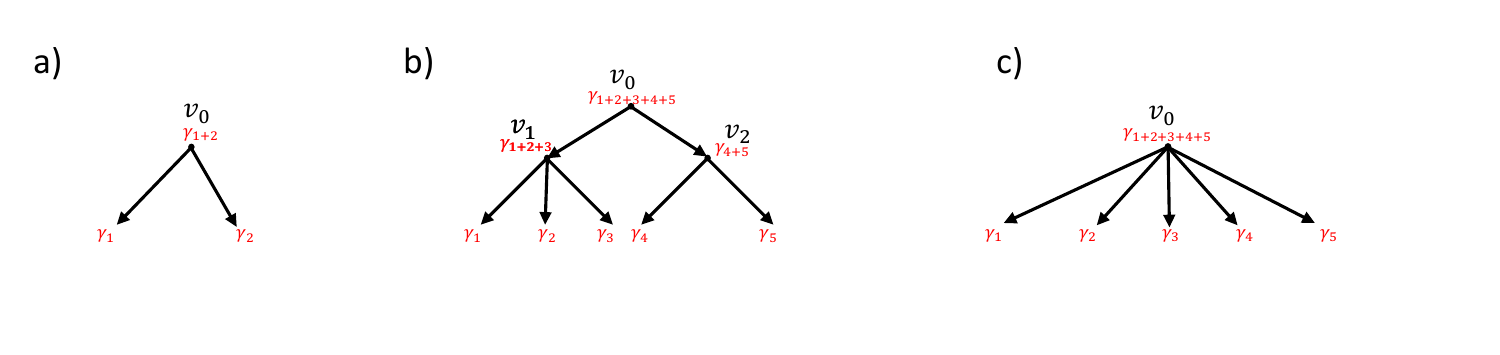}{15cm}{fig-3Schroder}{-1.2cm}

Next, we define the unrooted labeled trees. An element $T\in\LT_n$ is simply a planar unrooted tree with $n$ vertices, where to each vertex we attach a reduced charge vector $\hgam_i$ that serves as a label. The number of these trees is given by $n^{n-2}$.

Finally, we define the set of marked unrooted labeled trees $\LT_{n,m}$. 
The indices $n$ and $m$ denote respectively the number of vertices
and of $m$ marks where the marks are distributed across vertices and each one adds 2 labels to its vertex.
For example, $|\LT_{2,1}|=4$ with all trees having the same topology (2 vertices, one of which carries a mark),
4 labels (3 at the marked vertex and 1 at the non-marked vertex) and differing by the label assigned 
to the non-marked vertex. More generally, $|\LT_{n,m}|=n^{n-2+2m}$.
In our case the labels are again identified with the charges. Thus, given a set $\{\gama_1,\dots,\gama_{n+2m}\}$, 
the trees are decorated in the following way. 
Let $m_\ver\in \{0,\dots m\}$ be the number of marks carried by the vertex $\ver$, 
so that $\sum_\ver m_\ver=m$. Then a vertex $\ver$ with $m_\ver$ marks carries 
$1+2m_\ver$ charges $\gama_{\ver,s}$, $s=1,\dots,1+2m_\ver$ and 
we set $\gama_\ver=\sum_{s=1}^{1+2m_\ver}\gama_{\ver,s}$.
Given a tree $\cT\in \IT_{n,m}^\ell$, we denote the set of its edges by $E_{\cT}$, the set of vertices by $V_{\cT}$, 
the source and target vertex\footnote{
	Although various intermediate quantities can flip sign under a flip of the orientation, the final result does not change.} 
of an edge $e$ by $s(e)$ and $t(e)$, respectively. Finally, the two disconnected trees obtained from $\cT$ by removing the edge $e$ by $\cT_e^s$ and $\cT_e^t$, where the former contains $s(e)$ and the latter $t(e)$.

\subsection{Kernels}
\label{subsec-DTall-kernels}

The definitions in this subsection were initially given in \cite{Alexandrov:2018lgp}.
We start by defining the two functions $\cEp$ and $\cEf$. They are given by a single function $\cE_n$, which will be defined later, as  
\be
\cE_n(\{\gama_i\};\tau_2)=\cEf_n(\{\gama_i\})+\cEp_n(\{\gama_i\};\tau_2),
\label{twocEs}
\ee
where $\tau_2=\Im \tau$, the first term is $\tau_2$-independent and the second one is exponentially suppressed as $\tau_2\to\infty$\footnote{Take this limit while keeping the charges $\hgam_i$ fixed.}. Given a \sch tree $T$ and a vertex $v\in V_T$, we define $\cE_v\equiv\cE_{k_v}(\{\hgam_{v'}\};\tau_2)$ where $\{\hgam_v'\}$ is the set of charges labeling the $k_v$ children of $v$. Naturally, this induces similar notations $\cEp_v$ and $\cEf_v$. Then, the kernel $\scR_n(\{\hgam_i\};\tau_2)$ is given by a sum over the set of \sch trees with $n$ leaves 
\be
\scR_n(\{\gama_i\};\tau_2)=
\frac{1}{2^{n-1}} \sum_{T\in\IT_n^{\rm S}}(-1)^{n_T-1} 
\cEp_{v_0}\prod_{v\in V_T\setminus{\{v_0\}}}\cEf_{v}.
\label{expRn}
\ee
We give an example of a contribution in Fig. \ref{fig-Rtree}.
\lfig{An example of Schr\"oder tree contributing to $R_8$. 
	Near each vertex we showed the corresponding factor
	using the shorthand notation $\gamma_{i+j}=\gamma_i+\gamma_j$.
}
{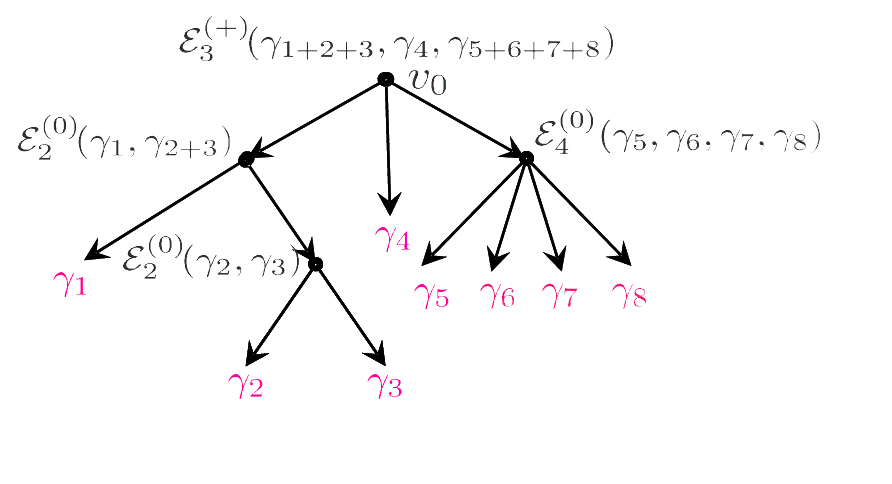}{9.75cm}{fig-Rtree}{0cm}

Next, we define the function $\cE_n$. Given a tree $\cT\in\LT_{n,m}$ we denote $E_\cT$ the set of its edges and associate to each $e\in E_\cT$ the vector
\be
\bfv_e=\sum_{i\in V_{\cT_e^s}}\sum_{j\in V_{\cT_e^t}}\bfv_{ij},
\label{defue}
\ee
where $\bfv_{ij}$ are the vectors with the following components
\be
v_{ij,k}^a=\delta_{ki} p^a_j-\delta_{kj} p^a_i.
\label{defvij}
\ee
Notice that a vector $\bfv_{ij}$ has $nb_2$ components ($a=0,\dots,b_2$ and $k=0,\dots,n-1$). We will use boldface script to denote $nb_2$-dimensional vector throughout this whole chapter. Another remark is that the definition \eqref{defue} depends on the orientation of the tree, but this dependence cancels inside the function $\cE_n$.
Then, we define 
\cite[Eq.(5.32)]{Alexandrov:2018lgp}
\be
\cEPhi_n(\bfx)=
\frac{1}{ n!}\sum_{m=0}^{[(n-1)/2]}
\sum_{\cT\in\, \IT_{n-2m,m}^\ell}
\[\prod_{\ver\in V_\cT}\cD_{m_\ver}(\{p_{\ver,s}\})\prod_{e\in E_\cT} \cD(\bfv_{s(e) t(e)})\]
\Phi^E_{\cT}(\bfx).
\label{rescEn}
\ee
Here $\Phi^E_{\cT}$ is a generalized error function defined by an unrooted labeled tree 
in terms of the usual (boosted) generalized error function $\Phi^E_n$ \eqref{generrPhiME} with parameters given by the vectors \eqref{defue},
\be 
\Phi^E_{\cT}(\bfx)=\Phi_{|E_\cT|}^E(\{\bfv_e\};\bfx).
\label{defPhiT-main}
\ee 
The functions $\Phi^E_n$ are reviewed
in subsection \ref{subsec-gen-erf} and defined with respect to the bilinear form
\be
\bfx\cdot\bfy= \sum_{i=1}^n \kappa_{abc}p_i^a x_i^b y_i^c.
\label{biform}
\ee
Besides, $\cD(\bfv)$ is the differential operator 
\be
\cD(\bfv)=\bfv\cdot\(\bfx+\frac{1}{2\pi}\,\p_\bfx\)
\label{defcD}
\ee
and $\cD_{m}(\{p_{s}\})$ is another differential operator associated with the existence of marks.
It is given by a sum over unrooted labeled trees with $2m+1$ vertices. To write its precise formula, 
we need to introduce rational coefficients $a_\cT$ determined recursively by the relation
\be
a_\cT=\frac{1}{n_\cT}\sum_{\ver\in V_\cT} (-1)^{n_\ver^+} \prod_{s=1}^{n_\ver} a_{\cT_s(\ver)},
\label{def-rec-aT}
\ee
where $n_\cT$ is the number of vertices, $n_\ver$ is the valency of the vertex $\ver$,
$n_\ver^+$ is the number of incoming edges at the vertex, and
$\cT_s(\ver)$ are the trees obtained from $\cT$ by removing the vertex $\ver$.
The relation \eqref{def-rec-aT} is supplemented by the initial condition for the trivial tree 
consisting of one vertex, $a_\bullet=1$, and one can show that $a_\cT=0$ for all trees with even number of vertices.
(One can find a table of these coefficients for trees with $n_\cT\leq 7$ in \cite[appendix B]{Alexandrov:2024jnu}.)
Then we have
\be
\cD_{m}(\{p_{s}\})=\sum_{\cT\in\, \IT_{2m+1}^\ell} a_{\cT}\prod_{e\in E_{\cT}}\cD(\bfv_{s(e)t(v)}).
\label{defcDcT}
\ee
Note that all vectors entering the definition of $\cD_{m_\ver}$ are orthogonal to the vectors 
appearing as parameters of the operators $\cD$ and the generalized error functions in \eqref{rescEn}.
Therefore, $\cD_{m_\ver}$ do not actually act on those factors and can be replaced by
usual functions obtained as $\cD_{m_\ver}\cdot 1$.

In terms of the functions \eqref{rescEn}, we finally set
\be
\cE_n(\{\gama_i\};\tau_2)=
\frac{\cEPhi_n(\bfx)}{(\sqrt{2\tau_2})^{n-1}}\, ,
\label{rescEnPhi}
\ee
where $\bfx$ is a vector with components $x^a_i=\sqrt{2\tau_2}\, \kappa_i^{ab}q_{i,b}$.
In particular, its scalar products with the vectors $\bfv_{ij}$ \eqref{defvij}
reproduce the Dirac products \eqref{Dirac-product}, i.e. $\bfv_{ij}\cdot\bfx=\sqrt{2\tau_2}\gamma_{ij}$.

This completes the definition of the functions $\scR_n(\{\gama_i\};\tau_2)$ determining the modular anomaly 
of the generating functions \eqref{exp-whh}. Given that their building blocks 
$\cE_n$ are constructed from the generalized error functions with parameters defined by charges $\hgam_i$,
they have the meaning of kernels of indefinite theta series providing completions 
for holomorphic theta series constructed from sign functions of Dirac products of these charges.

\section{The anomaly equation in depth}
\label{sec-DTall-simplify}
Although the previous definitions in principle fix the functions $\scR_n$ and thus the anomaly equation, in practice one needs to find the two functions giving the decomposition \eqref{twocEs}. This amounts to computing the large $\tau_2$ limit of $\cE_n$. 
This has been done in \cite[Eq.(5.29)]{Alexandrov:2018lgp} but the expression was wrong for some degenerate charge configurations. Then, in \cite{Alexandrov:2024jnu} it was revised and corrected. In addition, we found a significant simplification where \eqref{expRn} is expressed in terms of two new functions $\Ef_n,\Ep_n$ induced in the same lines as \eqref{twocEs} by a function $\Ev_n$ whose expression does not involve marked trees. 

In this section we start by giving the correct large $\tau_2$ limit of $\cE_n$ in a nice form and then we present the simplifications leading to the new function $\Ev_n$. We will prioritize a clear explanation of the results and omit some proofs and details. 

\subsection{Limit of large argument}
\label{subsec-DTall-largetau}

According to \eqref{rescEnPhi} and \eqref{rescEn} the large $\tau_2$ limit of $\cE_n$ is given directly by that of the $\Phi^E_\cT$ for a given tree $\cT\in\LT_{n,m}$, therefore we will focus on the limit of the latter. At a generic argument, the functions $\Phi_\cT^E(\bfx)$ at large $\bfx$ simply reduce to a product of sign functions. However, there are some degenerate charge configurations, namely when for some edge $e\in\cT$ the scalar product $\bfv_e\cdot\bfx$ vanishes, for which the product of signs vanishes\footnote{Assuming we take $\sgn(0)=0$.} and does not faithfully represent the limit.

An exact expression for the limit is then given by  
\begin{proposition}
	For a tree $\cT\in\LT_{n}$ we have
	\label{prop-largex}
	\be
	\lim_{\lambda\to\infty}\Phi_{\cT}^E\(\lambda\bfx\)=
	\esf_{\cT_{E_0}}(\{\bfv_e\})
	\prod_{e\in E_\cT\setminus E_0} \sgn (\bfv_i\cdot\bfx),
	\label{large-x-PhiE}
	\ee
	where $E_0\subseteq E_\cT$ is the subset of all edges $e$ such that $\bfv_e\cdot\bfx = 0$, $\cT_{E_0}$ is obtained from $\cT$ by contracting the edges $e\in E\setminus E_0$ and
	\be
	\esf_\cT(\{\bfv_e\})=\Phi_{\cT}^E(0).
	\ee
\end{proposition}

In general, the numbers $\esf_\cT(\{\bfv_e\})$ are irrational. This is not a desirable feature as one expects to only have rational numbers in the completion of a mock modular form with rational coefficients. Luckily, when they appear in combinations inside $\cEf_n$ they can, conjecturally, \cite[Conj. 1]{Alexandrov:2024jnu} be traded for rational coefficients $\esf_\cT\xrightarrow{} e_\cT\in \IQ$ that depend only on the topology of the tree and, up to a sign, on the orientation of its edges. Let's first give the limit of $\cEf_n$,  resulting from \eqref{large-x-PhiE}, in terms of the rational coefficients and then explain briefly how the conjecture works and give some supporting evidence. 

The $\tau_2$-independent part in \eqref{twocEs} is given by 
\be
\cEf_n(\{\gama_i\})
= \frac{1}{n!}
\sum_{m=0}^{[(n-1)/2]}\sum_{\cT\in\, \IT_{n-2m,m}^\ell}
\[\prod_{\ver\in V_\cT}\cV_{m_\ver}(\{\gama_{\ver,s}\})\]\kappa_\cT(\{\gama_{\ver}\})\,S_\cT(\{\gama_\ver\}),
\label{En0limit}
\ee
where 
\be
\label{def-all-modanom}
\begin{split}
\Gamma_e =& \sum_{i\in V_{\cT_e^s}}\sum_{j\in V_{\cT_e^t}}\gamma_{ij},
\qquad 
S_\cT(\{\gama_\ver\}) =
\sum_{\cJ\subseteq E_\cT}e_{\cT_\cJ}
\,\prod_{e\in \cJ}\delta_{\Gamma_e}
\prod_{e\in E_\cT\setminus \cJ} \sgn (\Gamma_e),
\\
\kappa_\cT(\{\gama_{\ver}\}) =& \prod_{e\in E_{\cT}}\gamma_{s(e) t(e)},\qquad 
\cV_{m}(\{\gama_{s}\}) = \sum_{\cT\in\, \IT_{2m+1}^\ell} a_{\cT}\kappa_\cT(\{\gama_{s}\}),
\end{split}
\ee
and we already made the replacement to rational coefficients $\bfe_\cT\to e_\cT$.

The formula above is still incomplete because the rational coefficients $e_\cT$ have not yet been specified. In fact they are given along $a_\cT$ in a table in \cite{Alexandrov:2024jnu} for trees with up to 7 vertices. But more generally we have a recursive formula for them 
\be
\label{def-rec-eT}
e_\cT=-\sum_{m=1}^{n_\cT-1}\sum_{\smash{\mathop{\cup}\limits_{k=1}^m\cT_k \simeq\cT }}\,
e_{\cT/\{\cT_k\}}\prod_{k=1}^m a_{\cT_k}.
\ee
Where $a_\cT$ are defined in \eqref{def-rec-aT} and $\cT/\{\cT_k\}$ denotes the tree obtained from $\cT$ by collapsing each subtree $\cT_k$ to a single vertex labeled by the sum of the collapsed charges. For $\cT=\bullet$ the coefficient is initialized to $1$. 

Now let's see how the replacement $\esf_\cT\xrightarrow{} e_\cT$ works. 
When $\cT$ has an even number of vertices (odd number of edges) the generalized error function is odd and thus $\esf_\cT=0$. The situation is more involved when the number of vertices is odd and we only proved the assertion for $n=3$ and $n=5$, which already provides good evidence for the conjecture. For $n=3$ this is trivial after using the identity \eqref{identPhiE} which implies 
\be
\label{PhiT-T123id}
\Phi^E_{\hat{\cT}_1}(\bfx)+\Phi^E_{\hat{\cT}_2}(\bfx)-\Phi^E_{\hat{\cT}_3}(\bfx)=-\Phi^E_{\cT}(\bfx).
\ee
where $\hcT_r$ ($r=1,2,3$) are the trees constructed from arbitrary unrooted trees $\cT_r$ as 
shown in Fig. \ref{fig-Vign3}, while $\cT$ is obtained from any tree $\hcT_r$ by collapsing both edges $e_i$. The minus sign before the third term is due to the difference in orientation of $\hcT_3$.
After setting $\bfx=0$ \eqref{PhiT-T123id} becomes
\be
\label{esf-id}
\esf_{\hat\cT_1}+\esf_{\hat\cT_2}-\esf_{\hat\cT_3}=-\esf_{\hat\cT}.
\ee 
In fact for $n=3$ all unrooted labeled trees have the same topology but differ by the label attached to the middle vertex. Namely, they are given by the $\hcT_r$ for $r=1,2,3$ constructed from $\cT_1=\cT_2=\cT_3=\bullet$ labeled by $\hgam_{1},\hgam_{2},\hgam_{3}$ respectively. Moreover, they contribute to $\cE^{(0)}$ through the combination on the l.h.s of \eqref{esf-id}. Since in this case $\hcT=\bullet$ and $\esf_{\bullet}=1$, we simply use the prescription $\esf_{\cT_3}\xrightarrow{}e_{\cT_3}=\frac13$.
For $n=5$ the procedure is much more involved \cite[Appendix C.]{Alexandrov:2024jnu}, but relies mostly on obvious identities between products of $\gamma_{ij}$ and on the three independent equations given by \eqref{esf-id} when $\hcT_i$ have five vertices.

\lfig{Unrooted trees constructed from the same three subtrees.}{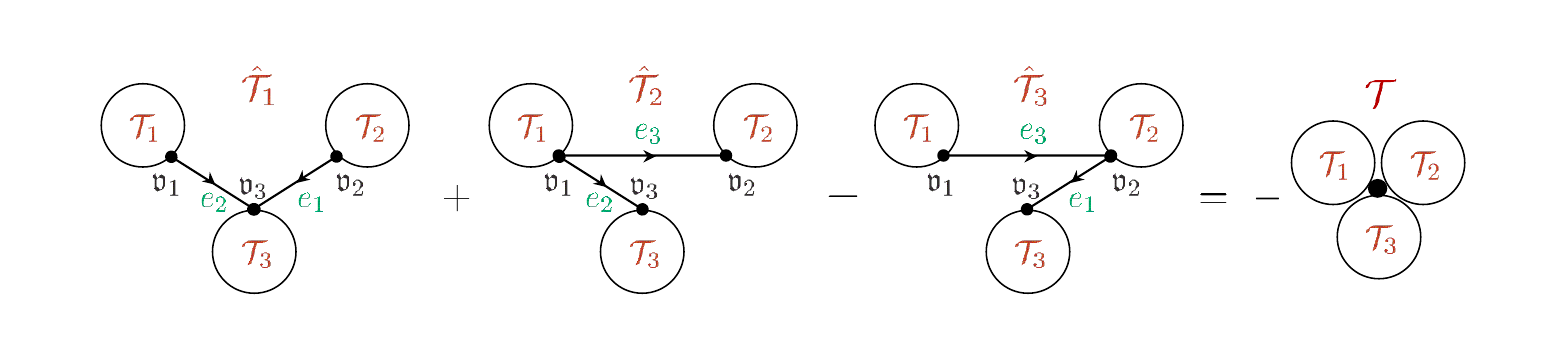}{19.5cm}{fig-Vign3}{-1.cm}

\subsection{Simplified anomaly}
In this subsection we show that the expression for the modular completion reviewed in the previous section
can be significantly simplified and we  follow closely \cite[Sec. 3]{Alexandrov:2024jnu} where the simplification was found\footnote{Note that we use different notations and normalizations compared to that paper.}. 

The equations \eqref{exp-whh} and \eqref{defrmRi} are unchanged.
Then, the expression for $\scR_n$ \eqref{expRn} keeps the same form but the simplification lies in replacing the building blocks $\cE_n$ by the more elementary $\Ev_n$ that we will introduce below. 

Precisely, we claim that the functions $\scR_n$ \eqref{expRn} can be rewritten as 
\be
\scR_n(\{\gama_i\};\tau_2)=\frac{1}{2^{n-1}} \sum_{T\in\IT_n^{\rm S}}(-1)^{n_T-1} 
\Ep_{v_0}\prod_{v\in V_T\setminus{\{v_0\}}}\Ef_{v},
\label{solRn-new}
\ee
where $\Ef_n$ and $\Ep_n$ are the constant and exponentially suppressed terms in a decomposition of 
the new functions $\Ev_n$,
\be
\Ev_n(\{\gama_i\};\tau_2)=\Ef_n(\{\gama_i\})+\Ep_n(\{\gama_i\};\tau_2).
\label{twocEs-new}
\ee
These functions are defined, similarly to \eqref{rescEnPhi}, as 
\be
\Ev_n(\{\gama_i\};\tau_2)=
\frac{\EPhi_n(\bfx)}{(\sqrt{2\tau_2})^{n-1}}\, ,
\label{rescEnPhi-new}
\ee
where $\EPhi_n(\bfx)$ replace $\cEPhi_n(\bfx)$ and take the following form
\be
\EPhi_n(\bfx)=
\frac{1}{n!}\sum_{\cT\in\, \IT_n^\ell}
\[\prod_{e\in E_\cT} \cD(\bfv_{s(e) t(e)},\bfy)\]
\Phi^E_{\cT}(\bfx)\Bigr|_{\bfy=\bfx},
\label{rescEn-new}
\ee
written in terms of a differential operator generalizing the one in \eqref{defcD} and given by
\be
\cD(\bfv,\bfy)=\bfv\cdot\(\bfy+\frac{1}{2\pi}\,\p_\bfx\).
\label{defcDif}
\ee
Furthermore, the constant part of $\Ev_n$ is given by $\Ef_n$ 
\be
\Ef_n(\{\gama_i\})
= \frac{1}{n!}
\sum_{\cT\in\, \IT_n^\ell} \kappa_\cT(\{\gama_{\ver}\})\, S_\cT(\{\gama_\ver\}),
\label{Efn0new}
\ee	
where $S_\cT$ and $\kappa_\cT$ are defined in \eqref{def-all-modanom}. 

Comparing \eqref{rescEn-new} and \eqref{Efn0new} with \eqref{rescEn} and \eqref{En0limit}, 
we observe that the simplification consists in dropping all contributions generated either by trees with 
non-vanishing number of marks or by the mutual action of the derivative operators $\cD(\bfv)$.
The latter is achieved through the use of the operators \eqref{defcDif} which do not act on each other in 
contrast to the original operators $\cD(\bfv)$. The price to pay for this is that the functions 
$\EPhi_n(\bfx)$ are {\it not} eigenfunctions of the Vign\'eras operator \eqref{Vigdif} and 
therefore appear to spoil modularity. However, the claim is that  
all anomalies are canceled in the combinations relevant for the completion.

The proof of the representation \eqref{solRn-new} proceeds in two steps that can be explained succinctly.
At the first step we note that the functions $\cE_n$ \eqref{rescEnPhi} can be expressed as follows
\be
\begin{split}
	\cE_n(\{\gama_i\};\tau_2)=&\, 	\frac{1}{n!}\sum_{m=0}^{[(n-1)/2]}
	\frac{1}{(\sqrt{2\tau_2})^{n-2m-1}} 
	\\
	&\times
	\sum_{\cT\in\, \IT_{n-2m,m}^\ell}
	\[ \prod_{\ver \in V_\cT} \cV_{m_\ver}(\{\hgam_{\ver,s}\})\prod_{e\in E_\cT}\cD(\bfv_{s(e) t(e)},\bfy) \]
	\Phi_{\cT}^E(\bfx)\Bigr|_{\bfy=\bfx}\, ,
\end{split}
\label{new-explicit-En}
\ee
where, as usual, $\bfx$ is a vector with components $x^a_i=\sqrt{2\tau_2}\, \kappa_i^{ab}q_{i,b}$.
This relation generalizes the one in \eqref{En0limit} to finite $\tau_2$.
Comparing to the original representation using the functions $\cEPhi_n(\bfx)$ \eqref{rescEn}, there are two changes:
the functions $\cD_{m}$ (see the comment below \eqref{defcDcT}) are replaced by $(2\tau_2)^m\cV_{m}$, and 
the operators $\cD(\bfv)$ are replaced by $\cD(\bfv,\bfy)$.
The proof of this statement is completely analogous to the proof of Proposition 5 in \cite{Alexandrov:2018lgp}
where the role of identity \cite[Eq.(F.20)]{Alexandrov:2018lgp} is played by \eqref{PhiT-T123id} and we make the replacement
\be
\prod_{e\in E_{\cT}}(\bfv_{s(e) t(e)},\bfx)\,\sgn(\bfv_e,\bfx)
\ \longrightarrow\ 
\[\prod_{e\in E_\cT}\cD(\bfv_{s(e) t(e)},\bfy) \]
\Phi_{\cT}^E(\bfx)\Bigr|_{\bfy=\bfx}.
\ee

\lfig{Combination of two Schr\"oder trees ensuring the cancellation of contributions generated by marked trees.}
{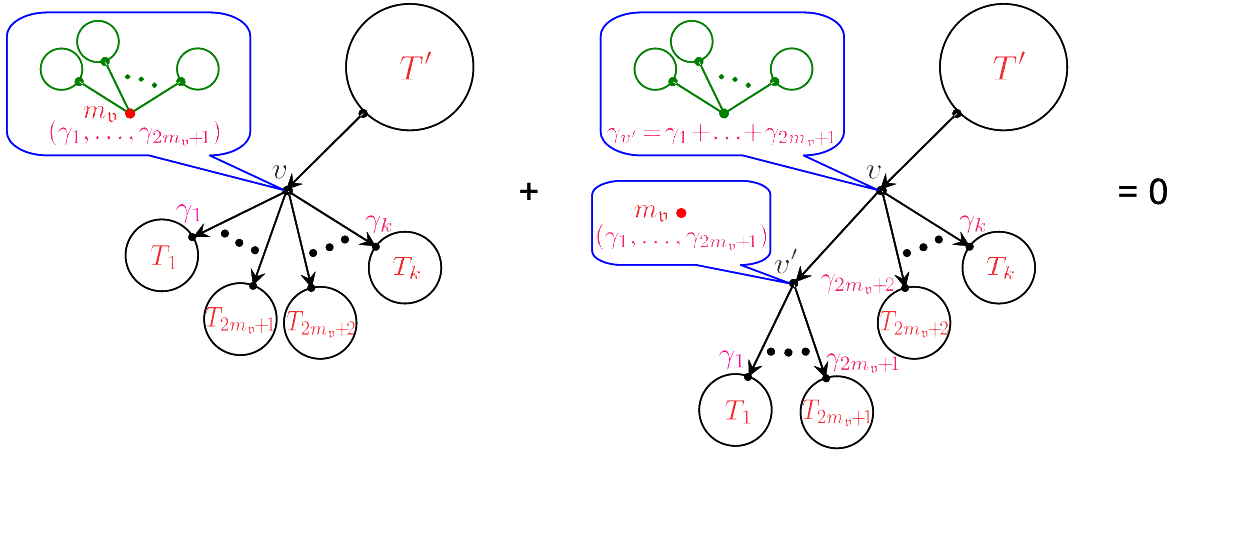}{17.cm}{fig-Mcancel}{-1.5cm}

The second step is essentially identical to the proof of Proposition 10 in \cite{Alexandrov:2018lgp}.
Namely, let us consider the original representation of $\scR_n$ \eqref{expRn} 
and pick up the contribution generated by a non-trivial marked tree $\cT$ in \eqref{new-explicit-En}, 
i.e. a tree having more than one vertex and at least one mark, 
which appears in the sum over marked unrooted trees producing the factor assigned
to a vertex $v$ of a Schr\"oder tree $T$. 
We denote $k=n_v$ the number of children of the vertex $v\in V_T$ and $\gamma_i$ ($i=1,\dots,k$) their charges.
Let us focus on a vertex $\ver\in V_\cT$ with $m_\ver>0$ marks and take 
$\gamma_s$ ($s=1,\dots,2m_\ver+1$) to be the charges labeling this vertex, 
so that its weight to our contribution, due to \eqref{En0limit} and \eqref{new-explicit-En}, 
is given by $\cV_{m_\ver}(\{\hgam_{s}\})$.
Note that $k\ge 2m_\ver+2$ because the tree $\cT$ 
has at least one additional vertex except $\ver$.
Then it is easy to see that the contribution we described is exactly canceled 
by the contribution coming from another Schr\"oder tree,
which is obtained from $T$ by adding an edge
connecting the vertex $v$ to a new vertex $v'$, whose children are the $2m_\ver+1$ children of $v$ in $T$ 
carrying charges $\gamma_s$ (see Fig. \ref{fig-Mcancel}).\footnote{The new tree is of Schr\"oder 
	type because its vertex $v$ has $k-2m_{\ver}\ge 2$ children and vertex $v'$ has $2m_\ver+1\ge 3$ children.}
Indeed, choosing in the sum over marked trees at vertex $v$ a tree $\cT$ which is the same as before except
that now it has 0 marks at vertex $\ver$, and
in the sum at vertex $v'$ the trivial tree having one vertex and $m_\ver$ marks, 
one gets exactly the same contribution as before,
but now with an opposite sign due to the presence of an additional vertex in the Schr\"oder tree.
Thus, all contributions from non-trivial marked trees are canceled.

As a result, we remain only with the contributions generated by trivial marked trees, 
i.e. having only one vertex and $m_\ver$ marks.
One has to distinguish two cases: either the corresponding vertex $v$ of the Schr\"oder tree is the root or not.
In the former case, this contribution is trivially canceled in the difference $\cE_{v_0}-\cEf_{v_0}=\cEp_{v_0}$ 
assigned to $v_0$ in \eqref{expRn}.
In the latter case, there are again two possibilities: whether the vertex $\ver$ of the unrooted tree $\cT$ assigned to
the parent of $v$, which carries the charge $\gamma_v$, has marks or not.
If not, this is precisely the contribution used above to cancel the contributions
from non-trivial marked trees. If $m_\ver>0$, then $\cT$ must be a trivial marked tree because 
the contributions corresponding to non-trivial ones have already been canceled. 
But then we can apply the same argument to this parent vertex. 
As a result, we continue in this way up to the root, but  
we already know that its weight $\cEp_{v_0}$ does not contain the contributions of trivial marked trees.
Thus, all contributions generated by unrooted trees in \eqref{En0limit} and \eqref{new-explicit-En}  
with at least one mark cancel and we arrive at the formula \eqref{solRn-new}.

\section{Anomalous coefficients}
\label{sec-DTall-gi}
Starting from this section we focus on solving the modular anomaly equation. In alignment with our needs, we will restrict to CY with second Betti number $b_2=1$. 

We start by rewriting the main equations and properties that we saw in section \ref{sec-DTall-def} in the specific case of one modulus. 
In this case, the generating functions have weight $-3/2$. But before rewriting the completion equation and the multiplier system, let us perform a redefinition of the generating function 
\be
\tlh_{r,\mu}(\tau) =
(-1)^{(\Nr-1)\mu }h_{r,\tmu(r)}(\tau),
\label{redef-th}
\ee
where
\be
\tmu(r) = \mu - \frac{\kappa \Nr(\Nr-1)}{2}\, .
\label{def-tmu}
\ee
The shift of $\mu$ replaces the quadratic term in the spectral flow decomposition \eqref{defmua} by a linear one
so that now it reads
\be
\label{defmu}
q = \tmu + \frac12\, \kappa r + \kappa r \eps,
\ee
with $\epsilon\in\IZ$ and $\kappa$ is the triple-intersection number of the CY.

The completion equation of these functions is given by
\be
\twhh_{r,\mu}(\tau,\btau)=\sum_{n=1}^r \sum_{\sum_{i=1}^n r_i=r}
\sum_{\bfmu}
\trmRi{\bfr}_{\mu,\bfmu}(\tau, \btau)
\prod_{i=1}^n \tlh_{r_i,\mu_i}(\tau).
\label{exp-twhh}
\ee
The $\trmRi{\bfr}_{\mu,\bfmu}$ are a redefined version of the (non-holomorphic) indefinite theta series over a dimension $n-1$ lattice \eqref{defrmRi}
\be
\trmRi{\bfr}_{\mu,\bfmu}(\tau, \btau)=
\sum_{\sum_{i=1}^n q_i=\mu \atop q_i\in \kappa r_i \IZ+\mu_i} 
\Sym\Bigl\{\scR_n(\bfhgam;\tau_2)\Bigr\}\, e^{\pi\I \tau Q_n(\bfhgam)},
\label{redefrmRi}
\ee
where the linear shifts of $q$ and $q_i$ can be canceled out since $\sum\Nr_i=\Nr$ and the sign was canceled by the one in \eqref{redef-th}. The kernel $\scR_n(\bfhgam;\tau_2)$ remains intact and its exact expression can be found in \eqref{solRn-new} but is not needed for what follows.

The redefinition \eqref{redef-th} also affects the multiplier system which is now for the function $\tlh_\Nr$
\be
\begin{split}
	\Mi{\tlh_\Nr}_{\mu \nu}(T) &= 
	e^{\frac{\pi\I}{\kappa\Nr}(1-\kappa\Nr)\mu^2 + \frac{\pi \I}{4}\(\kappa+\frac{c_2}{3}\)\Nr} \,
	\delta_{\mu\nu},
	\\
	\Mi{\tlh_\Nr}_{\mu \nu}(S) &= \frac{e^{\frac{\pi\I}{4}((2\kappa+c_2)\Nr-1)}}{\sqrt{\kappa\Nr}} 
	\,e^{-2\pi \I \frac{\mu\nu}{\kappa\Nr}}.
\end{split}
\label{mult-thr}
\ee

We expect that for each D4-brane charge $r$, the anomaly equation fixes the generating function
$\tlh_{r,\mu}$ up to a modular ambiguity. Indeed, we can add any holomorphic modular form with the correct transformation properties to a solution of \eqref{exp-twhh} and this gives us a new solution. This holomorphic modular ambiguity can be fixed by other means, e.g.
by computing the first few terms in the Fourier expansion of $\tlh_{r,\mu}$.
In other words, we can represent 
\be
\tlh_{r,\mu}=\than_{r,\mu}+\thh_{r,\mu},
\label{han}
\ee
where $\than_{r,\mu}$ is a depth $r-1$ mock modular form satisfying \eqref{exp-twhh},
while $\thh_{r,\mu}$ is pure modular. Then, we can find $\tlh_{r,\mu}$ by first solving the modular anomaly equation and then fixing the ambiguity $\thh_\Nr$. We would like to perform the first step for all charges $\Nr$. 

The problem however is that the r.h.s. of \eqref{exp-twhh} depends on the full generating functions 
$\tlh_{r_i,\mu_i}$ with $r_i<r$ and hence on all $\thh_{r_i,\mu_i}$, which remain unknown at this point.
Therefore, $\than_{r,\mu}$ must also depend on them, and what we can do at best is to find $\than_{r,\mu}$ 
up to these modular functions.
To achieve this goal, we first parametrize the dependence of $\tlh_{r,\mu}$ on $\thh_{r_i,\mu_i}$
by holomorphic functions $\gi{\bfr}_{\mu,\bfmu}(\tau)$ which we call anomalous coefficients,
characterize them by anomaly equations similar to \eqref{exp-twhh}, and then solve 
these equations. 
In this section we perform the first two steps and leave the third one to the subsequent sections.
The main result is captured by the following

\begin{theorem}\label{thm-ancoef}
	Let $\gi{\Nr}_{\mu,\mu'}=\delta_{\mu,\mu'}$ and $\thh_{\Nr,\mu}$ be a set of holomorphic modular forms.
	Then 
	\be
	\tlh_{\Nr,\mu}(\tau)= \sum_{n=1}^\Nr\sum_{\sum_{i=1}^n \Nr_i=\Nr}
	\sum_{\bfmu}
	\gi{\bfr}_{\mu,\bfmu}(\tau)
	\prod_{i=1}^n \thh_{\Nr_i,\mu_i}(\tau),
	\label{genansatz}
	\ee
	is a depth $r-1$ modular form with completion of the form \eqref{exp-twhh}
	provided $\gi{\bfr}_{\mu,\bfmu}$ are depth $n-1$ mock modular forms (where $n$ is the number of charges $r_i$)
	with completions satisfying
	\be
	\whgi{\bfr}_{\mu,\bfmu}=
	\Sym\Bigg\{\sum_{m=1}^n 
	\sum_{\sum_{k=1}^m n_k=n}
	\sum_{\bfnu}\trmRi{\bfs}_{\mu,\bfnu}
	\prod_{k=1}^m \gi{\frr_k}_{\nu_k,\frm_k}
	\Bigg\},
	\label{compl-gi}
	\ee
	where\footnote{Note that while the sets $\bfr$ and $\bfmu$ have $n$ elements, 
		the sets $\bfs$ and $\bfnu$ have only $m\le n$ elements. 
		To comprehend the structure of the equation \eqref{compl-gi}, it might be 
		useful to use the fact that the sum on its r.h.s. is equivalent to the sum over
		rooted trees of depth 2 with leaves labeled by charges $r_i$ and other vertices
		labeled by the sum of charges of their children. Using this labeling, we assign 
		the function $\trmRi{\bfs}_{\mu,\bfnu}$ to the root vertex and 
		the anomalous coefficients to the vertices of depth 1 with arguments determined by the charges 
		of their children.
		Then the contribution of a tree is given by the product of the weights of its vertices.
		See Fig. \ref{fig-g-trees}. } 
	\be
	j_k=\sum_{l=1}^{k-1} n_l,
	\qquad
	\Ms_k=\sum_{i=1}^{n_k} \Nr_{j_k+i},
	\qquad
	\begin{array}{c}
		\frr_k=(\Nr_{j_k+1},\dots, \Nr_{j_{k+1}}),
		\\
		\frm_k=(\mu_{j_k+1},\dots,\mu_{j_{k+1}}).
		\vphantom{\sum\limits^{a}}
	\end{array}
	\label{split-rs}
	\ee
\end{theorem}
\lfig{A representation of contributions to the r.h.s. of \eqref{compl-gi} in terms of rooted trees of depth 2.}
{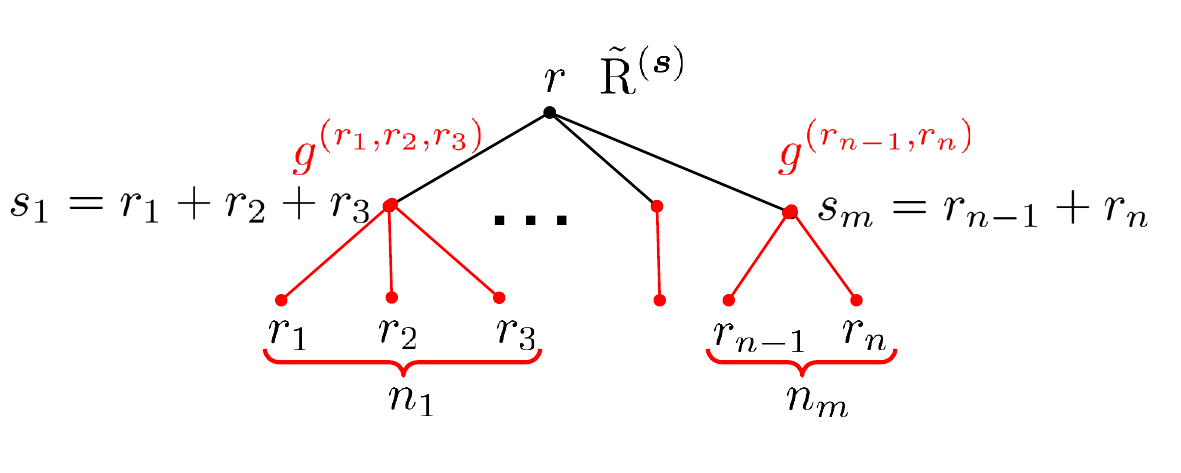}{10cm}{fig-g-trees}{-0.7cm}

The proof of this theorem can be found in \cite[Thm. 3.1]{Alexandrov:2024wla}. 

The main point is that now we have completely disentangled the dependence of the $\tlh_r$ on lower rank $\thh_{r_i}$. In practice, this means that there is no more obstruction to solving the completion equation, up to the lower rank ambiguities. The information about such a solution is completely encompassed by the family of holomorphic (higher depth) mock modular forms $\gi{\bfr}_{\mu,\bfmu}$. Furthermore, these functions satisfy their own modular completion equation \eqref{compl-gi}. One may wonder however about the point of replacing one modular completion equation by another. The crucial difference is that in solving for the anomalous coefficients, we can take \textbf{any} solution, as long as we use it as input for equations of higher charge. In fact, the only impact of taking a different solution $\gi{\bfr}_{\mu,\bfmu}$ is that we get a different decomposition \eqref{han} of $\tlh_{\Nr}$, where $\Nr=\sum_{\Nr_i\in\bfr}\Nr_i$. Therefore, in the rest of this chapter we will focus on solving for the anomalous coefficients.

We saw in theorem \ref{thm-ancoef} that $\gi{\bfr}_{\mu,\bfmu}$ are depth $n-1$ mock modular forms. In order to solve their completion equation we still need to specify their precise modular transformation properties. They follow directly from that of the $\tlh_\Nr$ and are given by
\be
\begin{split}
	w\(\gi{\bfr}\) &=3(n-1)/2, 
	\\
	\Mi{\gi{\bfr}}_{\mu, \bfmu, \nu, \bfnu}(T) &= 
	e^{\pi \I \( \mu - \sum_i \mu_i\)+ \pi \I \(\frac{\mu^2}{\kappa \Nr}-\sum_i \frac{\mu_i^2}{\kappa \Nr_i}\) }\,
	\delta_{\mu \nu} \delta_{\bfmu \bfnu},
	\\
	\Mi{\gi{\bfr}}_{\mu, \bfmu, \nu, \bfnu}(S) &= 
	\frac{e^{\frac{\pi\I}{4}(n-1)}}{\sqrt{\kappa^{n+1} \Nr\prod_i \Nr_i }} \,
	e^{-2\pi \I \(\frac{\mu \nu}{\kappa \Nr} - \sum_i \frac{\mu_i\nu_i}{\kappa \Nr_i}\)}.
\end{split}
\label{mult-gr}
\ee

Before proceeding with the different constructions, we give a few relevant notations
\be
\label{defr0}
\begin{split}
\Nr_0=\gcd(\bfr),\qquad 
\rdcr_i=\Nr_i/\Nr_0,\qquad 
&\kappa_{ij}=\hf\, \kappa\Nr\rdcr_i\rdcr_j, \qquad
\\
\Delta\mu=\mu-\sum_{i=1}^{n}\mu_i,\qquad 
\sum_{i=1}^{n}&\rho_i\,\rdcr_i=1,
\end{split}
\ee
where $\rho_i\in\IZ$ are any $n$-tuple satisfying the last equation.

\section{Two infinite families of solutions}
\label{sec-DTall-HeckeVW}
In this section we present the two infinite family of anomalous coefficients solving the completion equations. We start with the case of $\bfr=(\Nr_1,\Nr_2)$ that uses Hecke operators and afterwards we consider the case $\bfr=(1,\dots,1)$ related to generating functions of VW invariants. 

\subsection{Hecke-like operators}
\label{subsec-DTall-Hecke}

Let's take two arbitrary charges $r_1$ and $r_2$.
In this case the formula for the modular completion $\whgi{\Nr_1,\Nr_2}_{\mu,\mu_1,\mu_2}$ \eqref{compl-gi}, 
representing the anomaly equation, takes the simple form
\be
\whgi{\Nr_1,\Nr_2}_{\mu,\mu_1,\mu_2}(\tau,\btau)
=\gi{\Nr_1,\Nr_2}_{\mu,\mu_1,\mu_2}(\tau)+\trmRi{\Nr_1,\Nr_2}_{\mu, \mu_1, \mu_2}(\tau, \btau),
\label{whh2}
\ee
and $\gi{\Nr_1,\Nr_2}_{\mu,\mu_1,\mu_2}$ is required to be a mock modular form of weight 3/2
with the multiplier system \eqref{mult-gr} specialized to $n=2$.
The function $\trmRi{\Nr_1,\Nr_2}_{\mu, \mu_1, \mu_2}$ determining the completion 
is easily computable, but for our purposes it is sufficient to consider its derivative
with respect to $\btau$ which specifies the shadow of $\gi{\Nr_1,\Nr_2}_{\mu,\mu_1,\mu_2}$.
It is given by
\be
\p_{\btau}\trmRi{\Nr_1,\Nr_2}_{\mu, \mu_1, \mu_2}(\tau, \btau) =
\frac{\Nr_0\sqrt{\kappa_{12}}}{16\pi\I \tau_2^{3/2}}\,\delta^{(\kappa \Nr_0)}_{\Delta\mu} 
\, \overline{\theta^{(\kappa_{12})}_{\mu_{12}}(\tau)},
\label{shadowRN1N2}
\ee
where we used \eqref{defr0}, $\delta^{(n)}_x$ is the mod-$n$ Kronecker delta defined by
\be
\label{defdelta}
\delta^{(n)}_x=\left\{ \begin{array}{ll}
	1\  & \mbox{if } x=0\!\!\!\mod n,
	\\
	0\ & \mbox{otherwise}.
\end{array}\right.
\ee 
and 
\be
\label{def-mu12}
\mu_{12}=\rdcr_2\mu_1-\rdcr_1\mu_2+\rdcr_1\rdcr_2(\rho_1-\rho_2)\Delta \mu,
\ee
is a residue class and runs over $2\kappa_{12}$ values. The function $\theta^{(\kappa)}_{\mu}$ was defined in \eqref{deftheta} and in particular is periodic under shifts of $\mu$ by $2\kappa$, in line with the range of $\mu_{12}$ we just indicated. 
Notice that the tensor structure of the shadow is actually encoded in a vector-like object and a Kronecker delta. This suggests to look for a solution of the form
\be
\gi{\Nr_1,\Nr_2}_{\mu,\mu_1,\mu_2}(\tau)=\Nr_0\delta^{(\kappa\Nr_0)}_{\Delta\mu}
\Gi{\kappa_{12}}_{\mu_{12}}(\tau).
\label{sol-n=2-Hecke}
\ee
An important observation is that if $\Gi{\kappa}_\mu$ is a VV mock modular form
of weight 3/2 with a modular completion satisfying
\be
\tau_2^{3/2}\p_{\btau}\whGi{\kappa}_\mu(\tau,\btau) 
=\frac{\sqrt{\kappa}}{16\pi\I}\, \overline{\ths{\kappa}_\mu}(\btau),
\label{shadG}
\ee
where $\ths{\kappa}_\mu(\tau)$ is the theta series \eqref{deftheta} at $z=0$,
then it is trivial to see that \eqref{sol-n=2-Hecke} solves the anomaly equation \eqref{whh2}.
The only non-trivial fact to check is that it has the correct multiplier system. 
But this turns out to be guaranteed by the fact that 
\eqref{shadG} ensures that $\Gi{\kappa}_\mu$ has the multiplier system 
$\Mi{\kappa}_{\mu \nu}$ \eqref{STusual} conjugate to that of $\ths{\kappa}_\mu$ and the proposition 
\begin{proposition}
	\label{prop-multsys}
	If $\Gi{\kappa}_\mu$ ($\mu=0,\dots,2\kappa-1$) transforms with the multiplier system
	\be
	\Mi{\kappa}_{\mu \nu}(T) =
	e^{-\frac{\pi \I}{2\kappa} \,\mu^2}\delta_{\mu \nu},
	\qquad
	\Mi{\kappa}_{\mu \nu}(S)=
	\frac{e^{\frac{\pi\I}{4}}}{\sqrt{2\kappa}} \, e^{\frac{\pi \I}{\kappa}\,\mu \nu},
	\label{STusual}
	\ee
	then $\delta^{(\kappa\Nr_0)}_{\Delta\mu} \Gi{\kappa_{12}}_{\mu_{12}}$
	transforms with the multiplier system \eqref{mult-gr} specified for $n=2$.
\end{proposition}

As a result, we have reduced the problem of finding the anomalous coefficients for arbitrary two charges 
to exactly the same problem that was studied in \cite{Alexandrov:2022pgd} for charges $r_1=r_2=1$,
in which case $\kappa_{12}=\kappa$.
It was found that for any $\kappa$ equal to a power of a prime integer, $\Gi{\kappa}_\mu$
is determined by the generating series $H_\mu$ ($\mu=0,1$) of Hurwitz class numbers\footnote{An explicit formula
	for the generating series can be found in 
	\cite[Eq.(1.12)]{bringmann2007overpartitionsclassnumbersbinary}
	and its mock modular properties have been established in \cite{Zagier:1975,hirzebruch1976intersection}.}  
through the action on it by a certain modification of the Hecke-like operator introduced in 
\cite{Bouchard:2016lfg,Bouchard:2018pem}.
However, it turns out that a solution of this problem for generic $\kappa$ has already been found
in the seminal paper \cite{Dabholkar:2012nd}. 
More precisely, that paper looked for mock modular forms with 
shadow proportional to $\ths{\kappa}_\mu$
and further restricted to have the slowest possible asymptotic growth of their Fourier coefficients. 
Such functions have been called mock modular forms of optimal growth.
In our case we do not have to impose any restrictions on the asymptotic growth.
But since any solution of \eqref{shadG} is equally suitable, we can take the one provided 
by \cite{Dabholkar:2012nd}. All other solutions should differ just by a pure modular form.

We end this subsection by giving schematically the formula for the solution, without defining the precise action of the Hecke-like operators. The exact formula was found in \cite{Dabholkar:2012nd} and one can find its specification to our context and a minor correction of normalization factors in \cite{Alexandrov:2024wla}. Accordingly, we define the M\"obius function 
\be
\mu(d)=\left\{ \begin{array}{ll}
	+1 \quad & \mbox{if $d$ is a square-free with an even number of prime factors},
	\\
	-1 & \mbox{if $d$ is a square-free with an odd number of prime factors},
	\\
	0 & \mbox{if $d$ has a squared prime factor}.
\end{array}
\right.
\ee
We denote $T^{d}_{r}$ the Hecke-like operator. We don't give its precise definition but we say that it preserves modularity and that when it acts on VV modular forms, it multiplies the number of their components by $r$.

In terms of these quantities, the mock modular forms of optimal growth are given by
\be
\Gi{\kappa}_\mu=\sum_{d|\kappa\atop \mu(d)=1} \(T^{d}_{\kappa/d}\[\cGi{d}\]\)_\mu,
\label{DMZ}
\ee
where
$\cGi{d}$ are the \textit{seed functions}: VV mock modular forms of weight 3/2 with multiplier system $\Mi{d}_{\mu \nu}$.
Thus, for each square-free integer with an even number of prime factors, such as 1, 6, 10, 14, 15, etc.,
one needs to provide such a mock modular form.
The first two of them turn out to be well-known functions:
for $d=1$ it is (the doublet of) the generating series of Hurwitz class numbers,
\be
\cGi{1}_\mu(\tau)=H_\mu(\tau),
\ee
and for $d=6$ it has the following explicit expression 
\be
\cGi{6}_\mu(\tau)=\frac{\chi_{12}(\mu)}{12}\, \hi{6}(\tau),
\label{def-cG6}
\ee
where
\be
\chi_{12}(\mu)=\left\{ \begin{array}{ll}
	+1 \quad & \mbox{if } \mu=\pm 1 \mod 12,
	\\
	-1 & \mbox{if } \mu=\pm 5 \mod 12,
	\\
	0 & \mbox{if } \gcd(\mu,12)>0,
\end{array}
\right.
\ee
and 
\be
\hi{6}(\tau)=\frac{12 F_2^{(6)}(\tau)-E_2(\tau)}{\eta(\tau)}
\ee
is a mock modular form of weight 3/2 with shadow proportional to the Dedekind eta function $\eta(\tau)$, 
which is defined in terms of the quasimodular Eisenstein series $E_2(\tau)$ and the function
\be
F_2^{(6)}(\tau)=-\sum_{r>s>0}\chi_{12}(r^2-s^2)\, s\, \q^{rs/6}\, . 
\ee
For many other functions $\cGi{d}$, \cite{Dabholkar:2012nd} determined their first Fourier coefficients,
however we are not aware about any explicit expressions for their generating series.

\subsection{Relation to VW theory}
\label{subsec-DTall-VW}

Let us now consider the case of $n$ charges $r_i$ all equal to 1. In addition, we also restrict ourselves
to CYs with the intersection number $\kappa=1$. 
A crucial simplification in this case is that one can drop all indices $\mu_i$ 
because they take only $\kappa r_i=1$ value. Therefore, the corresponding anomalous coefficients
can be denoted simply as $g_{n,\mu}\equiv\gi{1,\dots,1}_{\mu}$.
Another feature of this set of anomalous coefficients is that the anomaly equations for $g_{n,\mu}$
form a closed system and do not involve other anomalous coefficients.
Moreover, it is easy to see that in this sector the anomaly equation \eqref{compl-gi} 
becomes identical to \eqref{exp-twhh} under the identification $g_{n,\mu}\leftrightarrow \tlh_{n,\mu}$ 
and thus takes the form 
\be
\whg_{n,\mu}=
\sum_{m=1}^n \sum_{\sum_{k=1}^m n_k=n}
\sum_{\bfmu}\trmRi{\bfn}_{\mu,\bfmu}
\prod_{k=1}^m g_{n_k,\mu_k}.
\label{exp-whtg}
\ee

The case $n=2$ has already been analyzed in the previous subsection. It follows from the results 
presented there, and in agreement with \cite{Alexandrov:2022pgd}, that
\be  
g_{2,\mu}=H_\mu,
\qquad 
\mu=0,1.
\ee
The vector valued function $H_\mu$ appearing here is known not only as the generating series of Hurwitz class numbers, 
but also as the (normalized) generating series of $SU(2)$ Vafa-Witten invariants on $\IP^2$, namely \cite{Vafa:1994tf}
\be  
\vwh_{2,\mu}=3 (h^{\rm VW}_1)^2 H_\mu ,
\ee
where $\vwh_{n,\mu}$ denotes the generating series of $SU(n)$ VW invariants and $h^{\rm VW}_1=\eta^{-3}$.
Combining the two relations, one obtains 
\be
g_{2,\mu}=\frac13\, \vwgi{2,\mu},
\label{relVW-n=2}
\ee
where we introduced the normalized generating series
\be
\vwgi{n,\mu}(\tau) = \eta^{3n}(\tau)\,\vwh_{n,\mu}(\tau).
\label{def-VWnorm}
\ee
As we show below, the relation \eqref{relVW-n=2} is not an accident, but a particular case 
of a more general relation between $g_{n,\mu}$ and $\vwgi{n,\mu}$.

Let us recall that the VW invariants count the Euler characteristic of moduli spaces of instantons 
in a topological supersymmetric gauge theory on a complex surface $S$ obtained from the usual $N=4$ super-Yang-Mills 
by a topological twist \cite{Vafa:1994tf}. The partition function of the theory reduces to the generating 
series of VW invariants and one could expect that it must be a modular form as a consequence of S-duality of 
the $N=4$ super-Yang-Mills. However, it turns out that on surfaces with $b_2^+(S)=1$, which includes $S=\IP^2$,
there is a modular anomaly \cite{Vafa:1994tf,Dabholkar:2020fde}. 
Its precise form can be established from the fact that the VW invariants on $S$ coincide with 
the D4-D2-D0 BPS indices on the non-compact CY given by the canonical bundle over $S$ 
\cite{Minahan:1998vr,Alim:2010cf,gholampour2017localized},
which in turn can be obtained from a compact CY given by an elliptic fibration over $S$
in the limit of large fiber. Since the modularity of the D4-D2-D0 BPS indices on such compact CY 
is governed by a generalization of \eqref{exp-whh} or \eqref{exp-twhh} to $b_2>1$, 
the generating series of VW invariants are subject to the same anomaly equation \cite{Alexandrov:2019rth}.

Furthermore, since in the local limit where the elliptic fiber becomes large 
the only divisor which remains finite is $[S]$, the D4-brane charges belong to the one-dimensional lattice,
and if $b_2(S)=1$, as is the case for $\IP^2$, the lattice of D2-brane charges is also one-dimensional.
Thus, for $S=\IP^2$ one reduces to the ``one-dimensional" case captured by the anomaly equation \eqref{exp-twhh}
with $\kappa=[H]^2=1$ where $[H]$ is the hyperplane class of $\IP^2$. The fact that we start with a compact CY with $b_2>1$ actually has consequences on the structure of the anomaly equation. Namely, the normalized generating functions of VW invariants for $SU(n)$ satisfy
\be
\whvwg_{n,\mu}=
\sum_{m=1}^n \sum_{\sum_{k=1}^m n_k=n}
\sum_{\bfmu}\trmRVWi{\bfn}_{\mu,\bfmu}
\prod_{k=1}^m \vwgi{n_k,\mu_k},
\label{exp-whVW}
\ee
where
\be 
\trmRVWi{\bfn}_{\mu,\bfmu}=3^{m-1}\trmRi{\bfn}_{\mu,\bfmu}.
\ee
where $m$ is the number of charges which the functions depend on.
Substituting this into \eqref{exp-whVW} and comparing to \eqref{exp-whtg},
one finds that the two equations become identical provided one identifies\footnote{The freedom
	to include in this relation a constant factor $c^n$ allowed by the equations 
	is fixed by the normalization conditions $g_{1}=\vwgi{1}=1$.} 
\be
g_{n,\mu}=3^{1-n}\vwgi{n,\mu}(\tau).
\label{rel-tggP2-n}
\ee
This result is consistent with \eqref{relVW-n=2} and provides an explicit solution for 
the anomalous coefficients with $r_i=\kappa=1$.

\section{General strategy: extend, solve, reduce}
\label{sec-DTall-Extensions}
It is natural to try to generalize the two previous methods for arbitrary $\bfr$. Our efforts in doing so were, however, unsuccessful. Therefore, we devise a new strategy for a solution that is general from the outset. This strategy actually does not solve the problem at hand directly. Instead, we perform a few generalizations that give us a new, more easily solvable problem. And at the same time we have to make sure that a solution to this generalized problem can always be reduced back to a solution of the original one (see Fig. \ref{fig-strategy}).

\lfig{The different extensions (blue) and reductions (green) in our strategy.}
{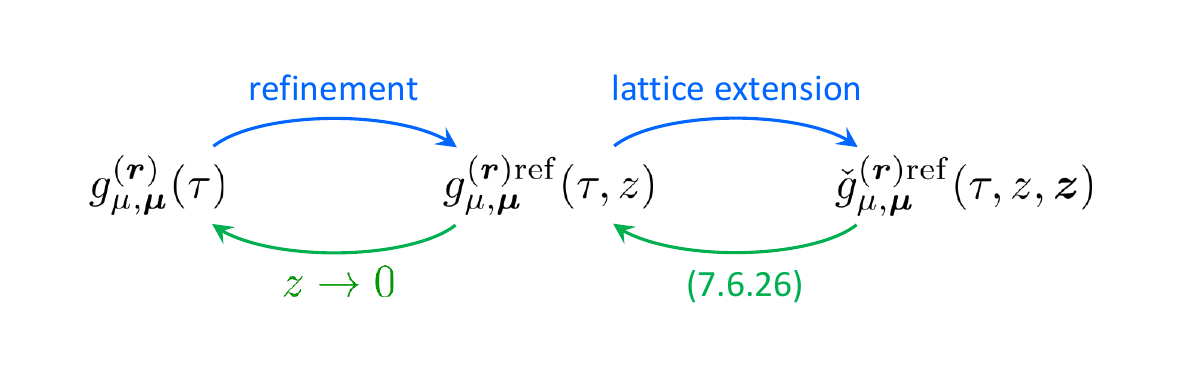}{10cm}{fig-strategy}{-0.3cm}

Once we formulate the new completion equations, we use indefinite theta series (cf \S \ref{sec-mod-theta}) to solve it, in a similar way to the solution of the same kind of modular anomaly equation for the generating functions of VW invariants constructed in \cite{Alexandrov:2020bwg,Alexandrov:2020dyy}.

The two generalizations that we apply to the problem are the {\it refinement} and the {\it lattice extension}. The former not only simplifies the equations \cite{Alexandrov:2019rth}, but also introduces a regularization for an otherwise divergent theta series, as we will see later. We will start in subsection \ref{subsec-DTall-refinement} by adding the refinement. Then in subsection \ref{subsec-DTall-motivbbm} we will motivate the lattice extension and then perform it in subsection \ref{subsec-DTall-bbm}. Finally, in subsection \ref{subsec-DTall-complchgirf} we will present the extended completion equation. We will use the boldface script $\bfx$ to denote $n$-dimensional objects and the blackboard script $\xbbm$ to denote objects on the extended $d$-dimensional lattice.

\subsection{Refinement}
\label{subsec-DTall-refinement}

A refinement has its physical origin 
in a non-trivial $\Omega$-background \cite{Nekrasov:2002qd,Nekrasov:2010ka}. It introduces a complex parameter $y=e^{2\pi\I z}$
which can be thought of as a fugacity conjugate to the angular momentum $J_3$ 
in uncompactified dimensions. At the same time, the BPS indices undergo refinement and are then given by \eqref{refined-BPS}. Their generating functions can be defined as was done in section \ref{sec-DTall-def}.
Crucially, the refinement preserves the modular properties of the generating series 
of BPS indices \cite{Alexandrov:2019rth}. More precisely, after refinement they become {\it mock Jacobi forms} 
for which the role of the elliptic argument is played by the refinement parameter $z$
and the formula for their modular completions takes exactly the same form as in \eqref{exp-whh}, but with the coefficients 
given now by\footnote{We give the coefficients {\it after} performing the same redefinition as in \eqref{redef-th},
	so that the formula to compare with is \eqref{redefrmRi} rather than \eqref{expRn},
	but we omit the tilde on $\rmRirf{\bfr}_{\mu,\bfmu}$ to avoid cluttering.} 
\be
\label{def-rmRirf}
\rmRirf{\bfr}_{\mu,\bfmu}(\tau,\btau,z)
=
\sum_{\sum_{i=1}^{n} q_i = \mu
	\atop q_i \in \kappa \Nr_i \IZ + \mu_i } 
\Sym \Bigl\{ 
\scRrf_n(\bfhgam;\tau_2,\beta) \, y^{\sum_{i<j} \gamma_{ij}}
\Bigr\} \,e^{\pi\I \tau Q_n(\bfhgam)},
\ee
where we set $z=\alpha-\tau\beta$ with $\alpha,\beta\in \IR$. The main difference here, besides the appearance of a power of $y$, lies in the form of the coefficients $\scRrf_n$ which are much simpler than \eqref{expRn}.
Indeed, while the coefficients $\scR_n$
involve a sum over two types of trees weighted by generalized error functions and their derivatives,
for $\scRrf_n$ one needs only one type of trees and no derivatives. We show how $\scRrf_n$ are constructed at the end of this subsection.

The main property that we need is that in the unrefined limit 
$\rmRirf{\bfr}_{\mu,\bfmu}$ develop a zero of order $n-1$ with a coefficient given by $\trmRi{\bfr}_{\mu,\bfmu}$:
\be
\trmRi{\bfr}_{\mu,\bfmu}(\tau,\btau)=
\lim_{y\xrightarrow{}1} (y-y^{-1})^{1-n} \rmRirf{\bfr}_{\mu,\bfmu}(\tau,\btau,z),
\label{unreflim-Rn}
\ee   
where the limit of $y$ should be taken while setting $\by=1$. 
Therefore, if we define {\it refined anomalous coefficients} as solutions of the following 
modular anomaly equation
\be
\begin{split}
	\whgirf{\bfr}_{\mu,\bfmu} 
	=&\,
	\Sym\Bigg\{
	\sum_{m=1}^n \sum_{\sum_{k=1}^m n_k=n} \sum_{\bfnu}
	\rmRirf{\bfs}_{\mu,\bfnu}
	\prod_{k=1}^m \girf{\frr_k}_{\nu_k,\frm_k} \Bigg\},
\end{split}
\label{complgirf}
\ee
where $\whgirf{\bfr}_{\mu,\bfmu}$ is required to be a VV Jacobi-like form of weight $\hf(n-1)$,
index\footnote{The weight 
	is obtained from the relation \eqref{lim-ancoef} by taking into account that the $y$-dependent factor 
	in the limit $y\to 1$ is proportional to $z^{1-n}$ and thus increases the weight by $n-1$.  
	The index instead follows from the index of the generating series of refined BPS indices
	which was established in \cite{Alexandrov:2019rth}.} 
\be 
m_{\bfr}=-\frac{\kappa}{6}\,\biggl(\Nr^3-\sum_{i=1}^n\Nr_i^3\biggr),
\label{index-mr}
\ee
and the same multiplier system as $\gi{\bfr}_{\mu,\bfmu}$ (see \eqref{mult-gr}), then 
a solution of \eqref{compl-gi} is obtained from these refined anomalous coefficients as 
\be
\gi{\bfr}_{\mu,\bfmu}(\tau)=
\lim_{y\xrightarrow{}1} (y-y^{-1})^{1-n} \girf{\bfr}_{\mu,\bfmu}(\tau,z).
\label{lim-ancoef}
\ee   
This is easily checked by multiplying \eqref{complgirf} by $(y-y^{-1})^{1-n}$ and taking the unrefined limit.
As a result, we have reformulated the problem of solving one anomaly equation in terms of solving another equation and subsequent evaluation of the unrefined limit.
Importantly, the relation \eqref{lim-ancoef} implies that the unrefined limit exists only if the refined solution
has a zero of order $n-1$ at $z=0$. In order to satisfy this condition, we will carefully choose our solution.

Despite the unclear status of refined BPS indices, our construction is consistent. This is because we do {\it not}
use the refined BPS indices or their generating functions, but only the coefficients 
\eqref{def-rmRirf} characterizing the refined completions.
In other words, we use the existence and properties of $\rmRirf{\bfr}_{\mu,\bfmu}$ as a mere trick 
to produce solutions to the anomaly equations \eqref{compl-gi}.\\

Let's see how the refinement makes the completion equations simpler. The main difference lies in the kernels $\scRrf_n$ which are now defined through \cite{Alexandrov:2019rth} the sum over \sch trees as in \eqref{expRn}, 
\be
\scRrf_n\(\bfhgam;\tau_2,\beta\) = \frac{1}{2^{n-1}}\sum_{T\in\IT_n^{\rm S}}(-1)^{n_T-1} 
\Eprf_{v_0} \prod_{v\in V_T \backslash \{v_0\}}\Efrf_v,
\label{refsolRn}
\ee
but now with the weights assigned to vertices determined by new functions $\Er_n(\bfhgam;\tau_2,\beta)$.
Although they depend on an additional parameter $\beta$, they are actually much simpler than their unrefined 
analogues $\Ev_n$ because in their definition there is no sum over trees.
Namely, they are given by
\be
\Er_n(\bfhgam;\tau_2,\beta)= \Phi^E_{n-1}\(\{ \bfv_{\ell}\};\sqrt{2\tau_2}\,(\bfq+\beta\bftet )\),
\label{Erefsim}
\ee
where $\Phi^E$ are defined in \eqref{generrPhiME} and
\be
\bfv_\ell= \sum_{i=1}^\ell\sum_{j=\ell+1}^n\bfv_{ij},
\qquad
\bftet = \sum_{i<j} \bfv_{ij}.
\label{def-bfvk}
\ee
As in the unrefined case, $\Eprf_n=\Er_n-\Efrf_n$, while $\Efrf_n$
is the large $\tau_2$ limit of $\Er_n$. However, before taking the limit, 
one should first set $\beta=0$, i.e. 
\be 
\Efrf_n(\bfhgam)\equiv  \lim_{\tau_2\to\infty}\Er_n(\{\gama_i\};\tau_2,0)
= S_{\cT_{\rm lin}}(\bfhgam),
\label{Efref}
\ee
where $S_{\cT_{\rm lin}}$ is defined in \eqref{def-all-modanom} and $\cT_{\rm lin}=
\bullet\!\mbox{---}\!\bullet\!\mbox{--}\cdots \mbox{--}\!\bullet\!\mbox{---}\!\bullet\,$
is the simplest linear tree.
Note that for the linear tree $e_{\cT_{\rm lin}}=\delta^{(2)}_{n-1}/n$ where $n$ is the number of vertices.

Before going further, we give, as an example, the simple expression
\be 
\scRrf_2(\gama_1,\gama_2)=
\hf\[
E_1\(\frac{\sqrt{2\tau_2}
	\( \gamma_{12} + \kappa \Nr \Nr_1 \Nr_2\beta \)}
{\sqrt{\kappa\Nr \Nr_1 \Nr_2}}\)
-\sgn (\gamma_{12}) \],
\label{expr-R2}
\ee
where $E_1(x)$ is defined in \eqref{generr-E} and coincides with the usual error function. We will use this expression in order to study the completion equation for anomalous coefficients with $n=2$ charges.

\subsection{Motivating the lattice extension}
\label{subsec-DTall-motivbbm}

Now let's motivate the second step in our construction: the {\it lattice extension}.
Let's take a look at the completion equation for the anomalous coefficients and argue for a solution in terms of indefinite theta series. We will pay special attention to the lattice on which the coefficients of the completion equation are defined, as it determines the lattice on which we ought to construct our solution. 

The term $\rmRirf{\bfr}_{\mu,\bfmu}$ in \eqref{complgirf}, up to some shifts in the arguments of its kernel, is given by a theta series \eqref{gentheta} defined on the $(n-1)$-dimensional lattice 
\be
\bfLami{\bfr}=\left\{\bfk\in \IZ^n \ :\ \sum_{i=1}^n \Nr_i k_i=0 \right\},
\label{def-bfLam}
\ee
with quadratic form $-Q_n(\bfhgam)$ given by
\be
\bfk^2=\kappa\sum_{i=1}^n \Nr_i k_i^2.
\label{qform-bfLam}
\ee
The agreement between this expression and \eqref{def-Qn} is ensured through
\be
\label{ki-to-qi}
q_i=\kappa\Nr_i\(k_i+\frac{\mu}{\kappa\Nr}+\hf\),
\ee
with $\bfk$ belonging to the dual lattice of $\bfLami{\bfr}$.  
Moreover, the characteristic vector is given by $\bfp=0$ and the elliptic vector by $z\bftet$ with $\bftet$ defined in \eqref{def-bfvk}.
The residue class is determined in terms of $(\mu,\bfmu)$. 
However, its kernel does not solve \vig equation. So, we look for a solution written, up to a holomorphic modular ambiguity, as a theta series on the same lattice such that its kernel is holomorphic and combines with that of $\rmRirf{\bfr}$ to give a new kernel that does satisfy \eqref{Vigdif}. In fact, assuming this choice for all lower rank $\gi{\frr_k}$, all the completing terms\footnote{The completing terms designate all terms in the right hand side of the expression for the completion equation except the function to be completed.} can be written on the same lattice. Looking closely at the sum of their kernels, we can see that it is written as a product of differences of generalized error functions and sign functions. As explained below \eqref{Phi-ex}, in order to get a completion we need to have either terms given by generalized error functions or by sign functions of null vectors. Thus, the kernel of the indefinite theta series in $\girf{\bfr}_{\mu,\bfmu}$ should contain sign functions with null vectors. Unfortunately, our lattice \eqref{def-bfLam} is of (positive) definite signature and does not contain null vectors. 
This problem can be addressed by performing a {\it lattice extension} and this is what we do in the next part.

\subsection{The extended lattice}
\label{subsec-DTall-bbm}
A lattice extension is a standard trick in the theory
of mock modular forms \cite{Zwegers-thesis}. The idea is that the original problem defined on a lattice $\Lambda$ 
is reformulated on a larger lattice $\bbLambda=\Lambda\oplus \Lambda_{\rm ad}$ 
that admits a solution in terms of indefinite theta series and, because $\bbLambda$ is a direct sum,
such solution is expected to be reducible to a solution on $\Lambda$. 
However, if the discriminant group $D_{\rm ad}=\Lambda_{\rm ad}^*/\Lambda_{\rm ad}$ is non-trivial,
the reduction to the original lattice is possible only if the solution on the extended lattice satisfies 
certain identities ensuring that components of the solution labeled by different elements of 
$D_{\rm ad}$ reduce to the same functions. 
In general, there is no guarantee that this is the case. 
Therefore, we should require triviality of $D_{\rm ad}$,
which in turn requires that, if $\Lambda_{\rm ad}=\IZ^{d_{\rm ad}}$, then the corresponding 
quadratic form is given by (minus) the identity matrix. 

In our case $\Lambda=\bfLami{\bfr}$
with quadratic form $-Q_n$ and $\Lambda_{\rm ad}$ should be chosen in a way to ensure the existence 
of a null vector on $\bbLambda$. Moreover, the null vectors must belong to the lattice, so we can't simply add one direction with quadratic form $-x^2$ since we need to cancel the norm of vectors $\bfv_e\in\bfLam$ which is, in general, not a perfect square. In fact there are various other conditions that need to be satisfied by the lattice extension. We will simply give the appropriate $\bbLambda$ and then explain its different features. 

We start by introducing:
\begin{itemize}
	\item 
	integer valued function $d_\Nr$ of the magnetic charge (and intersection number $\kappa$)
	such that $d_\Nr\geq 2$;
	\item 
	$d_\Nr$-dimensional vectors $\frt^{(\Nr)}$ such that their components 
	are all non-vanishing integers and sum to zero, $\sum_{\alpha=1}^{d_\Nr}\frt^{(\Nr)}_\alpha=0$.
\end{itemize}
Note that if $d_\Nr$ could be equal to 1, it would be impossible to satisfy the last condition on $\frt^{(\Nr)}$.
The main features of the construction below do not depend on a specific form of $\frt^{(\Nr)}$, therefore we will not specify it. Although the exact expression for $d_\Nr$ is also not crucial, it is nicer to follow the discussion if we give it
\be 
d_\Nr=\left\{\begin{array}{ll} 
	4\Nr, & \quad\kappa=1,
	\\ 
	\kappa\Nr, & \quad\kappa>1.
\end{array}\right.
\label{defdr}
\ee 
The extended lattice is then given by
\be
\bbLami{\bfr}=\bfLami{\bfr}\oplus \IZ^{d_\bfr},
\label{extlatNr}
\ee
where $d_\bfr=\sum_{\Nr_i\in\bfr}\di{i}$
and it carries the bilinear form
\be
\xbbm\ast\ybbm=\sum_{i=1}^n \(\kappa \Nr_i x_i y_i-\sum_{\alpha=1}^{d_{\Nr_i}} x_{i,\alpha} y_{i,\alpha}\),
\label{bb-r}
\ee
where $\xbbm=\{x_i,x_{i,\alpha_i}\}$ with $i=1,\dots ,n$ and $\alpha_i=1,\dots,d_{\Nr_i}$.

Let's describe this choice. For each charge $\Nr_i$ the original lattice has a direction associated to this charge, with a factor $\kappa\Nr_i$ in the quadratic form. In the lattice extension we associate to $\Nr_i$ the lattice $\IZ^{\kappa \Nr_i}$ with $-1$ in the quadratic form. Hence, we associate to one direction with eigenvalue $\kappa\Nr_i$ in its quadratic form, $\kappa\Nr_i$ directions with eigenvalues\footnote{This guarantees that $D_{\rm ad}$ is trivial, as required in our strategy (cf section \ref{subsec-DTall-bbm})} $-1$. So in a sense the lattice used for the extension is a "diagonal expansion" of the initial lattice. 

The inclusion of a factor of $4$ when $\kappa=1$ is to avoid having $d_\Nr=1$ for some charges and thus contradicting the first condition shown in the previous subsection and it will induce some extra factors of $2$ in the null vectors. However, in this dissertation we will focus on the simpler, more general case of $\kappa>1$, and we refer to \cite{Alexandrov:2024wla} for a treatment that encompasses both cases.

The choice \eqref{defdr} fulfills two goals. The first is that it preserves the recursive structure of the completion equation and it does so because $d_\Nr$ is additive. The second is that it guarantees the existence of null vectors $\wbbm_{ij}$, which will be defined shortly.
More generally, we will use two sets of vectors $\vbbm_{ij},\wbbm_{ij}$ belonging to the extended lattice $\bbLami{\bfr}$, both of which
are extensions of the vectors $\bfv_{ij}\in \bfLami{\bfr}$ 
defined as in \eqref{defvij}
\be 
(\bfv_{ij})_k=\delta_{ki}\Nr_j-\delta_{kj}\Nr_i.
\label{def-bfvij}
\ee 
They are given by 
\be 
\begin{split}
	(\vbbm_{ij})_k=&\,(\bfv_{ij})_k,
	\qquad\ \
	(\vbbm_{ij})_{k,\alpha}= 0,
	\\
	(\wbbm_{ij})_k=&\, (\bfv_{ij})_k,
	\qquad
	(\wbbm_{ij})_{k,\alpha}= (\bfv_{ij})_k.
\end{split}
\label{def-bfcvij}
\ee 
Below we will see how the existence of the null vectors $\wbbm_{ij}$ gives the possibility to construct 
holomorphic theta series associated with the extended lattice and satisfying 
the anomaly equation \eqref{complchgirf}.

\subsection{Extended and refined completion equation}
\label{subsec-DTall-complchgirf}
In this subsection we will first give the extended, refined completion equation. Then we will give the modified modular properties of the corresponding anomalous coefficients. Finally, we will show how a solution of this equation reduces to a solution of \eqref{complgirf}.

The lattice extension, at the level of the completion equation, is ensured by taking 
\be
\label{chgirf-init}
\chgirf{\Nr}_{\mu,\mu'}(\tau,z,z')=\delta_{\mu,\mu'}\prod_{\alpha=1}^{d_\Nr}\theta_1(\tau,\bft^{(r)}_\alpha z'),
\ee
in the case of a single charge and, for a general number of charges, the replacement
\be
\label{chgirf-replacement}
\girf{\bfr}_{\mu,\bfmu}(\tau,z) \to \chgirf{\bfr}_{\mu,\bfmu}(\tau,z,\bfz),
\ee
where the extended anomalous coefficients now depend on a vector of extra elliptic parameters $\bfz=(z_1,\dots,z_n)$.

The equation then becomes 
\be
\whchgirf{\bfr}_{\mu,\bfmu}(\tau,z,\bfz) 
=
\Sym \Bigg\{
\sum_{m=1}^n \sum_{\sum_{k=1}^m n_k=n} \sum_{\bfnu}
\rmRirf{\bfs}_{\mu,\bfnu}(\tau,z)
\prod_{k=1}^m \chgirf{\frr_k}_{\nu_k,\frm_k} (\tau,z,\frz_k)
\Bigg\},
\label{complchgirf}
\ee
where $\frz_k=(z_{j_k+1},\dots,z_{j_{k+1}})$.
Formally it looks the same as \eqref{complgirf}. However, the lattice is effectively extended due to \eqref{chgirf-init}. This can be most easily seen from the term with $m=n$. For terms with $m<n$ and thus involving anomalous coefficients with $1<n_i\leq n$, the extension is ensured by the fact that the solution for $\girf{\frr_k}$ is written as a theta series on the corresponding extended lattice.

After the lattice extension, the new completion terms are still of the form \eqref{gentheta} but now written on the lattice $\bbLami{\bfr}$ and with ingredients given by 
\be
\begin{split}
	\bbmu=&\,(\bfhmu;0,\dots, 0),
	\hspace{3.3cm}
	\pbbm=\, \(\boldsymbol{0};-1,\dots, -1\)\, ,
	\\
	\zbbm=&\, (\bftet z;-\frt^{(\Nr_1)}z_1;\dots;-\frt^{(\Nr_n)}z_n),
	\qquad
	\bftet=\sum_{i<j}\bfv_{ij},
\end{split}
\label{thetadata}
\ee
where $\Delta \mu=\mu-\sum_{i=1}^{n}\mu_i$ and $\rho_i$ verifying $\sum_{i=1}^{n}\rdcr_i\rho_i=1$. Note that one has the relation 
\be 
\bftet^2=-2m_\bfr,
\qquad
\bftet\cdot \bfk=\sum_{i<j}\gamma_{ij}, 
\label{rel-scpr-th}
\ee
where\footnote{Note that $\bftet\cdot\bfk=\bftet\cdot\bfq$ where $\bfq=\bigl(\frac{q_1}{\kappa \Nr_1}\, ,\dots ,\, \frac{q_n}{\kappa \Nr_n}\bigr)$ since the difference $\bfk-\bfq$ is orthogonal to $\bftet$.} $\bfk=\bigl(k_1\, ,\dots ,\, k_n\bigr)$
is related to the physical charges through \eqref{ki-to-qi}. They ensure that the factor $e^{2\pi\I \zbbm\ast\kbbm}$ 
in the theta series reproduces the $y$-dependent factor in \eqref{def-rmRirf}
and gives rise to the index \eqref{index-mr}. 
Let us also mention here another useful relation. The argument of the kernel in the theta series \eqref{gentheta} 
is $\xbbm=\sqrt{2\tau_2}(\kbbm+\bbbeta)$ where $\kbbm$ runs over the lattice.
Therefore, it is useful to introduce $\xbbm_\bbbeta=\xbbm-\sqrt{2\tau_2}\,\bbbeta$ which in our case 
takes the form $\xbbm_\bbbeta=\sqrt{2\tau_2}(\bfk; k_{1,1},\dots, k_{n,\di{n}})$.
With respect to the bilinear form \eqref{bb-r}, one finds that
\be 
\xbbm_\bbbeta\ast\vbbm_{ij}=\sqrt{2\tau_2}\gamma_{ij} .
\label{rel-scpr}
\ee

The additional factor in \eqref{chgirf-init} leads to a change in the modular properties of 
$\chgirf{\bfr}_{\mu,\bfmu}$ compared to $\girf{\bfr}_{\mu,\bfmu}$: they
should be higher depth {\it multi-variable} Jacobi-like\footnote{More precisely, it is a Jacobi-like form with respect to $z$ and Jacobi with respect to the other parameters $z_i$.} forms 
of weight, index (which is now a matrix since there are several elliptic arguments) 
and multiplier system 
\be
\begin{split}
	w\(\chgirf{\bfr}\)=&\, (n-1+d_\bfr)/2,
	\\
	m\(\chgirf{\bfr}\) =&\, 
	\hf\, \diag\biggl(-\frac{\kappa}{3}\biggl(\Nr^3 - \sum_{i=1}^n \Nr_i^3\biggr),(\frt^{(\Nr_1)})^2,\dots,(\frt^{(\Nr_n)})^2\biggr),
	\\
	\Mi{\chgirf{\bfr}}_{\mu, \bfmu, \nu, \bfnu}(T) =&\, 
	e^{\pi \I \( \mu - \sum_i \mu_i\)+ \pi \I \(\frac{\mu^2}{\kappa \Nr}-\sum_i \frac{\mu_i^2}{\kappa \Nr_i}\)
		+\frac{\pi\I}{4}\, d_\bfr}\,
	\delta_{\mu\nu} \delta_{\bfmu\bfnu},
	\\
	\Mi{\chgirf{\bfr}}_{\mu, \bfmu, \nu, \bfnu}(S) =&\, 
	\frac{e^{\frac{\pi\I}{4}(n-1-3d_\bfr)}}{\sqrt{\kappa^{n+1} \Nr\prod_i \Nr_i }} \,
	e^{-2\pi \I\(\frac{\mu \nu}{\kappa \Nr}-\sum_i\frac{\mu_i \nu_i}{\kappa \Nr_i} \) },
\end{split}
\label{multsys-chgirf}
\ee
where $(\frt^{(\Nr)})^2=\sum_{\alpha=1}^{d_\Nr}(\frt^{(\Nr)}_\alpha)^2$ and $d_\bfr=\sum_{i=1}^n d_{\Nr_i}$. 
The multiplier system \eqref{multsys-chgirf} can be easily obtained by combining 
\eqref{mult-gr} with the modular properties of the  Jacobi theta function given in \eqref{multi-theta-N}.

The important property of the system of equations \eqref{complchgirf} is that any solution
that is regular at $\bfz=0$ reduces to a solution of \eqref{complgirf}
with the required modular properties. The relation between the two solutions is given by
\be
\girf{\bfr}_{\mu,\bfmu}(\tau,z)= 
\frac{1}{\(-2\pi\eta^3(\tau)\)^{d_\bfr}} 
\(\prod_{i=1}^{n}\frac{\cD^{(d_{\Nr_i})}_{\hf(\frt^{(\Nr_i)})^2}(z_i)}
{d_{\Nr_i}!\prod_{\alpha=1}^{d_{\Nr_i}}\frt^{(\Nr_i)}_\alpha}\) 
\chgirf{\bfr}_{\mu,\bfmu}|_{\bfz=0},
\label{recover-gref}
\ee
where $d_\bfr=\sum_{i=1}^n d_{\Nr_i}$ and 
the modular differential operators $\cD_m^{(n)}$ are defined in \eqref{defcDmn}.
Indeed, due to Proposition \ref{prop-Jacobi-n} and the fact that $\theta_1(\tau,z)$ and $\eta^3(\tau)$
have identical multiplier systems, 
the product of the differential operators in \eqref{recover-gref} acting on the completion $\whchgirf{\bfr}_{\mu,\bfmu}$ 
produces a Jacobi-like form with weight and multiplier system as in \eqref{mult-gr} and index \eqref{index-mr}.
Then to see that $\girf{\bfr}_{\mu,\bfmu}$ defined by \eqref{recover-gref} 
satisfies the anomaly equation \eqref{complgirf}, it is sufficient 
to apply this product of the differential operators
to \eqref{complchgirf} and use the fact that each differential operator acts 
only on one of the functions $\chgirf{\frr_k}$ on the r.h.s. of this equation.\footnote{It was 
	to ensure this factorization property that we introduced the additional 
	refinement parameters $z_i$ for each magnetic charge.} 
Finally, the standard normalization for the case $n=1$ is reproduced due to the property 
\be
\frac{\cD^{(d_{\Nr})}_{\hf(\frt^{(\Nr)})^2}}
{d_{\Nr}!\prod_{\alpha=1}^{d_{\Nr}}\frt^{(\Nr)}_\alpha}
\prod_{\alpha=1}^{d_\Nr} \theta_1(\tau,\frt^{(\Nr)}_\alpha z)|_{z=0}=\(\p_z \theta_1(\tau,0)\)^{d_\Nr}
=\(-2\pi\eta^3(\tau)\)^{d_\Nr}.
\ee

\section{Solving the $n=2$ case}
\label{sec-DTall-sol2}

Here we will explain how all the steps work in the case of $n=2$ charges. Conceptually, this case only holds minor differences with the general case and we will indicate where they should arise. So this can be thought of as a toy model study of the full problem. 

There are a few steps involved in this section. First, we find a solution to the completion equation of refined anomalous coefficients on the extended lattice in terms of an indefinite theta series and argue for its convergence and modularity. Then, we perform a factorization of the lattice, that induces a factorization of the theta series. This proves instrumental in finding that our initial solution has a pole at $z=0$ and that we need to choose a holomorphic modular ambiguity that makes it regular. Subsequently, we apply the modular derivatives that map the solution to a solution of the equation on the non-extended lattice. Finally, we take the unrefined limit.

\subsection{Generic solution}
\label{subsec-DTall-compl2}

Before writing the extended completion equation for $n=2$, let's look at a simplification that happens in this case. 
In fact, we can easily solve the constraint in the double sum defining \eqref{def-rmRirf} and work with a single sum.
Namely, 
\be
q_1=\rdcr_2\ell+\frac{\Delta\mu}{\kappa\Nr_0}\rho_1+\mu_1, \qquad q_2=-\rdcr_1\ell+\frac{\Delta\mu}{\kappa\Nr_0}\rho_2+\mu_2
\ee
where $\ell\in\IZ$, solve the condition on the sum. Hence, we have an isomorphism $\bfLami{\Nr_1,\Nr_2}\sim \IZ$ given by 
\be
\label{bfLami2-isom}
k\in\IZ \to \bfk=(\rdcr_2,-\rdcr_1)k \in \bfLami{\Nr_1,\Nr_2},
\ee
with the quadratic form on $\IZ$ given by
\be
\label{Q2-solved}
2\kappa_{12}k^2.
\ee

Now we can write the extended completion equation, using the above simplification and notations introduced in \eqref{defr0} and \eqref{def-mu12}:
$\Nr_0,\,\rdcr_i,\, \rdcr,\,\kappa_{12}$ and $\mu_{12}$.
We have
\bea
\whchgirf{\Nr_1,\Nr_2}_{\mu, \mu_1, \mu_2} &=& 
\chgirf{\Nr_1,\Nr_2}_{\mu, \mu_1, \mu_2}
+\prod_{i=1}^2\(\prod_{\alpha=1}^{d_{\Nr_i}}\theta_1(\tau,\frt^{(d_{\Nr_i})}_\alpha z_i)\)
\rmRirf{\Nr_1,\Nr_2}_{\mu, \mu_1, \mu_2}
\nn\\
&=&
\chgirf{\Nr_1,\Nr_2}_{\mu, \mu_1, \mu_2}
+ \frac14\, \delta^{(\kappa\Nr_0)}_{\Delta\mu}
\sum_{\sigma=\pm 1}\sum_{k\in \IZ+\frac{\mu_{12}}{2\kappa_{12}}} 
\(\prod_{i=1}^2\prod_{\alpha=1}^{d_{\Nr_i}}\sum_{k_{i,\alpha}\in\IZ+\hf} \)
\biggl[E_1\bigl(2\sqrt{\kappa_{12}\tau_2}(\sigma k+\Nr_0 \beta)\bigr) 
\nn\\
&&\qquad
- \sgn (\sigma k) \biggr] (-1)^{\pbbm\ast\kbbm}
\q^{-\hf\kbbm^2} e^{2\pi \I \zbbm_\sigma \ast \kbbm},
\label{ext-cpmlg2}
\eea
with $\kbbm=(k;k_{1,1},\dots,k_{1,d_{\Nr_1}};k_{2,1},\dots,k_{2,d_{\Nr_2}})$ and the data \eqref{thetadata} specified to $n=2$. Namely,
$\pbbm=(0;-1,\dots,-1)$ and 
$\zbbm_\sigma=(\sigma\Nr_0 z ;-\frt^{(\Nr_1)}z_1;-\frt^{(\Nr_2)}z_2)$, 
which are contracted using the bilinear form
\be
\kbbm\ast\kbbm'=2\kappa_{12} kk'-\sum_{i=1}^2\sum_{\alpha=1}^{d_{\Nr_i}} k_{i,\alpha}k'_{i,\alpha}.
\label{qfbfk-zalpha}
\ee
This bilinear form is the image of \eqref{bb-r} upon the isomorphism 
$\bbLami{\Nr_1,\Nr_2}\simeq \IZ\oplus \IZ^{d_{\Nr_1}}\oplus \IZ^{d_{\Nr_2}}$
implied by \eqref{bfLami2-isom}. Under the same isomorphism, the vectors \eqref{def-bfcvij} become
\be
\vbbm_{12}=\(1;\vnc{0}{d_{\Nr_1}};\vnc{0}{d_{\Nr_2}}\),
\qquad
\wbbm_{12}=\bigl(\Nr_0;\vnc{\Nr_2}{d_{\Nr_1}};-\vnc{\Nr_1}{d_{\Nr_2}}\bigr),
\label{def-v0-2}
\ee 
where $\vnc{x}{n}$ denotes the $n$-dimensional vector with all components equal to $x$.
Using them 
the argument of the error function can be rewritten as $2\sqrt{\kappa_{12}\tau_2}(k+\sigma\Nr_0 \beta)
=\sqrt{2\tau_2}(\kbbm+\bbbeta_\sigma)\ast\vbbm_{12}/\sqrt{\vbbm_{12}^2}$ where 
we have done the usual decomposition $\zbbm_\sigma=\bbalpha_\sigma-\tau\bbbeta_\sigma$.
As a result, the second term in \eqref{ext-cpmlg2}, up to a $\sigma$-dependent factor 
and a $\beta$-dependent shift in the argument of the sign function, 
acquires the form of the theta series \eqref{gentheta}
associated with the lattice $\bbLami{\Nr_1,\Nr_2}$, residue class 
$\bbmu=\frac{\mu_{12}}{2\Nr_0\kappa_{12}}\vbbm_{12}$
and kernel 
\be
\PhiRi{\Nr_1,\Nr_2}(\xbbm)=E_1\(\frac{\xbbm\ast\vbbm_{12}}{||\vbbm_{12}||}\)-\sgn(\xbbm_\bbbeta\ast\vbbm_{12}),
\label{kerPhi1?}
\ee
where $||\vbbm||=\sqrt{\vbbm^2}$ is the norm of a vector and $\xbbm_\bbbeta=\xbbm-\sqrt{2\tau_2}\,\bbbeta_\sigma$.
More precisely, we get
\be
\whchgirf{\Nr_1,\Nr_2}_{\mu, \mu_1, \mu_2} = 
\chgirf{\Nr_1,\Nr_2}_{\mu, \mu_1, \mu_2}
+ \frac14\, \delta^{(\kappa\Nr_0)}_{\Delta\mu}
\sum_{\sigma=\pm 1}\sigma \,\vth_{\bbmu}(\tau, \zbbm_\sigma; \bbLami{\Nr_1,\Nr_2}, \PhiRi{\Nr_1,\Nr_2},\pbbm).
\label{ext-cpmlg2th}
\ee

Using the null vectors in our lattice, we can write a generic solution as
\be
\chgirf{\Nr_1,\Nr_2}_{\mu, \mu_1, \mu_2}=\chphi^{(\Nr_1,\Nr_2)}_{\mu, \mu_1, \mu_2}
+ \frac14\, \delta^{(\kappa\Nr_0)}_{\Delta\mu}
\sum_{\sigma=\pm 1}\sigma \,\vth_{\bbmu}(\tau, \zbbm_\sigma;\bbLami{\Nr_1,\Nr_2}, \Phii{\Nr_1,\Nr_2}, \pbbm),
\label{sol-cpmlg2th}
\ee
where we use the null vector $\wbbm_{12}$ to define the kernel
\be
\Phii{\Nr_1,\Nr_2}(\xbbm)=\sgn(\xbbm_\bbbeta\ast\vbbm_{12})-\sgn(\xbbm\ast\wbbm_{12}).
\label{kerPhi1}
\ee
as in \eqref{Phigi-ex}.
Here
$\chphi^{(\Nr_1,\Nr_2)}_{\mu, \mu_1, \mu_2}(\tau, z,\bfz)$ is a holomorphic Jacobi-like form
with the same modular properties as $\whchgirf{\Nr_1,\Nr_2}_{\mu, \mu_1, \mu_2}$.
It represents an inherent ambiguity of solution of the anomaly equation and will be fixed later
by requiring the existence of an unrefined limit.
The convergence of the theta series is ensured by Theorem \ref{th-conv} and the fact that 
$\vbbm_{12}\ast\wbbm_{12}>0$, and using \eqref{mult-genth}
it is straightforward to check that the weight, index and multiplier system agree with \eqref{multsys-chgirf}.

\subsection{Holomorphic modular ambiguity}
\label{subsubsec-modamb2}

In contrast to the original anomalous coefficients \eqref{compl-gi}, not every solution for $\chgirf{\bfr}_{\mu,\bfmu}$
suits our purposes. The restriction to be imposed is that it must have a well-defined unrefined limit.
More precisely, $\chgirf{\Nr_1,\Nr_2}_{\mu, \mu_1, \mu_2}$ must be regular at $z_i=0$ 
and have a first order zero at $z=0$. It is this condition that should be used to fix the 
holomorphic modular ambiguity $\chphi^{(\Nr_1,\Nr_2)}_{\mu, \mu_1, \mu_2}$.
As we will see below, the second term in \eqref{sol-cpmlg2th} is finite at small $z_i$,
but has a pole at small $z$, so that $\chphi^{(\Nr_1,\Nr_2)}_{\mu, \mu_1, \mu_2}$ has to be non-trivial.
To extract the pole explicitly and then choose the holomorphic ambiguity, we proceed in several steps.

First, we perform a factorization of the lattice which induces a factorization of the theta series. This then allows us to easily isolate the contributions giving rise to the pole in $z$ from the contributions having the correct behavior at $z=0$. We follow up by choosing a precise form for $\chphi^{(\Nr_1,\Nr_2)}_{\mu, \mu_1, \mu_2}$, using the Jacobi-like form \eqref{E2-Jacobilike}. Finally, we show how the unrefined limit can be obtained.

\paragraph{Factorization and split}\

\noindent
In this chapter, we assume for simplicity that $\gcd(\Nr_1,\Nr_2)=1$. Then, let us consider the sum of the two orthogonal sublattices of $\bbLami{\Nr_1,\Nr_2}$
\be
\bbLami{\Nr_1,\Nr_2}_{||}\oplus
\bbLami{\Nr_1,\Nr_2}_{\perp}
\label{parr-plus-perp}
\ee 
where the first is the span of the vectors $\vbbm_{12},\wbbm_{12}$ and the second is its orthogonal complement. In fact, the sum of these sublattices has the same dimension as $\bbLami{\Nr_1,\Nr_2}$ but is not equal to it. If this seems surprising, one can gain intuition by looking at a simple example. The lattice $\IZ^2=(k_1,k_2)$ with the canonical bilinear form, is not equal to the sum of the orthogonal sublattices $(1,1)\,\IZ $ and $(1,-1)\,\IZ$. Indeed, the element $(1,0)$ is not part of the sum of sublattices but is part of the original one. This element is then called a {\it glue vector} and it allows to write a modified sum that gives exactly $\IZ^2$. 

For our purposes let $\bbLami{\Nr_1,\Nr_2}_{||}$ and $\bbLami{\Nr_1,\Nr_2}_{\perp}$ be the lattices we defined above and let $\glueg_\Asf$ the set of their glue vectors, labeled by $\Asf=0,\dots,N_g-1$. Then we have the decomposition 
\be
\bbLami{\Nr_1,\Nr_2}=\bigcup\limits_{\Asf=0}^{N_g-1}
\(\bbLami{\Nr_1,\Nr_2}_{||}+\glueg^{||}_{\Asf}\)\oplus
\(\bbLami{\Nr_1,\Nr_2}_{\perp}+\glueg^{\perp}_{\Asf}\),
\label{latDecomp2}
\ee
where $\glueg^{||}_{\Asf}$ and $\glueg^{\perp}_{\Asf}$ are the projections of $\glueg_\Asf$ on $\bbLami{\Nr_1,\Nr_2}_{||}$ and $\bbLami{\Nr_1,\Nr_2}_{\perp}$ respectively\footnote{$\glueg^{||}_{\Asf}$ and $\glueg^{\perp}_{\Asf}$ belong to the respective dual of each sublattice and not to the sublattice itself.}. Studying the glue vectors, especially later when computing the pole, poses some challenges and we invite the reader interested in them to find more details in the paper.

On another front, the kernel \eqref{kerPhi1} depends only on the orthogonal projection of $\xbbm$ on the sublattice $\bbLami{\Nr_1,\Nr_2}_{||}$. Hence, the decomposition \eqref{latDecomp2} induces a factorization of the theta series \eqref{sol-cpmlg2th} 
\be
\chgirf{\Nr_1,\Nr_2}_{\mu, \mu_1, \mu_2}=\chphi^{(\Nr_1,\Nr_2)}_{\mu, \mu_1, \mu_2}
+\frac14\, \delta^{(\kappa\Nr_0)}_{\Delta\mu}
\sum_{\Asf}\(\sum_{\sigma=\pm 1}\sigma \,
\vthls{\kappa_{12}}_{\mu_{12},\tnu(\Asf)}(\tau,\sigma \Nr_0 z) \)\vthps{\bfr}_\Asf(\tau,\bfz),
\label{extg12-fact}
\ee
where $\tnu(\Asf)$ depends only on the glue vectors index. A similar decomposition happens for any number of charges $\bfr$. The first theta series factor in \eqref{extg12-fact} is given for $\bfr=(\Nr_1,\Nr_2)$ by
\be 
\vthls{\kappa}_{\nu,\tnu}(\tau,z)=
\sum_{\ell\in \IZ + \frac{\nu}{2\kappa}} 
\sum_{\tell\in \IZ+ \frac{\tnu}{2\kappa}} 
\Bigl(\sgn(\ell ) - \sgn (\ell-  \tell+\beta)\Bigr)
\, \q^{\kappa (\tell^2-\ell^2)} y^{2\kappa \ell},
\label{defvth2}
\ee 
while the second is
\be
\vthps{\bfr}_\Asf(\tau,\bfz)=\vth^{(d_\Nr)}_{\nu_0(\Asf)}(\tau) 
\prod_{i=1}^{2} \vthA{d_{\Nr_i}}_{\asf_i}(\tau,z_i;\frt^{(\Nr_i)}), 
\label{factor-perptheta2}
\ee
where
\bea
\label{3thetaZ2}
\vth^{(d)}_{\nu_0}(\tau) &=& \vths{d,1}_{\nu_0}(\tau,0)=
\sum_{\ell_0\in \IZ +\frac{\nu_0}{d}+\hf} 
(-1)^{d\ell_0} \, \q^{\frac{d}{2}\,\ell_0^2} ,
\\
\vthA{N}_\asf(\tau,z;\frt) &=&
\(\prod_{\alpha=1}^{N-1}\sum_{\ell_\alpha\in \IZ +\frac{\alpha\asf}{N}} \)
\q^{\sum\limits_{\alpha=1}^{N-1}\( \ell_{\alpha}^2-\ell_\alpha\ell_{\alpha+1} \)}
y^{\sum\limits_{\alpha=1}^{N-1 } \(\frt_{\alpha+1}-\frt_\alpha \)\ell_{\alpha}} ,
\label{3thetas2}
\eea
$\vths{d,p}_{\mu}(\tau,z)$ is the theta series \eqref{Vignerasth},
and in the last equation we used the convention $\ell_N=0$. Furthermore, the theta series \eqref{3thetas2} can be recognized as the theta series corresponding to the $A_{N-1}$ lattice.

The theta series \eqref{defvth2} is the main object that we will consider. Despite the orthogonal factor \eqref{factor-perptheta2} being equally important for the construction, it is quite easy to deal with. Before going further, a few remarks are in order about the representation \eqref{extg12-fact}. 
First, it was crucial that the kernel \eqref{kerPhi1} depend only on the projection of $\xbbm$ on $\bbLami{\Nr_1,\Nr_2}_{||}$ thus allowing the factorization to go through. One can recognize that the difference of signs in \eqref{defvth2} exactly reproduces the kernel. Moreover, $\vthls{\kappa}_{\nu,\tnu}(\tau,z)$ does not depend on the vector of elliptic parameters $\bfz$ and this ensures a crucial simplification to our problem. In fact if it did, the poles in $z$ that we will see later would also depend on $z_i$ and thus they would have been much more difficult to cancel. This crucial property is due to the fact that the vectors $\frt^{(\Nr_i)}$, associated to each elliptic parameter $z_i$, verify $\sum_{\alpha}\frt^{(\Nr_i)}_{\alpha}=0$ which was the second condition we required in our strategy for the lattice extension. At the same time, \eqref{factor-perptheta2} does not depend on $z$ and further decomposes into the $n+1$ independent theta series \eqref{3thetaZ2} and \eqref{3thetas2}. The fact that the former has no elliptic parameters and the other $n$ depend each on one single $z_i$ ensures that the action of the modular derivatives $\cD_m^{(n)}$ factorizes, which is an additional simplification.

As we said already, the solution develops a pole at $z=0$. In order to study it, we split the theta series \eqref{defvth2} into two parts,
$\vthls{\kappa}_{\nu,\tnu}=
\cvths{\kappa}_{\nu,\tnu}+\tvths{\kappa}_{\nu,\tnu}$, 
where in the first term one sums only over $(\ell,\tell)$ 
satisfying the condition $\ell= \tell$, which can also be written in 
geometric terms as
\be  
\kbbm_{||}\ast \wbbm_{12}=0,
\label{zeromode1}
\ee 
while in the second the sum goes over the rest of the lattice.
Then in $\tvths{\kappa}_{\nu,\tnu}$, for sufficiently small $z$ 
one can drop the shift by $\beta$ in the second sign function
and one obtains
\be  
\sum_{\sigma=\pm 1}\sigma \,\tvthls{\kappa}_{\nu,\tnu}(\tau, \sigma\Nr_0 z)=
\!\!\!\!\sum_{{\ell\in \IZ + \frac{\nu}{2\kappa}\atop 
		\tell\in \IZ+ \frac{\tnu}{2\kappa}}\; : \; \ell\ne  2^\eps \tell }  \!\!\!\!
\Bigl(\sgn(\ell ) - \sgn (\ell- \tell)\Bigr)
\, \q^{\kappa (\tell^2-\ell^2)}\Bigl(y^{2\Nr_0\kappa\ell  }-y^{-2\Nr_0\kappa\ell}\Bigr).
\label{nonzm-yy}
\ee 
This theta series is not only convergent for all $z$, but also 
vanishes at $z=0$. Thus, it has a well-defined unrefined limit and it remains to analyze only 
the function $\cvths{\kappa}_{\nu,\tnu}$ which we call ``zero mode contribution".

\paragraph{Pole evaluation}\

\noindent
The zero mode contribution is characterized by the condition $\ell= \tell$.
Importantly, it also restricts the set of glue vectors by requiring 
$\mu_{12}-\tnu(\Asf) \in 2\kappa_{12}\IZ$. 
We denote the set of the glue vectors satisfying this condition 
by $\cAr_0(\mu_{12})$ without giving them explicitly (cf \cite[Appendix. A.3]{Alexandrov:2024wla}).

Implementing the zero mode condition in \eqref{defvth2}, one finds for $\Asf\in\cAr_0(\nu)$ and $\tnu\equiv\tnu(\Asf)$,
\be
\sum_{\sigma=\pm 1}\sigma \,\cvths{\kappa_{12}}_{\nu,\tnu}(\tau,\sigma\Nr_0 z) 
=-\sum_{\sigma=\pm 1}\sum_{\tell\in \IZ+ \frac{\tnu}{2\kappa_{12}}} 
\Bigl(\sgn (\beta)-\sigma\sgn(\tell) \Bigr) 
\, y^{2 \sigma \Nr_0 \kappa_{12} \tell },
\label{geompr}
\ee
where the power of $\q$ became trivial.
This is just a simple geometric progression.
Assuming that $\beta>0$, so that $\Im z<0$ and $|y|>1$, it evaluates to
\be
-\sum_{\sigma=\pm 1}\(2\,\frac{y^{2\sigma\Nr_0\kappa_{12}\(\frac{\tnu}{2\kappa_{12}}
		-\left\lceil\frac{\tnu}{2\kappa_{12}}\right\rceil+\frac12\,(1-\sigma)\)}}
{1-y^{-2\Nr_0 \kappa_{12}}} -\sigma\delta_{\tnu}^{(2\kappa_{12})}\)
=-2\,\frac{y^{2\Nr_0 \kappa_{12}\lambda_{12} }+y^{-2\Nr_0 \kappa_{12}\lambda_{12} }}
{y^{\Nr_0 \kappa_{12}}-y^{-\Nr_0 \kappa_{12}}}\, ,
\label{zmod-sigmap1}
\ee
where we defined
\be 
\lambda_{12} = \left\lceil \frac{\tnu}{2\kappa_{12}} \right\rceil-\frac{\tnu}{2\kappa_{12}}-\hf\, ,
\label{def-lam0}
\ee
which depends on the glue vector index $\Asf$.
The same result holds for $\beta<0$ as well. The leading pole in the expansion of \eqref{geompr} is then given by
\be
\sum_{\sigma=\pm 1}\sigma \,\cvths{\kappa_{12}}_{\nu,\tnu}(\tau,\sigma\Nr_0 z)  = 
-\frac{1}{\pi \I \kappa_{12}\Nr_0 z} + O(z).
\ee

In fact, this pole is the most important effect of the elliptic parameter $z$. Namely, $z$ acts as a regularization parameter and without it the same direction in this sum would have been divergent. Hence, the refinement is not simply a trick to simplify computations but turns out to be crucial for the whole construction. Additionally, studying the pole is the main reason that we performed the factorization and split of the lattice as in \eqref{latDecomp2}.

The contribution of the zero mode to \eqref{extg12-fact} gives terms of order $O(1/z)$ and $O(z)$ while the contribution of $\tvths{\kappa}_{\nu,\tnu}$ contributes only at $O(z)$. This means that we need the function $\chphi^{(\Nr_1,\Nr_2)}_{\mu, \mu_1, \mu_2}$ to cancel the pole. It is also constrained by the modular properties of the solution. Nonetheless it is easy to find a Jacobi-like form with the correct modular properties and that cancels the pole and one has a lot of freedom in doing so. However, the choice we will make is very natural as it easily generalizes to higher $n$. 

First, we look for $\chphi^{(\Nr_1,\Nr_2)}_{\mu, \mu_1, \mu_2}$ in the form 
\be
\chphi^{(\Nr_1,\Nr_2)}_{\mu, \mu_1, \mu_2}=
\frac12\, \delta^{(\kappa\Nr_0)}_{\Delta\mu}\, \phi^{(\kappa_{12})}(\tau,\Nr_0 z)\sum_{\Asf \in \cAr_0}  
\vthps{\bfr}_\Asf(\tau,\bfz),
\label{def-phik}
\ee
where $\phi^{(\kappa_{12})}(\tau,z)$ is a scalar valued Jacobi-like form whose
modular properties can be obtained from \eqref{multsys-chgirf} and \eqref{def-phik}.
Using that $\frac{\kappa}{6} \(\Nr^3 - \Nr_1^3 - \Nr_2^3\)=\Nr_0^2 \kappa_{12}$,
one finds that it should have weight 1, index $-\kappa_{12}$
and a trivial multiplier system. The last fact follows from the
observation that the leading coefficient in the small $z$ expansion of a VV Jacobi-like form has
the same multiplier system as the form itself.

We make the choice
\be
\phi^{(\kappa_{12})}(\tau,z)
= \frac{1}{2\pi \I \kappa_{12}}\,\frac{e^{\frac{\pi^2}{3}\, \kappa_{12} E_2(\tau)z^2}}{z},
\label{choice-phi2}
\ee
using \eqref{E2-Jacobilike}. The first factor $1/z$ gives the correct weight, the index is ensured by the factors in the exponent and the constant prefactor ensures the cancellation of the pole. 

In the $n=2$ case it is enough to cancel the pole, since all our functions are even in $z$ and thus the next contribution is of order $z$ which is enough to get the unrefined limit. But in general, we need to make sure that all orders from $\frac{1}{z^{n-1}}$ up to $z^{n-1}$ are canceled. So as far as the $n=2$ case is concerned, we have already found the solution and we have shown how the unrefined limit can be obtained easily, but without giving explicitly each step.

\subsubsection{The unrefined limit}

Let us now reduce the solution \eqref{extg12-fact} on the extended lattice to 
the anomalous coefficient $\gi{\Nr_1,\Nr_2}_{\mu, \mu_1, \mu_2}(\tau)$ we are really interested in. 
At the first step we obtain the refined anomalous coefficient $\girf{\Nr_1,\Nr_2}_{\mu, \mu_1, \mu_2}(\tau,z)$ 
using the relation \eqref{recover-gref}.
As was already mentioned, the absence of $z_i$-dependence in $\vth^{||}$ and 
the factorized form \eqref{factor-perptheta2} of $\vth^{\perp}$ makes the application of \eqref{recover-gref}
almost trivial: one should simply apply each of the differential operators to the corresponding 
$A_{N-1}$ lattice theta series $\vthA{d_{\Nr_i}}_{\asf_i}(\tau,z_i)$.
This gives
\be
\begin{split}
	\girf{\Nr_1,\Nr_2}_{\mu, \mu_1, \mu_2}(\tau,z)= &\, \hf\,
	\delta^{(\kappa\Nr_0)}_{\Delta\mu}\sum_{\Asf=\{\asf_0,\asf_1,\asf_2\}}\Biggl[
	\frac12 \sum_{\sigma=\pm 1}\sigma \,\vthls{\kappa_{12}}_{\mu_{12},\tnu(\Asf)}(\tau,\sigma \Nr_0 z)
	\\
	&\, \qquad
	+ \delta_{\Asf \in \cAr_0}\,\phi^{(\kappa_{12})}(\tau,\Nr_0 z) 
	\Biggr]
	\vth^{(d_\Nr)}_{\nu_0(\Asf)}(\tau) 
	\prod_{i=1}^{2} \cD\vthA{d_{\Nr_i}}_{\asf_i}(\tau;\frt^{(\Nr_i)}),
\end{split}
\label{res-gref}
\ee
where 
\be
\cD\vthA{N}_{\asf}(\tau;\frt) =
\frac{\cD^{(N)}_{\frt^2/2}\vthA{N}_{\asf}(\tau,z;\frt)\bigr|_{z_i=0}}
{N!\(\prod_{\alpha=1}^{N}\frt_\alpha\)\(-2\pi\eta^3(\tau)\)^{N}}\,.
\label{def-cDvth}
\ee

Finally, we take the unrefined limit $z\to 0$ according to \eqref{lim-ancoef}. 
To this end, we split $\vth^{||}$ into the zero mode, whose pole cancels against \eqref{choice-phi2}, and the non-zero mode parts and then take the term of order $z$. 
One can find the explicit expression for the solution of the unrefined completion equation \eqref{compl-gi} for $n=2$ and $n=3$ in \cite{Alexandrov:2024wla}. In appendix H of that paper one can also find cross checks against solutions obtained from the independent methods of section \ref{sec-DTall-HeckeVW}.

\section{Solving the general case}
\label{sec-DTall-soln}

\subsection{Solution}

One of the main results of our work is a theorem that gives a family of solutions $\chgirf{\bfr}_{\mu,\bfmu}$. This theorem works for any choice of holomorphic ambiguity, but since we want to take the unrefined limit we follow up with a conjectural choice for a function $\chphi^{(\bfr)}_{\mu,\bfmu}$ that guarantees that the limit exists. Finally, we perform the step of taking the modular derivatives and reducing to a solution $\girf{\bfr}_{\mu,\bfmu}$ but we keep the unrefined limit non-evaluated.

\paragraph{Generic solution}\
\label{subsubsec-gengen}

\noindent
We start by presenting a solution to the anomaly equation 
\eqref{complchgirf}. Of course, for any set of charges this solution involves a holomorphic modular ambiguity
parametrized by a Jacobi-like form $\phi^{(\bfr)}_{\mu, \bfmu }$. Moreover, holomorphic modular ambiguities for subsets $\bfs\in\bfr$ enter the completion of $\chgirf{\bfr}_{\mu,\bfmu}$.

In fact, the holomorphic modular ambiguities we choose can be written, formally, as theta series over the full extended lattice
$\bbLami{\bfr}$
\be 
\chphi^{(\bfr)}_{\mu, \bfmu}=\frac{\delta^{(\kappa\Nr_0)}_{\Delta\mu}}{2^{n-1}}\,
\vth_{\bbmu}(\tau, \zbbm;\bbLami{\bfr}, \Phidi{\bfr}, \pbbm),
\label{phin-theta}
\ee 
where $\delta_{\Delta\mu}^{(\kappa\Nr_0)}$ ensures that $(\mu,\bfmu)$ gives a residue class of the lattice $\bbLambda$, $\bbmu$, $\zbbm$ and $\pbbm$ are as in \eqref{thetadata}, and
the expression for the kernel $\Phidi{\bfr}$ as well as a discussion about it can be found in \cite{Alexandrov:2024wla}.

Then, let us define the composite kernel
\be
\rPhi(\xbbm;\{\Fvi{\bfs}\})=
\sum_{m=2}^{n} \sum_{\sum_{k=1}^m n_k = n} 
\Fvi{\bfs}(\xbbm^{(0)})	\prod_{k=1}^{m} \Phidi{\frr_k}(\xbbm^{(k)},\tau,z),
\label{ker-thm}
\ee
where $\bfs$ are $\frr_k$ are the notations from \eqref{split-rs},
the upper indices $^{(0)}$ and $^{(k)}$ denote projections to $\bbLami{\bfs}$ and $\bbLami{\frr_k}$, respectively,
and for a single charge we set $\Phidi{\Nr}=1$.
Finally, we have the following

\begin{theorem}\label{thm-gensol}
	A solution of the anomaly equation \eqref{complchgirf} and its modular completion 
	can be expressed as
	\be
	\begin{split}
		\chgirf{\bfr}_{\mu, \bfmu }  = &\, \chphi^{(\bfr)}_{\mu,\bfmu}
		+\frac{\delta^{(\kappa\Nr_0)}_{\Delta\mu}}{2^{n-1}}\,
		\Sym\Bigl\{\vth_{\bbmu}(\tau, \zbbm;\bbLami{\bfr}, \rPhi(\{\Fvi{\bfs}\}), \pbbm)\Bigr\} ,
		\\ 
		\whchgirf{\bfr}_{\mu, \bfmu}  = &\,\chphi^{(\bfr)}_{\mu,\bfmu}+
		\frac{\delta^{(\kappa\Nr_0)}_{\Delta\mu}}{2^{n-1}}\,
		\Sym\Bigl\{\vth_{\bbmu}(\tau, \zbbm;\bbLami{\bfr}, \rPhi(\{\whFvi{\bfs}\}), \pbbm) \Bigr\},
	\end{split}
	\ee
	where the functions $\Fvi{\bfr}$ and $\whFvi{\bfr}$ are given by 
	\be
	\begin{split} 
		\Fvi{\bfr}(\xbbm)=&\,
		\sum_{\cJ\subseteq\Zv_{n-1}} e_{|\cJ|} \delta_\cJ
		\prod_{\ell \in \Zv_{n-1} \backslash \cJ}
		\Bigl(\sign (\xbbm_\bbbeta\ast\vbbm_\ell)-\sign ( \xbbm\ast\wbbm_{\ell,\ell+1}) \Bigr),
		\\
		\whFvi{\bfr}(\xbbm)=&\,
		\sum_{\cJ \subseteq \Zv_{n-1}}\Phi_{|\cJ|}^E \(\{\vbbm_l\}_{l \in \cJ}; \xbbm \) 
		\prod_{\ell \in \Zv_{n-1}\backslash\cJ}\bigl(-\sign(\xbbm\ast\wbbm_{\ell,\ell+1}) \bigr).
	\end{split} 
	\label{kern-manyz}
	\ee
	Here $\xbbm_\bbbeta=\xbbm-\sqrt{2\tau_2}\,\bbbeta$,
	$\Zv_{n}=\{1,\dots,n\}$, 
	\be
	\label{def-genthm}
	e_m=\left\{\begin{array}{ll}
		0 & \mbox{\rm if $m$ is odd},
		\\
		\frac{1}{m+1}\ & \mbox{\rm if $m$ is even},
	\end{array}\right.
	\qquad
	\delta_\cJ=\prod_{\ell\in\cJ}\delta_{\xbbm_\bbbeta\ast\vbbm_\ell},
	\quad\mbox{\rm and}\quad
	\vbbm_\ell=\sum_{i=1}^{\ell}\sum_{j=\ell+1}^{n} \vbbm_{ij} .
	\ee 
\end{theorem}

Although the functions \eqref{kern-manyz} might seem complicated, their structure is easy to understand.
First, if all scalar products $\xbbm_\bbbeta\ast\vbbm_\ell$ are non-vanishing, then the function $\Fvi{\bfr}(\xbbm)$
simplifies to 
\be  
\Fvi{\bfr}(\xbbm) =\prod_{\ell =1}^{n-1}
\Bigl(\sign (\xbbm_\bbbeta\ast\vbbm_\ell)-\sign ( \xbbm\ast\wbbm_{\ell,\ell+1}) \Bigr),
\ee 
which is the standard kernel ensuring convergence of indefinite theta series 
with quadratic form having $n-1$ positive eigenvalues
(see Theorem \ref{th-conv}).
If, however, some of the scalar products vanish, it is not sufficient 
to set the corresponding sign functions to zero. Instead, one gets additional contributions
similar \eqref{Efn0new}.
In the presence of refinement, only the linear tree is relevant (see \eqref{Efref})
and one can apply a simple recipe that $\sgn(0)^m\to e_m$ \cite{Alexandrov:2019rth}. 
This gives rise to the expression in \eqref{kern-manyz}.

The proof of Theorem \ref{thm-gensol}, including convergence of $\chgirf{\bfr}$ and modularity of $\whchgirf{\bfr}$, is discussed in \cite{Alexandrov:2024wla} by analogy to the proof of Theorem 1 in \cite{Alexandrov:2020bwg}.

\paragraph{Factorization and split}\

\noindent
As in the $n=2$ case, we decompose the lattice as 
\be
\bbLami{\bfr}_{||}\oplus
\bbLami{\bfr}_{\perp}
\label{parr-plus-perp-n}
\ee 
where the first is the span of the (normalized versions of the) vectors $\vbbm_{ij},\wbbm_{ij}$ and the second is its orthogonal complement.
This induces a factorization of the solution similar to \eqref{extg12-fact}
\be
\chgirf{\bfr}_{\mu,\bfmu} = \chphi^{(\bfr)}_{\mu, \bfmu} + 
\frac{\delta^{(\kappa\Nr_0)}_{\Delta\mu}}{2^{n-1}}\,\sum_{\Asf}
\Sym\bigl\{\vthls{\bfr}_{\bbmu,\Asf}(\tau, z) \bigr\}
\vthps{\bfr}_\Asf(\tau,\bfz),
\label{extgn-fact}
\ee
with a sum over glue vectors indexed by $\Asf$ and
\be
\vthps{\bfr}_\Asf(\tau,\bfz)=\vth^{(d_\Nr)}_{\nu_0(\Asf)}(\tau) 
\prod_{i=1}^{n} \vthA{d_{\Nr_i}}_{\asf_i}(\tau,z_i;\frt^{(\Nr_i)}).
\label{factor-perpthetan}
\ee
The theta series $\vthls{\bfr}_{\bbmu,\Asf}$ is quite complicated and we choose not to give its full expression here. We will, however, in the next paragraph describe some of its properties.

\paragraph{Holomorphic modular ambiguity}\

\noindent
When looking for a choice of the holomorphic modular ambiguity that ensures the existence of the unrefined limit, we do not use the formal form \eqref{phin-theta} and instead look for it in the factorized form
\be
\chphi^{(\bfr)}_{\mu, \bfmu}=
\frac{\delta^{(\kappa\Nr_0)}_{\Delta\mu}}{2^{n-1}}
\sum_{\Asf} \phi^{(\bfr)}_{\glueg_\Asf^{||}+\bbmu}(\tau,z) \,
\vthps{\bfr}_\Asf(\tau,\bfz), 
\label{def-phin}
\ee
where $\phi^{(\bfr)}_{\bbnu}(\tau,z)$ is a VV Jacobi-like form labeled 
by $\bbnu\in(\bbLami{\bfr}_{||})^*/\bbLami{\bfr}_{||}$,
and characterized by weight $n-1$, index $m_{\bfr}$ \eqref{index-mr},
and the multiplier system given by the Weil representation \eqref{mult-genth} 
associated with the lattice $\bbLami{\bfr}_{||}$.

In order to fix the VV Jacobi-like form $\phi^{(\bfr)}_{\bbnu}(\tau,z)$, as in \S\ref{sec-DTall-sol2}, we split the theta series $\vthls{\bfr}_{\bbmu,\Asf}$ 
into contributions with different zero mode order: from non-zero modes to maximal zero modes, of order $n-1$, and compute the poles they produce. The order of a zero mode is given by the number of
linearly independent vectors $\wbbm_{ij}$ having a vanishing scalar product $\kbbm\ast \wbbm_{ij}$.

First, we make a conjecture saying that for each set of charges $\bfr$, we only need to worry about the contribution from maximal zero modes in ensuring the existence of the full unrefined limit. 
\begin{conj}\label{conj-zm}
	Let us fix an integer $n_0$, and assume that for all sets of charges $\bfr$ with the number of charges $n<n_0$, 
	the functions $\phi^{(\bfr)}_{\bbnu}(\tau,z)$ are Jacobi-like forms that ensure the existence of the unrefined limit 
	for all $\chgirf{\bfr}_{\mu,\bfmu}$ so that they behave as $O(z^{n-1})$ at small $z$.	
	Then for $n=n_0$, the contributions to $\Sym\bigl\{\vthls{\bfr}_{\bbmu,\Asf}(\tau, z) \bigr\}$ 
	of any zero mode order different from the maximal one, given by $n-1$, behave as $O(z^{n-1})$.
\end{conj}

Since now we can focus on maximal zero modes only, we study the coefficient of the leading pole they produce and find it numerically up to $n=5$. Building on this, we conjecture the exact coefficient of this pole for all $n$. Then, we have 
\begin{conj}
\label{conj-leadingpole}
The choice
\be  
\phi^{(\bfr)}_\bbnu(\tau,z)=\frac{\Sym\{ \cbfr \}}{ z^{n-1}}\,e^{-\frac{\pi^2}{3}\,m_{\bfr} E_2(\tau)z^2}
\delta_{\bbnu\in\cA_0^{\bfr}},
\label{solfullphi}
\ee
where $\cA_0^{\bfr}$ is the set of residue classes $\bbnu$ that allow to give a solution\footnote{For $n=2$ this was given by the condition $\glueg_\Asf\in\cA_0(\mu_{12})$ where $\bbnu=(\mu_{12};0^{[\di{1}]};0^{[\di{2}]})+\glueg_\Asf^{(||)}$.} to the maximal zero mode condition, and normalization constant given by
\be  
\cbfr
=\frac{\Nr_0}{(2^{\eps}\pi\I \kappa)^{n-1}\Nr}
\prod\limits_{k=1}^{n-1}\(\sum\limits_{i=1}^k\Nr_i\sum\limits_{j=n-k+1}^n\Nr_j\)^{-1}\,,
\label{constphi}
\ee 
is conjectured to cancel the leading pole.
\end{conj}
With this proposal we have
\begin{theorem} \label{thm-ambig}
	Provided Conjectures \ref{conj-zm} and \ref{conj-leadingpole} hold, the holomorphic modular ambiguity given by \eqref{def-phin} and \eqref{solfullphi}
	ensures the existence of the unrefined limit.
\end{theorem}

We will not give the proof of this theorem, but we simply say that it is based on prop. \ref{prop-JacobiE2}, on the fact that (the factorized part of) maximal zero modes do not have $\tau$ dependence and on the fact that the choice \eqref{solfullphi} only depends on $\tau$ through the quasimodular form $E_2(\tau)$. In particular, this justifies our use of Jacobi-like forms.

\paragraph{The unrefined limit}\

\noindent
Theorem \ref{thm-gensol} and the choice \eqref{solfullphi} provide a solution for the functions $\chgirf{\bfr}_{\mu,\bfmu}$ 
satisfying the anomaly equation \eqref{complchgirf} and having a well-defined unrefined limit. 
It remains just to reduce it to the original anomalous coefficients $\gi{\bfr}_{\mu,\bfmu}$. 
The first step, the reduction to the refined anomalous coefficients $\girf{\bfr}_{\mu,\bfmu}$, is trivial
and done by applying the relation \eqref{recover-gref} to the factorized expression. 
This affects only the $A_{N-1}$ lattice theta series defined in \eqref{3thetas2} and results in
\be
\girf{\bfr}_{\mu,\bfmu}(\tau,z) =  \frac{\delta^{(\kappa\Nr_0)}_{\Delta\mu}}{2^{n-1}} \sum_{\Asf}
\(\phi^{(\bfr)}_{\glueg_\Asf^{||}+\bbmu}(\tau,z)+\Sym\bigl\{\vthls{\bfr}_{\bbmu,\Asf}(\tau, z) \bigr\}\)
\vth^{(d_\Nr)}_{\nu_0(\Asf)}(\tau) 
\prod_{i=1}^{n} \cD\vthA{d_{\Nr_i}}_{\asf_i}(\tau;\frt^{(\Nr_i)}),
\label{refgr}
\ee
where $\cD\vthA{d_\Nr}_\asf$ is defined in \eqref{def-cDvth}.

The last step is to evaluate the unrefined limit $z\to 0$. Unfortunately, we cannot accomplish it analytically in full generality
because this would require rewriting $\Sym\bigl\{\vthls{\bfr}_{\bbmu,\Asf}(\tau, z) \bigr\}$ in a form which makes manifest 
the existence of zero of order $n-1$ at small $z$ for all contributions except the zero modes of maximal order
and, in particular, would automatically provide a proof of Conjecture \ref{conj-zm}.
However, since evaluating a limit of a function should certainly be simpler than solving non-trivial anomaly equations,
we see this problem as just a technical obstacle and hope to return to it elsewhere.

\chapter{Conclusions}
\label{chap-concl}

\section{Summary}
The work presented in this thesis has covered a rich array of subjects at the crossroads of physics and mathematics. 
Our investigation proceeded along three main themes. First, we computed the one-instanton NS5-brane corrections to the hypermultiplet moduli space metric. Second, we revisited and solved a non-commutative quantum Riemann-Hilbert problem induced by refined BPS indices. Third, we studied and provided (recipe for) solutions to the anomalous part of the generating functions of D4-D2-D0 BPS indices. These topics are all related to the hypermultiplet moduli space in Type II string theory compactified on a Calabi-Yau. In this concluding chapter, we shall review the main results obtained, discuss their potential shortcomings, and outline promising directions for future research.

String dualities are among the most powerful tools for accessing the non-perturbative sector of the theory. This is why in \cite{Alexandrov:2010ca} S-duality was used, on the twistor space, to lay the ground for describing the NS5-instanton corrected hypermultiplet moduli space.
Our first result, in chapter \ref{chap-NS5HM}, was the computation of the NS5-brane one-instanton corrected metric, following the procedure outlined in chapter \ref{chap-twist}. 
This was achieved using the NS5-instanton twistor data obtained in \cite{Alexandrov:2010ca}.
This approach, while successful, has some limitations. The construction is consistent only at the one-instanton level. This is because the BPS rays associated with different instantons intersect, leading to inconsistencies in a multi-instanton expansion performed in this frame. 
Furthermore, symplectic invariance, which is a central property of Type IIA string theory, is not manifest in our final expressions. This arises because our method relies on Type IIB twistor data despite being ultimately applied to a Type IIA problem. 
It is plausible that a Poisson resummation, perhaps like in \cite{Alexandrov:2009qq}, is required to recast the result into a manifestly symplectic invariant formulation. 

To validate our computations, we performed two independent checks. First, we considered the specialization to a rigid Calabi-Yau manifold, where the moduli space reduces to the universal hypermultiplet. In this case, the QK metric can be described by a single holomorphic potential satisfying the Przanowski equation, a highly non-linear differential equation. Using a direct map from the twistor description to the Przanowski formulation established in \cite{Alexandrov:2009vj}, we confirmed perfect agreement between our result and the expected form. In a related check, we compared our findings to the families of physically relevant solutions to the linearized Przanowski equation found in \cite{Alexandrov:2006hx}. 
While one family of solutions therein perfectly reproduced the D-instanton deformation, the NS5-instanton case did not exhibit such a direct match. Reconciling these two descriptions would likely require finding a non-trivial integral transform. Discovering this transform is a key task for future work, as it could provide crucial inspiration for a simpler expression for the NS5-instanton corrections.

Our second check involved taking the weak string coupling limit, 
$g_s\ll 1$. After properly defining this limit by establishing the scaling of physical fields with respect to the supergravity action, we found that our metric possesses the precise structure predicted by direct string amplitude computations. Since the full expression in this limit was quite complicated, we took a further limit of small background Ramond-Ramond fields to successfully reproduce the known Gaussian NS5-instanton action. The one-forms that constitute the metric in this double-limit represent a result of this thesis, providing a direct prediction for sphere three-point functions in the presence of an NS5-brane and small background RR-fluxes.

The second part of this dissertation, presented in chapter \ref{chap-qRH}, focused on a non-commutative Riemann-Hilbert (RH) problem. 
By considering refined BPS indices, one is led to a quantum deformation of this classical problem \cite{Barbieri:2019yya}.

We formulated this quantum Riemann-Hilbert (qRH) problem in terms of complex functions endowed with a non-commutative Moyal star product.
Although this formalism can be applied in various contexts, we chose to focus on the one given by the conformal limit of $N=2$ gauge theories \cite{Gaiotto:2014bza,Bridgeland:2016nqw}. Building upon a proposal in \cite{Alexandrov:2019rth}, which introduced an integral equation whose solutions define refined analogs of the Joyce and \pleb potentials but do not solve the quantum RH problem itself, we introduced new variables, $\hcXr_\gamma$. We have shown that these new variables provide a formal solution to the qRH problem and correctly reduce to the solutions of the classical problem in the unrefined limit.

Inspired by the classical case, where the RH problem can be exchanged for a Thermodynamic Bethe Ansatz (TBA)-like integral equation, we looked for a similar structure for the qRH problem. We proposed such an integral equation, but found that it contains an infinite number of terms. Another interesting result emerged when we expressed the quantum RH problem in its adjoint form. In this formulation, the solution $\hcXr_\gamma$ is governed by a generating function $\psi$ that is, remarkably, independent of the charges. We proved that its logarithm has a well-defined unrefined limit and that its unrefined limit coincides with the unrefined generating function of Darboux coordinates. This provides an all-orders expression for the latter, which was previously known only up to the second order in a perturbative expansion \cite{Alexandrov:2017qhn}. As a potential future direction, this generating function can be used to prove the S-duality invariance of the D3-instanton corrected Darboux coordinates to all orders, generalizing the second-order proof in \cite{Alexandrov:2017qhn}. Finally, in section \ref{sec-qTBA-uncoupled} we computed explicitly the solution $\hcXr_\gamma$ in the case of an uncoupled BPS structure and wrote it in terms of the modified Barnes Gamma function, thus matching results from \cite{Barbieri:2019yya}.

The final part of this thesis, in chapter \ref{chap-DTall}, explored the modular properties of the generating functions of D4-D2-D0 BPS indices, which are important objects giving DT invariants and black hole entropy. These functions are known to satisfy modular anomaly equations involving indefinite theta series.

Our first contribution was a careful re-examination of these equations. We identified and recovered subtle but crucial contributions to the kernels of the theta series that were previously overlooked. Then we found a significant structural simplifications of the kernels, a result which may contribute to a deeper mathematical understanding of the origin of these modular properties.

We then undertook the task of solving these equations for arbitrary charge in Calabi-Yau spaces with 
$b_2=1$. 
The equations do not have unique solutions, but rather a family of solutions related by a holomorphic modular ambiguity. 
If we add their recursive nature, this ambiguity presents a significant challenge, as one cannot solve for a given charge without first solving and fixing the ambiguity for all lower charges. We circumvented this by parametrizing the solutions in terms of their dependence on the modular ambiguities, which we encoded in the anomalous coefficients $\gi{\bfr}_{\mu,\bfmu}$ where $\bfr=\{\Nr_1,\dots,\Nr_n\}$.
We show that these are vector-valued mock modular forms of depth $n-1$ and satisfy a similar modular completion equation to that of the generating functions of BPS indices. Furthermore, we show in theorem \ref{thm-ancoef} that solving the completion equation of $\gi{\bfr}_{\mu,\bfmu}$ is equivalent to solving that of $h_{\Nr,\mu}$ up to all lower-rank modular ambiguities.

We first provided two infinite families of solutions, one for two charges ($n=2$), constructed using some Hecke-like operators \cite{Dabholkar:2012nd}, and another for the case where all charges are equal to 1 as well as triple intersection number $\kappa=1$, which we showed is related to the generating functions of Vafa-Witten invariants on $\IP^2$.

To find a general solution, we employed a strategy where we extend the problem, solve it and then reduce back to a solution of the original problem. We extended the problem by introducing a refinement parameter and a lattice extension, constructed a general solution for this extended problem using indefinite theta series and Jacobi-like forms, and then prescribed a procedure to reduce it to a solution of the original problem. The reduction is contingent on certain conditions. For $n=2$ and $n=3$ we propose solutions and we prove that they satisfy these conditions. For an arbitrary number of charges, we propose solutions and conjecture that they also do satisfy these conditions and thus reduce to a solution of the non-extended unrefined problem. 
In \cite{Alexandrov:2024wla} a cross-check was performed by comparing our general solution for $n=2$ with the solution constructed using Hecke-like operators, for multiple pairs of charges. We showed that their difference, contracted with an appropriate theta series, is a scalar Jacobi form, providing strong evidence for the validity of both solutions and, as a byproduct, yielding new closed-form expressions for the "seed" functions used in the Hecke-based construction.

\section{Outlook}
The results summarized in the preceding paragraphs, while providing answers to several key questions, simultaneously illuminate a vast landscape for future investigation. The study of Type II string theory compactified on a Calabi-Yau manifold remains a very rich field, and it would be great if the progress made in this thesis opens new avenues to explore. We outline here several promising directions for research.

A central challenge that remains open is to find a description of NS5-instanton corrections to the hypermultiplet moduli space at all orders. The twistor space, which proved so crucial to our one-instanton computation, is the natural framework to attempt this. There are several strategies. One could do a Poisson resummation as mentioned earlier. Or investigate the existence of an alternative, manifestly S-duality invariant, twistor space construction that avoids the inconsistencies of the one we used. 
It may be impossible to find such a construction using BPS rays, as it may require to work in terms of closed contours.

Inspiration can be drawn from the Type IIB side. There, D3-instantons are incorporated using a BPS ray structure analogous to the NS5-instanton case we encountered. Understanding how the D3-instanton corrected twistor space is constructed to all orders could provide invaluable lessons for resolving the ray-crossing inconsistencies of the NS5-instanton problem. A completely different, yet potentially powerful, approach would be to leverage the principles of resurgence. The computation of large D-instanton effects on the hypermultiplet moduli space in \cite{Pioline:2009ia} strongly suggests that NS5-instanton contributions are linked by resurgence to the D-instanton ones. One could therefore use resurgence techniques to infer the all-orders structure of NS5-corrections from the D-brane sector. Even if this technique proves difficult to apply at all orders and only yields linear corrections, it could still be highly valuable, as these corrections would likely take a different form than those we derived in this work.

Finally, the Przanowski description offers another promising avenue. It may be simpler to formulate an all-orders description within this framework. In \cite[Eqns. (4.15),(4.17)]{Alexandrov:2025abc}, we have spelled out the precise form of the linear corrections to the coordinates $(h,z^1,z^2)$ within this description. These deformations (especially for $z^2$ and $h$) were relatively simple and one can attempt to upgrade the solution of the linearized equation to a solution of the full non-linear differential equation. 

Successfully finding the all order NS5-instanton corrections would be a monumental step, as it would complete our picture of the low-energy effective theory and is expected to resolve the known curvature singularity induced by the one-loop correction.

In this thesis we also considered a quantized hypermultiplet moduli space obtained after inclusion of a refinement parameter. A potential physical explanation for why the refinement quantizes the moduli space remains mysterious. Furthermore, in one of the setups we considered in \cite{Alexandrov:2025abc}, the large volume limit of D3-instanton corrected Type IIB theory, the refinement appears to be compatible with S-duality, a surprising feature that reinforces the expectation of a deep physical meaning to the refinement in this setup.
The quantum Riemann-Hilbert problem we formulated suggests that the new variables $\hcXr_\gamma$, should be interpreted as Darboux coordinates on a quantum analogue of the twistor space. A first step towards formalizing this notion was taken in \cite{Alexandrov:2023wdj} with the definition of a quantized contact structure. However, quantum twistor spaces remain mysterious and it would be very interesting to study them. 
Moreover, since our qRH problem is a deformation of the classical RH problem, it is natural to ask whether other deformations exist that effectively incorporate NS5-instanton corrections.

Finally, the modular properties of the generating functions of D4-D2-D0 BPS indices are a remarkable result, that we leveraged in order to get one step closer to systematically fixing these generating functions. The immediate next step is to adapt our solution to the modular anomaly to the case of arbitrary second Betti number, $b_2$, thus accommodating a much larger number of CYs. At the same time, for specific Calabi-Yau manifolds, one must work to fix the remaining holomorphic modular ambiguities by computing polar terms of the generating functions. This program has already shown promise \cite{Alexandrov:2022pgd, Alexandrov:2023ltz, Alexandrov:2023zjb} for some CYs and with D4-brane charge equal to $1$ or $2$.

Typically, these polar terms are computed using wall-crossing formulas and (rank 1) DT and PT invariants. This is done using wall-crossing behavior of an anti D6-brane giving another anti D6-brane and a D4-brane \cite{Alexandrov:2023ltz,Alexandrov:2023zjb}. The anti D6-branes correspond to PT invariants and they can be computed from the knowledge of GV invariants using the MNOP formula \cite{Pandharipande:2011jz}. The latter are obtained by the \textit{direct integration} method \cite{Huang:2006hq, Grimm:2007tm}, provided we fix a number of boundary conditions. Once this machinery is applied to fix a new generating function of rank 0 DT invariants, one can go in the opposite direction and obtain new GV invariants that can be used to fix new boundary conditions which in turn would allow to obtain an infinite number of GV invariants from direct integration. Unfortunately, one quickly hits a wall because the genus of the GV invariants needed to get the DT and PT invariants grows very quickly with the D4-brane charge thus halting any possible progress. An interesting project would be to find an independent way to compute polar terms of rank 0 DT invariants which does not require high genus GV invariants.

A promising avenue is given by the wall-crossing of one anti D6-brane and one D6-brane to a D4-brane. This method requires GV invariants of much lower genus. However, it requires some additional conditions that guarantee the existence of a chamber where the invariants vanish. Finding such a chamber can only be done in very few cases. But surprisingly, if one simply applies this method without worrying about said conditions, one finds that it works in many cases. There are however many cases where it fails and is thus not yet reliable. A possible remedy comes from \cite{VanHerck:2009ww} where a modification to this method, which gives the correct result, was argued on geometrical grounds. Understanding how this fix maps to the framework of wall-crossing is an open problem.

While the previous discussion focused on a summary and explicit ideas for the outlook, we would like now to touch on some abstract, far-reaching ideas. 
The non-perturbative effects studied in this thesis are a powerful tool to uncover potential universal properties of theories of quantum gravity. 
Moreover, their role in ensuring the consistency of various dualities shows that they are paramount to a potential non-perturbative definition of string theory.
Within these effects we see that geometry and topology are intimately related to combinatorics and algebra. This might be an indication that the former are not manifest in a complete formulation of quantum gravity.
Indeed, as black hole thermodynamics inspired the holographic principle, this may indicate that the usual process of quantizing spacetime will give way to another one where spacetime emerges as an effective description of a non-geometric theory.


\newcommand{\myscale}{1}
\newcommand{\xoffset}{-2cm}
\newcommand{\yoffset}{0mm}

\bibliography{combined}
\bibliographystyle{utphys}

\end{document}